\documentclass[ba,preprint]{imsart}
\usepackage{imsart}

\usepackage{graphicx} 
\usepackage{amsmath,amsfonts,amssymb,amsthm}
\RequirePackage[colorlinks,citecolor=blue,urlcolor=blue,backref=page,backref=page]{hyperref}

\usepackage{tabularx}
\usepackage{chngcntr}
\usepackage{subfigure}
\usepackage{mathtools}
\usepackage[authoryear]{natbib}
\usepackage{xcolor}
\usepackage{bbm}
\usepackage{overpic}
\usepackage{booktabs}
\usepackage{multirow}
\usepackage{adjustbox}
\usepackage{colortbl}
\usepackage[shortlabels]{enumitem}
\usepackage{algpseudocode}
\usepackage{algorithm}
\usepackage{algpseudocode}
\usepackage{ragged2e}
\usepackage{tikz}
\usetikzlibrary{cd} 
\usetikzlibrary{backgrounds}
\usetikzlibrary{fit,positioning,bayesnet,shapes.geometric,chains,matrix,scopes,calc,decorations.markings}
\tikzstyle{param}=[circle, minimum size = 0.7cm, thick, draw=black!100, fill = gray!10, node distance = 0.5cm]
\tikzstyle{data}=[rectangle, minimum size = 0.7cm, thick, draw =black!100, node distance = 0.5cm]
\tikzstyle{model}=[rectangle, minimum size = 1cm, thick, draw=black!100, node distance = 0.5cm]

\pubyear{2025}
\arxiv{2506.12532}
\volume{TBA}
\issue{TBA}
\firstpage{1}
\lastpage{35}

\startlocaldefs
\newtheorem{theorem}{Theorem} 
\newtheorem{corollary}{Corollary} 
 
\newtheorem{proposition}{Proposition} 

\newtheorem{assumption}{Assumption}

\DeclareMathOperator*{\argmin}{argmin}
\DeclareMathOperator*{\argmax}{arg\,max}

\def\x{{\mathbf{x}}}
\def\y{{\mathbf{y}}}
\def\z{{\mathbf{z}}}
\def\V{{\mathcal{V}}}
\def\R{{\mathbb{R}}}
\def\c{c} 
\def\pstar{p^*}
\def\deltaJ{\delta_J}

\newcommand{\one}{\mathbf{1}}

\newcommand{\E}{\mathbb{E}}

\newcommand{\Int}{\operatorname{Int}}

\endlocaldefs

\begin{document}
\setcounter{secnumdepth}{3}

\begin{frontmatter}
\title{Bayesian inference for the learning rate in Generalised Bayesian Inference}
\runtitle{Bayesian inference for the learning rate in GBI}

\begin{aug}
\author{\fnms{Jeong Eun} \snm{Lee}\thanksref{addr2}\ead[label=e2]{kate.lee@auckland.ac.nz}},
\author{\fnms{Sitong} \snm{Liu}\thanksref{addr1}\ead[label=e1]{sitong.liu@spc.ox.ac.uk}},
\author{\fnms{Geoff K.} \snm{Nicholls}\thanksref{addr1}\ead[label=e3]{nicholls@stats.ox.ac.uk}},

\address[addr1]{
Department of Statistics,
  University of Oxford,
  Oxford, UK.
}
\address[addr2]{
Department of Statistics,
  University of Auckland,
  Auckland, NZ.
}
\end{aug}

\begin{abstract}
In Generalised Bayesian Inference (GBI), the learning rate and hyperparameters of the loss must be estimated. These inference-hyperparameters can't be estimated jointly with the other parameters, from the data, by giving them a prior. 
However, in some settings there exist unknown ``true'' hyperparameter-values about which it is meaningful to have prior belief. It is then possible to use Bayesian inference with held-out data to get hyperparameter-posteriors. We define two hyperparameter posteriors, one based on an Expected Log Pointwise Predictive Density (ELPPD)-utility and one aiming to cover the pseudo-true parameter. 
The new framework supports estimation and uncertainty quantification for multiple hyperparameters jointly.
Experiments show that the resulting GBI-posteriors outperform Bayesian inference on simulated test data and select optimal or near-optimal hyperparameter values in a large real problem of text analysis. Generalised Bayesian inference is particularly useful for combining multiple data sets and most of our examples belong to that setting. We also give asymptotic results for some of the special ``multi-modular'' Generalised Bayes posteriors which we use in our examples.
\end{abstract}

\begin{keyword}[class=MSC]
\kwd[Primary ]{62C10}
\end{keyword}

\begin{keyword}
\kwd{Generalised Bayesian inference}
\kwd{Learning rate}
\kwd{Hyperparameter estimation}
\kwd{Posterior predictive}
\kwd{Multi-modular inference}
\end{keyword}

\end{frontmatter}

\section{Introduction}\label{sec:intro}

Generalised Bayesian Inference (GBI) with a Gibbs posterior \citep{Chernozhukov2003,Zhang2006,Jiang2008,Bissiri2016,Winter2023} gives a general framework for updating belief in parametric models which has proven useful in many settings. In GBI the usual parametric observation model $p(\x|\phi)$ for data $\x=(x_i)_{i=1}^{n},\ x_i\in \R^{d_x}$ given a parameter $\phi \in \R^{d_\phi}$ is replaced by a loss $\ell_\beta(\phi; \x)$ and the update from prior $\pi(\phi)$ to posterior becomes 
\begin{equation}\label{eq:GB-post-generic}
    \pi_s(\phi|\x)\propto \exp(-\eta\, \ell_\beta(\phi;\x))\,\pi(\phi),
\end{equation}
where $s=(\eta,\beta)$.
The \emph{learning-rate} $\eta>0$ controls the rate at which information about $\phi$ accumulates as data is gathered. The loss may also have hyperparameters $\beta$. We call $\eta$ and $\beta$ inference hyperparameters (as opposed to model hyperparameters).

The inference hyperparameters $\eta$ and $\beta$ must be chosen to support the goals of the inference. 
Bayesian inference for $s$, with an observation model 
$\exp(-\eta\,\ell_\beta)/c(\phi,s)$ normalised over $\x$, is not justified \citep{Bissiri2016, Miller2018a}. 
However, GBI with a separate loss for $s$ depending on held-out data is justified. The loss expresses a clear objective, to minimise the associated risk over an infinite population of test data points. 
In this paper we use held-out data to define a posterior for the inference hyperparameters. A Bayesian approach to this problem may be useful: we can pass uncertainty in the value of $s$ into a joint posterior in $s$ and $\phi$ and the prior applies shrinkage to the posterior mean $s$-estimate.

One area in which GBI has found use is in treating misspecification in models for multiple heterogeneous data sets. 
We illustrate Bayesian estimation of inference hyperparameters using a belief update tailored for multi-modular models called Semi-Modular-Inference (SMI). We show that the SMI-posterior is an instance of a Sequential-Gibbs posterior \citep{Winter2023}. This brings the theory established for Sequential-Gibbs posteriors to bear on SMI, alongside existing results \citep{frazier2023}.

Our main contribution is a Bayesian framework for learning
inference hyperparameters from calibration data. We give \emph{pooled calibration}, targeting parameter estimation, and \emph{product calibration}, targeting prediction,
and derive asymptotic results for interior and boundary optima. We introduce joint learning of $(\eta,\beta)$ in
SMI and identify SMI as a sequential-Gibbs posterior, allowing existing Bernstein--von Mises theory to be applied.

The paper is structured as follows. Section \ref{sec:bayes-for-s-main-idea} sets out the framework for joint inference of hyperparameters and derives the asymptotic behaviour of their posterior distribution under pooled and product losses. Section \ref{sec:SMI} presents the multi-modular inference setting and Section \ref{subsec:seqGibbs} gives asymptotic properties of SMI using sequential-Gibbs results. Section~\ref{sec:simulation-examples} begins by defining $s$-estimators and comparison criteria. Synthetic normal-mixture and state-space model examples (Sections \ref{sim:normal} and \ref{sec:ssm}) illustrate earlier theory, and we end with an application to a real sense-tagging problem in Section \ref{sec:edisc}. 


\section{Bayesian Inference for Inference Hyperparameters}\label{sec:bayes-for-s-main-idea}

When we define a loss $\ell_\beta(\phi;\x)$ for use in GBI, we often keep the parametric link 
between data and parameter but adjust the belief update to allow for misspecification (see Table~1 in \cite{Knoblauch22}). 
We assume an observation model
\begin{equation}\label{eq:lkd-product}\textstyle
p(\x|\phi)=\prod_{i=1}^n p(x_i|\phi)
\end{equation}
for conditionally independent data vectors $x_i\in\R^{d_x}$ has been given. Let $x\in \R^{d_x}$ be a generic component of $\x$ and denote by $\pstar(x)$ the true generative model for the data. 
One way to define the loss is to penalise departures from $\pstar$ with a divergence $D(\pstar(\cdot)||p(\cdot|\phi))$. \cite{Jewson2022} consider factors involved in choosing $D$ and \cite{Knoblauch22} discuss these in a wider class of generalisations of Bayesian inference. 

We seek a family of divergences in which $\pstar$ doesn't need to be evaluated, only simulated. \cite{Jewson2018} achieve this using a Bregman divergence,
\[
D_f(\pstar||p) =  \int f(\pstar(x))dx-\int f(p(x|\phi))dx - \int f'(p(x|\phi))(\pstar(x)-p(x|\phi))dx,
 \]
where $f:[0,\infty)\to \R$ is a strictly-convex continuously-differentiable function and
$f'$ is its derivative. The general ``Bregman'' loss function $\ell_{f}(\phi;\x)$, which depends on the parameterisation of $f$, is an unbiased estimate of $D_f$. Up to terms not depending on $\phi$,
\begin{equation}\label{loss:f}
\ell_{f}(\phi;\x) = n\mathbb{E}_{x|\phi}(f'(p(x|\phi))) -  n\int f(p(x|\phi))dx -\sum_{i=1}^{n}f'(p(x_i|\phi))\,. 
\end{equation}
When $f(p)=p\log(p)-p$ the loss is the negative log-likelihood and \eqref{eq:GB-post-generic} is the {\it power posterior}, $\pi_\eta(\phi|\x)\propto p(\x|\phi)^\eta \pi(\phi)$. If $f(p;\beta)=(p^\beta - 1)/\beta(\beta-1)$ then $\ell_\beta$ in \eqref{loss:f} is the $\beta$-loss \citep{Basu1998,Minami2002,stummer2012breg},
\begin{equation}\label{eq: beta loss}
    \ell_\beta(\phi;\x) = -\frac{1}{\beta-1} \sum_{i=1}^n p(x_i|\phi)^{\beta-1} + \frac{n}{\beta}\int p(x|\phi)^\beta dx.
\end{equation}
If $\ell_\beta(\phi;\x)$ is the $\beta$-loss then the posterior in \eqref{eq:GB-post-generic} is the {\it $\beta$-loss posterior}. The integral in \eqref{eq: beta loss} is evaluated by quadrature or simulation. Inference with the $\beta$-loss posterior is robust to outliers \citep{Knoblauch22,jewson24}. It interpolates the \cite{Itakura1968} loss ($\beta=0$), the negative log-likelihood ($\beta=1$) and the L2-norm divergence ($\beta=2$). We illustrate GBI for $s=(\eta,\beta)$ in the $\beta$-loss posterior.

We now consider the choice of $s$. A good value of $s$ depends on the goals of the inference, so we define one risk function, $l(s;\x)$ below, for parameter estimation and another, $\tilde l(s;\x)$, for prediction. Risk is defined for a population of unobserved test data $z\sim \pstar,\ z\in \R^{d_x}$. Suppose for fixed $\x$ and prior $\pi$ there is a unique $\widetilde{s}$ minimising the chosen population-level risk. Here $\widetilde{s}$ controls the degree of learning and robustness adjustment induced by the misspecified observation model. 
Our loss-functions are empirical estimates of the respective risks. Suppose that we have in addition to $\x$ a \emph{calibration} data set, $\y=(y_{j})^J_{j=1}$ which $y_j\sim \pstar$, $y_j\in \mathbb{R}^{d_x}$. We define \emph{pooled} and \emph{product} losses, respectively $l_{(1,J)}(s;\y,\x)$ and $l_{(J,1)}(s;\y,\x)$, measuring how well the GB-posterior performs for $\phi$-estimation and $z$-prediction respectively. We think of $\y$ as a sample of test data which we use to adjust the belief update to minimise the risk over the population. In practice we get $\y$ by splitting the data into training and calibration blocks.

{\bf Pooled loss}: Let $\widetilde{\phi}=\arg\max_\phi E_{z\sim \pstar}(\log(p(z|\phi))$ be the pseudo-true parameter. Our risk for estimation, 
\begin{equation*}
{l}(s;\x)=-\log(\pi_s(\widetilde{\phi}|\x))\,,
\end{equation*} 
favors $s$ such that $\pi_s(\phi|\x)$ covers the pseudo-true parameter. Let
\begin{equation}\label{eq:pooled_loss}
l_{(1,J)}(s;\y,\x)=-\log(p_{s}(\y|\x))
\end{equation}
be the loss for estimation. The joint posterior predictive in \eqref{eq:pooled_loss},
\begin{equation*}\label{eq:pp-joint-pooled}
p_{s}(\y|\x)=\int 
p(\y|\phi)\,\pi_s(\phi|\x)\,d\phi. 
\end{equation*}
is the intrinsic marginal likelihood of \cite{Berger1996}. 
Treating $l_{(1,J)}(s;\y,\x)$ as an empirical risk makes sense: under regular observation models and smooth $\phi$-priors, a rescaled-$p_{s}(\y|\x)$ converges to  $\pi_s(\widetilde{\phi}|\x)$ at large $J$ (see Corollary~\ref{cor:posterior_pool}). With a prior $\rho(s)$ and hyperparameter space $s\in\Omega_s$, we get the pooled posterior, 
\begin{equation}\label{eq:poolpost}
    \rho_{(1,J)}(s|\y;\x) \propto \rho(s) p_s(\y|\x) \,, s\in\Omega_{s}.
\end{equation}
The value of $\widetilde s (\pstar,\pi,p,\x)$ minimising ${l}(s;\x)$ will depend on the true generative model $\pstar(z)$, the prior $\pi(\phi)$, the likelihood $p(z|\phi)$ and the realised data $\x$; all these elements are fixed and the role of $\widetilde s$ is to work with these fixed quantities and adapt the posterior $\pi_s(\phi|\x)$ to better cover the pseudo-true parameter. 

{\bf Product loss}: Our risk for prediction is the (negative) Expected Log Pointwise Predictive Density (ELPPD, \cite{Vehtari2016}),
\begin{equation}\label{eq:gb-risk}
\tilde{l}(s;\x)=-\mathbb{E}_{z\sim \pstar}(\log(p_s(z|\x))\,,
\end{equation} 
where $p_s(z|\x)$ is the posterior predictive for one observation. The hyperparameter $\widetilde s(\pstar,\pi,p,\x)$ minimising $\tilde{l}(s;\x)$ adapts $p_s(z|\x)$ to cover test data. The empirical risk,
\begin{equation} \label{eq:gb-product-loss}
    l_{(J,1)}(s;\y,\x)=-\sum_{j=1}^J \log(p_{s}(y_{j}|\x))\,,
\end{equation}
converges (when rescaled by $J$) to the risk $\tilde{l}(s;\x)$ as $J\!\to\!\infty$ and is our loss for prediction. We take a prior $\rho(s)$ and update belief using the product posterior,
\begin{equation}\label{eq:GB-hyper-posterior}
    \rho_{(J,1)}(s|\y;\x)\propto \rho(s)\prod_{j=1}^J p_{s}(y_{j}|\x)\,,\ s\in\Omega_{s}.
\end{equation}



The hyperparameter domain $\Omega_{s}$ is often low-dimensional. In the power posterior (where $\beta=1$) we have $s=\eta$ and $\Omega_s=(0,\infty)$ and in the $\beta$-loss posterior with $s=(\eta,\beta)$ we have
$\Omega_s=(0,\infty)^2$. In SMI, $0\le \beta <\infty$ but $0\le \eta\le 1$, as that interpolates between a posterior conditioning on two data sets $\x_1,\x_2$ and a posterior conditioning only on $\x_1$. This is discussed in Section~\ref{sec:SMI}. 
When $\Omega_s$ has dimension just one or two, $p_s(\y|\x)$ in \eqref{eq:poolpost} and $p_s(y_{j}|\x)$ in \eqref{eq:GB-hyper-posterior} can be estimated, using MCMC samples 
$\phi^{(t)}\sim \pi_s(\cdot|\x),\ t=1,\dots,T$, at points on a lattice of $s$ values and all further calculations carried out using smoothing and quadrature. This is what we do below. For higher dimensional $\Omega_s$, evaluation on a lattice isn't feasible, but we can target the joint posteriors $\rho_{(1,J)}(s,\phi|\y;\x)\propto \rho(s)p(\y|\phi)\pi_s(\phi|\x)$ and $\rho_{(J,1)}(s,\phi_{1},\dots,\phi_J|\y;\x)$ using the nested MCMC algorithm given in Appendix~\ref{Appendix:mcmc}. 

\subsection{Simple example illustrating hyperparameter posteriors}\label{sec:simple-first-example}

We illustrate our scheme for hyperparameter estimation in a simple normal example. See Appendix~\ref{app:simple-first-example} for calculations.
Suppose the true data distribution is
$\pstar(z) = N(z;\mu^*,1)$ so
$x_i,y_j \sim \pstar$ iid for $i=1,\dots,n$ and $j=1,\dots,J$.
We fit the observation model
$
x_i \mid \phi \sim N(\phi,\sigma^2)
$
with fixed $\sigma^2$ and prior $\phi \sim N(0,1)$. The pseudo-true parameter value is $\widetilde{\phi}=\mu^*$. The loss in \eqref{eq:GB-post-generic} is the negative log-likelihood so $\pi_\eta(\phi|\x)\propto p(\x|\phi)^\eta\pi(\phi)$ which is \(
\pi_\eta(\phi|\x)=
N\left(\phi;
\frac{n\eta}{\sigma^2+n\eta}\bar x,
\frac{\sigma^2}{\sigma^2+n\eta}
\right) \) where $\bar x$ is the mean of $\x$. 
This is actually just standard Bayesian inference with a rescaled likelihood,
but it provides a simple setting in which to illustrate inference
for the learning rate. A genuine GBI non-Bayesian normal/inverse-Gamma
example is given in Appendix~\ref{app:gbi-normal-example}.

We reparameterize $\kappa={n\eta}/(\sigma^2+n\eta)\in(0,1)$ so that $\pi_\kappa(\phi|\x)=N\{\phi;\kappa\bar x,1-\kappa\}$. Varying
$\eta>0$ is equivalent to varying $\kappa\in(0,1)$.
The pooled loss for $\kappa$ is
\begin{align}
-\log(p_\kappa(\y|\x))
=
\frac12 \log\!\big(\sigma^2+J(1-\kappa)\big)
+
\frac{J(\bar y-\kappa\bar x)^2}{2(\sigma^2+J(1-\kappa))}+c_J,
\label{eq:first-example-pool-loss}
\end{align}
up to a function $c_J$ of $J$ which is independent of $\kappa$.
The pooled risk is
\begin{align}
-\log \pi_\kappa(\mu^*\mid \x)
=
\frac12\log(1-\kappa)
+
\frac{(\mu^*-\kappa\bar x)^2}{2(1-\kappa)} + c.
\label{eq:first-example-pool-risk}
\end{align}
At large $J$, $\bar y\to \mu^*$ and the scaled posterior predictive $J^{1/2}p_\kappa(\y|\x)$ has the same $\kappa$-dependence as $\pi_\kappa(\mu^*|\x)$.
The pooled posterior for $\kappa$ is
\begin{align}
\rho_{(1,J)}(\kappa|\y;\x)
\propto
N\left(
\bar y;\kappa\bar x,1-\kappa+{\sigma^2}/{J}
\right)\rho(\kappa).
\label{eq:first-example-pool-post}
\end{align} 
The product loss in \eqref{eq:gb-product-loss} converges to its risk in \eqref{eq:gb-risk} by the weak law. Finally, the posterior predictive density for one observation is $p_\kappa(y|\x)=N\!\left(y;\kappa\bar x,\sigma^2+1-\kappa\right)$ so, 
\begin{align*}
\rho_{(J,1)}(\kappa\mid \y;\x)
\propto
\rho(\kappa)
\prod_{j=1}^J N\!\left(y_j;\kappa\bar x,\sigma^2+1-\kappa\right).
\end{align*} 
The two constructions use the same training and calibration data,
but combine the calibration observations in different ways, so their large-$J$ behaviour is different. We return to this in Section~\ref{blocksize}, after establishing the corresponding asymptotic results. 

\subsection{Related literature}

 We categorise methods for conducting Bayesian analysis on misspecified models into three groups. The first group comprises methods that apply bootstrapping within a Bayesian framework, such as Weighted Likelihood Bootstrap \citep{Newton1991,Newton1994,Lyddon2019} and BayesBag \citep{Buhlmann2014,Huggins2021}. The second group includes the tempered likelihood, resulting in power posteriors \cite{Walker2001,Grunwald2012,Miller2018a}. The third group generalises the second and involves replacing the likelihood with a more general loss function that mediates data and parameters, including Gibbs posteriors \citep{Zhang2006, Jiang2008, Martin2022}, Generalised Bayes \citep{Bissiri2016, Grunwald2017} and approaches that take the PAC-Bayes loss \citep{Germain2016,Zhang2006,McAllester1998,Shawe1997, Knoblauch22} as their starting point.  We focus on the power posterior, but also consider more general losses which are empirical estimates of the Bregman divergence, so our work falls in the second and third groups.

Several different methods for estimating the learning rate $\eta$ have been given, reflecting different goals. See \cite{Wu2020calib-comparison} for an overview.
Popular methods include matching the frequentist coverage probability \citep{Syring2018, Winter2023}, matching the Hyv\"arinen score \citep{Yonekura2023} and matching the asymptotic covariance using the Fisher information \citep{Lyddon2019, Holmes2017, frazier2023}. 
One straightforward approach is to choose $\eta$ to maximise a utility evaluated on held-out calibration data. Methods based on Leave-One-Out Cross-Validation (LOOCV) estimates of the ELPPD \citep{carmona20} and SafeBayes (\cite{Grunwald2017}, who use sequential prediction) are related. \cite{battaglia2025} minimises the Posterior Mean Squared Error for held-out class labels in a supervised classification task. \cite{Carmona2022scalable} give variational methods for scalable inference when the dimension of $s$ is large, using stochastic gradient descent on the Watanabe-Akaike Information Criterion (WAIC, \cite{watanabe10,Vehtari2016}). However, there is to date no joint estimation of the learning rate $\eta$ with loss hyperparameters, such as $\beta$ in the $\beta$-loss. 

\cite{shenCarvalho2024} estimate $\eta$ in a ``power prior'' setting \citep{carvalhoIbrahim2021} with many of the same elements but different goals. We have training data $\x$, calibration data $\y$ and test data $\z$. These all have some unknown true distribution $\pstar$ and we use the same observation model $p(\cdot|\phi)$ with a common parameter $\phi$. 
In \cite{shenCarvalho2024}, $\x\sim p(\cdot|\phi_x)$ is a well-specified model for the ``old'' data set which they use to make a prior for the parameters of a ``new'' data set $\y\sim p(\cdot|\phi_y)$. For example, $\x$ and $\y$ measure a treatment response in two drug trials involving adults and children respectively and $\phi_x$ and $\phi_y$ are model parameters for the two data sets. Misspecification arises if $\phi_x$ differs significantly from $\phi_y$. As this is a concern, the authors construct a prior for $\phi_y$ by down-weighting the information from the old data, taking
\[
\pi(\phi_y,\eta|\x,\y)\propto \rho(\eta) p(\y|\phi_y)\pi_\eta(\phi_y|\x).
\]
This joint posterior is the joint pooled posterior $\rho_{(1,J)}(s,\phi|\y;\x)$. The authors call $\pi_\eta(\phi_y|\x)$ the ``power prior'' for $\phi_y$. 
\cite{shenCarvalho2024} consider the asymptotic behaviour of $\eta$ as $n$ and $J$ grow at fixed $n/J$ for GLMs with possibly different parameters. They show $\eta\to 0$ if the true parameters differ.

Recently, \cite{McLatchie2024} 
showed that for bounded $\eta$, the power-posterior predictive $p_\eta(y|\x)$ converges to the MLE plug-in predictive as the $\x$-sample size $n\to\infty$. 
Predictive performance is insensitive to $\eta$ in this regime, so the product loss in \eqref{eq:gb-product-loss} doesn't inform $\eta$. 
Our setting differs from theirs in two main ways. 
First, we consider the case where the number of calibration points $\y=(y_1,\dots,y_J)$ goes to infinity for a fixed training sample $\x=(x_1,\dots,x_n)$; in \cite{McLatchie2024} it is the other way round.
In the product posterior the risk is the ELPPD, and the ``likelihood'' for $s$ is $\prod_j p_s(y_j|\x)$
so information about $s$ accumulates when we fix $n$ 
and take $J\to \infty$. 
Second, in the problems we consider the number of parameters grows with the number of observations ($\theta_M$ in Section~\ref{sec:ssm} and $c_1,\dots,c_n$ in \ref{sec:edisc}), and so the large-$n$ regime is not so relevant. 

\subsection{Asymptotic distribution of hyperparameters}\label{sec:bayes_s}

The following translates results for the asymptotics of Generalised Bayes posteriors from \cite{Millar21} to our setting (Section \ref{asymptotic_regular}), develops some intuition on the different behaviors (Section \ref{blocksize}) and gives asymptotics for non-regular models (Section \ref{asymptotic_boundary}), using results from \cite{Bochkina2014}.

For simplicity in notation, partial derivatives evaluated at particular values are simplified by omitting subscripts. For example, $\nabla_\varphi f(\varphi,\x)|_{\varphi=\bar{\varphi}} := \nabla_\varphi f(\bar{\varphi},\x)$. A probability density function and distribution are denoted by $p$ and $P$ respectively and $N(\mu,\Sigma)$ is the multivariate normal distribution with mean $\mu$ and covariance matrix $\Sigma$.

\subsubsection{Regular asymptotics}\label{asymptotic_regular}

The posterior for the hyperparameters $s$ may concentrate as we gather more calibration data $\y=(y_1,\dots,y_J)$. How this happens depends on whether we have a pooled loss (Theorem \ref{thm:loss_s_single}) or a product loss (Theorem \ref{thm:product}). We begin with the pooled loss. 
All probability limits are with
respect to $J$ and the calibration data unless stated otherwise.

\begin{assumption}[Pooled model]
\label{A:laplace} 
Let $\y=(y_1,\ldots,y_J)$ be iid samples from $\pstar$. 
Let
$K_s\subset\Int(\Omega_s)$ be a compact set, which may depend
on fixed data $\x$ but not on $J$ or $y$.
\begin{enumerate}[itemsep=2pt,topsep=2pt]
 \item[(i)] For every $s\in \Int(\Omega_s)$, both
 $\pi_s(\phi|\x)$ and $\log p(\y|\phi)$ are $C^\infty$ in $\phi$.
 For every compact $N_\phi\subset \Int(\Omega_\phi)$,
 $(s,\phi)\mapsto\pi_s(\phi|\x)$ is jointly continuous on
 $K_s\times N_\phi$.
 \item[(ii)] With probability tending to one, there is a unique
 calibration-sample MLE, \\
 \(
  \bar\phi_J\in\argmax_{\phi\in \Int(\Omega_\phi)}
  \sum_{j=1}^J\log p(y_j|\phi).\\[-2ex]
 \)
 \item[(iii)] At $\bar\phi_J$, the matrix $H_{\bar\phi_J}
  :=-\frac1J\nabla_\phi^2
  \sum_{j=1}^J\log p(y_j|\phi)
  \big|_{\phi=\bar\phi_J}$ is finite and positive definite with probability tending to one.

 \item[(iv)] The expected log likelihood has a unique maximiser
 \(
  \widetilde\phi
  =\argmax_{\phi\in \Int(\Omega_\phi)}
  \E_{Z\sim \pstar}\{\log p(Z|\phi)\},
 \) and the calibration-sample MLE is consistent so
 $
  \bar\phi_J\xrightarrow{p}\widetilde\phi.
 $

 \item[(v)] The Laplace approximation $p_s(\y| \x)=\int p(\y|\phi) \pi_s(\phi|\x) d\phi\approx C_J(\y) \pi_s(\bar{\phi}_J|\x)$ with
 \[
  C_J(\y)
  :=\frac{(2\pi)^{d_\phi/2}
  \exp\{\sum_{j=1}^J\log p(y_j|\bar\phi_J)\}}
  {J^{d_\phi/2}|H_{\bar\phi_J}|^{1/2}},
 \]
holds in the sense that, uniformly in $K_s$, as $J\to\infty$,
 \[
  \sup_{s\in K_s}
  \left|
  p_s(\y|\x)/C_J(\y)
  -\pi_s(\bar\phi_J| \x)
  \right|\xrightarrow{p}0.
  \label{eq:laplace-defn-A1}
 \]
\end{enumerate}
\end{assumption}

Assumptions~\ref{A:laplace}(i-iv) hold for exponential family likelihoods under additional regularity conditions. This may be shown by applying
\citet[Propositions~3.1 \& 3.2 and Theorem~3.3,
Section~3.5]{wainwright2008graphical}. See Appendix~\ref{app:ass1-reg-exp-fam} for a proof. 
\begin{proposition}[Regular exponential families]
\label{prop:ass1-reg-exp-fam}
Suppose the observation model is
a full, regular, minimal exponential family in natural-parameter form with base measure $\nu$,
\[
p(y|\phi)
=
\exp\{\phi^\top T(y)-A(\phi)\},
\qquad
\phi\in\Omega_\phi,
\]
where $\Omega_\phi$ is open. If
$\E_{Y\sim \pstar}\{\|T(Y)\|\}<\infty$ and $\pstar$ has a strictly positive
density with respect to $\nu$
then Assumptions~\ref{A:laplace}(ii - iv) hold and the likelihood part of (i). 
\end{proposition}

Assumption~\ref{A:laplace}(v) assumes $\bar\phi_J$ exists uniquely and $H_{\bar\phi_J}$ is finite and positive definite; the probability for this tends to one by Assumptions~\ref{A:laplace}(ii) and (iii). It is the uniform form of the Laplace approximation needed
when an optimiser over $s$ is being studied; see
\citet{KassTierneyKadane1990} for the background Laplace expansion at fixed $s$ (noting that the likelihood in that paper is our $p(\y|\phi)$ and their prior is our posterior $\pi_s(\phi|\x)$, as we are approximating the posterior predictive $p_s(\y|\x)$ at large $J$, not the marginal likelihood $p(\x)$ at large $n$). For a regular likelihood and compact $K_s$, Assumption~\ref{A:laplace}(v) holds under the usual Laplace conditions imposed
uniformly over $s\in K_s$. We illustrate this in Proposition~\ref{prop:ef-uniform-laplace} in Appendix~\ref{app:ass1-reg-exp-fam} which gives explicit sufficient
conditions under which Proposition~\ref{prop:ass1-reg-exp-fam} may be extended to cover Assumption~\ref{A:laplace}(v).

The following corollary shows that the $s$-posterior for the pooled loss doesn't concentrate with $J$ but has a limiting distribution which is only as concentrated as $\pi_s(\widetilde\phi|\x)\rho(s)$.  

\begin{corollary}\label{cor:posterior_pool} 
If Assumption~\ref{A:laplace} holds, the prior $\rho(s)$ is supported on $K_s$ and \(D_\infty=\int_{K_s}\pi_s(\widetilde\phi|\x)\rho(s)\,ds
>0\) then at large $J$ and in total variation, the pooled-posterior 
is
\def\mystrut{\rule[-0.2\baselineskip]{0pt}{\baselineskip}}
\begin{align}
\rho_{(1,J)}(s;\y,\x)\approx \dfrac{\pi_s(\bar{\phi}_J|\x)\rho(s)}{\mystrut\int \pi_s(\bar{\phi}_J|\x)\rho(s) ds}.\label{eq:loss_laplace}
\end{align}
\end{corollary}
\begin{proof} Assumptions (i) and (iv) give
$\int_{K_s}\pi_s(\bar\phi_J|\x)\rho(s)\,ds
\xrightarrow{p} D_\infty>0$. Normalizing $\sup_{s\in K_s}
  \left|
  p_s(\y|\x)/C_J(\y)
  -\pi_s(\bar\phi_J|\x)
  \right|$ in Assumption (v) preserves the approximation and gives the stated approximation. \end{proof}

\begin{assumption}[Pooled minimisers]
\label{A:log_min_pooled}
\phantom{ }\\[-0.1in]
\begin{enumerate}[itemsep=0pt,topsep=-3pt]
 \item[(i)] The pooled risk minimiser
$ \widetilde{s}_{(1,\infty),\,\x}
 =\argmax_{s\in\Omega_s}\pi_s(\widetilde\phi|\x)
 $ 
is unique, lies in the interior of the compact set $K_s$ in
Assumption~\ref{A:laplace} and satisfies
$\pi_{\widetilde{s}_{(1,\infty),\,\x}}(\widetilde\phi|\x)>0$.
\item[(ii)] A global pooled-loss minimiser
$\bar s_{(1,J),\,\x}
 =\argmin_{s\in\Omega_s}l_{(1,J)}(s;\y,\x)$
exists for all sufficiently large $J$. Moreover,
$\Pr\{\bar s_{(1,J),\,\x}\in K_s\}\longrightarrow 1$.
\end{enumerate}
\end{assumption}
We need $K_s$ to contain $\widetilde{s}_{(1,\infty),\,\x}$ and $\bar s_{(1,J),\,\x}\in K_s$ eventually. Apart from that, $K_s$ is just any large compact subset of $\Int(\Omega_s)$. 
This is enough for the following consistency result which we prove in Appendix~\ref{app:loss_s_single-proof}.

\begin{theorem}[Pooled consistency]\label{thm:loss_s_single}
Suppose Assumptions~\ref{A:laplace} and \ref{A:log_min_pooled} hold. Then
\[
 \bar s_{(1,J),\,\x}\xrightarrow{p}\widetilde{s}_{(1,\infty),\,\x}
 \qquad\text{as }J\longrightarrow\infty.
\]
\end{theorem}

We now turn to the product posterior. We take $J\to \infty$ and show that the posterior for $s$ concentrates. 
We map our setup onto the convergence theory given for GBI in \cite{Millar21}. 
We work with the scaled loss $J^{-1}l_{(J,1)}(s;\y,\x)$ and set out the assumptions and theorem from \cite{Millar21} separately so we can discuss them below.
\begin{assumption}[Product regularity]
\label{A:product-regular}
Let \(\Omega_s\subseteq\mathbb{R}^{d_s}\) be open. There is an open
bounded set \(E_s\subseteq\Omega_s\), which may depend on the fixed
training data \(\x\) but not on \(J\) or \(\y\), such that
\(\widetilde s_{(\infty,1),\,\x}\in E_s\). We assume that, with probability
one, the following hold.\\[-0.2in]
\begin{enumerate}[itemsep=0pt]
\item[(i)] At every \(s\in\Omega_s\),
$
J^{-1}l_{(J,1)}(s;\y,\x)
\longrightarrow
\widetilde l(s,\x)
$.
\item[(ii)] The risk $\widetilde l(s,\x)$ is twice differentiable at
\(\widetilde s_{(\infty,1),\,\x}\),
$
\nabla_s\widetilde l\bigl(\widetilde s_{(\infty,1),\,\x},\x\bigr)=0
$, and the Hessian \(
H_s
:=
\nabla_s^2
\widetilde l\bigl(\widetilde s_{(\infty,1),\,\x},\x\bigr)
\)
is positive definite.

\item[(iii)] The loss 
$
l_{(J,1)}(s;\y,\x)
$
has continuous third derivatives on \(E_s\) at each $J$, and
\[
\sup_J\sup_{s\in E_s}\sup_{i,j,k} \frac{1}{J}
\left|
\nabla^3_{(s_i,s_j,s_k)}
l_{(J,1)}(s;\y,\x)
\right|
<\infty.
\]
\end{enumerate}
\end{assumption}

\begin{assumption}[Product separation]
\label{A:product-separation}
There is a compact set \(K_s\subseteq E_s\) such that
$
\widetilde s_{(\infty,1),\,\x}
\in\operatorname{Int}(K_s)
$, with
$
\widetilde l(s,\x)
>
\widetilde l\bigl(\widetilde s_{(\infty,1),\,\x},\x\bigr)$ for $s\in
K_s\setminus\{\widetilde s_{(\infty,1),\,\x}\}
$ and
\[
\liminf_{J\to\infty}
\inf_{s\in\Omega_s\setminus K_s}
\left\{
J^{-1}l_{(J,1)}(s;\y,\x)
-
\widetilde l\bigl(\widetilde s_{(\infty,1),\,\x},\x\bigr)
\right\}
>0 \qquad \mbox{almost surely.}
\]
The set $K_s$ may depend on the
fixed training data \(\x\) but not on \(J\) or \(\y\). 
\end{assumption}

\begin{theorem}[Product concentration]
\label{thm:product}
Suppose Assumptions~\ref{A:product-regular} and
\ref{A:product-separation} hold. If \(\rho(s)\) is a probability density
on \(\Omega_s\) that is continuous and strictly positive at
\(\widetilde s_{(\infty,1),\,\x}\), then, almost surely, for all sufficiently
large \(J\), the loss minimiser
\[\bar s_{(J,1),\,\x}
=
\argmin_{s\in\Omega_s}
l_{(J,1)}(s;\y,\x)\]
is unique and
$
\bar s_{(J,1),\,\x}
\xrightarrow{p} 
\widetilde s_{(\infty,1),\,\x}
$.
If $s\sim\rho_{(J,1)}(s|\y;\x)$ defined in (\ref{eq:GB-hyper-posterior}),
\[
\sqrt{J}
\left(
s-\bar s_{(J,1),\,\x}
\right)
\xrightarrow{\mathrm{t.v.}}
N\left(0,H_s^{-1}\right).
\]
\end{theorem}
The proof in Appendix~\ref{app:loss_s_product_proof} maps our notation onto Theorem 5 in \cite{Millar21}. 
We have not verified Assumptions~\ref{A:product-regular} and \ref{A:product-separation} on any large class of models. However, the derivative and convergence conditions would hold for our experiments. For example, in Appendix~\ref{Appendix:ssm-finite-eta-star} we prove that $\widetilde{\eta}_{(\infty,1),\,\x}<\infty$ in the State Space Model example in Section~\ref{sec:ssm}, and the concentration and asymptotic normality anticipated in Theorem~\ref{thm:product} appear to hold in many of our experiments. The easiest way to break Assumptions~\ref{A:product-regular}(ii) and \ref{A:product-separation} is to arrange things so that $\widetilde s_{(\infty,1),\,\x}$ is on the boundary of $\Omega_s$. We illustrate this with simple examples in Section~\ref{blocksize} and Appendix~\ref{app:gbi-normal-example-product} where we show that $\widetilde{\eta}_{(\infty,1),\,\x}=\infty$ can occur. If $\widehat\phi_\x$ is the MLE for the training-data likelihood $p(\x|\phi)$ and the calibration likelihood $p(\y|\phi)$ happens to be maximised at $\phi$ near $\widehat\phi_\x$, then sending $\eta\to\infty$ makes the posterior $\pi_\eta(\phi|\x)$ concentrate on $\phi$-values closer to $\widehat\phi_\x$ so the predictive $p_\eta(\y|\x)$ increases with $\eta$. This can't happen so easily for the pooled loss: if the ordinary Laplace approximation to the power posterior holds and $\widehat\phi_\x\ne \widetilde\phi$ then sending $\eta\to\infty$ sends $\pi_\eta(\phi|\x)$ to a Delta-function at $\widehat\phi_\x$ and $\pi_\eta(\widetilde \phi|\x)$ must eventually decrease.


\subsubsection{Pooled loss vs Product loss}\label{blocksize}

The product and pooled losses correspond to different inferential targets: pooled calibration moves the parameter-posterior onto the pseudo-true parameter; product calibration moves the predictive distribution onto the data-generating distribution. 
The product loss for $s$ is an additive loss; information about $s$ accumulates with increasing $J$ and, under Assumptions~\ref{A:product-regular} and \ref{A:product-separation}, the posterior concentrates at rate $1/\sqrt{J}$. 
The pooled loss targets the marginal likelihood of an intrinsic Bayes factor \citep{Berger1996}, where a data-dependent prior is formed from the training data and inference proceeds via the marginal likelihood of the remaining data. This provides a Bayesian, evidence-based justification for using the pooled loss. 
Although both estimators have well-defined inferential objectives, $\tilde{s}_{\infty,1}$ and $\tilde{s}_{1,\infty}$ are not equal in general.
We compare the two estimators using a simple normal-inverse gamma/power posterior example in Appendix~\ref{app:gbi-normal-example}. The remarks at the end of Section~\ref{asymptotic_regular} explain why we expect Assumptions~\ref{A:laplace} and \ref{A:log_min_pooled} to hold quite generally. In practice, although Assumptions~\ref{A:product-regular} and \ref{A:product-separation} are fragile, the conclusion in Theorem~\ref{thm:product} is robust across our main examples.

{\bf Normal example:} We continue the normal example of Section~\ref{sec:simple-first-example}. See Appendix~\ref{app:simple-first-example} for further details. For $\kappa\in (0,1)$, the pooled posterior $\rho_{(1,J)}(\kappa|\y;\x)$ in \eqref{eq:first-example-pool-post}
converges to $\rho_{(1,\infty)}(\kappa|\x)\propto 
N\left(
\mu^\star;\kappa\bar x,1-\kappa
\right)\rho(\kappa)$ so the limiting posterior for $\kappa$ is not a delta-function. In the product posterior, the scaled loss $L_J=-J^{-1}\sum_j\log p_\kappa(y_j|\x)$ converges to the risk $\widetilde{l}(\kappa,\x)=\E_{z\sim \pstar}(-\log(p_\kappa(z|\x))$ so
 $\rho_{(J,1)}(\kappa|\y;\x)
\propto \exp(-JL_J(\kappa))$
converges to $\rho_{(\infty,1)}(\cdot|\x)=\delta_{\widetilde\kappa_{(\infty,1),\,\x}}(\cdot)$, with $\widetilde\kappa_{(\infty,1),\,\x}=\argmin \widetilde{l}(\kappa,\x)$ playing the role of $\widetilde{s}(\pstar,\pi,p,\x)$ above.
The takeaway here is that the pooled \(\kappa\)-posterior converges to a density which is diffuse in $\kappa$, whereas the product \(\kappa\)-posterior
concentrates at the product-risk minimiser. 

The behavior of the risk functions is revealing. Consider optima at boundary limits of $\kappa\in (0,1)$. As $\kappa\to 1$ (i.e., $\eta\to\infty$), $\pi_\kappa(\phi|\x)$ concentrates at the $\x$-data MLE $\hat\phi_x=\bar x$ so the limiting pooled risk $-\log \pi_\kappa(\widetilde\phi|\x)\to\infty$ unless $\hat\phi_x=\widetilde\phi$, or $\bar x= \mu^*$ here. That means 
$\widetilde\kappa_{(1,\infty),\,\x}<1$ and hence $\widetilde\eta_{(1,\infty),\,\x}$ is almost surely finite; this holds more widely, not just in this example. 

By contrast, the $\kappa$-dependence of the product risk,
\[
\E_{z\sim \pstar}(-\log(p_\kappa(z|\x))=\operatorname{KL}\!\left\{
p^\star(\cdot)
\,\middle\|\,
p_\kappa(\cdot\mid\x)
\right\}-\E_{z\sim \pstar}(\log(\pstar(z))
\] is controlled by the KL-divergence between $N(\mu^*,1)$ and $N\!\left(y;\kappa\bar x,\sigma^2+1-\kappa\right)$. If the fitted model is overdispersed, so $\sigma^2>1$, and $\bar x$ happens to be close to $\mu^*$ then taking $\kappa\to 1$ tends to reduce the divergence, which means that in the product case we may see $\widetilde\eta_{(\infty,1),\,\x}=\infty$ (we give the calculation in Appendix~\ref{app:simple-first-example}). If we believe in the ELPPD as a utility, we interpret this as meaning we should discard the prior and predict $y$-values by plugging in the MLE for the $\x$-data using $y\sim p(\cdot|\hat\phi_x)$. Appendix~\ref{app:simple-first-example} shows that $\widetilde\eta_{(\infty,1),\,\x}= 0$ can also optimise the risk; this happens for both pooled and product risks when $(\mu^*-\bar x)^2$ is large. As $\bar y\to\mu^*$, the training and calibration data conflict; evidence against updating belief using $\x$ accumulates with increasing calibration samples to the point where, if we believe in our elicited risk, we should set $\kappa=0$ and discard $\x$. 



\subsubsection{Non-Regular asymptotics}\label{asymptotic_boundary}

Sometimes $\Omega_s$ is half-closed or closed, and $\widetilde{s}$ lies on the boundary of the domain. This is particularly common in SMI, discussed below, as $\eta=0$ corresponds to the Cut-Model, which is often selected. In such cases, standard interior asymptotic normality results are not directly applicable, the loss-minimiser need not have zero gradient, so the leading term in the loss-expansion is linear rather than quadratic. 

To handle this non-regular setting, we convert the posterior concentration theory on the boundary by \cite{Bochkina2014} to our setting, considering a special  case where $\widetilde{s}=0_{d_s}$ lies on the closed boundary and all components of the risk-gradient are positive there. In Section~\ref{sec:post.boundary.example} we give the corresponding result for $\widetilde{s}=\infty$. In contrast to the usual interior case, the limiting posterior distribution is a Gamma distribution, with potentially faster convergence rates and explicit dependence on prior behaviour near the boundary. We consider the product loss only. Assumption~\ref{A:s_boundary} is given in Appendix~\ref{app:non-reg-assume}.

\begin{theorem}\label{thm:post.boundary} 
Suppose $\Omega_s = [0,\infty)^{d_s}$ and $\widetilde{s}=0$, so the optimal $s$ lies on the boundary. With Assumptions \ref{A:product-regular}(i) and \ref{A:s_boundary}, and additionally that $\Delta_{\rm tail}(\delta_J)\xrightarrow{\pstar}0$ as $J\to\infty$ where
$$\Delta_{\rm tail}(\deltaJ) = J^{\sum_{j=1}^{d_s} \alpha_j}\int_{\Omega_s \setminus \overline{\Omega}_{(\widetilde{s},\deltaJ)}} e^{-l_{(J,1)}(s;\y,\x)+l_{(J,1)}(\widetilde{s};\y,\x)} \rho(s) ds \,.$$
For $s\sim \rho_{(J,1)}(s|\y,\x)$ and $T=Js$ let $\mathcal{L}_T(dt|\y,\x),\ t\in (0,\infty)^{d_s}$ be the distribution of $T$. In the limit $J\to\infty$, it holds that
\begin{align}
P_{\y\sim \pstar} \bigg(\| \mathcal{L}_T(\cdot|\y,\x) -\prod_{j=1}^{d_s}\Gamma(\cdot;\alpha_j,\nabla_{s_j} \widetilde{l}(\widetilde{s};\x)) \|_{TV}>\epsilon\bigg) \to 0\,,
\label{eq:post.boundary}
\end{align}
for all $\epsilon>0$, where $\Gamma(\cdot;\alpha,a)$ is the Gamma distribution with shape $\alpha$ and rate $a$.
\end{theorem}
We give a short and informal proof in Appendix~\ref{app:non-reg-assume}. This maps our notation onto the notation in \cite{Bochkina2014} and checks Assumption~\ref{A:s_boundary} and the conditions in Theorem~\ref{thm:post.boundary} cover all the assumptions required in that work.


\subsubsection{Examples of boundary optima}\label{sec:post.boundary.example}

Suppose that the inference hyperparameter is scalar with $s\in\Omega_s=(0,\infty]$, its prior $\rho(s)=Ds^{-\alpha-1}(1+o(1))$ has a polynomial tail at $s\to\infty$ with $\alpha>0$ in Assumption~\ref{A:s_boundary}(iii) and the product risk $\widetilde l(s;\x)$ has its unique minimum $\tilde s=\infty$. Setting
$u=1/s$, the prior becomes $\rho_u(u)=u^{-2}\rho(1/u)$ on $\Omega_u=[0,\infty)$ and the product loss and risk for $u$ are $l^{\,u}_{(J,1)}(u;\y,\x)=l_{(J,1)}(1/u;\y,\x)$ and $\widetilde l_u(u;\x)=\widetilde l(1/u;\x)$. Suppose these functions have continuous extensions to $u=0$.
Suppose 
all conditions of Theorem \ref{thm:post.boundary} hold with $\tilde u=0$ for the loss $l^{\,u}_{(J,1)}$ and prior $\rho_u$. If $\rho^{\,u}_{(J,1)}(\cdot\mid\y,\x) \propto \exp\{-l^{\,u}_{(J,1)}(u;\y;\x)\}\rho_u(u)$ and $U\sim \rho^{\,u}_{(J,1)}(\cdot| \y,\x)$ then $JU \xrightarrow{d} \operatorname{Gamma}(\alpha,a)$ in the sense of \eqref{eq:post.boundary} with
$a=\partial_u\widetilde l_u(0+;\x)$. 
Equivalently, provided the relevant derivative limit exists, $a=\lim_{s\to\infty} -s^2\partial_s \widetilde l(s;\x)>0$.
 Hence, if $S=1/U$ then ${S}/{J}\xrightarrow{d}\operatorname{IG}(\alpha,a)$ so 
the posterior for $s$ goes to infinity at rate $J$.

In the normal example in Section~\ref{sec:simple-first-example}, if the prior for $\eta$ has a polynomial tail and the quantity $c_x$ 
in Appendix~\ref{app:simple-first-example} is greater than zero then all conditions of Theorem~\ref{thm:post.boundary} can be shown to hold so if $\eta\sim\rho_{(J,1)}(\cdot|\y,\x)$ then $\eta/J\xrightarrow{d} \operatorname{IG}\!\left(
\alpha,{\sigma^2c_x}/{n}
\right)$ in the sense above. In Appendix~\ref{app:gbi-normal-example-product} we verify the conditions of Theorem~\ref{thm:post.boundary} in a slightly more general example.


\section{Generalised Bayes for multi-modular models}\label{sec:SMI}

\subsection{Related work in multi-modular settings}\label{sec:smi-lit-rev}

A module is the set of variables needed to define the observation model for one data set in a posterior distribution based on more than one data set. A posterior is multi-modular if it has more than one module and the modules share parameters. See \cite{liu2024generalframeworkcuttingfeedback} for a formal definition. Misspecification is common in multi-modular Bayesian inference: more can go wrong. \citet{Liu2009} provide an early analysis of ``modularization" and methods such as Bayesian melding \citep{Poole2000} and Markov melding \citep{Goudie2019melding} address conflicts between priors across modules.   

We assume that modules have been identified as either misspecified or well-specified. For example, we may have a mix of experimental and observational data, or an initial analysis which showed misfit. {\it Cut-model} inference \citep{Plummer2015} is effective in this setting. It is a sequential imputation procedure, in which the uncertainty for a shared parameter is estimated from a first well-specified module and the distribution passed as a prior for the shared parameter in a second module. This is not Bayesian inference, as information from the second module does not inform the shared parameter. An early form of Cut model inference has been available in WinBUGS \citep{Spiegelhalter2014} and cut models have found many applications \citep{BLANGIARDO2011379, Finucane2016, Styring2017, Li2017, Nicholson2021covid, Teh2021covid, Lunn2013, Kaizar2015, Zigler2014}. \citet{Jacob2017b} provides an overview of modularised Bayesian analysis, including Cut-models, from the perspective of statistical decision theory, while \citet{frazier2023} considered general losses and prove a BVM-type result for the generalised cut posterior. Nested MCMC \citep{Plummer2015, carmona20} is commonly used to fit Cut models. Recent advances include Laplace \citep{frazier2023} and variational approximations \citep{Yu2021variationalcut,song2025neuralvariationalinferencecutting}, unbiased MCMC \citep{Jacob17} and a computationally efficient variant of nested MCMC \citep{Liu2020sacut}.

Cutting feedback gives a bias-variance trade-off. For example, if the parameters of a well-specified module lack sufficient ``local" information, incorporating limited information from misspecified modules reduce uncertainty without introducing significant bias. \citet{carmona20} address this issue using Semi-Modular Inference (SMI). In this framework, a loss with a learning rate is designed to control the flow of information from a misspecified module, while the loss for a well-specified module is just the negative log-likelihood. This gives a sequence of posteriors parameterised by the learning rate and interpolating the conventional posterior and the cut posterior. Various parameterisations of the sequence have been suggested \citep{Nicholls2022,Nott2024} and choices of loss include tempered likelihood \citep{carmona20} and tempered marginal likelihood \citep{Nicholls2022,frazier2023}. SMI has proven effective in applications \citep{Liu2021generalized,styring22,battaglia2025} involving misspecified models for high dimensional latent variables.

The asymptotic properties of the Gibbs posterior in the single module case, including posterior consistency and a Bernstein von Mises (BVM) type result, have been established \citep{Martin2022, Millar21}. \citet{frazier2023} show similar results for the generalised cut posterior. Recently, \citet{Winter2023} presented a sequential-Gibbs posterior and its asymptotic properties, extending \citet{Millar21}. Their setup includes SMI posteriors, though they do not discuss the connection. 

\subsection{Semi-Modular inference}\label{sec:SMI-intro}

\begin{figure}[t]
  \begin{center}
  \scalebox{1}{
    \begin{tikzpicture}
      \node (Z) [data] {$\x_1$};
      \node (Y) [data, right=of Z, xshift=0.5cm] {$\x_2$};
      \node (phi) [param, below=of Z] {$\varphi$};
      \node (theta) [param, below=of Y] {$\theta$};
      \edge {phi} {Y, Z};
      \edge {theta} {Y};
      \draw[dashed,red] (0.85,0) to (0.85,-1.25);
      \draw[red,->] (0.55,-1.3) to (1.15,-1.3);
      \node[text width=3cm] at (0.5,1) {Module 1};
      \node[text width=3cm] at (2.75,1) {Module 2};
    \end{tikzpicture}
    }
  \end{center}
  \caption{Graphical representation of a simple multi-modular model. The Bayes posterior for this model is given in \eqref{eq:eta-smi-joint} with $\eta=1$. The dashed vertical line and arrow indicate that information flow from right to left is modulated in SMI.}
  \label{fig:Multimodular_model}
\end{figure}

Inference hyperparameters control the flow of information from different sources. Consider the two-module setup in Figure~\ref{fig:Multimodular_model} with two data-vectors $\x_1=(x_{1,1},\dots,x_{1,n_1})$ and $\x_2=(x_{2,1},\dots,x_{2,n_2})$ and model parameter vectors $\varphi$ and $\theta$ and suppose the sample spaces are $x_{i,j}\in \R^{d_i},\ j=1,\dots,n_i$ and $i=1,2$, $\varphi\in \Omega_\varphi \subseteq\R^{d_\varphi}$ and $\theta\in \Omega_\theta \subseteq \R^{d_\theta}$. The sample sizes grow with a fixed ratio, $\lim_{n_2\to\infty} n_1/n_2=\alpha$. Let $\pstar$ be the true generative model for $\x=(\x_1,\x_2)$ and let $\phi=(\varphi,\theta)$ to match previous notation. 

The belief update in Semi-Modular Inference (SMI, \cite{carmona20}) is a sequential-Gibbs posterior \citep{Winter2023}. The loss $\ell_{s}(\varphi;\x)$ for $\varphi$ controls the flow of information from the misspecified $\x_2$-module into the posterior for $\varphi$. The belief update for $\theta$ is Bayesian, with loss $\ell(\theta;\x_2,\varphi)=-\log(p(\x_2|\theta,\varphi))$ conditioned on $\varphi$ and learning rate equal one.  The posterior has the general form
\begin{align}\label{eq:SeqGibbsPosterior}
    \pi_s(\varphi,\theta|\x)&=\pi_s(\varphi|\x)\pi(\theta|\x_2,\varphi)\\
    &=\frac{\exp(-\ell_{s}(\varphi;\x))\pi(\varphi)}{p_s(\x)}\frac{\exp(-\ell(\theta;\x_2,\varphi))\pi(\theta)}{p(\x_2|\varphi)},\nonumber
\end{align}
and it is clear from the first equation that this is just a special case of the sequential-Gibbs posterior given in Definition~1 of \cite{Winter2023}.
Different choices of $\ell_{s}(\varphi;\x)$ give different variants of SMI.
\begin{itemize}
    \item If $\ell_{s}(\varphi;\x)=-\log(p(\x_1|\varphi)$ then $\pi_s(\varphi,\theta|\x)$ is the Cut model of \cite{Plummer2015}.
    \item If $\ell_{s}(\varphi;\x)=\ell_\eta(\varphi;\x)$ where
    \begin{equation}\label{loss:eta_varphi}
     \ell_\eta(\varphi;\x) = -\displaystyle\log \int p(\x_2|\varphi,\theta')^\eta \pi(\theta') d\theta' - \log p(\x_1|\varphi)\,,
    \end{equation}
    then $\pi_s(\varphi,\theta|\x)$ is the $\eta$-SMI posterior of \cite{carmona20}.
    \item If $\ell_{s}(\varphi;\x)=\ell_\gamma(\varphi;\x)$ where
    \begin{equation}\label{loss:gamma_varphi}
    \ell_\gamma(\varphi;\x)=- \displaystyle\gamma\log \int p(\x_2|\varphi,\theta') \pi(\theta') d\theta' - \log p(\x_1|\varphi)\,,
    \end{equation}
    then $\pi_s(\varphi,\theta|\x)$ is the $\gamma$-SMI posterior of \cite{Nicholls2022} (the marginalised loss posterior of \citet{frazier2023}).
    \item If $\ell_{s}(\varphi;\x)=\ell_{(\eta,\beta)}(\varphi;\x)$ where
    \begin{align*}
        \ell_{(\eta,\beta)}(\varphi;\x)=-\log\left(\int\exp(-\eta\ell_\beta(\varphi,\theta';\x_2))\pi(\theta')d\theta'\right)  - \log p(\x_1|\varphi)\,,
    \end{align*}
    and $\ell_\beta(\varphi,\theta';\x_2)$ is the $\beta$-loss in \eqref{eq: beta loss} with $\phi\to (\varphi,\theta')$ and $\x\to \x_2$ then $\pi_s(\varphi,\theta|\x)$ is the $(\eta,\beta)$-SMI posterior. This new variant of SMI is designed to be robust to outliers in $\x_2$ following the discussion in \cite{jewson24}. 
\end{itemize}
 
These SMI-posteriors interpolate between the Cut model posterior (at $\eta=0$ and $\gamma=0$) and Bayes (at $\eta=1$ and $\gamma=1$). The $(\eta,\beta)$-SMI posterior is a larger family coinciding with the Cut model posterior at $\eta=0$ (and any $\beta$) and Bayes at $\eta=1$ in the limit $\beta\to 1$. More general Bregman-SMI posteriors may be defined by taking a loss $\ell_\beta(\varphi,\theta';\x_2)$ defined by other choices of the convex function $f$ in \eqref{loss:f}.

These Cut-model and SMI posteriors involve intractable integrals, but the $\eta$- and $(\eta,\beta)$-posteriors may be written in terms of an auxiliary parameter and sampled using the same ``nested'' MCMC simulation \citep{Plummer2015,carmona20} used for Cut-models. The joint $\eta$-SMI posterior is 
\begin{align}\label{eq:eta-smi-joint}
\pi_\eta(\varphi,\theta',\theta|\x)&= \pi_\eta(\varphi,\theta'|\x)\ \times\ \pi(\theta|\x_2,\varphi) \\
&=\frac{\pi(\varphi)\pi(\theta')\exp(-\ell_\eta(\varphi,\theta';\x))}{c_\eta(\x)}\frac{\pi(\theta)p(\x_2|\varphi,\theta)}{p(\x_2|\varphi)},\nonumber
\intertext{where $c_\eta(\x)$ is an intractable normalising constant and}
\ell_\eta(\varphi,\theta';\x))&=-\log(p(\x_1|\varphi))-\eta\log(p(\x_2|\varphi,\theta'))\,.\label{eq:eta-smi-joint-loss}
\end{align}
Integration over $\theta'$ in \eqref{eq:eta-smi-joint} gives the $\eta$-SMI posterior for $\varphi$ and $\theta$. Similarly, the joint $(\eta,\beta)$-SMI posterior is
\begin{align}\label{eq:eta-beta-smi-joint}
\pi_{(\eta,\beta)}(\varphi,\theta',\theta|\x)&= \pi_{(\eta,\beta)}(\varphi,\theta'|\x)\ \times\ \pi(\theta|\x_2,\varphi) \\
\intertext{with}
\pi_{(\eta,\beta)}(\varphi,\theta'|\x)&=\frac{\pi(\varphi)\pi(\theta')\exp(-\ell_{(\eta,\beta)}(\varphi,\theta';\x_2))}{c_{(\eta,\beta)}(\x)}\,,\nonumber
\intertext{and}
\ell_{(\eta,\beta)}(\varphi,\theta';\x_2)&=-\log(p(\x_1|\varphi))+\eta\ell_\beta(\varphi,\theta';\x_2) \label{eq:eta-beta-smi-joint-loss}
\end{align}
Integrating over $\theta'$ in \eqref{eq:eta-beta-smi-joint} gives the $(\eta,\beta)$-SMI posterior. The marginal $p(\x_2|\varphi)$ is intractable but $\pi_s(\varphi,\theta',\theta|\x)$ can be sampled by sampling the first factor $\varphi,\theta'\sim \pi_s(\cdot|\x)$ in \eqref{eq:eta-smi-joint} or \eqref{eq:eta-beta-smi-joint} and then for each $\varphi$ we sample $\theta\sim \pi(\theta|\x_2,\varphi)$. This fails for $\gamma$-SMI.
\cite{frazier2023} use a variational approximation and below we analyse $\gamma$-SMI in cases where numerical evaluation or exact calculation (Appendices~\ref{Appendix:computing-gamma-SMI} and \ref{Appendix:Normal}) is feasible.

When we estimate $s=\eta, s=\gamma$ or $s=(\eta,\beta)$ there is a choice of loss depending on which dataset we want to predict. The calibration data $\y=(\y_1,\y_2)$ has a component for each module, and $z=(z_1,z_2)$ is test data, so we can take $l(s;\y_i,\x)=-\log(p_s(\y_i|\x))$ with $i=1$ or $i=2$ or we can take $l(s;\y,\x)=-\log(p_s(\y|\x))$. In the following we use $\y_1$ to predict $z_1$ as this is the simplest case. The product-posterior for $s$ in \eqref{eq:GB-hyper-posterior} is $\rho_{(J,1)}(s|\y_1,\x)$ with 
\[
p_s(y_{1,j}|\x)=\int p_s(y_{1,j}|\varphi)\pi_s(\varphi|\x) d\varphi\,,
\]
for $j=1,\dots,J$ and $\pi_s(\varphi|\x)$ given in \eqref{eq:SeqGibbsPosterior}. The pooled-posterior is $\rho_{(1,J)}(s|\y_1,\x)\propto \rho(s)p_s(\y_1|\x)$.
These can be computed using the same methods we gave for the single-module case in Section~\ref{sec:bayes-for-s-main-idea}. Calibration using $\y_2$ to predict $z_2$ is only slightly more complicated: for example in the product posterior $\rho_{(J,1)}(s|\y_2,\x)$ it is necessary to sample $(\varphi,\theta',\theta)$ in order to sample the marginal for $(\varphi,\theta)$ and evaluate  $p_s(y_{2,j}|\varphi,\theta)$; only $(\varphi,\theta')$ are needed to sample the marginal for $\varphi$ and evaluate $p_s(y_{1,j}|\varphi)$.

\section{Asymptotic behaviour of SMI-posteriors} \label{subsec:seqGibbs}

We consider how $\pi_s(\varphi|\x)$ and $\pi(\theta|\x_2,\varphi)$ in \eqref{eq:SeqGibbsPosterior} evolve as the size $n=n_2$ of the training data $\x$ grows with $\lim_{n\to\infty} n_1/n=\alpha$ fixed. Convergence results from \cite{Winter2023} carry over straightforwardly for $\eta$-SMI and $(\eta,\beta)$-SMI (and $f$-SMI, see last paragraph). The results we need are set out for completeness in our Appendix~\ref{app:winter-theory-main}. The manifold setting in \cite{Winter2023} is the full parameter space here so the chart there is the identity map here. Many results here are essentially the same as those in \citet{Martin2022}, \citet{Millar21} and \cite{frazier2023}.

We add two results to those in \cite{frazier2023}: we give consistency and a Bernstein-Von Mises theorem for $(\eta,\beta)$-SMI, subject to Assumptions \ref{A:Winter0-old}-\ref{A:Winter3-new}, \ref{A:Winter1-old} and \ref{A:Winter3-old} given in Appendix~\ref{app:winter-theory-all}. Assumption~\ref{A:Winter0-old} does not hold in general for $\gamma$-SMI so in Proposition~\ref{prop:eta-gamma-first-order} in Appendix~\ref{Appendix:gamma-is-eta} we show that the risk functions for $\varphi$ in the first stage of $\gamma$-SMI and $\eta$-SMI converge as $n$ goes to infinity at fixed $\eta=\gamma$. This may give similar posterior distributions for $\eta$ and $\gamma$ in \eqref{eq:GB-hyper-posterior} and this is visible in our experiments.

In the following Corollary the SMI-posteriors in \eqref{eq:eta-smi-joint} and \eqref{eq:eta-beta-smi-joint} have the parameterisation $\pi_s(\varphi,\theta'|\x)\pi(\theta|\x_2,\varphi)$, so the SMI-variables $(\varphi,\theta')$ correspond to $\varphi$ in  Theorems~\ref{thm:strong} and \ref{thm2} (and $\theta$ has the same meaning in both places). For the two SMI-variants with $s=\eta$ and $s=(\eta,\beta)$ all definitions carry over by substitution.

\begin{corollary}\label{cor:aux_loss_joint}
If the conditions of Theorem~\ref{thm:strong} in Appendix~\ref{app:winter-theory-main} are satisfied when $\ell_{s}(\varphi;\x)$ is replaced with $\ell_{s}(\varphi,\theta';\x)$ (and $\ell(\theta;\varphi,\x_2)$ is unchanged) then (1) it holds that \[P^{(s)}(d((\varphi,\theta',\theta),(\widetilde{\varphi}_s,\widetilde{\theta}'_s,\widetilde{\theta}_{\widetilde{\varphi}_s}))<\epsilon|\x) \xrightarrow{a.s.} 1.\] 
If $(\varphi,\theta') \sim \pi_{s}(\varphi,\theta'|\x)$, $\theta|\varphi \sim \pi(\theta|\x_2,\varphi)$ and, under the same substitutions, the conditions of Theorem~\ref{thm2} in Appendix~\ref{app:winter-theory-main} hold then (2) it holds, given $\x$, that \[\sqrt{n}((\varphi,\theta') -  (\overline{\varphi}_s,\overline{\theta}'_s)) \xrightarrow{d} N(0,H^{-1}_{\varphi,\theta'})\,,\] where $H_{\varphi,\theta'}= \nabla^2_{\varphi,\theta'} \widetilde{\ell}_{s}(\widetilde{\varphi}_s,\widetilde{\theta}'_s)$. 
This implies (3) the convergence of the $\varphi$-marginal, $\sqrt{n}(\varphi-\overline{\varphi}_s) \xrightarrow{d} N(0,H^{-1}_{\varphi}) $ where $H^{-1}_\varphi$ is the first $d_\varphi$ rows and columns of $H^{-1}_{\varphi,\theta'}$.
\end{corollary}

\begin{proof} Under the assumptions, (1) holds by Theorem~\ref{thm:strong} and (2) holds by Theorem~\ref{thm2} so
(3) holds as weak convergence of the joint gives weak convergence of marginals.
\end{proof}

Corollary~\ref{cor:aux_loss_joint} gives sufficient conditions for the marginal $\eta$-SMI posterior $\pi_\eta(\varphi|\x)$ to be approximately normal, a result \cite{frazier2023} have already established. Its mean is the loss minimiser $\bar{\varphi}_\eta$ and this converges to $\widetilde{\varphi}_\eta$ satisfying $\nabla_{\varphi,\theta'} \widetilde{\ell}_\eta(\widetilde{\varphi}_\eta,\widetilde{\theta}'_\eta)=0$. The equivalent result holds for $(\eta,\beta)$-SMI. In fact, if we replace $\ell_\beta$ in the $(\eta,\beta)$-SMI loss for $(\varphi,\theta')$ in \eqref{eq:eta-beta-smi-joint-loss} with any loss $\ell_f$ of the form \eqref{loss:f} (derived from a Bregman divergence $D_f$ with $f$ having a finite dimensional parameterisation) then Assumption~\ref{A:Winter0-old} holds whenever the convergence of the loss (an average) to the risk (an expectation) satisfies a strong law. This is discussed in Appendix~\ref{app:winter-theory-main} for $\eta$- and $(\eta,\beta)$-SMI. If Assumptions~\ref{A:Winter2-old}
-\ref{A:Winter3-old} also hold then Corollary~\ref{cor:aux_loss_joint} will hold for this ``$f$-SMI''. 


\section{Simulation study}\label{sec:simulation-examples}

We report experiments with two simulated data sets, a normal mixture model in Section \ref{sim:normal} and a state space model in Section~\ref{sec:ssm}. These examples illustrate convergence of $\rho(s|\y,\x)$ as set out in Section~\ref{sec:bayes_s}, generally superior performance of SMI over Bayes and Cut-model belief updates (in risk estimated on test data) and best results from $(\eta,\beta)$-SMI. We illustrate the methods on real data in Sections~\ref{sec:hpv} (a small epidemiological data set) and \ref{sec:edisc} (a realistic text classification task). We begin by setting out the estimators we use to summarise the posterior and the criteria we use to compare them.

\subsection{Hyperparameter estimation and evaluation criteria}\label{sec:hyperparam-estimation-and-R}

In practice, the calibration data is obtained by splitting the data into training $\x$ and calibration $\y$. This is discussed in Appendix~\ref{Appendix:normal4-TC-split}. Once we have an estimate $\hat s$ for $\widetilde{s}$ we pool the training and calibration data sets and go forward with $\pi_{\hat s}(\phi|\x,\y)$ using all the data. We are assuming that if $\hat s$ adjusts the predictive to make $p_{\hat s}(\cdot|\x)$ match $\pstar(\cdot)$ then $\hat s$ will also adjust $p_{\hat s}(\cdot|\x,\y)$ to match $\pstar(\cdot)$. This assumption is supported by experiments reported in Appendices~\ref{sec:normal3} and  \ref{Appendix:risk-ratios-state-space} (see Figures~\ref{fig:normal_b_gamma_eta_prod_pred} and \ref{fig:diff_eta}).

We consider five $s$-estimators: (i) the posterior mean $\hat s_{(J,1)}=E_{s\sim \rho_{(J,1)}(\cdot|\y;\x)}(s)$, (ii) the loss-minimiser $\bar{s}_{(J,1)}=\arg\min_s l_{(J,1)}(s;\y,\x)$, (iii) the WAIC $\eta_{\rm \scriptscriptstyle WAIC}$ and (iv) the Bayes estimator
$\widehat{s}_{KL(J,1)}$ minimising the KL divergence between posterior-predictive distributions at $\widetilde s$ and $s$. 
This is defined in Appendix~\ref{Appendix:s-KL-eta-beta} where we give a Monte Carlo estimator. In Appendix~\ref{Appendix:KL-harmonic} we show using a normal approximation to power posterior for $\phi$ that the KL-estimate for $\eta$ is approximately the harmonic mean, $\widehat{\eta}_{KL(J,1)}\simeq \widehat{s}_{hm(J,1)}=1/E_{\eta\sim \rho_{(J,1)}(\cdot|\y;\x)}(1/\eta)$ when $n$ and $J$ are large; this is the fifth estimator.  
The pooled loss $l_{(1,J)}$ leads to four corresponding estimators $\hat s_{(1,J)}$, 
$\bar{s}_{(1,J)}$, 
$\widehat{s}_{KL(1,J)}$ and $\widehat{s}_{hm(1,J)}$. The WAIC doesn't make sense for the pooled case -- it would correspond to an ELJPD with a single observation, the entire data set. From here on we drop the subscript $\x$ on hyperparameter estimates, writing $\hat s_{(J,1)}$ instead of $\hat s_{(J,1),\,\x}$ etc, for simplicity.

When we compare different methods for estimating $s$ and updating belief in Sections~\ref{sim:normal}, \ref{sec:ssm}, \ref{sec:edisc} and Appendix~\ref{sec:hpv}, we compute a risk ratio $R(\hat s_1,\hat s_2)$ estimated on test data. If we have $J_z$ test replicates from $\pstar$ and $\z=(z_1,\dots,z_{J_z})$ then
\begin{equation}\label{eq:s_relative}
    R_J(\hat s_1,\hat s_2) = \mathbb{E}_{\{z_{j}\}^{J_z}_{j=1} \sim \pstar} \left[\prod_{j=1}^{J_z}\frac{p_{\hat s_1}(z_{j}|\x,\y)}{p_{\hat s_2}(z_{j}|\x,\y)}\right] \,.
\end{equation}
When we compare pooled-loss estimators we use a different measure,
\begin{equation}\label{eq:s_relative_pooled}
    R_1(\hat s_1,\hat s_2) = \mathbb{E}_{\{z_{j}\}^{J_z}_{j=1} \sim \pstar} \left[\frac{p_{\hat s_1}(\z|\x,\y)}{p_{\hat s_2}(\z|\x,\y)}\right] \,,
\end{equation}
as the inferential objective is different. Here $R_1$ is the average intrinsic Bayes factor \citep{Berger1996}.
The expectations in \eqref{eq:s_relative} and \eqref{eq:s_relative_pooled} are estimated using Monte Carlo and if $R(\hat s_1,\hat s_{2})>1$ then $\hat s_1$ is favored. The ratios condition on $\x$ \emph{and} $\y$ (the full data set) to predict test data although $\hat s_1$ will be estimated to predict calibration data off $\x$ alone. We split the data to get $\y$ and want to compare using a measure which favors methods like Bayesian inference which do not split data. 



\subsection{Normal Mixture example}\label{sim:normal}

Following the framework in Figure \ref{fig:Multimodular_model}, we set up an example in which the $\x_1$-module is well-specified and the $\x_2$-module is misspecified due to a small group of outlying observations. In this section we take a simple version of $(\eta,\beta)$-SMI with $\eta=1$ which we refer to $\beta$-SMI. The loss for $\varphi$ controls the feedback from the $\x_2$-module and interpolates the cut and conventional posteriors. We consider $\eta$, $\gamma$ and $\beta$-losses. 

The setup is adapted from \cite{jewson24} where it was used to illustrate the efficacy of inference with a $\beta$-loss. See Figure~\ref{fig:normal_true}. Let $x_i$ be a generic component of $\x_i,\ i=1,2$. The true data generative models are $\pstar_1(x_1)=N(x_1;\varphi_A^*, \sigma^2_1)$ and \[\pstar_2(x_2)=\lambda^* N(x_2;\varphi_A^*,\sigma_2^{2}) + (1-\lambda^*) N(x_2;\theta^*,\sigma_2^{2}),\] all jointly independent. Here $\sigma^2_1=4^2, \varphi_A^*=0,  \theta^*=6,\sigma_2^{2}=1$ and varying values of $\lambda^*$ are considered.
The fitted models miss the mixture component with mean $\theta^*$, taking $p(x_1|\varphi) = N(x_1;\varphi,\sigma^2_1)$ and $p(x_2|\varphi,\theta) = N(x_2;\varphi+\theta,\sigma^2_2)$. Here $\sigma_1$ and $\sigma_2$ are fixed to the true values. A uniform prior is assigned for $\varphi$ and $\theta\sim N(0,s^2_\theta)$ with $s^2_\theta=0.33^2$. 

We explore $s$-estimation for $\eta,\gamma$ and $\beta$-SMI.
We reparameterise the $\beta$-SMI posterior using $b=1/\beta$ values so increasing $b$ gives a similar effect to increasing $\eta$. 
Interpolation between the cut and conventional posterior distributions motivates restricting $\eta$ and $\gamma$ to the unit interval. The range of $b$ is not restricted, but upper bounded using a bound chosen from preliminary numerical experiments to support the optimal values we actually see; the hyperparameter priors $\rho(s)$ are $\gamma\sim U[0,1]$, $\eta\sim U[0,1]$ and $b\sim U[0,3]$ (upper bound chosen to cover support of $b$-posteriors in all experiments). Simulation details for the $\eta$, $\gamma$ and $\beta$-SMI posteriors, including explicit expressions for losses, are summarised in Appendix \ref{Appendix:Normal}. 

 Misspecification arises due to the discrepancy between $\pstar_2$ and the parametric model $p(x_2|\varphi,\theta)$. As $\lambda^*$ decreases, the degree of misspecification decreases as the unmodeled outlier component goes away. However, when $\lambda^*=1$, the model is over-parameterised. This leads to interesting behaviour in the $\beta$-SMI posterior with the optimal $\beta<1$ ($b>1$, so up-weighting rather than down-weighting the information from the data). Given the training data $\x=\{\x_1,\x_2\}$, the calibration data for the loss-hyperparameter $s=\eta,\gamma$ or $s=b$ is taken from the $\x_1$-module, so we have $\y=(y_{1,1},\dots,y_{1,J})$ with $y_{1,j} \sim \pstar_1$ iid for $j=1,\dots,J$ (here and below we drop the $1$ subscript on $\y$ as we always calibrate using the well-specified first module).

Three studies are carried out: $s$ estimation with different levels of misspecification in Appendix~\ref{sec:normal1}, asymptotic property validation in Section~\ref{sec:normal2}, and an expected risk ratio comparison using test data in Appendix \ref{sec:normal3}. 

\subsubsection{Asymptotic property validation}\label{sec:normal2}

We illustrate convergence per Section~\ref{sec:bayes_s} by taking different values of $J$, the dimension of the calibration data $\y$. The training data $\x=\{\x_{1,1:n_1},\x_{2,1:n_2}\}$ has $n_1=30$ and $n_2=60$ samples. The misspecification parameter is set at $\lambda^*=0.9$ (ten percent outliers).

The $J$-dependence of the posterior densities for $b$, $\gamma$ and $\eta$ are illustrated in Figure \ref{fig:normal_loss_b_gam_eta}. 
In Appendix~\ref{Appendix:s-post-asym-further} we plot equivalent graphs for $\lambda^*=0.99$ and $\lambda^*=0.1$. These experiments report posteriors for just a single fixed training data set $\x$. Experiments in Appendix~\ref{sec:normal1} across many replicates of $\x$ show similar behaviour. 
\begin{figure}[ht]
    \centering
    {\begin{overpic}[ width=.27\linewidth,height=0.17\textheight ]
    {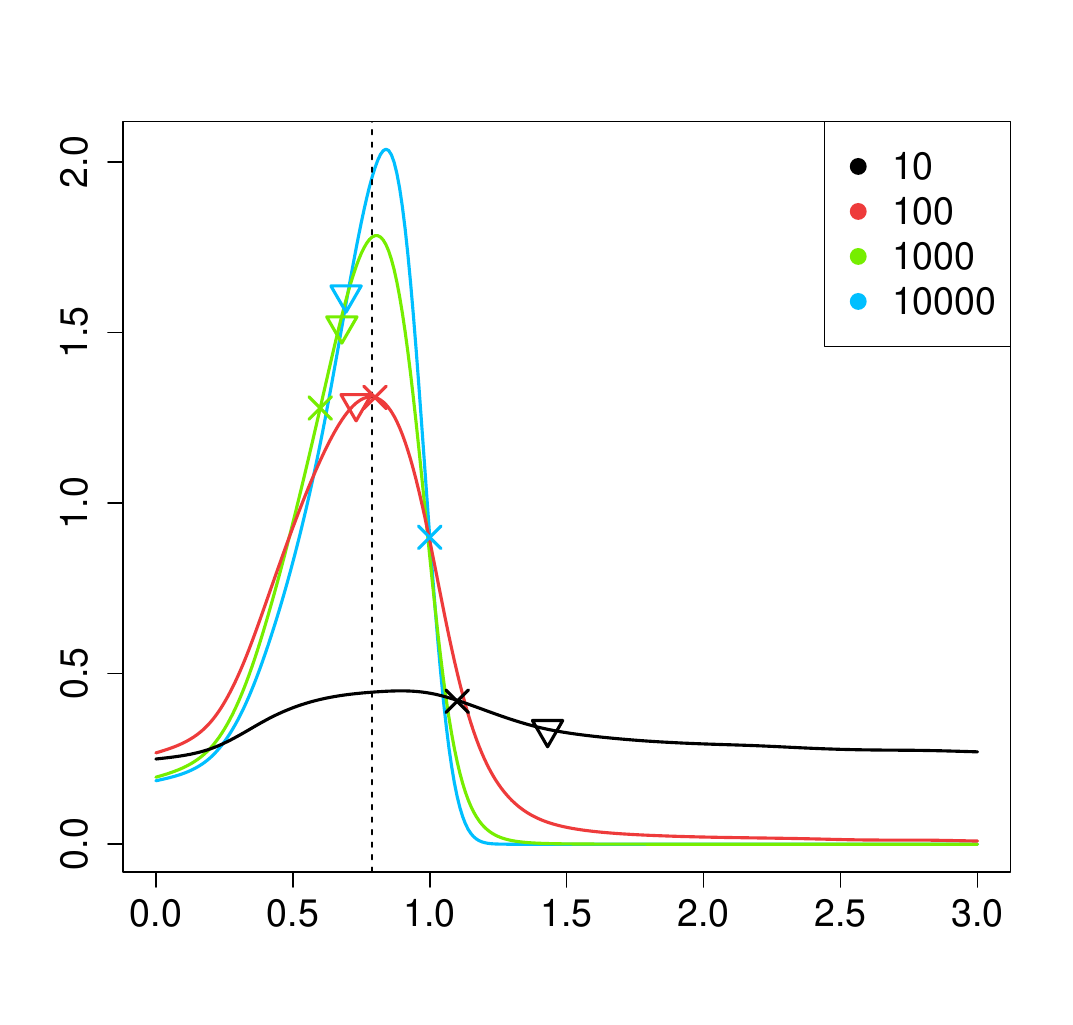}  
      \put(-5,20){\rotatebox{90}{\tiny $\rho_{(1,J)}(b|\y,\x)$}}
      \put(40,0){\tiny b values}
    \end{overpic}}
    \hspace{0.3cm}
    \begin{overpic}[ width=.27\linewidth,height=0.17\textheight ]{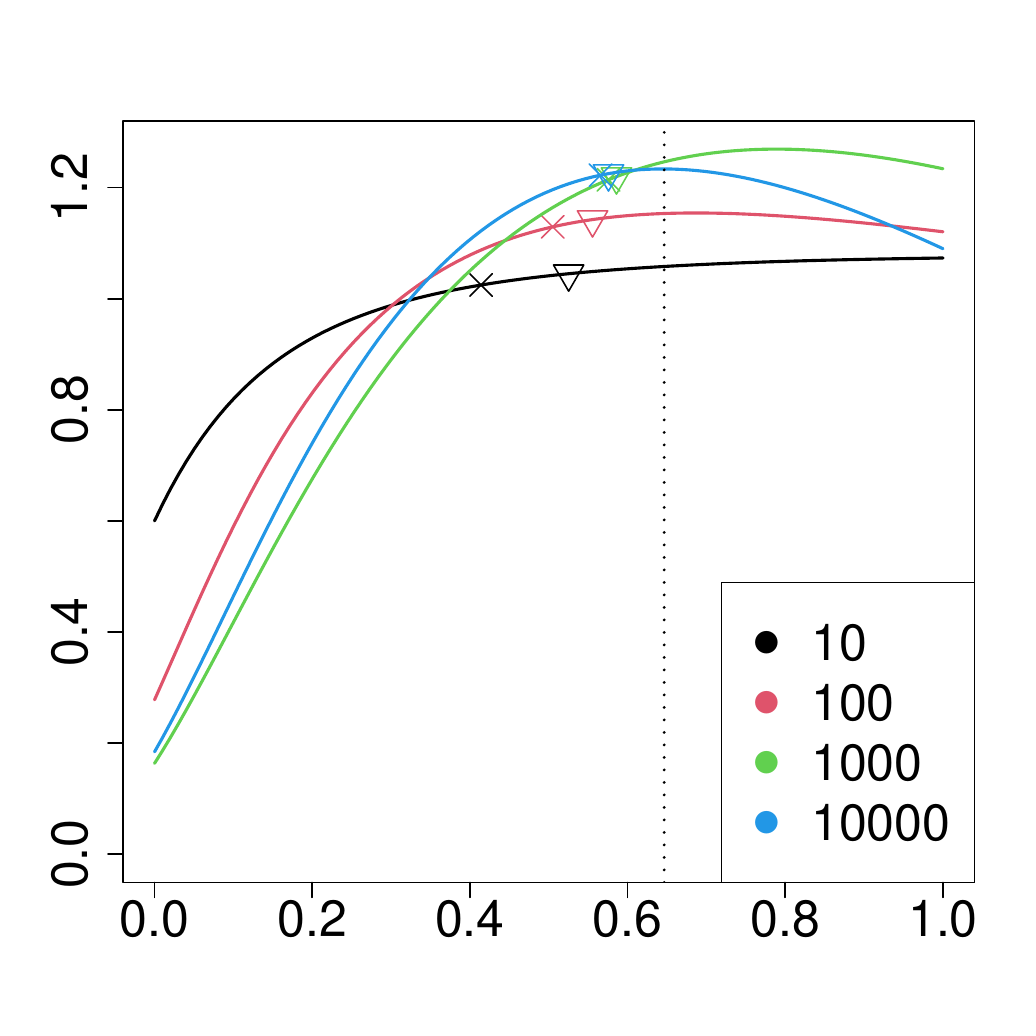}  
     \put(40,0){\tiny $\gamma$ values}  \put(-5,25){\rotatebox{90}{\tiny $\rho_{(1,J)}(\gamma|\y,\x)$}}
    \end{overpic}
    \hspace{0.3cm}
    \begin{overpic}[ width=.27\linewidth,height=0.17\textheight ]{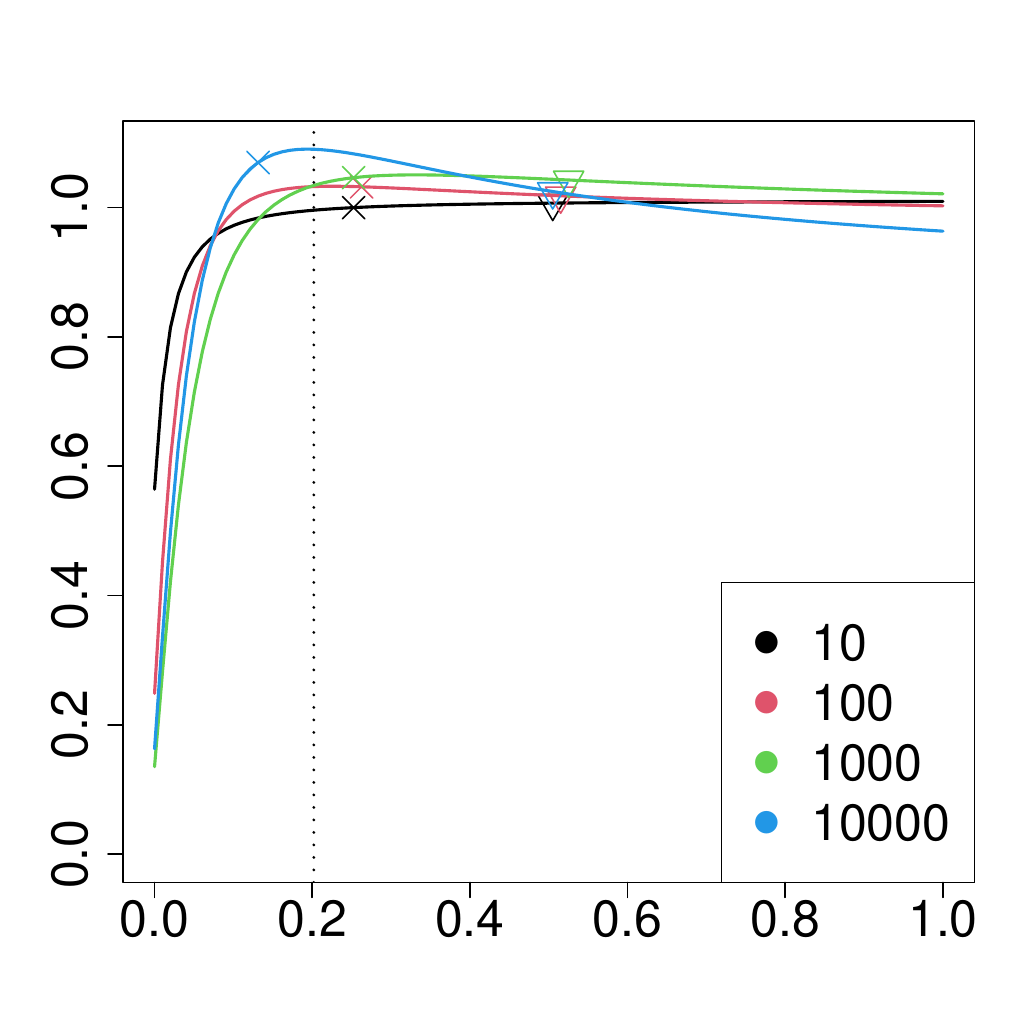}  
     \put(40,0){\tiny $\eta$ values}  \put(-5,25){\rotatebox{90}{\tiny $\rho_{(1,J)}(\eta|\y,\x)$}}
    \end{overpic}
    \vspace{0.5cm}
    {\begin{overpic}[ width=.28\linewidth,height=0.17\textheight ]{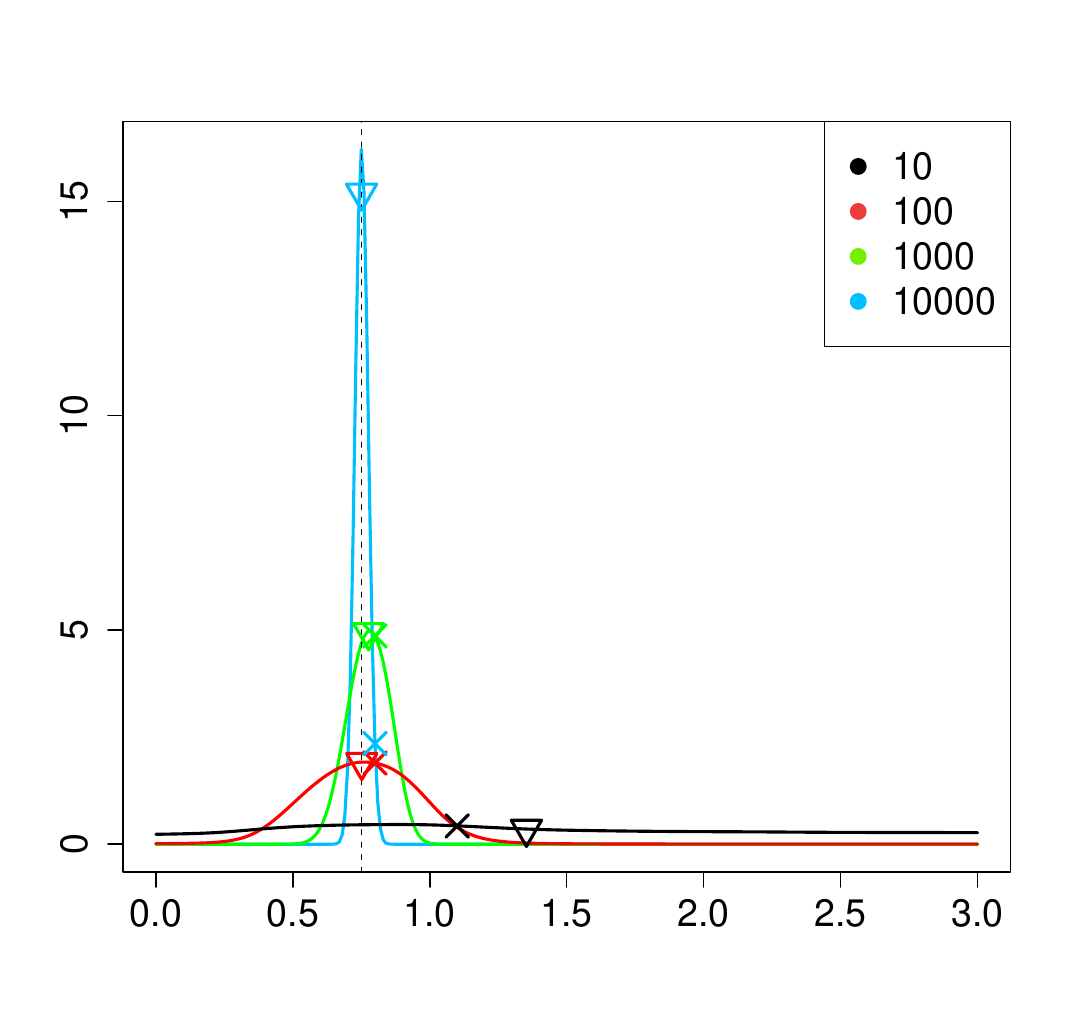}  
      \put(-5,20){\rotatebox{90}{\tiny $\rho_{(J,1)}(b|\y,\x)$}}
      \put(40,0){\tiny b values}
    \end{overpic}}
     \hspace{0.3cm}
    \begin{overpic}[width=.27\linewidth,height=0.17\textheight ]{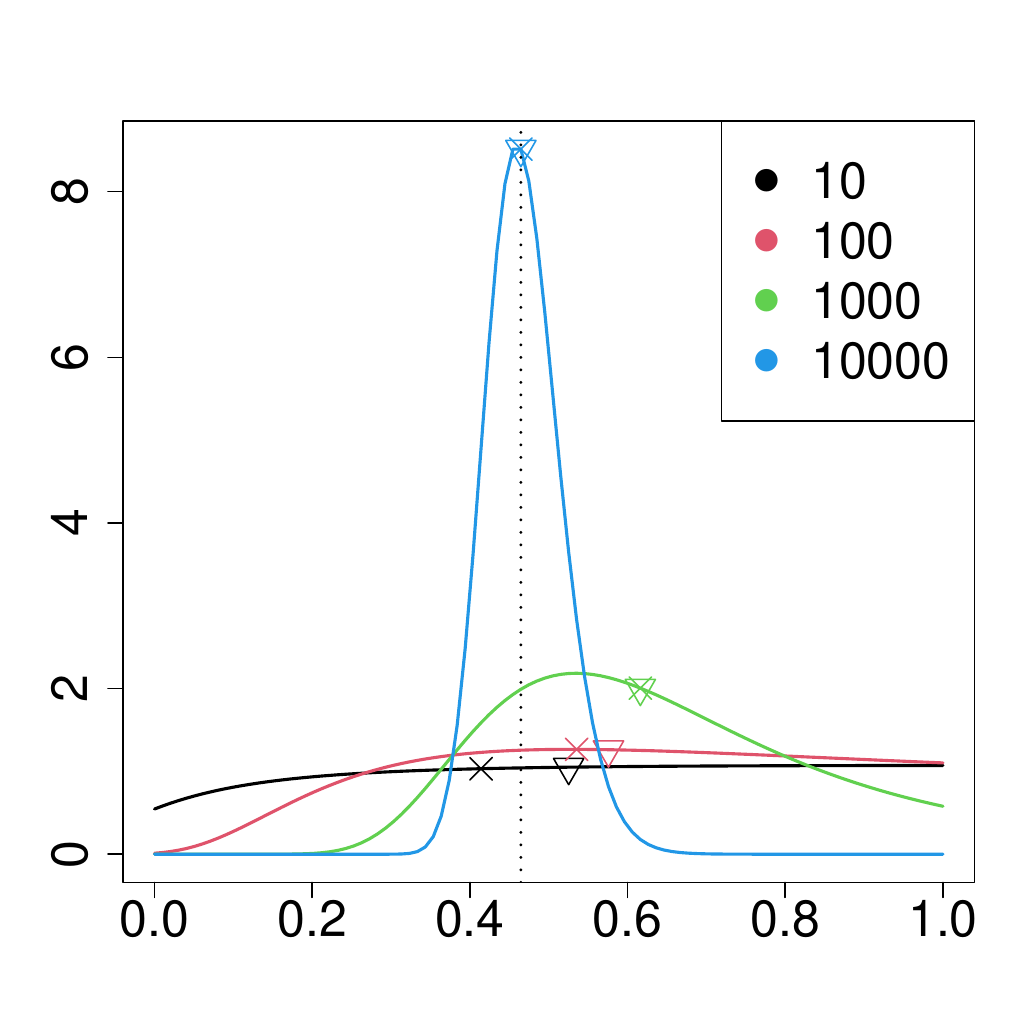}  
     \put(40,0){\tiny $\gamma$ values} \put(-5,25){\rotatebox{90}{\tiny $\rho_{(J,1)}(\gamma|\y,\x)$}}
    \end{overpic}
     \hspace{0.3cm}
    \begin{overpic}[width=.27\linewidth,height=0.17\textheight ]{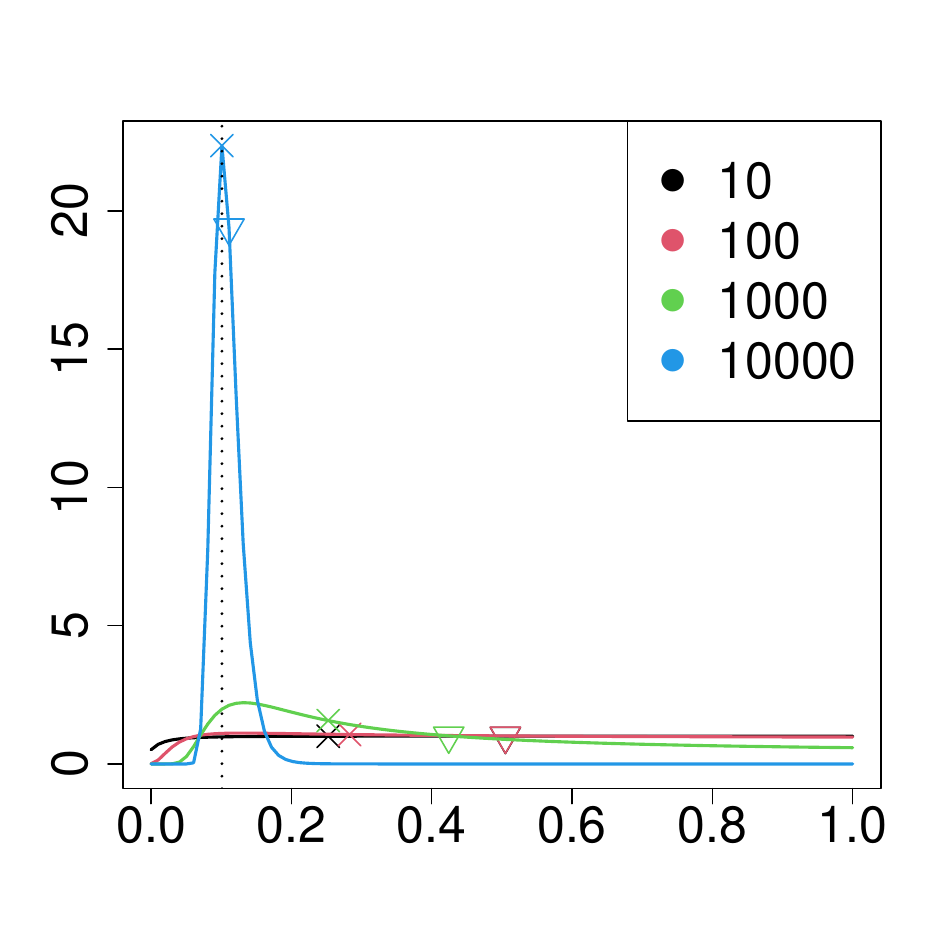}  
     \put(40,0){\tiny $\eta$ values} \put(-5,25){\rotatebox{90}{\tiny $\rho_{(J,1)}(\eta|\y,\x)$}}
    \end{overpic}
    \caption{Normal model of Section~\ref{sec:normal2}: columns give posterior densities for $b$ (left), $\gamma$ (middle) and $s=\eta$ (right) at varying $J$.
    Rows show pooled (top row, $\rho_{(1,J)}$)  and product (bottom row, $\rho_{(J,1)}$) posteriors. 
    Line color indicates $J$-value. Vertical dotted lines, triangle marks and cross marks respectively indicate the optimal value $\widetilde{s}$, posterior mean $\widehat{s}$ and $\widehat{s}_{KL}$ (see Appendix~\ref{Appendix:s-KL-eta-beta}) for $s=b$, $s=\gamma$ and $s=\eta$ in columns left to right. 
    }
    \label{fig:normal_loss_b_gam_eta}
\end{figure}

The points to note here are (1) qualitatively similar behaviour of $\gamma$ and $\eta$-posteriors in the two right columns of Figure~\ref{fig:normal_loss_b_gam_eta} per Proposition~\ref{prop:eta-gamma-first-order}, (2) optimal pooled $\widetilde{s}_{(1,\infty)}$ and product $\widetilde{s}_{(\infty,1)}$ values are not equal (see Section~\ref{blocksize}) but are located in the interior of $\Omega_s$ in line with Assumptions~\ref{A:laplace}-\ref{A:product-separation}, (3) convergence and concentration of the product-loss posteriors $\rho_{(J,1)}(s|\y,\x)$ to a normal centered at $\widetilde{s}_{(\infty,1)}$ in the bottom row, and ultimately a point mass, as required by Theorem~\ref{thm:product} and (4) convergence of the pooled $s$-posteriors $\rho_{(1,J)}(s|\y,\x)$ to a diffuse distribution coinciding with $\pi_s(\tilde\phi|\x)$ in the top row as required by Corollary~\ref{cor:posterior_pool}. 
See Appendix~\ref{Appendix:s-post-asym-further} for further discussion. In particular Figure~\ref{fig:smipost}
illustrates the near-identity of the $\gamma$ and $\eta$-SMI posteriors $\pi_\gamma(\varphi,\theta|\x)$ and $\pi_\eta(\varphi,\theta|\x)$ when the size of the training data $n_1,n_2$ is very large, in line with Proposition~\ref{prop:eta-gamma-first-order}.


\subsection{State-space example}\label{sec:ssm}

This is a Hidden Markov Model in which some of the ``hidden'' states are known and others are missing. There is misspecification due to different true observation models for the observed and missing values in the Markov chain. In this synthetic-data example the number of unknown 
arameters grows with the number of observations. It is a simple one-dimensional version of the spatial data analysis described in \cite{battaglia2025}.

\begin{figure}[!hb]
  \begin{center}
    \scalebox{0.85}{
    \begin{tikzpicture}
      \node (x1) [text=black] {$x_{1,1}$};
      \node (x2) [text=black, right=of x1, xshift=-0.45cm] {$x_{1,2}$};
      \node (x3) [text=black, right=of x2, xshift=-0.45cm] {$x_{1,3}$};
      \node (x4) [text=black, right=of x3, xshift=-0.45cm] {$x_{1,4}$};
      \node (x5) [text=black, right=of x4, xshift=-0.45cm] {$x_{1,5}$};
      \node (x6) [text=black, right=of x5, xshift=-0.45cm] {$x_{2,1}$};
      \node (x7) [text=black, right=of x6, xshift=-0.45cm] {$x_{2,2}$};
      \node (x8) [text=black, right=of x7, xshift=-0.45cm] {$x_{2,3}$};
      \node (x9) [text=black, right=of x8, xshift=-0.45cm] {$x_{2,4}$};
      \node (x10) [text=black, right=of x9, xshift=-0.45cm] {$x_{{2,5}}$};
      \node (xn) [text=black, right=of x10, xshift=-0.05cm] {$x_{n,5}$};
      \node (t1) [text=black, below=of x1] {$\theta_{1,1}$};
      \node (t2) [param, right=of t1, xshift=-0.02cm] {$\theta_{1,2}$};
      \node (t3) [param, right=of t2, xshift=-0.13cm] {$\theta_{1,3}$};
      \node (t4) [param, right=of t3, xshift=-0.13cm] {$\theta_{1,4}$};
      \node (t5) [text=black, below=of x5] {$\theta_{1,5}$};
      \node (t6) [text=black, below=of x6] {$\theta_{2,1}$};
      \node (t7) [param, right=of t6, xshift=0.0cm] {$\theta_{2,2}$};
      \node (t8) [param, right=of t7, xshift=-0.13cm] {$\theta_{2,3}$};
      \node (t9) [param, right=of t8, xshift=-0.13cm] {$\theta_{2,4}$};
      \node (t10) [text=black, below=of x10] {$\theta_{2,5}$};
      \node (tn) [text=black, below=of xn] {$\theta_{n,5}$};
      \edge[->] {t1} {t2, x1};
      \edge[->] {t2} {t3, x2};
      \edge[->] {t3} {t4, x3};
      \edge[->] {t4} {t5, x4};
      \edge[->] {t5} {t6, x5};
      \edge[->] {t6} {t7, x6};
      \edge[->] {t7} {t8, x7};
      \edge[->] {t8} {t9, x8};
      \edge[->] {t9} {t10, x9};
      \edge[->] {t10} {x10};
      \edge[dashed,->] {t10} {tn};
      \edge[->] {tn} {xn};
      \node[shape=coordinate] (b1t) [left=of x1, xshift=0.9cm] {};
      \node[shape=coordinate] (b1b) [left=of t1, xshift=0.9cm, yshift=-0.5cm] {};
      \edge[black, very thick, dashed, -] {b1t}{b1b}
      \node[shape=coordinate] (b2t) [left=of x6, xshift=0.6cm] {};
      \node[shape=coordinate] (b2b) [left=of t6, xshift=0.6cm, yshift=-0.5cm] {};
      \edge[black, very thick, dashed, -] {b2t}{b2b}
      \node[shape=coordinate] (b3t) [right=of x10, xshift=-0.8cm] {};
      \node[shape=coordinate] (b3b) [right=of t10, xshift=-0.8cm, yshift=-0.5cm] {};
      \edge[black, very thick, dashed, -] {b3t}{b3b}
      \node[shape=coordinate] (bJt) [right=of xn, xshift=-0.9cm] {};
      \node[shape=coordinate] (bJb) [right=of tn, xshift=-0.9cm, yshift=-0.5cm] {};
      \edge[black, very thick, dashed, -] {bJt}{bJb}
    \end{tikzpicture}
    } 
  \end{center}
  \caption{State Space Model, latent variable process $\theta_{1,1},\theta_{1,2}\dots,\theta_{n,d_x}$ and response variables $x_{1,1},x_{1,2},\dots,x_{n,d_x}$. Parameters in circles are missing; all other parameters and responses are observed. The vertical dashed lines indicate conditionally independent blocks of data. In this figure the block size is $d_x=5$ so the anchors are $A=\{(1,1),(1,5),(2,1),(2,5),\dots, (n,5)\}$. We took $d_x=6$ in the study reported below.}
  \label{fig:ssm-hmm-block-structure}
\end{figure}

The true and fitted observation models have a block structure illustrated in Figure~\ref{fig:ssm-hmm-block-structure}: the data in each block
are independent of data in other blocks so they contribute to one factor in the product in \eqref{eq:lkd-product}. We have $n$ blocks with $d_x$ observations in each block for $nd_x$ observations in all. Let $(\theta_{i,j},x_{i,j})$ denote the hidden and emitted HMM state pair at step $j=1,\dots,d_x$ in block $i=1,\dots,n$. Let $N=[n]\times [d_x]$ be the index set. The first and last ``hidden state'' parameters $\theta_{i,1}$ and $\theta_{i,d_x}$ in each block of the HMM have been measured and the rest are missing so $\theta_{i,j}$ is observed for $(i,j)\in A=[n]\times\{1,d_x\}$ and missing for $(i,j)\in M=N\setminus A$. In our experiment we take $d_x=6$ and we consider a range of $n$-values.
Let $\theta_A=\{\theta_{i,j}\}_{(i,j)\in A}$ and 
$\theta_M=\{\theta_{i,j}\}_{(i,j)\in M}$ and
let $\x_M$ and $\x_A$ be the corresponding emitted values (the data). Here $\x_A$ plays the role of $\x_1$ in Section~\ref{sec:SMI-intro}, $\x_M$ is $\x_2$ and $\theta_M$ is the $\theta$-parameter of the $\x_2$-module. Let $\x=(\x_A,\x_M)$.

The fitted observation model (Figure~\ref{fig:ssm-true-fitted-graph-models}, right) is $x_{i,j}\sim N(\theta_{i,j},\varphi^2)$ for $(i,j)\in N$ with $\varphi$ and $\theta_M$ unknown parameters to be estimated.
The hidden states $\theta$ in the HMM follow a stationary autoregressive model of order one so $\theta_{i,j}\sim N(\nu \theta_{i,j-1},\sigma^2)$ with $\nu=0.5$ and $\sigma=0.7$ (and the process continues across block boundaries). We take an inverse Gamma prior for $\varphi^2\sim \operatorname{IG}(2,1)$ and the true AR(1) process as a prior for $\theta_M$. The $\eta$-SMI posterior with auxiliary variable $\theta'_M$ (compare \eqref{eq:eta-smi-joint}) is
\begin{align}
    \label{eq: eta SMI SSM}\pi_{\eta}(\varphi,\theta'_M,\theta_M|\x) &=  \pi_{\eta}(\varphi,\theta'_M|\x)\, \pi(\theta_M|\x_M,\varphi)
    \\&\propto
     p(\x_A|\theta_A,\varphi)p(\x_M|\theta'_M, \varphi)^\eta \pi(\theta'_M) \pi(\varphi)\,\pi(\theta_M|\x_M,\varphi) \,,
    \nonumber
\end{align}
and the $(\eta,\beta)$-SMI posterior (compare \eqref{eq:eta-beta-smi-joint}) is
\begin{align}\label{eq: eta beta SMI SSM}
    \pi_{(\eta,\beta)}(\varphi,\theta'_M,\theta_M|\x) 
    &= \pi_{(\eta,\beta)}(\varphi,\theta'_M|\x)\,\pi(\theta_M|\x_M,\varphi)
\\&\propto 
p(\x_A|\theta_A,\varphi) \exp\{-\eta \ell_{\beta}(\varphi,\theta'_M;\x_M)\} \pi(\theta'_M) \pi(\varphi)\,\pi(\theta_M|\x_M,\varphi) \,,
\nonumber
\end{align}
where $\ell_{\beta}(\varphi,\theta'_M;\x_M)$ is the $\beta$-loss in \eqref{eq: beta loss} with $\phi\to (\varphi,\theta'_M)$ and $\x\to \x_M$. Conditioning on the known $\theta_A$ is implicit throughout.

\begin{figure}[!ht]
  \begin{center}
  \scalebox{0.75}{
    \begin{tikzpicture}
      \node (ta) [data] {$\theta_A$};
      \node (tm) [param, right=of ta, xshift=0.5cm] {$\theta_M$};
      \node (xm) [param, below=of tm] {$\x_M$};
      \node (xa) [data, left=of xm, xshift=-0.5cm] {$\x_A$};
      \node (sm) [param, below=of xm] {$\varphi^*_M$};
      \node (sa) [param, left=of sm, xshift=-0.35cm] {$\varphi^*_A$};
      \edge[-] {ta} {tm};
      \edge {ta} {xa};
      \edge {tm} {xm};
      \edge {sa} {xa};
      \edge {sm} {xm};
      \node [text=black, below= of sa, xshift=1.0cm, yshift=0.7cm] {True model};
    \end{tikzpicture}
    \hspace*{0.5in}
    \begin{tikzpicture}
      \node (ta) [data] {$\theta_A$};
      \node (tm) [param, right=of ta, xshift=0.5cm] {$\theta_M$};
      \node (xm) [param, below=of tm] {$\x_M$};
      \node (xa) [data, left=of xm, xshift=-0.5cm] {$\x_A$};
      \node (s) [param, below=of xa, xshift=0.95cm, yshift=-0.3cm] {$\varphi$};
      \edge[-] {ta} {tm};
      \edge {ta} {xa};
      \edge {tm} {xm};
      \edge {s} {xa};
      \edge {s} {xm};
      \node [text=black, below= of s, xshift=0.05cm, yshift=0.7cm] {Fitted model};
    \end{tikzpicture}
    }
  \end{center}
  \caption{The true (left) and fitted (right) State Space Models differ. In the fitted model the observation models for the observed and missing parts of the process have the same variance parameter $\varphi$. In fact they have different true variances $\varphi^*_A$ and $\varphi^*_M$.}
  \label{fig:ssm-true-fitted-graph-models}
\end{figure}

Misspecification arises because the fitted model assumes the variance $\varphi$ is the same for missing and observed cases. The AR(1) process realising the hidden variables $\theta$ is the same in the true and fitted models. However, in the true model $\x|\theta\sim \pstar$ (Figure~\ref{fig:ssm-true-fitted-graph-models}, left) observations $x_{i,j}\sim N(\theta_{i,j}, (\varphi_A^*)^2),\ (i,j)\in A$ at anchor locations have $\varphi_A^*=1$, and this is different from the generative model for observations at missing states where $x_{i,j}\sim N(\theta_{i,j}, (\varphi_M^*)^2),\ (i,j)\in M$. We give examples with 
$\varphi^*_M=0.5$ (strong misspecification) $\varphi^*_M=0.7$ (mild) and
$\varphi_M^*=1$ (well-specified). 

We suppose a calibration data set $\y,\theta^{(y)}$ in $J$ blocks is available. Following the discussion at the end of Section~\ref{sec:SMI-intro}, we adjust $s$ to predict the data in the well-specified module so $\y,\theta^{(y)}$ have the same distribution as $\x_A,\theta_A$. They are simulated by taking an independent realisation $\psi_{1,1},\dots,\psi_{J,d_x}\sim$AR(1) and $\x'|\psi\sim \pstar$ of length $n^{(y)}=Jd_x$ ($J$ blocks, with observed $\psi_{i,j}$ at $(i,j)\in A'=[J]\times\{1,d_x\}$) and collecting pairs $(y_{i,j}=x'_{i,j},\ \theta^{(y)}_{i,j}=\psi_{i,j})$ with $(i,j)\in A'$. 

Four numerical experiments are carried out: asymptotic behaviour of $\eta$-SMI (Appendix~\ref{sec: asym ssm}, similar to the normal mixture example in Section~\ref{sec:normal2}), expected risk ratio comparison of $\eta$-SMI and Bayes and Cut-models using test data (Section~\ref{sec: post pred ssm}) and experiments with $(\eta,\beta)$-SMI (Section~\ref{sec: eta-beta ssm}). 
In these SMI experiments we interpolate the Cut Model and Bayesian inference so we restrict $\eta\in [0,1]$. Finally, in Appendix~\ref{Appendix:zero_constraint} we look at $\rho(s|\y;\x)$ for $\eta\ge 0$. Posteriors for multiple replicates of $\x,\y$ are shown in  Figure \ref{fig:post_eta_100rep}. They cluster in $\eta\in [0,1]$, or overlapping $\eta=1$ in the well-specified case. We imposed $0\le\eta\le 1$ in order to interpolate Cut and Bayes posteriors and the results in Appendix~\ref{Appendix:zero_constraint} support that choice. In Appendix~\ref{Appendix:ssm-finite-eta-star} we show $\widetilde{\eta}_{(1,\infty)}$ and $\widetilde{\eta}_{(\infty,1)}$ are almost surely finite using an expansion at large $\eta$.

\subsubsection{Expected risk ratio comparison for \(\eta\)-SMI}\label{sec: post pred ssm}

We compare the expected risk ratio (\ref{eq:s_relative}) of our fitted $\eta$-SMI belief update (\ref{eq: eta SMI SSM}) with Cut and Bayes updates in Figure~\ref{fig:ssm_pred}. Comparisons with $\beta$-SMI and $(\eta,\beta)$-SMI are given in Appendix~\ref{Appendix:risk-ratios-state-space}. We generate 100 independent datasets $\bar{\x}=(\x,\y) $ each of $\bar{n}=60$ blocks which we split into $n=10$ blocks of training data $\x$ and $J=50$ blocks of calibration data, $\y$. We compute $\overline{\eta}$, $\widehat{\eta}$ and $\widehat{\eta}_{hm}$ for the product and pooled posteriors and for each training data set $\x$. 
We then estimate (for each of these estimates of $\eta$) the expected risk ratio $R(\eta,\eta_o)$ defined in (\ref{eq:s_relative}), based on a further 30 test data sets each with $J^{(z)}=100$ blocks. We take reference values of $\eta_o$ corresponding to the Bayes posterior ($\eta_o=1$, left column in Figure~\ref{fig:ssm_pred}) and the Cut posterior ($\eta_o=0$, centre column). 
We consider three levels of misspecification $\varphi^*_M=0.5$ (red), $0.7$ (green) and $1$ (blue). 

\begin{figure}[t]
    \centering
\begin{overpic}[ width=.32\linewidth,height=0.22\textheight ]
{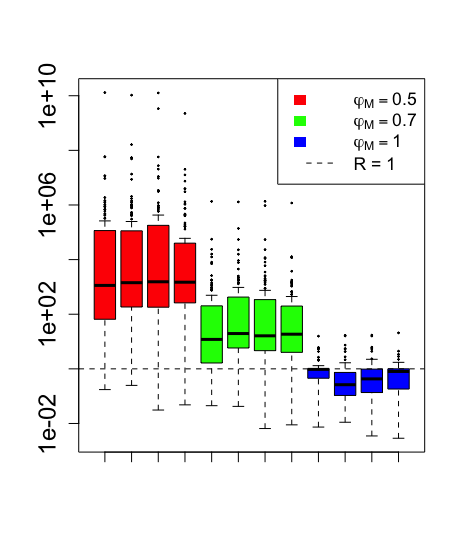} 
\put(32,92){\tiny $R(\eta_{(J,1)},1)$}
\put(16,5){\tiny a}
\put(23,5){\tiny b}
\put(28,5){\tiny c}
\put(33,5){\tiny d}
\end{overpic}
\begin{overpic}[ width=.32\linewidth,height=0.22\textheight ]{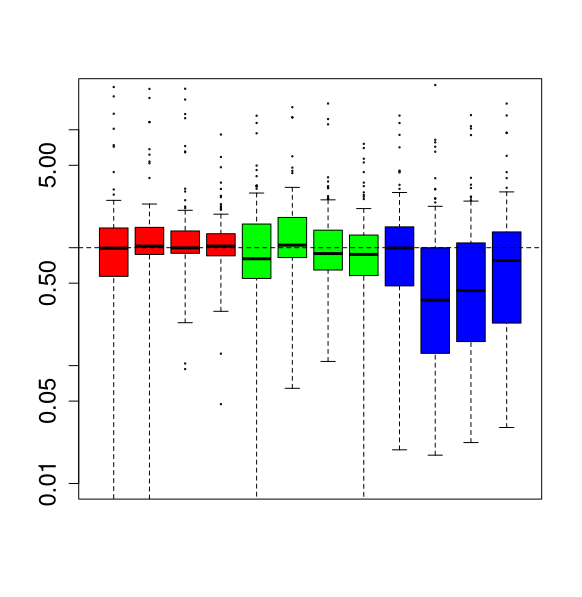} 
\put(32,92){\tiny $R(\eta_{(J,1)},0)$}
\end{overpic}
\begin{overpic}[ width=.32\linewidth,height=0.22\textheight ]{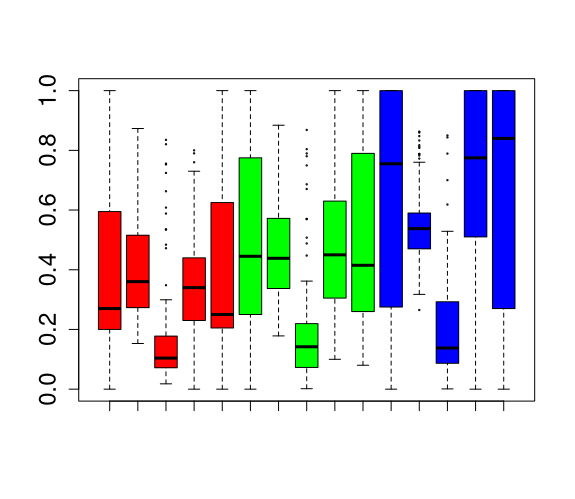} 
\put(32,92){\tiny $\eta$-estimates}
\put(16,5){\tiny a}
\put(21,5){\tiny b}
\put(25,5){\tiny c}
\put(30,5){\tiny d}
\put(35,5){\tiny e}
\end{overpic}
\begin{overpic}[ width=.32\linewidth,height=0.22\textheight ]{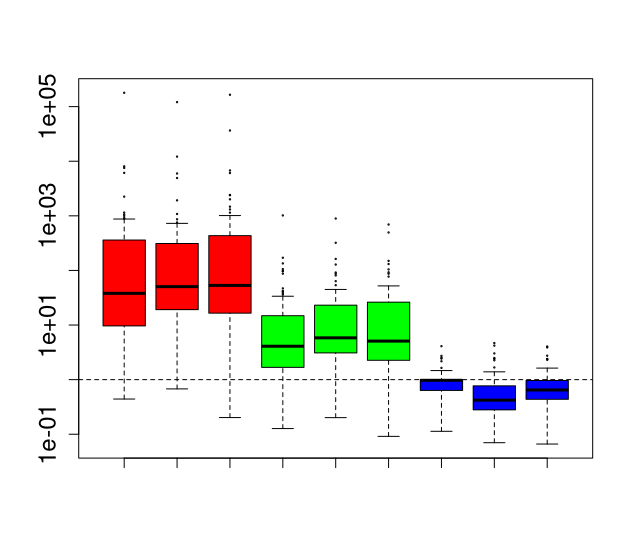} 
\put(32,92){\tiny $R(\eta_{(1,J)},1)$}
\put(18,5){\tiny a}
\put(28,5){\tiny b}
\put(35,5){\tiny c}
\end{overpic}
\begin{overpic}[ width=.32\linewidth,height=0.22\textheight ]
{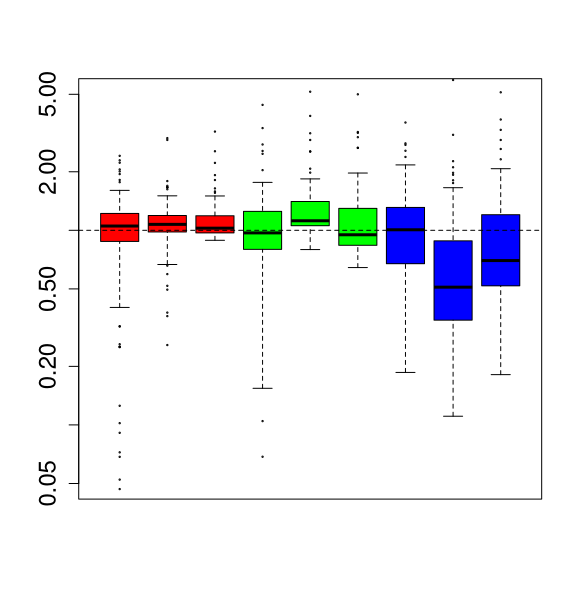} 
\put(32,92){\tiny $R(\eta_{(1,J)},0)$}
\end{overpic}
\begin{overpic}[ width=.32\linewidth,height=0.22\textheight ]{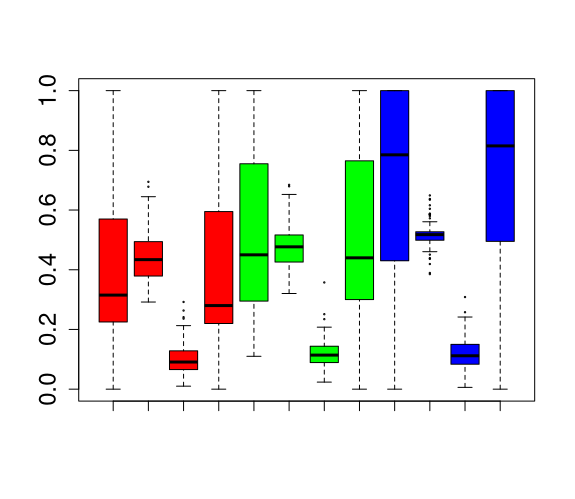}  
\put(32,92){\tiny $\eta$-estimates}
\put(17,5){\tiny a}
\put(23,5){\tiny b}
\put(28,5){\tiny c}
\put(34,5){\tiny e}
\end{overpic}
\caption{State-space model, Section~\ref{sec: post pred ssm}: Distributions of expected risk ratio \eqref{eq:s_relative} and $\eta$ estimates for the three misspecification levels; $\varphi^*_M=0.5$ (red, high), $\varphi^*_M=0.7$ (green, medium), and $\varphi^*_M=1$ (blue, no misspecification). Rows show comparisons against estimates of $\eta$ computed using product (top row) and pooled (bottom row) $\eta$-posteriors. For each misspecification level, four $\eta$ estimators, (a) $\overline{\eta}$ (posterior mode), (b) $\widehat{\eta}$ (mean), (c) $\widehat{\eta}_{hm}$  (\ref{eq:harmonic}) and (d) $\eta_{\rm \scriptscriptstyle WAIC}$ are compared against Bayes ($\eta=1$, left column) and cut ($\eta=0$, middle column) and (e) a high precision estimate of $\widetilde{\eta}_\x$ (in each color group, estimated using $J=10^3$ test blocks, right column). We use $\eta_{(1,J)}$ and $\eta_{(J,1)}$ to stand for these four estimators in the pooled and product cases respectively.
}
\label{fig:ssm_pred}
\end{figure}

When there is misspecification (red and green), the expected risk ratios in the left and centre panels of Figure~\ref{fig:ssm_pred} are distributed above one, so the proposed estimators give better predictions of test data than the Bayes or Cut model belief updates. In each case, the posterior mean gives a slight improvement (slightly higher) over the loss minimiser (which is the posterior mode here) and harmonic mean. When there is no misspecification (blue), standard Bayesian inference  is preferred (the blue boxes are under the dotted horizontal line in the left column of Figure~\ref{fig:ssm_pred}).

In the right column of Figure~\ref{fig:ssm_pred} we plot the distributions of the three $\eta$-estimates and the distribution of a high precision estimate of the risk minimiser $\widetilde\eta_\x = \argmin_{\eta} l(\eta;\z,\x)$ for pooled (top) and product (bottom) cases (so $\widetilde\eta_\x$ can be $\tilde{\eta}_{(1,\infty)}$ or $\tilde{\eta}_{(\infty,1)}$ depending on the chosen risk). The medians (of $\widetilde{\eta}$ for the two losses) over 30 replicates with $J^{(z)}=1000$ samples tend to increase as the level of misspecification decreases. The $\eta$ estimators show similar behaviour to the normal mixture example, though less clearly. The harmonic mean tends to underestimate, the posterior mean is shrunk to central values and the mode shows the greatest inter-quartile range. 
In Figure~\ref{fig:ssm_pred}, $\widetilde\eta_\x$ takes $\x$ as the training data because that is the optimal value targeted by the tested estimator. As explained in Section~\ref{sec:hyperparam-estimation-and-R}, we would really like to estimate $\widetilde\eta_{\overline{\x}}=\argmin_{\eta} l(\eta;\z,\overline{\x})$, taking $\overline{\x}=(\x,\y)$, the full data set before splitting, as the training data. In Figure~\ref{fig:diff_eta} we repeat the measurements of Figure~\ref{fig:ssm_pred} but plot differences, including $\widetilde\eta_{\x}-\widetilde\eta_{\overline{x}}$, between estimators and the full-data optimal $\eta$. Results support the assumptions in Section~\ref{sec:hyperparam-estimation-and-R}. See Appendix~\ref{Appendix:risk-ratios-state-space} for details.

The WAIC estimator behaves similarly to the product mode and mean estimators in the top (product-loss) row. It estimates the maximum of the ELPPD, so this is expected.

\subsubsection{Asymptotics and risk for $(\eta,\beta)$-SMI}\label{sec: eta-beta ssm}

One of the strengths of using Bayesian inference with calibration data is that we jointly estimate the learning rate and loss hyperparameters (e.g., $\beta=1/b$), rather than fixing the learning rate at 1 as is commonly done when using generalized losses beyond the likelihood.
In this experiment, we illustrate this using $(\eta,\beta)$-SMI \eqref{eq: eta beta SMI SSM}.

\begin{figure}[ht]
    \centering
    \begin{overpic}[ width=.32\linewidth,height=0.25\textheight ]{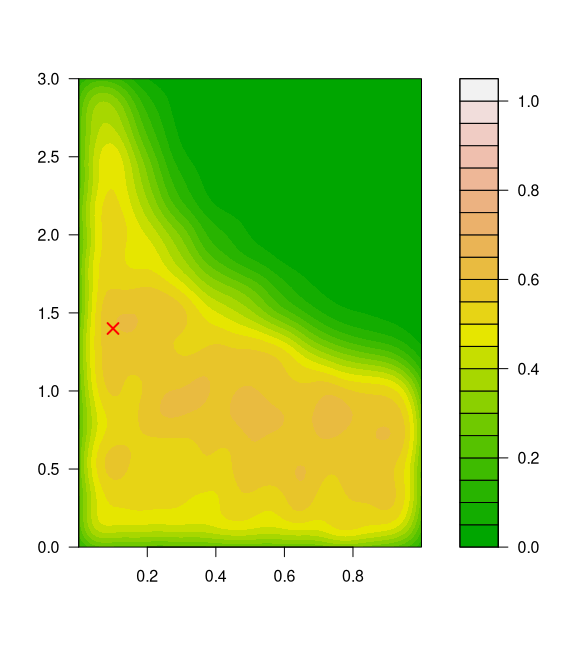}  
    \put(25,90){\tiny $J=10$}
    \end{overpic}
     \begin{overpic}[ width=.32\linewidth,height=0.25\textheight ]{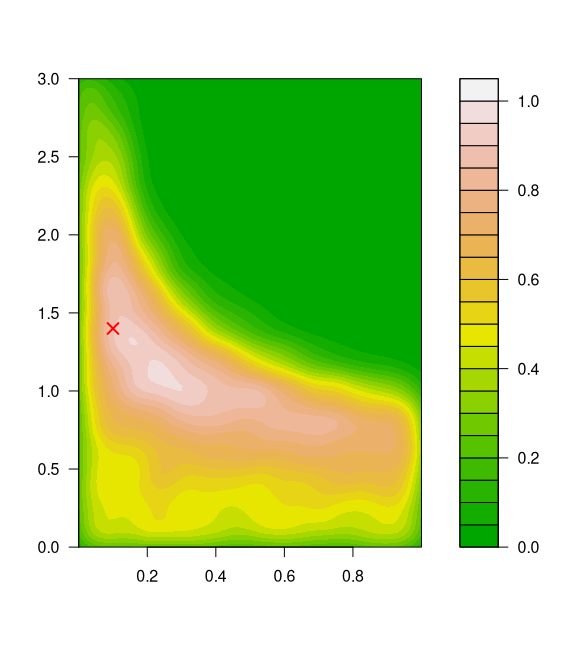} 
     \put(25,90){\tiny $J=100$}
    \end{overpic}
    \begin{overpic}[ width=.32\linewidth,height=0.25\textheight ]{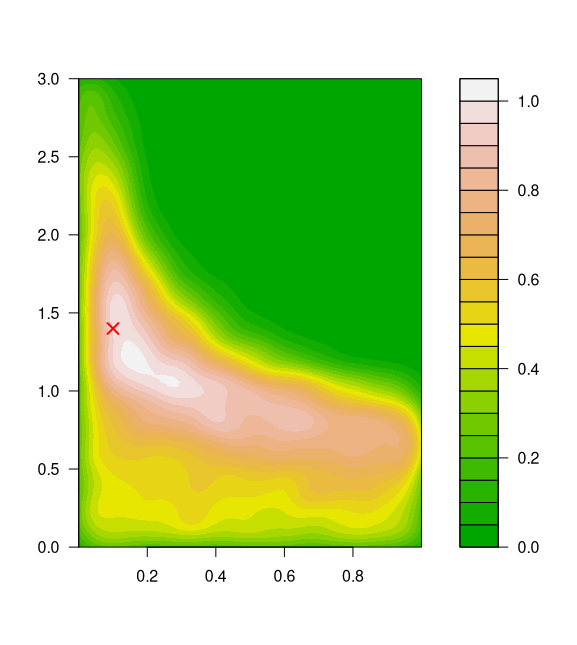} 
    \put(25,90){\tiny $J=1000$}
    \end{overpic}

    \vspace*{-0.2in}

    \begin{overpic}[ width=.32\linewidth,height=0.25\textheight ]{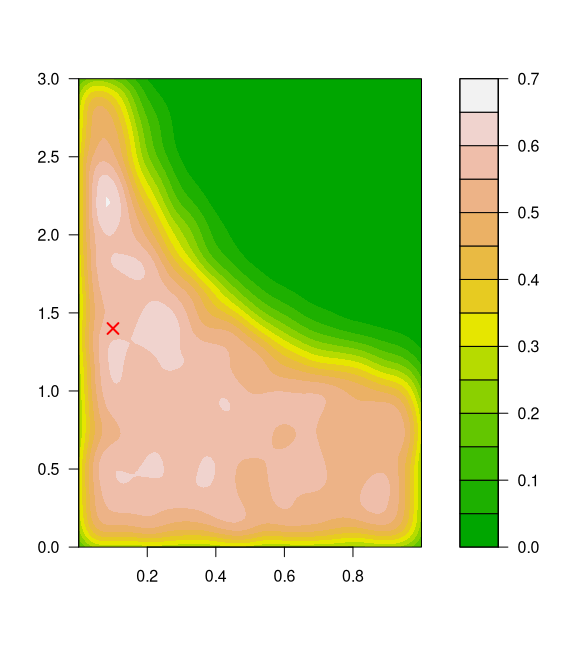}
    \put(25,90){\tiny $J=10$}
    \end{overpic}
     \begin{overpic}[ width=.32\linewidth,height=0.25\textheight ]{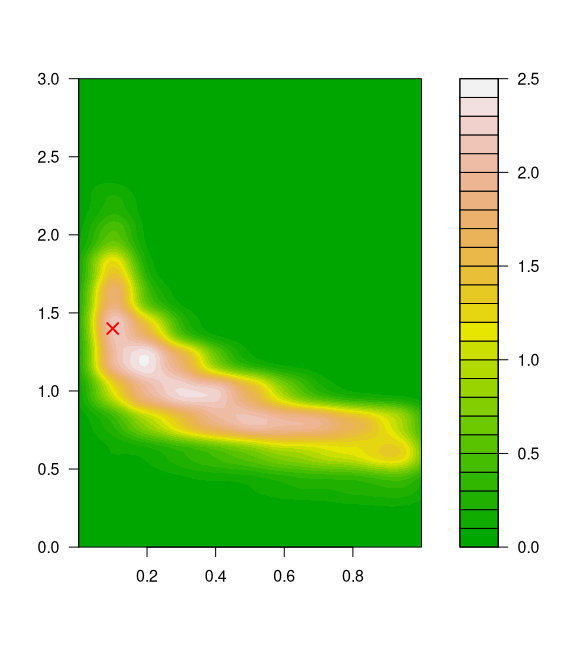}  
     \put(25,90){\tiny $J=100$}
    \end{overpic}
    \begin{overpic}[ width=.32\linewidth,height=0.25\textheight ]{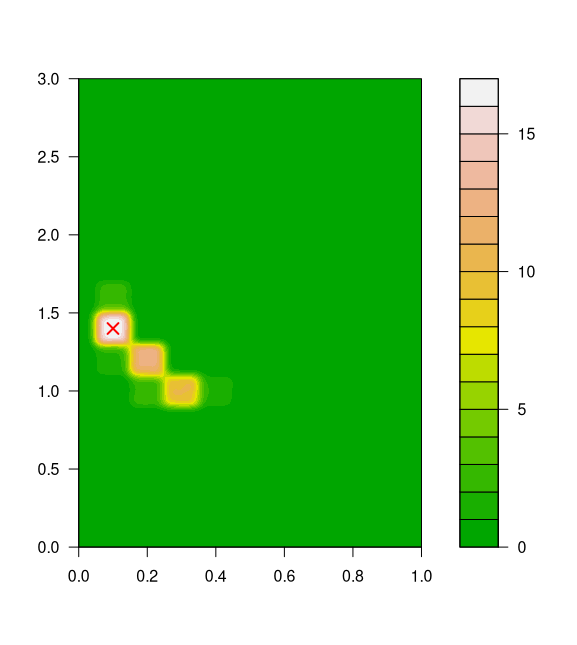}  
    \put(25,90){\tiny $J=1000$}
    \end{overpic}
        \caption{State-space model, Section~\ref{sec: eta-beta ssm}: Joint posterior of $(\eta,b)$ at $\varphi^*_M=0.7$; $b$-values ($y$-axis) against $\eta$ values ($x$-axis); posterior density, indicated by color bar, for pooled (top row) and product calibration (bottom row). The red cross marks indicate $\widetilde{s}$.
         }
    \label{fig:eta b joint}
\end{figure}

We take a single training dataset with $n=10$, scale the calibration data as $J=10,10^2,10^3$ and make a high-precision estimate of $\widetilde{s}$ using 
test data with $J^{(z)}=10^4$ blocks. The joint posteriors $\rho(s|\y,\x)$ for $s=(\eta,b)$ in Figure~\ref{fig:eta b joint} show convergence to a diffuse distribution and convergence to a point in the pooled and product cases, matching the theory in Section~\ref{sec:bayes_s} and Appendix~\ref{Appendix:zero_constraint}, which holds for $s$ of any finite dimension. The posteriors are evaluated at a grid of $(\eta,b)$ values then smoothed. This is impractical when the dimension of $s$ is large. The MCMC in Appendix~\ref{Appendix:mcmc} targeting $\rho(s,\phi|\y;\x)$ may help in the pooled case and for the product-loss at moderate $J$. For example, in Figure~\ref{fig:nested mcmc check}, agreement between smoothing and MCMC is good at $J=40$.

In Figure~\ref{fig:eta b ratio} in Appendix~\ref{Appendix:state-space} we compare the risk for $(\eta,\beta)$-SMI to $\eta$-SMI, $\beta$-SMI, Bayes and Cut for estimation using the pooled loss.
We find $(\eta,\beta)$-SMI has lower risk than Bayes, Cut and $\beta$-SMI and almost the same as $\eta$-SMI. This is interesting because work with the $\beta$-loss (eg \cite{jewson24}) often sets $\eta=1$ in the Gibbs posterior. 



\subsection{Estimation of Sense Change using the Power Posterior}\label{sec:edisc}

In this section we illustrate Bayesian estimation of $\eta$ in a single-module power posterior of the form \eqref{eq:GB-post-generic} using the EDiSC model of \cite{zafar2024exploring}. The data, summarised in Table~\ref{tab:EDiSC-experiment-settings} in Appendix~\ref{Appendix:EDiSC}, come from an English text corpus from 1810 to 2010 described in \cite{zafar2024embedded}. The aim is to cluster occurrences of a target word (for example, ``bug'') according to their meaning whilst allowing for meanings to change over time. Following \cite{frermann2016bayesian}, each occurrence of the target word contributes a short ``snippet'' made up of surrounding context words; different senses have different context-word distributions. For example, ``bug'' has four main meanings and appears in $N=522$ snippets made up of words from a vocabulary $\V=\{1,\dots,V\}$ with $V=2475$ words. We analyze the same ten datasets as \cite{zafar2024exploring}: five approximately equal subsets of ``bank'' and datasets for ``chair'', ``apple'', ``gay'', ``mouse'' and ``bug''. 

The statistical task is to carry out unsupervised clustering of snippets by meaning. In all the snippets in a cluster, the target word is used in the same sense. The EDiSC model does this by modelling the sense-dependent co-occurence of context words in snippets. It gives a posterior distribution over partitions with a user-specified number of senses $K$. EDiSC is a bag-of-words model resembling Latent Dirichlet Allocation. For each occurrence $i$ of the target word, the snippet $x_i\subset\V$ contains the $L$ surrounding context words, with $\V=\cup_i x_i$ and $\x=(x_1,\ldots,x_n)$. The latent variable $c_i\in\{1,\ldots,K\}$ gives the sense assignment of snippet $i$. Further model details, including $K,n,V,L$ and the power posterior \eqref{eq:EDiSC-gb-posterior}, are given in Appendix~\ref{Appendix:EDiSC}.

This is a challenging problem with sparse, temporally structured data, many latent variables and a moderately high-dimensional, misspecified model. \cite{zafar2024exploring} showed that a power posterior improves sense-label prediction over Bayes, with $\eta$ selected by posterior-predictive matching. In contrast, we estimate $\eta$ using Bayesian inference based on held-out calibration data.

\begin{figure}
    \centering
    \hbox{
    \begin{overpic}
    [width=.45\linewidth,height=0.25\textheight]{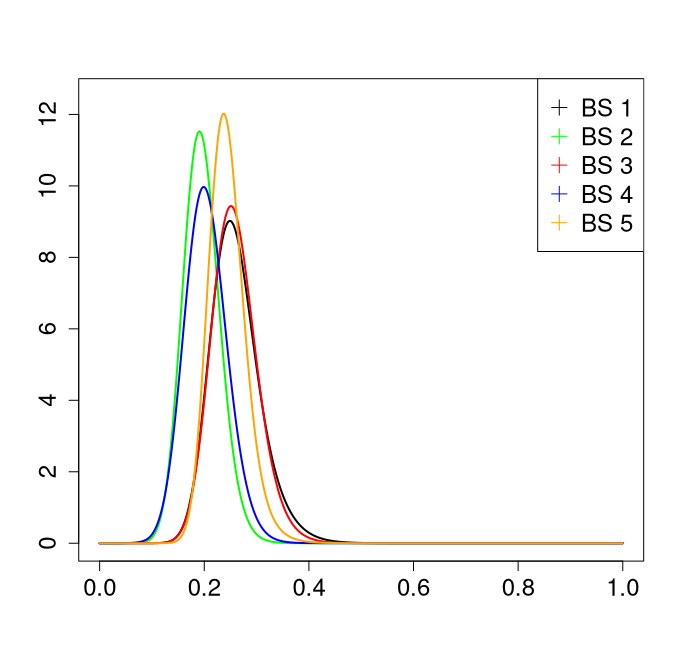}
    \put(-5,25){\rotatebox{90}{\small $\rho_{(J,1)}(\eta|\y, \x)$}}
        \put(50,2){\small $\eta$}
    \end{overpic}
    \hspace{0.2cm}
    \begin{overpic}
    [width=.45\linewidth,height=0.25\textheight]{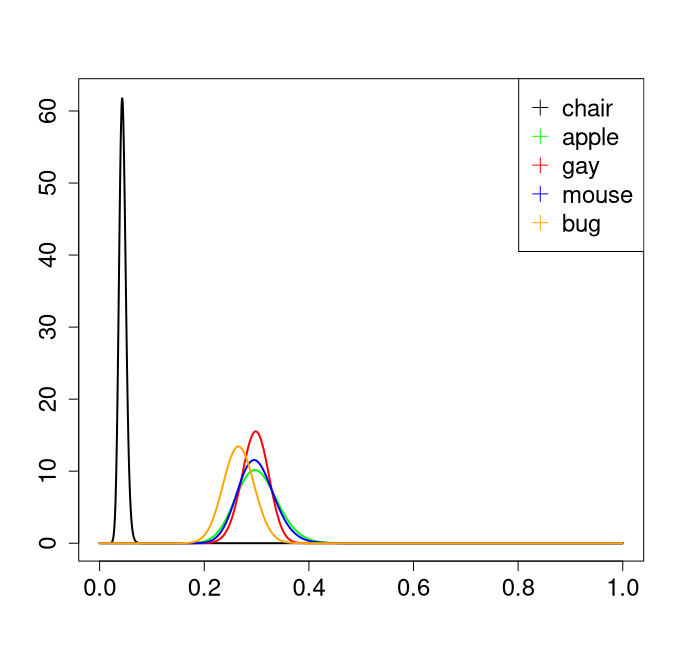}
    \put(-5,25){\rotatebox{90}{\small $\rho_{(J,1)}(\eta|\y, \x)$}}
        \put(50,2){\small $\eta$}
    \end{overpic}
    }
    \caption{Posterior densities for $\eta$ for five subsets of data all with target word ``bank'' (left, BS1-BS5) and five other target words (right). Posteriors are calculated using a third of the data for calibration, the product posterior and smoothing (Appendix~\ref{Appendix:mcmc}).}
    \label{fig:bank split}
\end{figure}

Figure~\ref{fig:bank split} shows the product posteriors for $\eta$, calculated using a $2{:}1$ train-to-calibration split with no test data. The five bank-posteriors are similar, with means around $0.25$, while most other target words favour slightly larger values with the exception of ``chair''. The posteriors are approximately normal, consistent with Theorem~\ref{thm:product}. When the calibration data is reduced (in Figure~\ref{fig:bank split sixth}), the posteriors become more dispersed and skewed but remain similarly located. Table~\ref{tab:edisc-R-bank} shows that the resulting estimates give larger posterior-predictive probabilities on held-out test data than Bayes or the values selected by \cite{zafar2024exploring}. 
Figure~\ref{fig:bank1_split_results} (left) reports a wider sensitivity analysis to the training/calibration split on the word ``bank''. For train/calibrate-splits ranging from $1{:}4$ to $4{:}1$ the location of the \(\eta\)-posterior mean varies in the range $\hat\eta=0.225$ to $0.325$ across splits; at $1{:}5$ and $5{:}1$ the peak varies from $0.2$ to $0.4$. This degree of variation doesn't greatly influence downstream analysis: for example, posterior expected conditional Brier scores are stable across splits in Figure~\ref{fig:bank1_split_results} (right) and plugging in the posterior mean $\hat\eta$ gives very similar results (Appendix~\ref{Appendix:EDiSC-pooled-loss}).

We can also use the true sense labels denoted by $o_i\in\{1,2,\dots,K\}$ for snippet $i=1,\dots,N$ to measure success. Let $\hat{p}_\eta(\c_i = k)$ be an estimate of the posterior probability $\pi_\eta(\c_i=k|\overline{\x})$ that snippet $i$ is in cluster $k$ (using MCMC samples
). Let $K^*$ give the true number of senses. Following \citet[Section 5.1]{zafar2024embedded}, we assess predictive accuracy under each learning rate $\eta$ using the Brier score:
\begin{equation}
\text{BS}_\eta = \frac{1}{N} \sum_{i=1}^{N} \sum_{k=1}^{K^*} \left( \hat{p}_\eta(\c_i = k) - I(o_i = k) \right)^2,\label{eq:edisc-brier-score}
\end{equation}
a proper scoring rule for multi-category probabilistic predictions $\hat{p}_\eta(\c_i = k)$, ranging from 0 (best) to 2 (worst). The expression in \eqref{eq:edisc-brier-score} has to be modified in cases where the number of fitted senses is greater than the number of labeled senses so $K>K^*$ (any given cluster may have more than one label).

Table~\ref{table: BS} compares Brier scores with those of \cite{zafar2024exploring}. Missing entries indicate ``collapse'' with uninterpretable clusters as illustrated in Table~\ref{tab:context_words}.
 We run the analysis at different $\eta$-values and report the \emph{a posteriori} most probable context-words in each cluster. The improvement in the clarity of the grouping of context-words by meaning is striking when we use the optimal $\eta$, so this shows the power-posterior can do a very good job of treating misspecification. At the Brier-optimal value $\eta=0.4$, the four clusters recover the four meanings ``insect'', ``disease'', ``software'' and ``wiretap''. By contrast, at $\eta=0.1$ one cluster has no clear interpretation and only the ``insect'' and ``software'' senses are recovered. Our method selects $\eta=0.3$, close, but not as good as $\eta=0.4$ (the ``wiretap'' sense is lost and ``software'' is repeated, as for $\eta=1$). 
Overall, the $BS_{\widehat{\eta}}\,$-values for our selected rates are comparable to those given in \cite{zafar2024exploring} (four are better, four are worse and two are tied).  

\cite{zafar2024exploring} estimate $\eta_{\rm \scriptscriptstyle WAIC}=1$ for all the words in Table~\ref{table: BS}. The WAIC depends on the effective parameter dimension and this is hard to estimate reliably, because the variance of the log-likelihood over snippets is large. This is a known issue for the WAIC \citep{Vehtari2016}. Although the WAIC is not a reliable estimator for the ELPPD here, the product loss works well, as it directly estimates the ELPPD, rather than relying on an approximation good at large $n$.


\renewcommand{\arraystretch}{0.95}
\begin{table}[t]
\centering
\begin{tabular}{c|c|c|c|c|c}
\toprule
$\eta$ 
& bank split 1
& bank split 2
& bank split 3
& bank split 4
& bank split 5 \\
\midrule
1   & 0.135 & 0.115 & 0.106 & 0.144 & 0.189 \\
0.6 & 0.128 & 0.111 & 0.103 & 0.139 & 0.178 \\
0.5 & 0.125 & 0.110 & 0.103 & 0.138 & 0.173 \\
0.4 & \fcolorbox{blue}{white}{0.124} & \fcolorbox{blue}{white}{0.109} & \fcolorbox{blue}{white}{\textcolor{red}{0.103}} & \fcolorbox{blue}{white}{0.138} & \fcolorbox{blue}{white}{0.168} \\
0.3 & \fcolorbox{red}{white}{0.121} & \textcolor{red}{0.108} & \fcolorbox{red}{white}{0.104} & 0.137 & 0.162 \\
0.2 & \textcolor{red}{0.118} & \fcolorbox{red}{white}{0.109} & 0.106 & \fcolorbox{red}{white}{\textcolor{red}{0.137}} &  \fcolorbox{red}{white}{0.158}\\
0.1 & 0.122 & 0.113 & 0.113 & 0.141 & \textcolor{red}{0.152} \\
\midrule
$\eta$ & chair & apple & gay & mouse & bug \\
\midrule
1   & 0.140  & 0.0510 & 0.343 & 0.0474 & 0.300\\
0.6 & \fcolorbox{blue}{white}{0.133} & 0.0496 & 0.299 & 0.0382 & 0.268\\
0.5 & 0.129& 0.0488 & \fcolorbox{blue}{white}{0.291} & 0.0359 & \fcolorbox{blue}{white}{0.267} \\
0.4 & 0.126 & 0.0484  & \textcolor{red}{0.282}& 0.0348& \textcolor{red}{0.265} \\
0.3 & 0.124& \fcolorbox{red}{white}{0.0479}  & \fcolorbox{red}{white}{0.295}  & \fcolorbox{blue}{white}{\fcolorbox{red}{white}{\textcolor{red}{0.0335}}}  &  \fcolorbox{red}{white}{0.330}\\
0.2 & \textcolor{red}{0.120} & \fcolorbox{blue}{white}{0.0467}  &  & 0.0337 &  0.424 \\
0.1 & \fcolorbox{red}{white}{0.125}& \textcolor{red}{0.0461}&  & 0.0393  &  \\
\bottomrule
\end{tabular}
\caption{Brier scores $BS_\eta$ for candidate learning rates $\eta$ on test data. Optimal scores $BS_{\eta*}$ are in red text. Scores $BS_{\widehat{\eta}}$ for the learning rate $\widehat{\eta}$ selected using our method are boxed in red. Scores $BS_{\eta^\dagger}$ selected using the posterior-predictive check of \cite{zafar2024exploring} are boxed in blue. Missing values for $BS_\eta$ indicate collapse (see text).}
\label{table: BS}
\end{table}

\section{Discussion}

We have given a Bayesian framework for learning inference hyperparameters $s=(\eta,\beta)$ in
generalised Bayesian inference. The idea is to treat the
posterior-predictive $p_s(\y|\x)$ for held-out data $\y$ as a likelihood-like quantity for $s$, so that
uncertainty about the learning rate and loss hyperparameters can be quantified. This gives a coherent
route from calibration data to a posterior for $s$, and extends the usual one-dimensional learning-rate problem to joint inference on several hyperparameters. The real-data analyses in Section~\ref{sec:edisc}, which are single-module GBI, support our approach. In this sense-change application the data are sparse and the parameter space is high-dimensional (just over $5500$ for the word ``chair''). The power posterior improves held-out predictive performance (see Table~\ref{tab:edisc-R-bank}) with compatible Brier scores in Table~\ref{table: BS} and is stable across data splits for central ratios ranging from $1{:}3$ to $3{:}1$ in Figure~\ref{fig:bank1_split_results} (left).

The framework extends naturally to multi-modular models. 
In the normal-mixture example in Section~\ref{sim:normal},
$\eta$- and $\gamma$-SMI behave very similarly (see Figure~\ref{fig:smipost}), and this agreement becomes stronger as
the training size grows, consistent with Proposition~\ref{prop:eta-gamma-first-order} in Appendix~\ref{Appendix:gamma-is-eta}. This result holds up in the HPV example with real data in Appendix~\ref{sec:hpv}. In the state-space example in Section~\ref{sec:ssm}, when there is misspecification the learned $\eta$-SMI values improve prediction ($R_J$-values in \eqref{eq:s_relative}) and intrinsic Bayes-factors ($R_1$-values in \eqref{eq:s_relative_pooled}) for $\eta$-SMI relative to fixed Bayes and
Cut updates  (see Figure~\ref{fig:ssm_pred}), while in the well-specified case ordinary Bayes is correctly preferred. The new $(\eta,\beta)$-SMI separates two forms of robustness: $\eta$ controls feedback from a suspect module, while $\beta$ changes the loss to reduce sensitivity to outliers.
In the same example, joint calibration in $(\eta,\beta)$-SMI is clearly better than
Bayes, Cut and $\beta$-SMI, and is almost indistinguishable from $\eta$-SMI (see Figure~\ref{fig:eta b ratio}), so the
extra flexibility does not incur a substantial cost. 

The main limitations are computational cost at high hyperparameter dimension and information loss from splitting data. Computing
$p_s(y_j| \x)$ at a grid of $s$ values for high-dimensional $\phi$-parameter spaces is practical when $\dim(s)$ is one or two,
as we illustrate in Sections~\ref{sec:ssm} and \ref{sec:edisc}.
However, product-loss calibration becomes expensive when either $J$ or the dimension of $\phi$ increase, because repeated simulation from $\pi_s(\cdot| \x)$ is required. In Appendix~\ref{Appendix:mcmc} a joint/nested MCMC algorithm gives similar posterior estimates to smoothing (Figure~\ref{fig:nested mcmc check}) but we had to restrict $J\le 40$, and when we fit the $(\eta,\beta)$ model to multiple replicate data sets to get risk ratio distributions in Appendix~\ref{Appendix:state-space} we only did that for the pooled case. 
Also, data splitting must drop information, though our experiments suggest that the resulting
estimator generalises well to the posterior based on the combined data, and risk ratio distributions are fairly insensitive to the split ratio in our examples in Figure~\ref{fig:normal_b_gamma_eta_prod_pred}. The issue is discussed in Appendix~\ref{Appendix:normal4-TC-split}. 

There are also methodological limitations. First, we plug-in to $\pi_s(\phi|\x,\y)$ an estimate of $s$, so downstream posterior summaries are conditioned on a single $s$-value. This is standard. In Figure~\ref{fig:bank1_split_results} in Appendix~\ref{Appendix:EDiSC-pooled-loss} we take a downstream summary computed at fixed $s$ and average over $s$, illustrating uncertainty propagation. A more systematic treatment  
would propagate uncertainty using a ``model averaged'' belief update $\int_{\Omega_s}
\pi_s(\phi|\x)\rho(s|\y;\x)\,ds$. This is supported by decision theory and is also reparameterisation-invariant: replacing $s$ by $r=g(s)$ is just a change of variables in an integral.
It would be good to explore this new belief update in future work. Second, we assume a simple conditional independence structure holds in \eqref{eq:lkd-product}. The calibration data $y_j$ are $d_x$-component vectors with $d_x$ defined by identifying this independence structure. \cite{Cooper2023} considers models for timeseries where $d_x$ would have to be chosen with care, as \eqref{eq:lkd-product} doesn't hold, and this leads to the Expected Log Joint Predictive Density (ELJPD). Replacing the ELPPD with the ELJPD, and increasing $d_x$ to merge factors in \eqref{eq:lkd-product} would interpolate between the product and pooled losses.
 Third, the rapid concentration of the product-posterior (for example, Figures~\ref{fig:eta b joint} and \ref{fig:bank split}) is appealing. However, it may give an infinite learning rate when the observation model is overdispersed: this means the product loss favors prediction of test data using the plug-in MLE. Although this did not occur in any of our main examples, we saw it in a simple model in Appendix~\ref{app:gbi-normal-example-product}.  

Other natural next steps are to use more informative priors for $s$, especially $\eta$. In SMI, $\eta=0$ (the Cut-model) is often optimal, so it would be natural to put an atom there and use reversible jump MCMC. In our own ongoing work on this topic we reduce computational costs using a variational approximation amortised over hyperparameters, following \cite{Carmona2022scalable}. This works well, even for the EDiSC posterior in Section~\ref{sec:edisc}. 
Alternative calibration losses tailored to task-specific
objectives such as classification \citep{battaglia2025} would also fit in our framework. 

\bibliographystyle{ba}
\bibliography{ref}

\clearpage

\appendix

\counterwithin{theorem}{section}
\counterwithin{assumption}{section}
\counterwithin{corollary}{section}
\renewcommand{\thetable}{\Alph{section}\arabic{table}}
\setcounter{table}{0}
\renewcommand{\thefigure}{\Alph{section}.\arabic{figure}}
\setcounter{figure}{0}
\renewcommand{\theequation}{\Alph{section}.\arabic{equation}}
\setcounter{equation}{0}

\section{Sampling hyperparameter-posteriors}\label{Appendix:mcmc}

We propose two methods to sample from $\rho(s|\y,\x)$. We illustrate this on the product-posterior in \eqref{eq:GB-hyper-posterior}. The pooled case is actually more straightforward. 
 
{\bf Smoothing} This works well when the dimension of $s$ is low (one or two, which is actually typical). For the first method, we evaluate $p_s(\y|\x),\ \y\in \R^J$ or $p_s(y|\x),\ y\in \R$ on a lattice of $s$ values and use a smoothing spline to evaluate them at $s$-values off lattice. 
Marginalisation over $\phi$-parameters in $p_s$ and $\tilde p_s(\y|\x)=\prod_j p_s(y_j|\x)$ is approximated numerically or using Monte-Carlo. For example, we use Gaussian quadrature in \ref{Appendix:computing-gamma-SMI}. 

{\bf MCMC} This still works if the dimension of $s$ is larger, though it is slow for the product-loss. We target $\rho$ using a nested MCMC algorithm. We first illustrate how the algorithm works for the product loss, i.e. we target $\rho_{(J,1)}(s,\phi_{1:J}|\y,\x)$.
Dropping the $(J,1)$ subscript, the joint distribution of $(s,\phi_{1:J})$ in the product posterior is
\begin{align}\label{eq:s phi joint prod}
\rho(s,\phi_{1:J}|\y;\x) &\propto
    \rho(s) \prod_{j=1}^J p(y_{j}|\phi_j) \frac{\pi(\phi_j)p_s(\x|\phi_j)}{p_s(\x)},
\end{align}
Let $q(s'|s)$ be a proposal distribution for $s$. The following algorithm is ergodic for $\rho(s,\phi_{1:J}|\y;\x)$ in the limit that $I$ is taken to infinity and otherwise approximate.

\begin{enumerate}[topsep=0pt,itemsep=0ex]
    \item[] Let $X_t=(s,\phi_{1:J})$ with $\phi_{1:J}=(\phi_1,\dots,\phi_J)$.
    \item Sample $s'\sim q(\cdot|s)$;
    \item For $j=1,\dots,J$,
    run an MCMC chain $\widetilde\Phi_{j,i},\ i=0,1,\dots I$ with $\Phi_{j,0}=\phi_j$ targeting 
    \[\pi_{s'}(\phi'|\x)=\frac{\pi(\phi')p_{s'}(\x|\phi')}{p_{s'}(\x)}\]
    and set $\phi_j'=\widetilde\Phi_{j,I}$ (with $I$ large enough so $\phi_j'\sim \pi_{s'}(\cdot|\x)$ to a good approximation);
    \item Accept $\phi'_j$ with probability
    \begin{align}
        \alpha(s',\phi_1',\dots,\phi_J'|s,\phi)&=\min\left\{1,\frac{\rho(s',\phi_1',\dots,\phi_J'|\y;\x)q(s|s')\pi_{s}(\phi_1|\x)\cdot \dots \cdot \pi_s(\phi_J|\x)}{\rho(s,\phi_1,\dots,\phi_J|\y;\x)q(s'|s)\pi_{s'}(\phi'|\x)\dots \pi_{s'}(\phi_J'|\x)}\right\} \nonumber\\ 
        & \qquad\mbox{(use Equation~\ref{eq:s phi joint prod})}\nonumber\\
        &=\min\left\{1,\frac{\rho(s')q(s|s')\prod_{j=1}^J p(y_{j}|\phi_j')}{\rho(s)q(s'|s)\prod_{j=1}^J p(y_{j}|\phi_j)}\right\}\label{eq:ap-mcmc-prod}
    \end{align}
    set $X_{t+1}=(s',\phi_1',\dots,\phi_J')$ and otherwise set $X_{t+1}=(s,\phi_1,\dots,\phi_J)$.
\end{enumerate}
The version of this for the pooled loss with posterior
\begin{align*}\label{eq:s phi joint pool}
\rho(s,\phi|\y;\x) &\propto
    \rho(s) p(\y|\phi) \frac{\pi(\phi)p_s(\x|\phi)}{p_s(\x)},
\end{align*}
is similar but the state is $X_t=(s,\phi)$, it just needs one $\phi'\sim \pi_s(\phi'|\x)$ at step 2 and the final acceptance probability has $p(\y|\phi')/p(\y|\phi)$ replacing $p(y_j|\phi'_j)/p(y_j|\phi_j)$ in \eqref{eq:ap-mcmc-prod}.

\begin{figure}[!t]
    \centering
 \begin{overpic}
    [ width=.3\linewidth,height=0.2\textheight ]{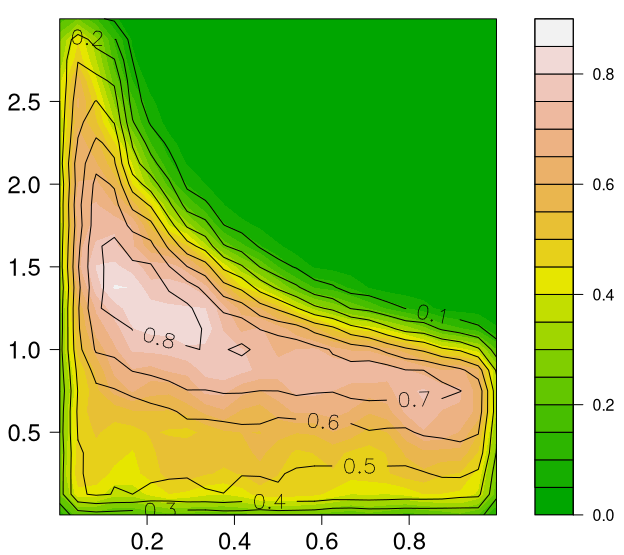} 
    \put(37,97){\tiny $J=10$}
    \put(-8,30)
    {\rotatebox{90}
    {\tiny $\rho_{(1,J)}(s|\y,\x)$}}
\end{overpic}
\hspace{0.3cm}
\begin{overpic}
    [ width=.3\linewidth,height=0.2\textheight ]{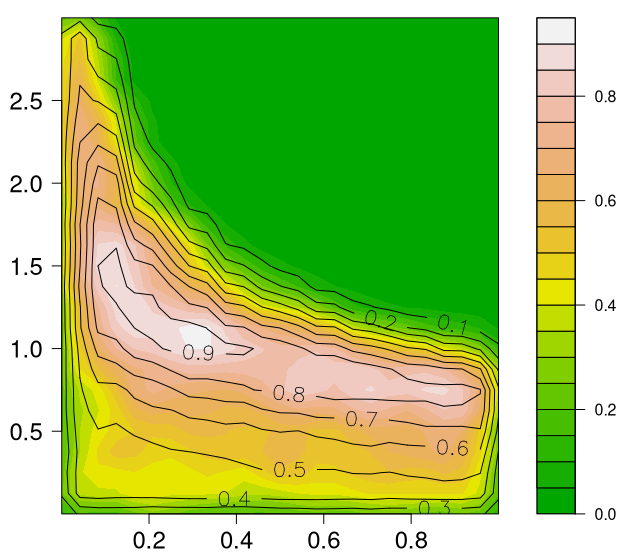} 
    \put(37,97){\tiny $J=20$}\put(-8,30)
    {\rotatebox{90}
    {\tiny $\rho_{(1,J)}(s|\y,\x)$}}
\end{overpic}
\hspace{0.3cm}
\begin{overpic}
    [ width=.3\linewidth,height=0.2\textheight ]{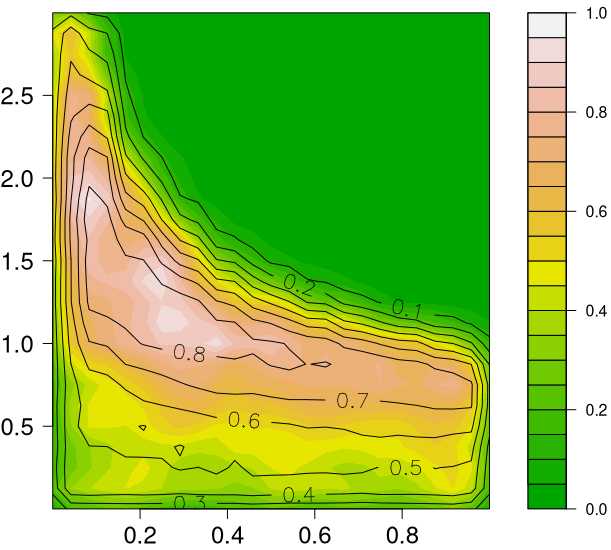} 
    \put(37,97){\tiny $J=40$}\put(-8,30)
    {\rotatebox{90}
    {\tiny $\rho_{(1,J)}(s|\y,\x)$}}
\end{overpic}

\vspace{1cm}
\begin{overpic}
    [ width=.3\linewidth,height=0.2\textheight ]{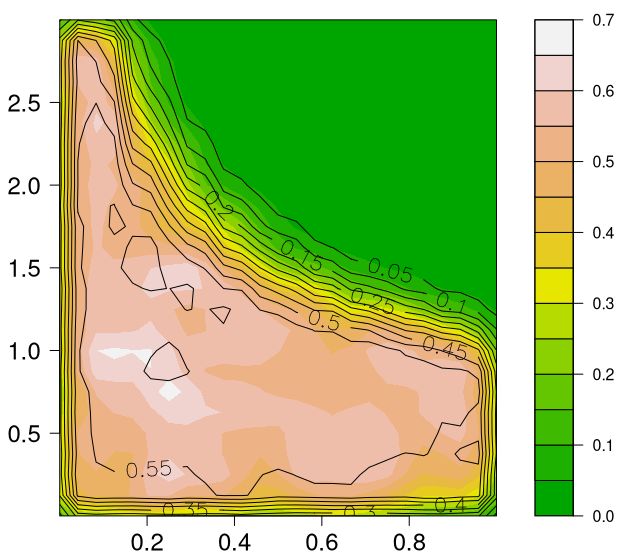} 
    \put(37,97){\tiny $J=10$}
    \put(-8,30)
    {\rotatebox{90}
    {\tiny $\rho_{(J,1)}(s|\y,\x)$}}
\end{overpic}
\hspace{0.3cm}
\begin{overpic}
    [ width=.3\linewidth,height=0.2\textheight ]{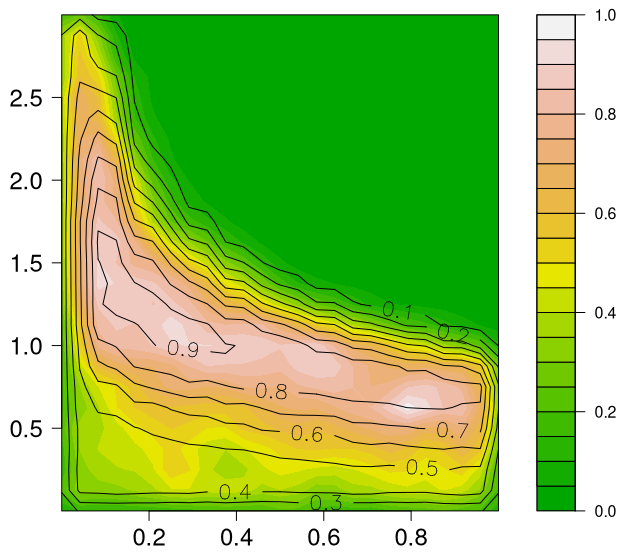} 
    \put(37,97){\tiny $J=20$}
    \put(-8,30)
    {\rotatebox{90}
    {\tiny $\rho_{(J,1)}(s|\y,\x)$}}
\end{overpic}
\hspace{0.3cm}
\begin{overpic}
    [ width=.3\linewidth,height=0.2\textheight ]{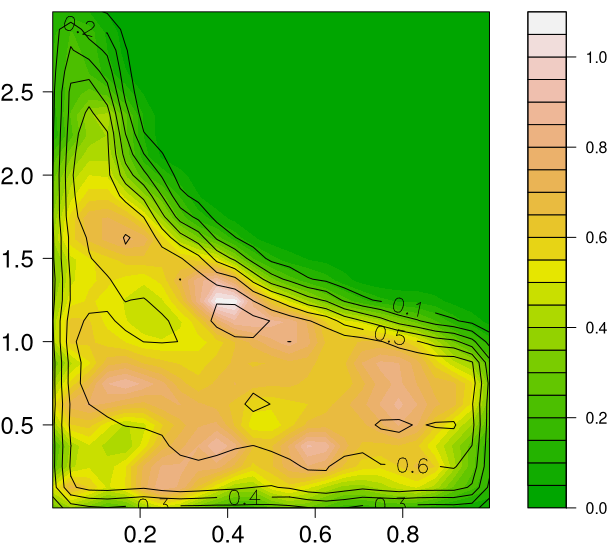} 
    \put(37,97){\tiny $J=40$}
    \put(-8,30)
    {\rotatebox{90}
    {\tiny $\rho_{(J,1)}(s|\y,\x)$}}
\end{overpic}
    \caption{
    Contour plots for $(\eta,\beta)$-SMI illustrating the smoothing and MCMC methods for approximating $\rho(s|\y,\x)$ for the SSM example in Section~\ref{sec:ssm}. The colour-level corresponds to the posterior density estimated using nested MCMC simulation, and the black contour lines correspond to the smoothing method. Training data is fixed and $n=10$. The calibration data has $J=10$ (first column), $J=20$ (second column), and $J=40$ (third column) respectively.
    The top row is pooled loss and the bottom row is product loss.
    For $J=10$, the pooled version is generated using 60,000 MCMC samples (without thinning) and product version is using 40,000 MCMC samples. For $J=20$, the pooled version is generated using 40,000 MCMC samples and product version is using 37,000 MCMC samples. For $J=40$, the pooled version is generated using 50,000 MCMC samples and product version is using 12,000 MCMC samples.
    The side chain length is $I=200$. }
    \label{fig:nested mcmc check}
\end{figure}

We compare the $s$ posterior generated using the two methods (smoothing and MCMC) for the State-Space model example from Section~\ref{sec:ssm} in Figure~\ref{fig:nested mcmc check}. We work on $(\eta,\beta)$-SMI with $n=10$ and $\sigma_{y,M}^*=0.7$.
The number of calibration data samples vary across columns as $J=10,20,40$, respectively. 
The color-map density obtained using MCMC and the black contour line from smoothing method agree well. The irregular structure bottom right is not a Monte-Carlo artifact as the agreement between the two estimators confirms. When we made the experimental verification of concentration for the product loss posterior, Figure~\ref{fig:eta b joint} bottom right, in Section~\ref{sec:ssm} we used the smoothing method rather than MCMC as we wished to take $J=1000$ which is impractical for MCMC specifically on the product loss due the $\phi_{1:J}$ sampling required.

\section{Example $s$-posteriors in Section~\ref{sec:simple-first-example}}\label{app:simple-first-example}

We refer the reader to the setup in Section~\ref{sec:simple-first-example} and \ref{blocksize}, picking things up following the definition of $\kappa$ with $\pi_\kappa(\phi|\x)=N\{\phi;\kappa\bar x,1-\kappa\}$; varying
$\eta>0$ is equivalent to varying $\kappa\in(0,1)$. We treat optima $\bar\kappa_{(J,1),\,\x},\ \widetilde \kappa_{(J,1),\,\x}\in \{0,1\}$ as boundary limits. Let
\[\textstyle
\bar x = \frac1n \sum_{i=1}^n x_i,
\qquad
\bar y = \frac1J \sum_{j=1}^J y_j,
\qquad
s_y = \frac1J \sum_{j=1}^J (y_j-\bar y)^2.
\]
The loss minimiser, $\bar\kappa_{(1,J),\,\x}=\argmax_\kappa p_\kappa(\y|\x)$, for the pooled loss in \eqref{eq:first-example-pool-loss} converges to the risk minimiser $\widetilde\kappa_{(1,\infty),\,\x}=\argmax_\kappa\pi_\kappa(\widetilde{\phi}|\x)$ for the pooled risk in \eqref{eq:first-example-pool-risk},
which is
\[
\widetilde\kappa_{(1,\infty),\,\x}
=
\max\left\{0,
\frac{
1+2\bar x^2
-
\sqrt{1+4\bar x^2(\mu^*-\bar x)^2}
}{
2\bar x^2
}.
\right\}.
\]

This limiting $\kappa$ is what we called $\widetilde{s}=\widetilde{s}(\pstar,\pi,p,\x)$ at the start of Section~\ref{sec:bayes-for-s-main-idea}; it is the ``true'' hyperparameter value given the models and data $\x$ if we evaluate using the pooled risk. It is almost surely less than one (which means that $\widetilde{\eta}_{(1,\infty),\,\x}$ is finite) and, if $(\mu^*-\bar x)^2<1+\bar x^2$, greater than zero. If $(\mu^*-\bar x)^2\ge 1+\bar x^2$ then $\bar y$ will conflict with $\bar x$ (as $\bar y\to \mu^*$ at large $J$). Evidence against updating belief using $\x$ accumulates to the point where we set $\kappa=0$ and discard $\x$. 

The product loss sums over $J$ copies of the pooled loss with $J=1$ in each so
\[
-\frac 1J\sum_{j=1}^J \log(p_\kappa(y_j|\x))
=
\frac 12 \log(\sigma^2+1-\kappa)
+
\frac{s_y+(\bar y-\kappa\bar x)^2}{2(\sigma^2+1-\kappa)}+c'_J
\]
up to additive constants independent of $\kappa$. 
The product risk is the negative ELPPD,
\begin{align*}
\mathbb E_{z\sim \pstar}\big[-\log p_\kappa(z\mid \x)\big]
=
\frac12\log(\sigma^2+1-\kappa)
+
\frac{1+(\mu^*-\kappa\bar x)^2}{2(\sigma^2+1-\kappa)} + c'.
\end{align*}
As $\bar y\to \mu^*$ and $s_y\to 1$,
the scaled loss converges to the risk. Separate calculations show that $\bar\kappa_{(J,1),\,\x}\to \widetilde\kappa_{(\infty,1),\,\x}$ at large $J$ with $\widetilde\kappa_{(\infty,1),\,\x}= \min\left\{1,\max\left\{0,\kappa_0\right\}\right\}$ and, if $\bar x\neq 0$,
\[
\kappa_0
=
\frac{
1+2(\sigma^2+1)\bar x^2
-
\sqrt{
1+4\bar x^2\Big(1+\big(\mu^*-(\sigma^2+1)\bar x\big)^2\Big)
}
}{
2\bar x^2
}.
\]
Let $\Delta=\bar x-\mu^*$. Conditions for $0<\widetilde\kappa_{(\infty,1),\,\x}<1$ are
\begin{align}
\widetilde\kappa_{(\infty,1),\,\x}<1 \ &\Leftrightarrow\  
c_x:=\frac{\sigma^2-2\sigma^2\Delta\bar x-1-\Delta^2}{2\sigma^4}<0,\label{eq:simple-first-example-fin}\\
\widetilde\kappa_{(\infty,1),\,\x}>0 \ &\Leftrightarrow\ \mu^{*2}<\sigma^2+2(\sigma^2+1)\mu^*\bar x.\label{eq:simple-first-example-gz}
\end{align}
Positivity in \eqref{eq:simple-first-example-gz} grows more probable with the number $n$ of samples in $\x$, as in the pooled case. If $\mu^*=0$ and $\sigma\le 1$ then $\widetilde{\eta}_{(\infty,1),\,\x}$ is almost surely finite in \eqref{eq:simple-first-example-fin} but if $\sigma> 1$ then $\widetilde{\eta}_{(\infty,1),\,\x}$ may be infinite for any $\mu^*$. 
For example, if $\sigma>1$ and $\mu^*=0$,
\[
\Pr\{\widetilde\eta_{(\infty,1),\,\x}=\infty\}
=
2\Phi\!\left(
\sqrt{\frac{n(\sigma^2-1)}{1+2\sigma^2}}
\right)-1.
\]

\section{Propositions~\ref{prop:ass1-reg-exp-fam} \& \ref{prop:ef-uniform-laplace} and Theorems~\ref{thm:loss_s_single}, \ref{thm:product} \& \ref{thm:post.boundary}}

\subsection{Proof of Proposition~\ref{prop:ass1-reg-exp-fam} and Proposition~\ref{prop:ef-uniform-laplace}}\label{app:ass1-reg-exp-fam}

{\bf Proposition~\ref{prop:ass1-reg-exp-fam}} (Regular exponential families) 
{\it
Suppose that, with respect to a base measure $\nu$, the observation model is
a full, regular, minimal exponential family in natural-parameter form,
\[
p(y|\phi)
=
\exp\{\phi^\top T(y)-A(\phi)\},
\qquad
\phi\in\Omega_\phi,
\]
where $\Omega_\phi$ is open. If
$\E_{Y\sim \pstar}\{\|T(Y)\|\}<\infty$ and $\pstar$ has a strictly positive
density with respect to $\nu$
then Assumptions~\ref{A:laplace}(ii - iv) hold and the likelihood part of (i). 
}
\begin{proof}
The following comes from \citet[Section~3.5]{wainwright2008graphical}. ``Minimal'' gives interior risk minimisers. Let $Q_\nu$ be the class of all probability densities with respect to $\nu$, let the mean-parameter space be
$\mathcal M=\left\{
\int T(y)q(y)\,\nu(dy): q\in Q
\right\}$, and let
\[
\tau^*
=
\E_{Y\sim \pstar}\{T(Y)\},\qquad\mbox{and} \qquad\widehat\tau_J
=
\textstyle \frac{1}{J}\sum_{j=1}^J T(y_j).
\]
First of all, $\tau^*\in\operatorname{Int}(\mathcal M)$: if $\tau^*$ is on the boundary of $\mathcal M$, a supporting
hyperplane would give a nonzero vector $a$ and a constant $c$ such that
$a^\top T(y)\leq c$ holds
$\nu$-almost everywhere,
with $a^\top\tau^*=c$; 
since the density of $\pstar$ is strictly positive, equality would then hold
$\nu$-almost everywhere, contradicting minimality.

Next, in a full, regular, minimal exponential family, $A(\phi)$ is $C^\infty$, 
$\nabla A(\phi)=\E_\phi\{T(Y)\}$,
and $\nabla A$ is a one-to-one map from $\Omega_\phi$ onto
$\operatorname{Int}(\mathcal M)$. Consequently,
the likelihood part of (i) holds, and $\widetilde\phi
=
(\nabla A)^{-1}(\tau^*)$
exists, is unique and maximises
\[
\phi^\top\tau^*-A(\phi)=\E_{Y\sim \pstar}\{\log p(Y|\phi)\}.
\]
This is the first part of (iv). For $\bar\phi_J$, the strong law gives
$\widehat\tau_J\xrightarrow{\mathrm{a.s.}}\tau^*$.
Because $\tau^*\in\operatorname{Int}(\mathcal M)$, we have 
$\widehat\tau_J\in\operatorname{Int}(\mathcal M)$ almost surely for sufficiently
large $J$ so the calibration-data MLE $\bar\phi_J
=
(\nabla A)^{-1}(\widehat\tau_J)$ is interior. This is  condition (ii).
Continuity of the inverse mean map gives 
$\bar\phi_J
\xrightarrow{\mathrm{a.s.}}\widetilde\phi$ completing condition (iv).
Moreover,
\[
H_{\bar\phi_J}
=
\operatorname{Var}_{\bar\phi_J}\{T(Y)\}=\nabla^2A(\bar\phi_J),
\]
which is finite and positive definite by regularity and minimality. This gives (iii). Further mild conditions on the prior, which we don't discuss, would be needed to ensure the posterior $\pi_s(\phi|\x)$ satisfies (i).
\end{proof}

\begin{proposition}[Uniform Laplace approximation]
\label{prop:ef-uniform-laplace}
Suppose the conditions of Proposition~\ref{prop:ass1-reg-exp-fam} hold.
For $\tau\in\operatorname{Int}\mathcal{M}$ let
$\phi(\tau):=(\nabla A)^{-1}(\tau)$ and $F_\tau(\phi):=A(\phi)-\tau^\top\phi$. 

If (A) $U_\tau$ is a compact neighborhood  of $\tau^\star$ such
that, for all $c<\infty$ the closure of
\[
\bigcup_{\tau\in U_\tau}
\left\{
\phi\in\Omega_\phi:
F_\tau(\phi)-F_\tau\{\phi(\tau)\}\leq c
\right\}
\]
is a bounded subset of $\Omega_\phi$, and (B) $\pi_s(\phi|\mathbf{x})$ and
$\nabla_\phi\pi_s(\phi|\mathbf{x})$
are jointly continuous on $K_s\times N_\phi$ with $N_\phi$ a compact convex neighborhood of $\widetilde\phi$, then Assumption~\ref{A:laplace}(v) holds.
\end{proposition}

\begin{proof}
If $\widehat\tau_J:=J^{-1}\sum_{j=1}^J T(y_j)$ and
we define $L_J(\phi):=p(\y|\phi)$ then exponential family form gives 
$L_J(\phi)
=
\exp\{-JF_{\widehat\tau_J}(\phi)\}$. Consequently,
\begin{equation}
\frac{L_J(\phi)}{L_J(\bar\phi_J)}
=
\exp\left[
-J\left\{
F_{\widehat\tau_J}(\phi)
-
F_{\widehat\tau_J}(\bar\phi_J)
\right\}
\right].
\label{eq:ef-likelihood-ratio}
\end{equation}
Since $\bar\phi_J\to\widetilde\phi$, almost surely
$\bar\phi_J\in\operatorname{Int}(N_\phi)$ for all sufficiently large
$J$. Condition (A) implies that there is a $\delta>0$ such that,
almost surely for all sufficiently large $J$,
\begin{equation}
\inf_{\phi\notin N_\phi}
\left\{
F_{\widehat\tau_J}(\phi)
-
F_{\widehat\tau_J}(\bar\phi_J)
\right\}
\geq\delta,
\label{eq:ef-global-separation}
\end{equation}
as condition (A) says any sequence $\phi_k\notin N_\phi$ satisfying $F_{\widehat\tau_J}(\phi_k)-F_{\widehat\tau_J}(\bar\phi_{J_k})\rightarrow0$ for data $\widehat\tau_J\to\tau^\star$ is confined to a compact set, with 
$F_{\tau^\star}(\phi_\infty)=F_{\tau^\star}(\widetilde\phi)$ holding at $J\to\infty$. This implies $\phi_\infty=\widetilde\phi$, contradicting
$\phi_k\notin N_\phi$.
It follows from \eqref{eq:ef-likelihood-ratio} and
\eqref{eq:ef-global-separation} that
\begin{equation}
\sup_{\phi\notin N_\phi}L_J(\phi)
\leq
L_J(\bar\phi_J)e^{-J\delta}.
\label{eq:ef-tail-likelihood}
\end{equation}
For a full, regular, minimal exponential family, $H_\phi=\nabla^2 A(\phi)$ is smooth on
$\Omega_\phi$ and positive definite so $H_{\bar\phi_J}\to H_{\widetilde\phi}$
with a positive definite limit. Equation \eqref{eq:ef-tail-likelihood} gives
\begin{equation}
\frac{1}{C_J(\mathbf{y})}\displaystyle\sup_{\phi\notin N_\phi}L_J(\phi)
=
O\left(J^{d_\phi/2}e^{-J\delta}\right)
\longrightarrow0
\label{eq:ef-tail-laplace}
\end{equation}
almost surely for $C_J(\mathbf{y})$ in Assumption~\ref{A:laplace}(v). The ordinary local
Laplace expansion, applied on $N_\phi$, gives
\begin{equation}
\frac{1}{C_J(\mathbf{y})}
\int_{N_\phi}L_J(\phi)\,d\phi
\longrightarrow1
\label{eq:ef-unit-laplace}
\end{equation}
and
\begin{equation}
\frac{1}{C_J(\mathbf{y})}
\int_{N_\phi}
\lVert\phi-\bar\phi_J\rVert L_J(\phi)\,d\phi
=
O(J^{-1/2})
\label{eq:ef-first-moment-laplace}
\end{equation}
almost surely. 

We control the dependence on $s$. As $N_\phi$ is convex, joint continuity of
$\nabla_\phi\pi_s(\phi|\mathbf{x})$ in condition (B), compactness of
$K_s\times N_\phi$ and the mean-value theorem give
\begin{equation}
\sup_{s\in K_s}
\left|
\pi_s(\phi|\mathbf{x})
-
\pi_s(\bar\phi_J|\mathbf{x})
\right|
\leq
M\lVert\phi-\bar\phi_J\rVert,
\qquad
\phi\in N_\phi.
\label{eq:ef-prior-lipschitz}
\end{equation}
Because $N_\phi$ has finite and positive measure, for each
$s$ there is a $\phi_s\in N_\phi$ such that
\[
\pi_s(\phi_s|\mathbf{x})
\leq
{1}/{\operatorname{vol}(N_\phi)},
\]
so \eqref{eq:ef-prior-lipschitz} gives
\begin{equation}
\sup_{(s,\phi)\in K_s\times N_\phi}
\pi_s(\phi|\mathbf{x})
\leq
\frac{1}{\operatorname{vol}(N_\phi)}
+
M\operatorname{diam}(N_\phi).
\label{eq:ef-prior-bounded}
\end{equation}
Let
$
a_{J,s}
:=
\pi_s(\bar\phi_J|\mathbf{x})
$.
Splitting the integral in $p_s(\mathbf{y}|\mathbf{x})$ over $N_\phi$ and
$N_\phi^c$ gives
\begin{equation}
\begin{split}
\frac{p_s(\mathbf{y}|\mathbf{x})}{C_J(\mathbf{y})}
-a_{J,s}
&=
\int_{N_\phi}
\frac{L_J(\phi)}{C_J(\mathbf{y})}
\left\{
\pi_s(\phi|\mathbf{x})-a_{J,s}
\right\}\,d\phi
\\
&
+
a_{J,s}
\left\{
\frac{1}{C_J(\mathbf{y})}
\int_{N_\phi}L_J(\phi)\,d\phi-1
\right\}
+
\int_{N_\phi^c}
\frac{L_J(\phi)}{C_J(\mathbf{y})}
\pi_s(\phi|\mathbf{x})\,d\phi.
\end{split}
\label{eq:ef-laplace-decomposition}
\end{equation}
Using \eqref{eq:ef-prior-lipschitz} and \eqref{eq:ef-prior-bounded} in
\eqref{eq:ef-laplace-decomposition} gives
\begin{equation*}
\begin{split}
\sup_{s\in K_s}
\left|
\frac{p_s(\mathbf{y}|\mathbf{x})}{C_J(\mathbf{y})}
-
\pi_s(\bar\phi_J|\mathbf{x})
\right|
&\leq
\frac{M}{C_J(\mathbf{y})}
\int_{N_\phi}
\lVert\phi-\bar\phi_J\rVert L_J(\phi)\,d\phi
\\
&\quad+
B
\left|
\frac{1}{C_J(\mathbf{y})}
\int_{N_\phi}L_J(\phi)\,d\phi-1
\right|
+
\frac{\displaystyle\sup_{\phi\notin N_\phi}L_J(\phi)}
{C_J(\mathbf{y})}.
\end{split}
\end{equation*}
The RHS tends to zero by
\eqref{eq:ef-first-moment-laplace},
\eqref{eq:ef-unit-laplace}, and
\eqref{eq:ef-tail-laplace}, proving Assumption~\ref{A:laplace}(v).
\end{proof}

\subsection{Proof of Theorem~\ref{thm:loss_s_single}}\label{app:loss_s_single-proof}

{\bf Theorem~\ref{thm:loss_s_single}} {\it 
If Assumptions~\ref{A:laplace} and \ref{A:log_min_pooled} hold then
$\bar s_{(1,J),\,\x}\xrightarrow{p}\widetilde s_{(1,\infty),\,\x}$ as $J\longrightarrow\infty.$
}\\[-0.2in]
\begin{proof}
By Assumption~\ref{A:laplace}(iv),
$\bar\phi_J\xrightarrow{p}\widetilde\phi$. Choose a compact neighborhood
$N_\phi\subset Int(\Omega_\phi)$ of $\widetilde\phi$. With probability
tending to one, $\bar\phi_J\in N_\phi$. Assumption~\ref{A:laplace}(i)
and compactness imply that $(s,\phi)\mapsto\pi_s(\phi|\x)$ is uniformly
continuous on $K_s\times N_\phi$. It follows that
\begin{equation}\label{eq:pi-phibar-to-pi-phitild}
 \sup_{s\in K_s}
 \left|
 \pi_s(\bar\phi_J|\x)-\pi_s(\widetilde\phi|\x)
 \right|\xrightarrow{p}0.
 \end{equation}

Let $Q_J(s)={p_s(\y|\x)}/{C_J(\y)}$ and $Q(s)=\pi_s(\widetilde\phi|\x)$.
Since $C_J(\y)$ is positive and does not depend on $s$ so by Assumption~\ref{A:laplace}(v) and
\eqref{eq:pi-phibar-to-pi-phitild}
\begin{align*}
 \sup_{s\in K_s}|Q_J(s)-Q(s)|
 &\leq
 \sup_{s\in K_s}
 \left|
 Q_J(s)-
 \pi_s(\bar\phi_J|\x)
 \right|\nonumber\\
 &\quad+
  \sup_{s\in K_s}
 \left|
 \pi_s(\bar\phi_J|\x)-
 Q(s)
 \right| \xrightarrow{p}0.
\end{align*}

In Assumption~\ref{A:log_min_pooled}, $\bar s_{(1,J),\,\x}$, is a global maximiser of $Q_J$ so (i) $Q$ is uniquely maximized at $\widetilde s_{(1,\infty),\,\x}$, (ii) $K_s$ is compact, (iii) $Q$ is continuous and (iv) $Q_J$ converges in probability to $Q$ uniformly over $K_s$. If in addition $\bar s_{(1,J),\,\x}\in K_s$ then \citet[Theorem 2.1]{newey94} gives $\bar s_{(1,J),\,\x}\to\widetilde s_{(1,\infty),\,\x}$ as $J\to\infty$.
However, for every $\epsilon>0$,
\begin{align*}
\Pr\!\left(
\left\|\overline s_{(1,J),\,\x}-\widetilde s_{(1,\infty),\,\x}\right\|>\epsilon
\right)
\ \leq\ &
\Pr\!\left(
\left\|\overline s_{(1,J),\,\x}-\widetilde s_{(1,\infty),\,\x}\right\|>\epsilon\mid
\overline s_{(1,J),\,\x}\in K_s
\right)
\\[2pt]& +\  \Pr\{\overline s_{(1,J),\,\x}\notin K_s\}.
\end{align*}
The second term tends to zero by Assumption~\ref{A:log_min_pooled}(ii) and the first by \citet[Theorem 2.1]{newey94}, proving the claimed consistency.
\end{proof}

\subsection{Proof of Theorem~\ref{thm:product}}
\label{app:loss_s_product_proof}

{\bf Theorem~\ref{thm:product}} (Product calibration) {\it 
Suppose Assumptions~\ref{A:product-regular} and
\ref{A:product-separation} hold. If \(\rho(s)\) is a probability density
on \(\Omega_s\) that is continuous and strictly positive at
\(\widetilde s_{(\infty,1),\,\x}\), then, almost surely, for all sufficiently
large \(J\), the loss minimiser
\[\bar s_{(J,1),\,\x}
\in
\argmin_{s\in\Omega_s}
l_{(J,1)}(s;\y,\x)\]
is unique and
$
\bar s_{(J,1),\,\x}
\rightarrow
\widetilde s_{(\infty,1),\,\x}
$.
Moreover, if
$
\rho_{(J,1)}(s|\y;\x)
\propto
\exp\{-l_{(J,1)}(s;\y,\x)\}\rho(s),
$
then, almost surely, for $s\sim\rho_{(J,1)}(s|\y;\x)$,
\[
\sqrt{J}
\left(
s-\bar s_{(J,1),\,\x}
\right)
\xrightarrow{\mathrm{t.v.}}
N\left(0,H_s^{-1}\right).
\]
}

\begin{proof}
Fix a calibration-data sequence in the probability-one set on which
Assumptions~\ref{A:product-regular} and
\ref{A:product-separation} hold. In the notation of
\citet[Theorem~5]{Millar21}, make the identifications
\[
n=J,
\qquad
\theta=s,
\qquad
\Theta=\Omega_s,
\qquad
\theta_0=\widetilde s_{(\infty,1),\,\x},
\]
\[
f_n(\theta)
=
J^{-1}l_{(J,1)}(s;\y,\x),
\qquad
f(\theta)
=
\widetilde l(s,\x),
\qquad
\pi(\theta)=\rho(s).
\]
The product posterior then has the required form
$
\rho_{(J,1)}(s|\y;\x)
\propto
\exp\{-Jf_J(s)\}\rho(s)
$.

Assumption~\ref{A:product-regular} gives the pointwise convergence,
positive-definite limiting Hessian, continuous third derivatives and
uniform third-derivative bound required by
\citet[Theorem~5]{Millar21}. Assumption~\ref{A:product-separation} is
case~(1) of that theorem. It follows that there is a sequence
\(t_J\) such that
$
t_J
\longrightarrow
\widetilde s_{(\infty,1),\,\x}
$
and, for all sufficiently large \(J\),
$
\nabla_s
l_{(J,1)}(t_J;\y,\x)=0
$.
Furthermore,
\[
\sqrt{J}(s-t_J)
\xrightarrow{\mathrm{t.v.}}
N\left(0,H_s^{-1}\right).
\]

It remains to identify \(t_J\) with $\bar s_{(J,1),\,\x}$. The proof of \citet[Theorem~5]{Millar21} shows that the
conditions of \citet[Theorem~4]{Millar21} hold at \(t_J\). In
particular, for every \(\epsilon>0\),
\begin{align}
\liminf_{J\to\infty}
\inf_{s\notin B_\epsilon(t_J)} J^{-1}
\big(\,
l_{(J,1)}(s;\y,\x)
-
l_{(J,1)}(t_J;\y,\x)\,
\big)
>0.
\label{eq:product-diff-loss}
\end{align}

By the regular-convergence result in
\citet[Theorem~7]{Millar21},
$
J^{-1}\nabla_s^2
l_{(J,1)}(s;\y,\x)
$
converges uniformly to
\(\nabla_s^2\widetilde l(s,\x)\) on a sufficiently small ball around
\(\widetilde s_{(\infty,1),\,\x}\). Since \(H_s\) in Assumption~\ref{A:product-regular}(ii) is positive definite, the
empirical Hessian is uniformly positive definite on a possibly
smaller ball for all sufficiently large \(J\). Hence the empirical loss
is strictly convex there, and \(t_J\) is its unique minimiser in
that ball. The separation in \eqref{eq:product-diff-loss} shows that every point outside the
ball has strictly larger empirical loss. Thus \(t_J\) is the
unique global product-loss minimiser for all sufficiently large \(J\) and consequently,
$
\bar s_{(J,1),\,\x}=t_J
$
eventually, which proves both
$
\bar s_{(J,1),\,\x}
\longrightarrow
\widetilde s_{(\infty,1),\,\x}
$
and
\[
\sqrt{J}
\left(
s-\bar s_{(J,1),\,\x}
\right)
\xrightarrow{\mathrm{t.v.}}
N\left(0,H_s^{-1}\right).
\]
Since the argument applies on a probability-one set of calibration-data
sequences, the asserted convergence is almost sure.
\end{proof}

\subsection{Assumptions for non-regular asymptotics in Theorem~\ref{thm:post.boundary}}\label{app:non-reg-assume}

Assumption~\ref{A:s_boundary} collects together several assumptions from \cite{Bochkina2014}. They treat three different types of parameter-components jointly. These would correspond to the components of $s$ here: for components with boundary risk-minimisers, the risk may have zero or positive slope; they also allow for components with interior risk-minimisers. We consider only the case where all components of the gradient of the risk are positive $\widetilde s=0_{d_s}$. 
As we take rather a special case of that more general theory, the following selects only the assumptions we need.

\begin{assumption}\label{A:s_boundary}
Suppose $\Omega_s=[0,\infty)^{d_s}$. Let $\overline{\Omega}_{(\widetilde{s},\deltaJ)} = \{s\in\Omega_s: \|s-\widetilde{s}\|_\infty <\deltaJ \}$ for $\deltaJ>0$. As $J\to\infty$, $\deltaJ\to 0$ and $\deltaJ J \to\infty$. We assume the following:\\[-0.25in]
\begin{enumerate}[(i),itemsep=0pt]
\item the risk defined in \eqref{eq:gb-risk} has a unique minimum $\widetilde{s}=\argmin_{s\in\Omega_s} \widetilde{l}(s;\x) $ at $\widetilde{s}=0$ and positive one-sided gradient $\nabla_s\widetilde{l}(\widetilde{s};\x)>0$ componentwise at $\widetilde{s}$. 
\item for all $s\in\overline{\Omega}_{(\widetilde{s},\deltaJ)}$, both $\nabla_s\widetilde{l}(s,\x)$ and $\nabla_s l_{(J,1)}(s;\y,\x)$ exist $P_{\y\sim \pstar}$ almost everywhere for large $J$
and $$P_{\y\sim \pstar} \left(\sup_{s\in \overline{\Omega}_{(\widetilde{s},\deltaJ)}} \left\| \frac{1}{J} \nabla_s l_{(J,1)}(s;\y,\x) - \nabla_s\widetilde{l}(\widetilde{s},\x) \right\|_\infty >\epsilon\right)\to 0$$
as $J\to\infty$ for any $\epsilon>0$.
\item the $\sigma$-finite prior measure $\rho(s)$ gives a proper posterior 
for $s$ for $P_{\y\sim \pstar}$ almost all $\y$ for large $J$ and for $s\in\overline{\Omega}_{(\widetilde{s},\deltaJ)}$, there exist  $c_0>0$, $\alpha_j>0$ not depending on $J$ and $\Delta_\rho(\deltaJ)\ge 0$ satisfying \(\Delta_\rho(\delta_J)\longrightarrow0\) such that $$c_0(1-\Delta_\rho(\deltaJ))\leq \rho(s)\prod^{d_s}_{j=1} s_j^{-(\alpha_j-1)}\leq c_0(1+\Delta_\rho(\deltaJ)).$$
\end{enumerate}
\end{assumption}

The asymptotic behaviour of $\rho_{(J,1)}(s|\y;\x)$ in Theorem~\ref{thm:post.boundary} in Section~\ref{asymptotic_boundary} is not governed by an interior quadratic expansion. It is instead driven by a first-order expansion in the boundary directions. In the following proof we simply give the mapping and between our notation and assumptions and the notation and assumptions in \cite{Bochkina2014}.

\paragraph{Proof of Theorem \ref{thm:post.boundary}.} This follows the proof of Theorem 1 in \cite{Bochkina2014}, considering only $\widetilde{s}$ on the closed boundary: our $1/J$ is their $\sigma^2$ and their $l_y(\theta)$ is our $-l_{(J,1)}(s;\y,\x)/J$; almost sure convergence in Assumption \ref{A:product-regular}(i) can be relaxed to convergence in probability and, with Assumption \ref{A:s_boundary}(i), covers their Assumption M;
Convergence of loss derivatives in Assumption \ref{A:s_boundary}(ii) covers the parts of their Assumption~S relating to minima on the boundary with positive gradient; our assumed prior control at $s=0$ in Assumption \ref{A:s_boundary}(iii) covers the corresponding parts of their Assumption P. Additional assumptions made in Theorem~\ref{thm:post.boundary} itself concerning the local geometry of the rescaled parameter space around $\widetilde{s}$ and the assumed behaviour of $\Delta_{\rm tail}$ (which enforces the required rate of posterior consistency) cover their Assumptions~B and L respectively.

\section{The KL-estimator}

\subsection{$\widehat{s}_{KL}$ definition and estimation}\label{Appendix:s-KL-eta-beta}

In this section, we define $\widehat{s}_{KL}$ and give an estimation procedure. 

The estimator $\widehat{s}_{KL}$ is a Bayes estimator minimising the expected posterior loss for miss-match between
$p_{\widetilde{s}}(z|\y,\x)$ and $p_{s'}(z|\y,\x)$ as measured by KL-divergence:
\begin{align*}
\widehat{s}_{KL} &=\argmin_{s'} \mathbb{E}_{s\sim \rho(s|\y,\x)}[KL( p_s(z|\y,\x) \| p_{s'}(z|\y,\x)]
\\
&=\argmin_{s'} \mathbb{E}_{s\sim \rho(s|\y,\x)}[\mathbb{E}_{z \sim p_s(z|\y,\x)}\left( -\log  p_{s'}(z|\y,\x)\right )]
\end{align*}

\begin{algorithmic}
\State Let $(s'_1,...,s'_K)$ be equally spaced values of $s\in \Omega_s$. 
\State 1. For each $s'_i$, we estimate $w_i = \frac{1}{TJ} \sum^T_{t=1} \sum^J_{j=1}\log p_{s'_i}(z_j|\y,\x)$ where $s_t \sim \rho(s|\y, \x)$, $t=1,...,T$ and $z_j|s_t \sim p_{s_t}(z|\y,\x)$, $j=1,...,J$.
\State 2. Smooth $w_1,...,w_K$ as a function of $s'$ and find $\widehat{s}_{KL} = \argmin_{s'} \widehat{KL}(s')$.
\end{algorithmic}

For example, to estimate $\widehat{b}_{KL}$ in Fig.~\ref{fig:normal_loss_b_gam_eta}, we took $T=400$ and $J=1000$.

\subsection{Harmonic mean estimator}\label{Appendix:KL-harmonic}

The harmonic mean estimator 
\begin{equation}\label{eq:harmonic}
    \widehat{\eta}_{hm} =\frac{1}{\mathbb{E}_{\eta\sim \rho(\eta|\y,\x)}[1/\eta]}\,.
\end{equation}
for $\widetilde{\eta}$ approximates a Bayes estimator similar to $\widehat{\eta}_{KL}$, but minimising the expected posterior loss given by the KL-divergence between GBI posterior distributions 
\[
\pi_{\eta'}(\phi|\x,\y)\propto \pi(\phi)\,e^{-\eta' \ell(\phi;\x,\y)}
\]
and $\pi_{\widetilde{\eta}}(\phi|\x,\y)$. Here $\ell(\phi;\x,\y)$ is a loss, in our case the negative log-likelihood and $\pi_{\eta'}(\phi|\x,\y)$ the power posterior. The two posteriors at $\eta=\eta'$ and $\eta=\widetilde{\eta}$ condition on both $\x$ (with $n$ samples) and $\y$ (with $J$). If $n'=n+J$ is large then the GBI-posterior is approximately $N(\bar{\phi},\frac{1}{\eta n'}H^{-1})$ where $\overline{\phi}=\argmin_\phi \ell(\phi;\x,\y)$, $H=\nabla^2 \widetilde{\ell}(\widetilde{\phi})$ and $\widetilde{\phi}=\argmin_\phi \widetilde{\ell}({\phi})$.  

This second expected posterior loss is approximated
\begin{equation} \label{eq:KL-epl-approx}
    \mathbb{E}_{\eta\sim \rho(\eta|\y,\x)}[KL( \pi_\eta(\phi|\y,\x) \| \pi_{\eta'}(\phi|\y,\x)] \approx  \mathbb{E}_{\eta\sim \rho(\eta|\y,\x)} \left( \log \frac{\eta}{\eta'} + \frac{\eta'}{\eta} \right) + c
\end{equation}
(up to constant terms not depending on $\eta'$) and finally
\[
    \argmin_{\eta'}\mathbb{E}_{\eta\sim \rho(\eta|\y,\x)} \left( \log \frac{\eta}{\eta'} + \frac{\eta'}{\eta} \right) =\widehat{\eta}_{hm}
\]
so $\widehat{\eta}_{hm}$ in \eqref{eq:harmonic} is approximately equal to the Bayes estimator for the loss on the LHS of Equation~\eqref{eq:KL-epl-approx}.
This motivates the harmonic mean for GBI only, as in Section~\ref{sec:edisc}, not SMI, as in Section~\ref{sec:ssm}. We explored (\ref{eq:harmonic}) for SMI simply as an alternative estimator.

\section{Asymptotic behaviour of SMI-posteriors}
\label{app:winter-theory-all}

\subsection{Theory from \protect\cite{Winter2023}}
\label{app:winter-theory-main}

We quote here Theorems~2 and 3 of \cite{Winter2023}, respectively Theorems~\ref{thm:strong} and \ref{thm2} below. These are applied in the SMI-setting in Section~\ref{subsec:seqGibbs}. They have been adapted to our SMI setting with mainly notational changes. 

\begin{assumption}\label{A:Winter0-old}  
$\frac{1}{n}\ell_{s}(\varphi;\x) \xrightarrow{a.s.} \widetilde{\ell}_{s}(\varphi)$ and $\frac{1}{n}\ell (\theta;\varphi,\x_2)\xrightarrow{a.s.} \widetilde{\ell}(\theta;\varphi)$ $\forall$ $\varphi\in\Omega_\varphi, \theta\in\Omega_\theta$.
\end{assumption}

\begin{assumption}\label{A:Winter2-old} There exist open sets $U_\varphi\subseteq \Omega_\varphi$ and $U_\theta\subseteq\Omega_\theta$ such that:
(a) there exists $\widetilde{\varphi}_s\in U_\varphi$ and for all $\varphi\in U_\varphi$ there exists $\widetilde{\theta}_{\varphi}\in U_\theta$ satisfying
$\nabla_\varphi \widetilde{\ell}_{s}(\widetilde{\varphi}_s) =0$,
 and $\nabla_\theta \widetilde{\ell}(\widetilde{\theta}_{\varphi};\varphi)=0$;
(b) there exists $\overline{\varphi}_s\in U_\varphi$ and for all $\varphi\in U_\varphi$ there exist $\overline{\theta}_{\varphi}\in U_\theta$ satisfying
$\nabla_\varphi \ell_{s}(\overline{\varphi}_s;\x)=0$ and
$\nabla_\theta \ell(\overline{\theta}_{\varphi};\varphi,\x_2)=0$ almost surely.
\end{assumption}
\begin{assumption}\label{A:Winter3-new}
For $U_\varphi, \widetilde{\varphi}_s$ and $U_\theta, \widetilde{\theta}_{\widetilde{\varphi}_s}$ satisfying Assumption~\ref{A:Winter2-old}, there are compact sets $K_\varphi\subset U_\varphi$ and $K_\theta\subset U_\theta$ such that $\widetilde\varphi_s\in\operatorname{Int}(K_\varphi)$ and 
$\widetilde\theta_{\widetilde\varphi_s}
\in\operatorname{Int}(K_\theta)$ and $\widetilde{\ell}_s(\varphi)>\widetilde{\ell}_s(\widetilde\varphi_s)$ for all $\varphi\in K_\varphi\setminus\{\widetilde\varphi_s\}$ and $\widetilde{\ell}(\theta;\widetilde\varphi_s)>\widetilde{\ell}(\widetilde{\theta}_{\widetilde{\varphi}_s};\widetilde\varphi_s)$ for all $\theta\in K_\theta\setminus\{\widetilde{\theta}_{\widetilde{\varphi}_s}\}$.
\end{assumption}

Assumptions~\ref{A:Winter0-old}, \ref{A:Winter2-old} and \ref{A:Winter3-new} are Assumptions~1, 2 and 3 of \cite{Winter2023}.
Two further regularity conditions, Assumptions~\ref{A:Winter1-old}
(S3) and \ref{A:Winter3-old} (S4-8), are given in Appendix~\ref{Appendix:SMI-convergence-WD-assumptions}. These are conditions on loss continuity, differentiability, Hessian and uni-modality. Theorems~2 and 3 of \cite{Winter2023} give the SMI-asymptotics. 

\begin{theorem}\label{thm:strong} (Theorem 2 of \cite{Winter2023})
If Assumptions~\ref{A:Winter0-old}, \ref{A:Winter2-old}(a) and \ref{A:Winter1-old} hold and the prior probabilities over neighborhoods $N_{\varphi,\epsilon}$ and $N_{\theta,\epsilon}$ defined in Assumption~\ref{A:Winter1-old} are positive for all $\epsilon>0$ then $P^{(s)}(d((\varphi,\theta),(\widetilde{\varphi}_s,\widetilde{\theta}_{\widetilde{\varphi}_s}))<\epsilon |\x)\xrightarrow{a.s.} 1$.
\end{theorem}


\begin{theorem}\label{thm2}
(Theorem 3 of \cite{Winter2023}) Suppose that $\varphi\sim \pi_{s}(\varphi|\x)$ and $\theta \sim \pi(\theta|\varphi,\x_2)$. The priors $\pi(\varphi)$ and $\pi(\theta)$ are continuous and strictly positive at $\widetilde{\varphi}_s$ and $\widetilde{\theta}_{ \widetilde{\varphi}_s}$, with support a subset of $U_\varphi$ and $U_\theta$ respectively. If Assumptions~\ref{A:Winter0-old},\ref{A:Winter2-old}, \ref{A:Winter3-new} and \ref{A:Winter3-old} hold then, $\sqrt{n}(\varphi-\overline{\varphi}_s) \xrightarrow{d} N(0,H_\varphi^{-1})$ and $\sqrt{n}(\theta-\overline{\theta}_{\varphi}) \xrightarrow{d} N(0,H_\theta^{-1})$. 
\end{theorem}

The weak-convergence in Theorem \ref{thm2} shows that the posterior is normally approximated at large $n$. Here the variations $\tau_\varphi=\varphi-\overline{\varphi}_s$ and $\tau_{\theta|\varphi}=\theta-\overline{\theta}_{\varphi}$ are independent asymptotically. Intuitively $\varphi$ concentrates at $\widetilde{\varphi}_s$ and the conditional posterior $\pi(\theta|\x_2,\varphi)$ is similar to $\pi(\theta|\x_2,\widetilde{\varphi}_s)$ for large $n$. 
As \citet{Winter2023} and \cite{frazier2023} remark, a similar convergence does not necessarily hold for the joint Gibbs posterior.

Unlike the $\gamma$-SMI posterior, the $(\eta,\beta)$ and $\eta$-SMI posterior distributions can be given jointly with an auxiliary variable $\theta'$ as in \eqref{eq:eta-smi-joint} and \eqref{eq:eta-beta-smi-joint}. The large $n$ behaviour of $\pi_s(\varphi,\theta'|\x)$ is derived in Corollary \ref{cor:aux_loss_joint} below, with $\ell_{s}(\varphi,\theta';\x)=\ell_{\eta}(\varphi,\theta';\x)$ in \eqref{eq:eta-smi-joint-loss} or $\ell_{s}(\varphi,\theta';\x)=\ell_{(\eta,\beta)}(\varphi,\theta';\x)$ in \eqref{eq:eta-beta-smi-joint-loss} substituted for $\ell_{s}(\varphi;\x)$ throughout.
These losses are sums over independent data components so Assumption~\ref{A:Winter0-old} is often satisfied. For the $\eta$-SMI posterior with $n_1$ and $n_2$ both $O(n)$, 
$n^{-1}\ell_\eta(\varphi,\theta';\x_1,\x_2)\xrightarrow{a.s.}\widetilde{\ell}_\eta(\varphi,\theta')$ 
where 
\[
\widetilde{\ell}_\eta(\varphi,\theta') = -\eta \mathbb{E}_{\pstar_{x_2}}[\log p(x_2|\varphi,\theta')]-\alpha \mathbb{E}_{\pstar_{x_1}}[\log p(x_1|\varphi)].\] 
Similarly for $(\eta,\beta)$-SMI using the loss in \eqref{eq:eta-beta-smi-joint-loss},
$n^{-1}\ell_{(\eta,\beta)}(\varphi,\theta';\x) \xrightarrow{a.s.}\widetilde{\ell}_{(\eta,\beta)}(\varphi,\theta')$ where \[\widetilde{\ell}_{(\eta,\beta)}(\varphi,\theta')=-\alpha \mathbb{E}_{\pstar_{x_1}} [\log p(x_1|\varphi)]
-\frac{\eta}{\beta-1} E_{\pstar_{x_2}}\!\left(p(x_2|\varphi,\theta')^{\beta-1}\right) + \frac{\eta}{\beta}\int \! p(x_2|\varphi,\theta')^\beta dx_2\,.
\]

\subsection{Assumptions required for Theorems~\ref{thm:strong} and \ref{thm2}}\label{Appendix:SMI-convergence-WD-assumptions}

The following conditions from \cite{Winter2023} are required for the asymptotics in Theorems~\ref{thm:strong} and \ref{thm2} in Appendix~\ref{app:winter-theory-main} to hold. 
Assumptions~\ref{A:Winter1-old} and \ref{A:Winter3-old} are respectively Assumption S3 and Asummptions S4-8 of \cite{Winter2023}. Limiting statements and inequalities below hold almost surely unless otherwise stated.

\begin{assumption}\label{A:Winter1-old}  
Let $d$ be the Euclidean norm. Let $N_{\varphi,\epsilon}=\{ \varphi: d(\varphi,\widetilde{\varphi}_s)<\epsilon \}$ and $N_{\theta,\epsilon} =\{ \theta: d(\theta,\widetilde{\theta}_{\widetilde{\varphi}_s})<\epsilon \}$. \\
(i) The risk $\widetilde{\ell}_{s}(\varphi)$ is continuous at $\widetilde{\varphi}_s$ and for every $\epsilon>0$, it holds that $$\liminf_n \inf_{\varphi\in N^c_{\varphi,\epsilon}} ( \ell_{s}(\varphi;\x)/n - \widetilde{\ell}_{s}( \widetilde{\varphi}_s)) >0.$$ \\
(ii) The risk $\widetilde{\ell}(\widetilde{\theta}_\varphi;\varphi)$ is continuous in $\varphi$ at 
$\varphi=\widetilde{\varphi}_s$, there exists $\delta>0$ such that $\widetilde{\ell}(\theta;\varphi)$ is continuous in $\theta$ at $\theta=\widetilde\theta_\varphi$ for all $\varphi \in N_{\varphi,\delta}$ and, for every $\epsilon>0$, it holds that 
$$\liminf_n \inf_{\varphi\in N_{\varphi,\delta}} \inf_{\theta\in N^c_{\theta,\epsilon}} ( \ell(\theta;\varphi,\x_2)/n - \widetilde{\ell}(\widetilde{\theta}_{\varphi};\varphi)) >0.$$ 
\end{assumption} 


\begin{assumption}\label{A:Winter3-old} 
\def\Ef{E}
(i) The Hessian $H_\varphi=\nabla^2_\varphi \widetilde{\ell}_{s}(\widetilde{\varphi}_s)$ is positive definite. There exists an open bounded set $\Ef_\varphi\subseteq\Omega_\phi$ such that $\widetilde{\varphi}_s \in \Ef_\varphi$, $\nabla^3_{\varphi}\ell_{s}(\varphi;\x)$ is continuous in $\Ef_\varphi$ and 
$$\sup_n \sup_{\varphi\in \Ef_\varphi} \sup_{ijk} | \nabla^3_{\varphi_i,\varphi_j,\varphi_k} \ell_{s}(\varphi;\x)/n|<\infty.$$ There is a compact set $K_\varphi\subseteq\Ef_\varphi$ such that $\widetilde{\varphi}_s\in K_\varphi$ and, for $U_\varphi$ in Assumptions~\ref{A:Winter2-old},\ref{A:Winter3-new}, $$\liminf_n \inf_{\varphi\in U_\varphi\setminus K_\varphi} (\ell_{s}(\varphi;\x)/n-\widetilde{\ell}_{s}(\widetilde{\varphi}_s))>0,$$
$\bar\varphi_s\to\widetilde\varphi_s$ and $\ell_s(\varphi;\x)/n\rightarrow\widetilde{\ell}_s(\varphi)$ for every $\varphi\in U_\varphi$.
\\
(ii) The Hessian $H_\theta=\nabla^2_\theta \widetilde{\ell}(\widetilde{\theta}_{\widetilde{\varphi}_s};\widetilde{\varphi}_s)$ is positive definite. For all $\varphi\in\Omega_\varphi$ there exists an open bounded set $\Ef_\theta$ such that $\widetilde{\theta}_{ \widetilde{\varphi}_s} \in \Ef_\theta$, $\nabla^3_{\theta}\ell(\theta;\varphi,\x_2)$ is continuous in $\Ef_\theta$ and $$\sup_n \sup_{\theta\in \Ef_\theta} \sup_{ijk} | \nabla^3_{\theta_i,\theta_j,\theta_k} \ell(\theta;\varphi,\x_2)/n|<\infty.$$ 
 There is a compact set $K_\theta\subseteq \Ef_\theta$ such that $\widetilde{\theta}_{ \widetilde{\varphi}_s}\in K_\theta$
 and, for $U_\varphi$ in Assumptions~\ref{A:Winter2-old},\ref{A:Winter3-new} 
 and, for every fixed \(h\in\mathbb R^{d_\varphi}\) and $\varphi_n(h)
=
\overline\varphi_s+\frac{h}{\sqrt n},
$
when \(\varphi_n(h)\in U_\varphi\),
 $$\liminf_n \inf_{\theta\in U_\theta\setminus K_\theta} (\ell(\theta;\varphi_n(h),\x_2 )/n-\widetilde{\ell}(\widetilde{\theta}_{ \widetilde{\varphi}_s};\widetilde{\varphi}_s ))>0,$$ 
 $\bar\theta_{\varphi_n(h)}\to\widetilde\theta_{\widetilde{\varphi}_s}$ and $\ell (\theta;\varphi_n(h),\x_2)/n\rightarrow\widetilde{\ell}(\theta;\widetilde{\varphi}_s)$ for every $\theta\in U_\theta$.
\end{assumption}

\subsection{On the relation between $\eta$-SMI and $\gamma$-SMI}
\label{Appendix:gamma-is-eta}

We compare the $\eta$-SMI and $\gamma$-SMI losses when
$\eta=\gamma$. 
Looking at the $\eta$ and $\gamma$-losses in \eqref{loss:eta_varphi} and \eqref{loss:gamma_varphi}, write $n=n_2$, suppose ${n_1}/{n}\longrightarrow\alpha$,
and, for $b>0$, define \[M_{b,n}(\varphi)
:=
\int
p(\x_2\mid\varphi,\theta)^b
\pi(\theta)\,d\theta,\] so that $\ell_\eta(\varphi|\x)
=
-\log M_{\eta,n}(\varphi)
-\log p(\x_1|\varphi) $ and  
$
\ell_\gamma(\varphi|\x)
=
-\gamma \log M_{1,n}(\varphi)
-\log p(\x_1|\varphi)$.

Let 
$r_n(\varphi,\theta)
:=
-\log\{p(\x_2\mid\varphi,\theta)\}/n$ be the second-stage loss with population limit
\[
\widetilde r(\varphi,\theta)
:=
\mathbb{E}_{p^\star_{x_2}}
\bigl\{-\log p(X_2\mid\varphi,\theta)\bigr\}.
\]
For all $\varphi\in\Omega_\varphi$, let
$
\overline\theta_n(\varphi)
:=
\argmin_{\theta\in Int(\Omega_\theta)}
r_n(\varphi,\theta)
$
and
$
\widetilde\theta(\varphi)
:=
\argmin_{\theta\in Int(\Omega_\theta)}
\widetilde r(\varphi,\theta)
$,
\[
H_n(\varphi)
:=
\nabla_\theta^2
r_n\{\varphi,\theta\}\big\vert_{\theta=\overline\theta_n(\varphi)}
\]
and define \(
A_n(\varphi)
:=
(2\pi)^{d_\theta/2}
\pi\{\overline\theta_n(\varphi)\}
|H_n(\varphi)|^{-1/2}.
\)

\begin{assumption}[Laplace regularity]
\label{A:eta-gamma-laplace}
The following hold for fixed $\varphi\in\Omega_\varphi$ and $a>0$.
\begin{enumerate}[itemsep=0pt,topsep=0pt]
\item[(i)]
With probability tending to one,
$r_n(\varphi,\cdot)$ has the unique interior minimiser
$\overline\theta_n(\varphi)$ and $H_n(\varphi)$ is positive
definite. Then as $n\to\infty$, 
\[
M_{b,n}(\varphi)
=
A_n(\varphi)(bn)^{-d_\theta/2}
\exp\left[
-bn
r_n\{\varphi,\overline\theta_n(\varphi)\}
\right]
\{1+o_p(1)\}.
\]
\item[(ii)]
The population loss $\widetilde r(\varphi,\cdot)$ has the unique
minimiser $\widetilde\theta(\varphi)$, and $r_n\{\varphi,\overline\theta_n(\varphi)\}
\xrightarrow{p} \widetilde r\{\varphi,\widetilde\theta(\varphi)\}$. Moreover,
$
A_n(\varphi)
\xrightarrow{p}
A(\varphi)
$ for some $0<A(\varphi)<\infty$.
\item[(iii)]
The first-module likelihood satisfies
\[
-\frac{1}{n}
\log p(\x_1\mid\varphi)
\xrightarrow{p}
\alpha\,
\mathbb{E}_{p^\star_{x_1}}
\bigl\{-\log p(X_1\mid\varphi)\bigr\}.
\]
\end{enumerate}
\end{assumption}

Assumption~\ref{A:eta-gamma-laplace}(i) is a stochastic analogue of the
usual leading Laplace approximation; see, for example,
\citet[Theorem~3, Section~5 of Chapter~9]{Wong1989}. Wong considers a
deterministic integrand, whereas here
\(r_n(\varphi,\theta)\) is random and changes with \(n\). For each fixed \(\varphi\),
convenient sufficient conditions for (i) are that the integrals are
finite and that Wong's conditions hold uniformly in probability along the sequence \(r_n\).
In particular, \(r_n(\varphi,\cdot)\) has an interior minimiser
\(\overline\theta_n(\varphi)\), its Hessian there is positive definite,
the relevant local derivative bounds are uniform in probability, and,
for every \(\epsilon>0\),
\[
\delta_{\epsilon,n}(\varphi)
:=
\inf_{\substack{\theta\in\Omega_\theta\\
\|\theta-\overline\theta_n(\varphi)\|\geq\epsilon}}
\left[
r_n(\varphi,\theta)
-
r_n\{\varphi,\overline\theta_n(\varphi)\}
\right]
\]
is bounded away from zero in probability. These uniformity conditions
ensure that the remainder in the relative Laplace approximation is
\(o_p(1)\), rather than merely giving a separate deterministic
approximation for each \(n\). Additional conditions giving Assumption~\ref{A:eta-gamma-laplace}(ii) are
$
\overline\theta_n(\varphi)
\xrightarrow{p}
\widetilde\theta(\varphi)$,
$
r_n\{\varphi,\overline\theta_n(\varphi)\}
\xrightarrow{p}
\widetilde r\{\varphi,\widetilde\theta(\varphi)\},
$
and
\[
H_n(\varphi)
:=
\nabla_\theta^2r_n
 \{\varphi,\overline\theta_n(\varphi)\}
\xrightarrow{p}
\widetilde H(\varphi),
\]
with \(\widetilde H(\varphi)\) positive definite. These standard
M-estimation conditions, together with continuity and strict positivity of
\(\pi(\theta)\) at \(\widetilde\theta(\varphi)\), give Assumption~\ref{A:eta-gamma-laplace}(ii).
Condition (iii) follows separately from the law of large numbers for the
first data module.

Once the uniform Laplace conditions hold for
\(r_n(\varphi,\cdot)\), they also hold for
\(b\,r_n(\varphi,\cdot)\) for every fixed \(b>0\): multiplication by
\(b\) leaves the minimiser unchanged and replaces \(H_n(\varphi)\) by
\(bH_n(\varphi)\), producing the factor
\((bn)^{-d_\theta/2}\). No derivatives with respect to \(\varphi\) are
required because Proposition~\ref{prop:eta-gamma-first-order} is
pointwise at each fixed \(\varphi\). At \(a=0\), the two losses are equal because \(M_{0,n}(\varphi)=1\), and they also coincide at
\(a=1\), so the proof covers
\(a>0\).

\begin{proposition}[Equality of the $\eta$-SMI and
$\gamma$-SMI risks at large $n$]
\label{prop:eta-gamma-first-order}
Let $a=\eta=\gamma>0$, and suppose that
Assumption~\ref{A:eta-gamma-laplace} holds. Then, for every fixed
$\varphi\in\Omega_\varphi$,
$\ell_\eta(\varphi\mid\x)/n
\xrightarrow{p}
\widetilde\ell_a(\varphi)$
and
$\ell_\gamma(\varphi\mid\x)/n
\xrightarrow{p}
\widetilde\ell_a(\varphi),
$
where
\begin{align}\label{eq:common-risk-limit}
\widetilde\ell_a(\varphi)
=
\alpha\,
\mathbb{E}_{p^\star_{x_1}}
\bigl\{-\log p(X_1\mid\varphi)\bigr\}
+
a\,
\widetilde r\{\varphi,\widetilde\theta(\varphi)\},
\end{align}
so
$
\widetilde\ell_\eta(\varphi)
=
\widetilde\ell_\gamma(\varphi)
=
\widetilde\ell_a(\varphi)$.
\end{proposition}

\begin{proof}
Taking logarithms of the Laplace approximation gives
\[
\log M_{b,n}(\varphi)
=
-bn
r_n\{\varphi,\overline\theta_n(\varphi)\}
-\frac{d_\theta}{2}\log(bn)
+
\log A_n(\varphi)
+
o_p(1).
\]
The common first-module term cancels, and hence
\begin{align*}
\ell_\eta(\varphi\mid\x)
-
\ell_\gamma(\varphi\mid\x)
&=
-\log M_{a,n}(\varphi)
+
a\log M_{1,n}(\varphi)\\
&=
\frac{d_\theta}{2}\log(an^{1-a})
+
(a-1)\log A_n(\varphi)
+
o_p(1).
\end{align*}
By Assumption~\ref{A:eta-gamma-laplace} 
$\log A_n(\varphi)\xrightarrow{p} \log A(\varphi)$ so
$
\left\{
\ell_\eta(\varphi\mid\x)
-
\ell_\gamma(\varphi\mid\x)
\right\}/n
\xrightarrow{p}0
$.
For the individual limiting risks, the same Laplace expansion gives
\[
-n^{-1}a\log M_{1,n}(\varphi)=-n^{-1}\log M_{a,n}(\varphi)
=
a\,r_n\{\varphi,\overline\theta_n(\varphi)\}
+
o_p(1).
\]
By Assumption~\ref{A:eta-gamma-laplace}(ii), the RHS
converges to
$a\,\widetilde r\{\varphi,\widetilde\theta(\varphi)\}$ and Assumption~\ref{A:eta-gamma-laplace}(iii) gives the first-module limit so both loss-averages converge to $\widetilde\ell_a(\varphi)$ in \eqref{eq:common-risk-limit}.
\end{proof}

\section{Simulation study}


\subsection{Computation for marginalized loss}\label{Appendix:computing-gamma-SMI}

In the $\gamma$-SMI posterior, the module-specific parameter $\theta'$ in the belief update for $\varphi$ is marginalized. The integral is not always solved analytically and is often numerically estimated. There is a wide range of works on marginal likelihood methods, and these methods are straightforwardly adapted to marginalized loss estimation. \citet{frazier2023} used the marginalized marginal likelihood identity \citep{Chib1995} with the Laplace approximated posterior distribution.

We use the adaptive Gauss-Hermite quadrature (AGHQ) method to estimate a marginalized loss. Gaussian quadrature is a well-known numerical method that integrates the product of the $d$-dimensional Gaussian density and any polynomial of total order $2k-1$ or less for a carefully chosen number of quadrature points, $k$. An integral is approximated by a weighted sum of $d$-dimensional integrands at quadrature points; points and weights are generated from the family quadrature rules. The relative rate of AGHQ is known to be $O(n^{-\lfloor (k+2)/3\rfloor })$ \citep{Liu1994,Jin2020}. Recently \citet{Bilodeau2024} showed that this relative rate still holds for approximating the marginal likelihood. The Laplace approximation is a special case of AGHQ with one quadrature point for AGH, and the relative error is known to be $O(n^{-1})$ \citep{Kass1989}. We used the R-package \texttt{aghq} \citep{Stringer2021} to approximate a marginalized loss using the AGHQ.

\subsection{Normal Mixture example}\label{Appendix:Normal}

Suppose we have two datasets, $\x_1=x_{1,1:n_1}$ and $\x_2=x_{2,1:n_2}$, informing unknown parameters $\varphi$ and $\theta$. Following the framework in Figure \ref{fig:Multimodular_model}, the $x_1$-module is reliable and $x_2$-module is misspecified. Parametric models are normal for both, $f(x_1) = N(\phi,\sigma^2_1)$ with known $\sigma^2_1$ and $f(x_2) = N(\varphi+\theta,\sigma^2_2)$ with known $\sigma^2_2$. A uniform prior is assigned for $\varphi$ and $\theta\sim N(0,s^2_\theta=0.33^2)$. Suppose that the proportion of data points from the $x_1$ over $x_2$ modules is $\alpha = \lim\limits_{n_2\to\infty} n_1/n_2=0.5$. The true data generative models are $\pstar_{x_1}=N(\varphi^*, \sigma^{2}_{1})$ and $\pstar_{x_2}= \lambda^* N(\varphi^*,\sigma_{2}^{2}) + (1-\lambda^*) N(\theta^*,\sigma_2^2)$ where $\varphi^*=0$, $\theta^*=6$, $\sigma_{2}^{2}=1$ and $\sigma_{1}^{2}=4^2$. 
\begin{figure}[h]
\centering
    \includegraphics[width=6.5cm,height=5cm]{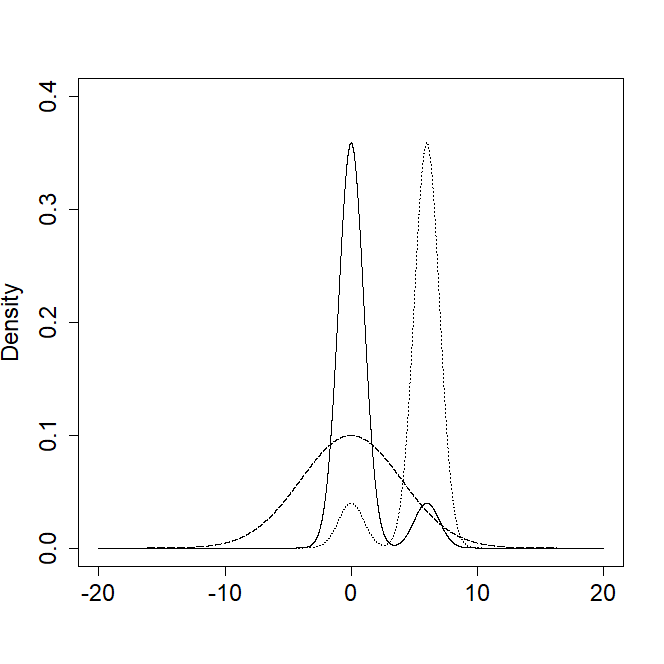}
    \caption{True data generative models; $\pstar_{x_1}$ is shown as a dashed line and $\pstar_{x_2}$ for $\lambda^*=0.9$ as a solid line and $\lambda^*=0.1$ as a dotted line.}
    \label{fig:normal_true}
\end{figure}
Here $\lambda^*$ determines the level of misspecification. See Figure \ref{fig:normal_true} for the densities. 

The $\gamma$-SMI posterior is    
$$\pi_\gamma(\varphi|\x)\pi(\theta|\x_2,\varphi) = N(\varphi; \mu_{\varphi,\gamma}, \sigma^2_{\varphi,\gamma}) N(\theta; \mu_{\theta|\varphi}, \sigma^2_{\theta|\varphi})$$

   where $\mu_{\varphi,\gamma} = \sigma^2_{\varphi,\gamma}\left( \frac{\gamma \sum_{i=1}^{n_2} \x_{2,i}}{\sigma^2_2+n_2 s^2_\theta} + \frac{\sum_{j=1}^{n_1} \x_{1,i}}{\sigma_1^2}\right)$, $\sigma^2_{\varphi,\gamma} = \left( \frac{\gamma n_2}{\sigma^2_2+n_2 s^2_\theta} + \frac{ n_1}{\sigma_1^2} \right)^{-1}$, \\ $\mu_{\theta|\varphi } = \rho(\sum^{n_2}_{i=2}\x_{2,i}/n_2-\varphi)$, $\sigma^2_{\theta|\varphi }= \left( \frac{1}{s_\theta^2} + \frac{n_2  }{\sigma_2^2} \right)^{-1}$ and $\rho = n_2 \sigma^2_{\theta|\varphi}/\sigma^2_2$.

    The $\eta$-SMI posterior is 
       $$\pi_\eta(\varphi|\x)\pi(\theta|\x_2,\varphi) = N(\varphi; \mu_{\varphi,\eta},\sigma^2_{\varphi,\eta}) N(\theta;\mu_{\theta|\varphi},\sigma_{\theta|\varphi}^2)$$
    where $\sigma^2_{\varphi,\eta} = \left( \dfrac{\eta n_2 }{ \sigma_2^2+n_2 \eta s_\theta^2} +\dfrac{n_1}{\sigma_1^2} \right)^{-1}$, 
 $\mu_{\varphi,\eta}=\sigma^2_{\varphi,\eta} \left(  \dfrac{\eta \sum_{i=1}^{n_2} \x_{2,i}}{ \sigma_2^2+n_2 \eta s_\theta^2} + \dfrac{\sum_{j=1}^{n_1} \x_{1,j}}{\sigma_1^2} \right) $ and $\pi(\theta|\x_2,\varphi)$ is the same as for the $\gamma$-SMI posterior. 

The $\beta$ SMI posterior is not analytically trackable, and we simulate the joint posterior, including an auxiliary variable $\theta'$ then marginalize it. The average loss function for $\varphi$ with an auxiliary variable $\theta'$ is simplified 
\begin{align*}
\frac{1}{n_2}\ell_{\beta}(\varphi,\theta';\x)=&
\frac{-1}{n_2(\beta-1)}\sum_{i=1}^{n_2} \left( \frac{e^{-(\x_{2,i}-\varphi-\theta)^2/2\sigma_2^2}}{\sqrt{2\pi \sigma_2^2}}\right)^{\beta-1}\\
&\quad + \quad \frac{1}{n_2\beta}\beta^{-\frac{3}{2}}(2\pi\sigma_2^2)^{-\frac{\beta-1}{2}}-\frac{\alpha}{n_1}\sum_{i=1}^{n_1}\log \frac{e^{(\x_{1,i}-\varphi)^2/2\sigma_1^2}}{\sqrt{2\pi \sigma_1^2}}
\end{align*}
With an auxiliary variable $\theta'$, the joint $\beta$-SMI posterior is 
\[ \pi_\beta(\theta,\varphi,\theta'|\x) = \pi_\beta(\varphi,\theta'|\x) \pi(\theta|\x_2,\varphi) \propto
        \exp(-\ell_{\beta}(\varphi, \theta'; \x))\pi(\varphi,\theta') \pi(\theta|\x_2,\varphi) \,. \]
For each $\beta$ value, a total of 8000 posterior samples are simulated after discarding 2000 as a burn-in and thinning by 10. 

\subsubsection{Loss for $s$}\label{Appendix:Normal_joint}

Let $\x=\{\x_1,\x_2\}$ be the training data learning parameters and $\y=\{\y_1,\y_2\}$ be the calibration data. Let $\overline{\y}_1$ and $\overline{\y}_2$ be the mean of $y_{1,1:J}$ and $y_{2,1:J}$ respectively. For simplicity, we give details for the pooled case $l(s;\y,\x)$ and write down results for the product case. 

We consider the $\gamma$-SMI posterior predictive for $\y_1$ and $\y_2$ for simplicity in notation. Since $\pi_\gamma(\varphi|\x)$ is Gaussian, the posterior predictive for $\y_{1}$ is analytically tractable, and $l(\gamma;\y_1,\x)$ is simplified 
\begin{align*}
 l(\gamma;\y_1,\x) & = -\log \int \prod^{J}_{i=1} p(y_{1,i}|\varphi,\sigma_1^2) \pi_\gamma(\varphi|\x) d\varphi \\ 
&= \frac{J}{2}\log(2\pi\sigma^2_1)-\frac{1}{2}\log\Big(\frac{\sigma_1^2}{J\sigma^2_{\varphi,\gamma}+\sigma^2_1}\Big)+\frac{1}{2\sigma_1^2}\sum^J_{i=1}(y_{1,i}-\overline{\y}_1)^2 + \frac{J(\overline{\y}_1-\mu_{\varphi,\gamma})^2}{2(J\sigma^2_{\varphi,\gamma}+\sigma^2_1)} \,.
\end{align*}
For the pooled loss, with Theorem 1, the optimal hyperparameter ($J\to\infty$) is 
\[ \widetilde{\gamma}_{(1,\infty),\,\x} = \argmin_\gamma \log \sigma^2_{\varphi,\gamma} + \frac{\mu_{\varphi,\gamma}^2}{\sigma^2_{\varphi,\gamma}} \,.\] 
For the product loss $l(\gamma;y,\x)=-\sum_i\log(p_\gamma(y_{1,i}|\x))$, the optimal hyperparameter (in the limit $J\to\infty$) is 
\begin{align*}
\widetilde{\gamma}_{(\infty,1),\,\x} & = \argmin_\gamma[\mathbb{E}_{y\sim \pstar_{x_1}}[ l(\gamma;y,\x)]] = \argmin_\gamma \frac{1}{2(\sigma^2_{\varphi,\gamma}+\sigma^2_1)} \Big( {\mathbb{E}[Y^2]}+\mu^2_{\varphi,\gamma}-2\mu_{\varphi,\gamma}\mathbb{E}[Y] \Big) \\ & = \argmin_\gamma \frac{1}{2(\sigma^2_{\varphi,\gamma}+4^2)} (4^2+\mu^2_{\varphi,\gamma})\end{align*}




The loss function for $\eta$ associated with the $\eta$-SMI posterior and the optimal values are obtained by replacing $\sigma^2_{\varphi,\gamma}$ by $\sigma^2_{\varphi,\eta}$ and $\mu_{\varphi,\gamma}$ by $\mu_{\varphi,\eta}$. 

The loss function for $b=1/\beta$ is associated with $\beta$-SMI. Given $\beta$-SMI posterior samples $\varphi_1,..,\varphi_T \sim \pi_\beta(\varphi|\x)$, the posterior predictive $p_\beta(\y_1|\x)$ for the pooled loss is estimated using the Monte Carlo method,
\[e^{-l_{(1,J)}(\beta;\y_{1},\x)} \approx \frac{1}{T} \sum^T_{j=1} \prod^J_{i=1}p(y_{1,i}|\varphi_j,\sigma^2_1)\,, \]
and in the product case
\[
e^{-l_{(J,1)}(\beta;\y_{1},\x)}\approx \prod^J_{i=1} \left[\frac{1}{T} \sum^T_{j=1} p(y_{1,i}|\varphi_j,\sigma^2_1)\right]\,.
\]

\subsubsection{$s$-estimation under varying levels of misspecification}\label{sec:normal1}

In this experiment the training data is $\x=\{\x_{1,1:n_1},\x_{2,1:n_2}\}$ with $n_1=30$ and $n_2=60$. The calibration data $\y$ has a fixed size $J=10^4$ blocks (and $J=10^3$ for $\beta$-SMI simulation). Viewing the $\eta$/$\gamma$-SMI posterior as interpolating between the cut and conventional posterior distributions, we assign $\eta$ and $\gamma$ uniform priors on $[0,1]$. Figures \ref{fig:normal_loss_gamma_rep} and \ref{fig:normal_loss_b_rep} illustrate the behaviour of the posterior mean estimate of the $\gamma$-SMI posterior and $\beta$-SMI posterior respectively under varying levels of misspecification, based on 100 repetitions (i.e., 100 realizations of $\x$). Both pseudo optimal parameters remain small, with low variation, until $\lambda^*$ becomes sufficiently high ($\approx 1$), with $\widetilde{b}_{(1,\infty)}$ showing less sensitivity on average. Under over-parametrization ($\lambda^*=1$), the variation in both optimal parameters increases. Specifically $\widetilde{\gamma}_{(1,\infty)}$ centers around 1 and $\widetilde{b}_{(1,\infty)}$ centers around 1.7. The $\beta$-loss encompasses a broader range of belief updates (as discussed in Section 2.1) and effectively distinguishes between underfitting ($\lambda^*<1$) and overfitting ($\lambda^*=1$). 

\begin{figure} 
    \centering
    \begin{overpic}[ width=.4\linewidth,height=0.15\textheight ]{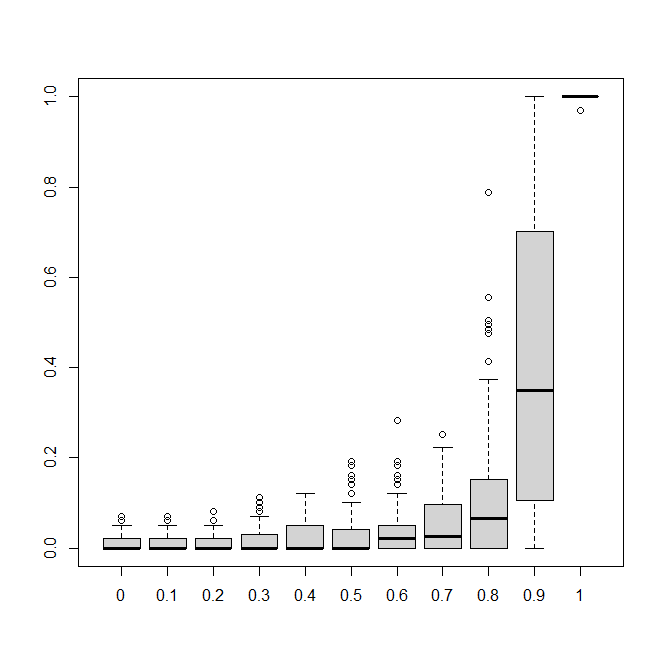}  
    \put(40,0){\tiny $\lambda^*$ values} \put(0,20){\rotatebox{90}{\tiny $\widetilde{\gamma}_{(1,\infty)}$}}
    \end{overpic}
    \begin{overpic}[ width=.4\linewidth,height=0.15\textheight ]{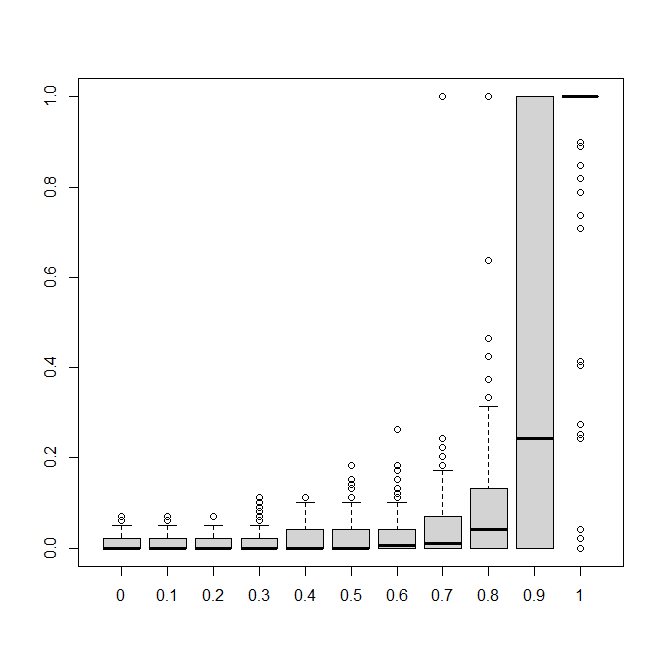}  
    \put(40,0){\tiny $\lambda^*$ values} \put(0,20){\rotatebox{90}{\tiny $\widetilde{\gamma}_{(\infty,1)}$}}
    \end{overpic}
    \begin{overpic}[ width=.4\linewidth,height=0.15\textheight ]{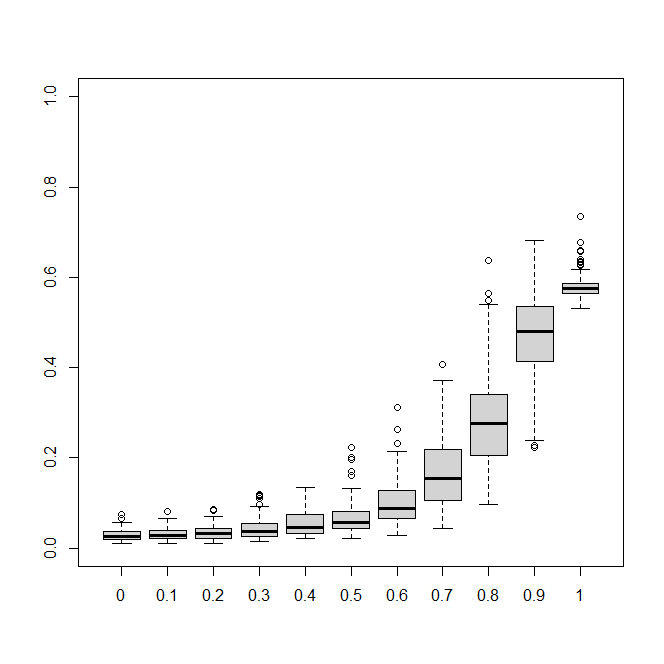}  
     \put(40,0){\tiny $\lambda^*$ values}  \put(0,20){\rotatebox{90}{\tiny $\widehat{\gamma}_{(1,J)}$}}
    \end{overpic}
    \begin{overpic}[ width=.4\linewidth,height=0.15\textheight ]{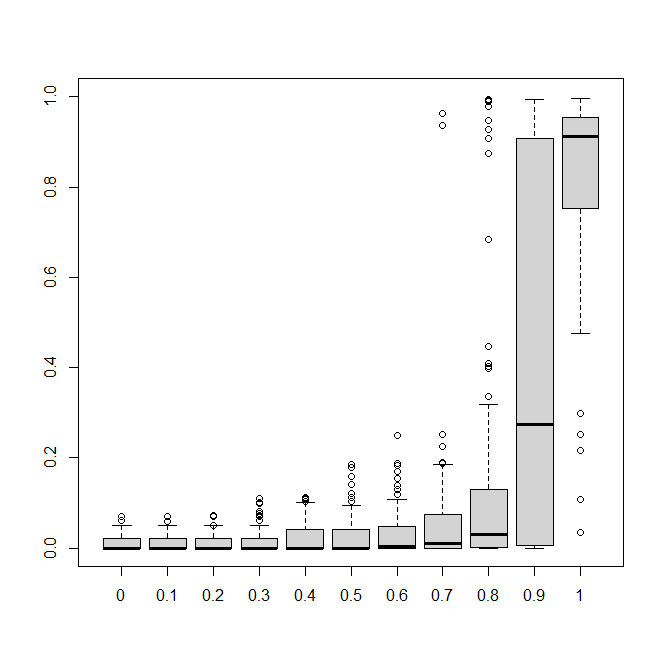}  
    \put(40,0){\tiny $\lambda^*$ values}  \put(0,20){\rotatebox{90}{\tiny $\widehat{\gamma}_{(J,1)}$}}
    \end{overpic}
     \caption{Summary of $\widetilde{\gamma}_{(1,\infty)}$ and posterior mean estimates, $\widehat{\gamma}_{(1,J)}$ and $\widehat{\gamma}_{(J,1)}$ with $J=10^4$ from 100 replicates as described in Section \ref{sec:normal1}. }
    \label{fig:normal_loss_gamma_rep}
\end{figure}

\begin{figure}
    \centering
    \begin{overpic}[ width=.4\linewidth,height=0.2\textheight ]{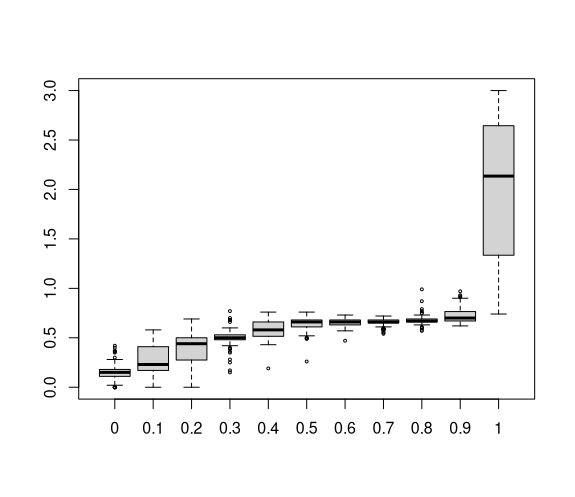}  
    \put(40,0){\tiny $\lambda^*$ values} \put(-5,30){\rotatebox{90}{\tiny $\tilde{b}_{(1,\infty)}$}}
    \end{overpic}
    \hspace{0.4cm}
    \begin{overpic}[ width=.4\linewidth,height=0.2\textheight ]{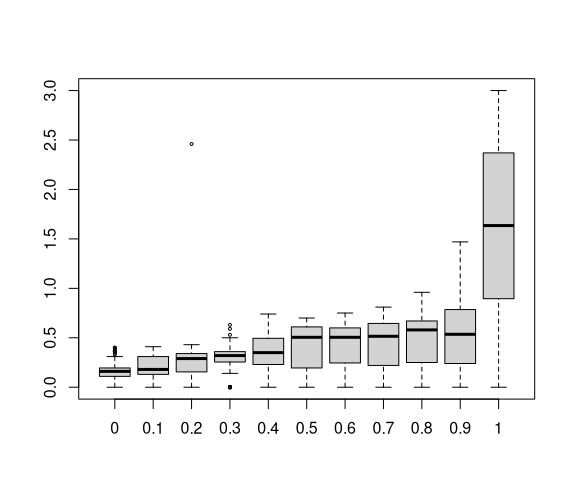}  
    \put(40,0){\tiny $\lambda^*$ values} \put(-5,30){\rotatebox{90}{\tiny $\widetilde{b}_{(\infty,1)}$}}
    \end{overpic}
    
    \begin{overpic}[ width=.4\linewidth,height=0.2\textheight ]{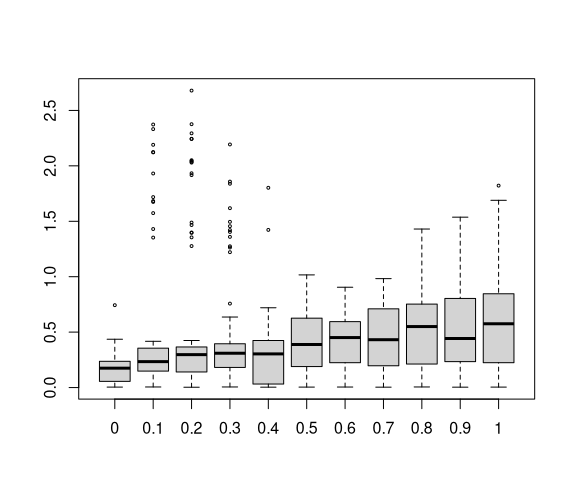}  
     \put(40,0){\tiny $\lambda^*$ values}  \put(-5,30){\rotatebox{90}{\tiny $\widehat{b}_{(1,J)}$}}
    \end{overpic}
    \hspace{0.4cm}
    \begin{overpic}[ width=.4\linewidth,height=0.2\textheight ]{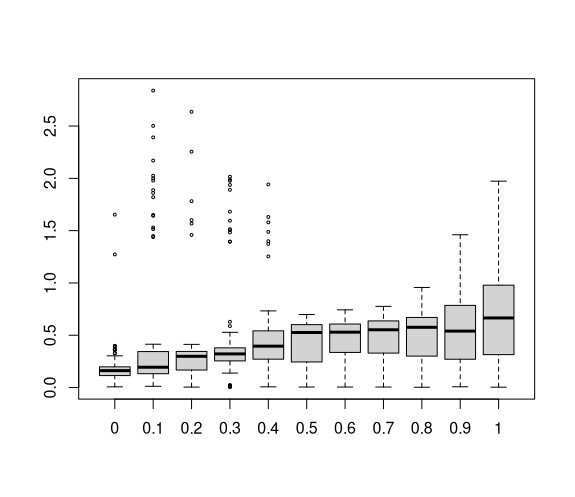}  
    \put(40,0){\tiny $\lambda^*$ values}  \put(-5,30){\rotatebox{90}{\tiny $\widehat{b}_{(J,1)}$}}
    \end{overpic}
     \caption{Summary of $\widetilde{b}_{(1,\infty)}$ and posterior mean estimates, $\widehat{b}_{(1,J)}$ and $\widehat{b}_{(J,1)}$ with $J=10^3$ from 100 replicates as described in Section \ref{sec:normal1}.}
    \label{fig:normal_loss_b_rep}
\end{figure}

Simulation results for $\eta$-SMI in Figure~\ref{fig:normal_eta_rep} show some similarities with the $\gamma$-SMI (left column, product-posterior in Figure~\ref{fig:normal_loss_gamma_rep}) as we might expect from Proposition~\ref{prop:eta-gamma-first-order} (despite the sample sizes $n_1,n_2$  in $\x$, which are only moderate and non-asymptotic). 

\begin{figure}[ht]
    \centering
    \begin{overpic}[ width=.8\linewidth,height=0.2\textheight ]{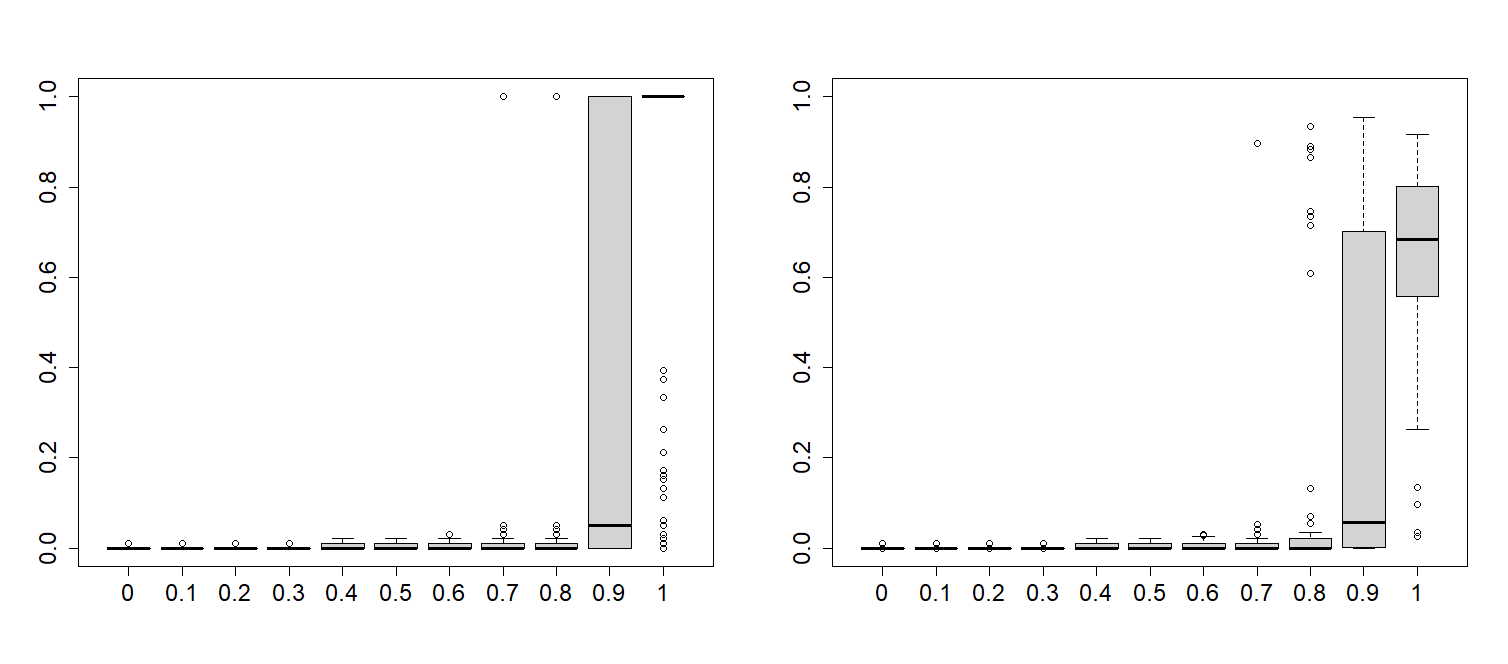}  
    \put(20,0){\tiny $\lambda^*$ values} \put(20,35){\tiny $\widetilde{\eta}_{(\infty,1)}$}
    \put(70,0){\tiny $\lambda^*$ values} \put(70,35){\tiny $\widehat{\eta}_{(J,1)}$}
    \end{overpic}
     \caption{Summary of $\widetilde{\eta}_{(\infty,1)}$ and posterior mean estimates, $\widehat{\eta}_{(J,1)}$, with $J=10^4$ from 100 replicates as described in Section \ref{sec:normal1}. }
    \label{fig:normal_eta_rep}
\end{figure}

\subsection{Asymptotic distribution of $s$ - further experiments}\label{Appendix:s-post-asym-further}

The $J$-dependence of the posterior density for $b$, $\gamma$ and $\eta$ was illustrated in Figure~\ref{fig:normal_loss_b_gam_eta} in Section~\ref{sec:normal2} for mild misspecification. In Figures \ref{fig:normal_loss_cut} and \ref{fig:normal_loss_overparam} we take two extremes of misspecification, $\lambda^*=0.08$ and 1. For $\eta$ and $\gamma$ in the product-posteriors in Figures \ref{fig:normal_loss_cut} and \ref{fig:normal_loss_overparam} the optimal $s$-values are on the boundary at 0 for $\lambda^*=0$ and at 1 for $\lambda^*=1$. These cases are governed by the non-standard analysis in Section~\ref{asymptotic_boundary}.

The flexible representation of misspecification in $\beta$-divergence becomes more evident in the asymptotic regime; $\widetilde{b}_{(1,\infty)}$ tends to be larger than $\widetilde{\gamma}_{(1,\infty)}$ and the overfit (signalled by $b>1$) and misfit are identified. A higher posterior concentration for $b$ is observed in comparison to $\gamma$ and $\eta$ posteriors. 

In general $\widehat{s}_{KL}$ (the cross marks) shows the expected downward shift relative to the mean (triangles).
This is clearest in Figure~\ref{fig:normal_loss_b_gam_eta} at mild misspecification and small calibration sample sizes. This is often a shift towards $\widetilde{s}_{(1,\infty)}$.

\begin{figure}  
    \centering
    \begin{overpic}[ width=.32\linewidth,height=0.17\textheight ]{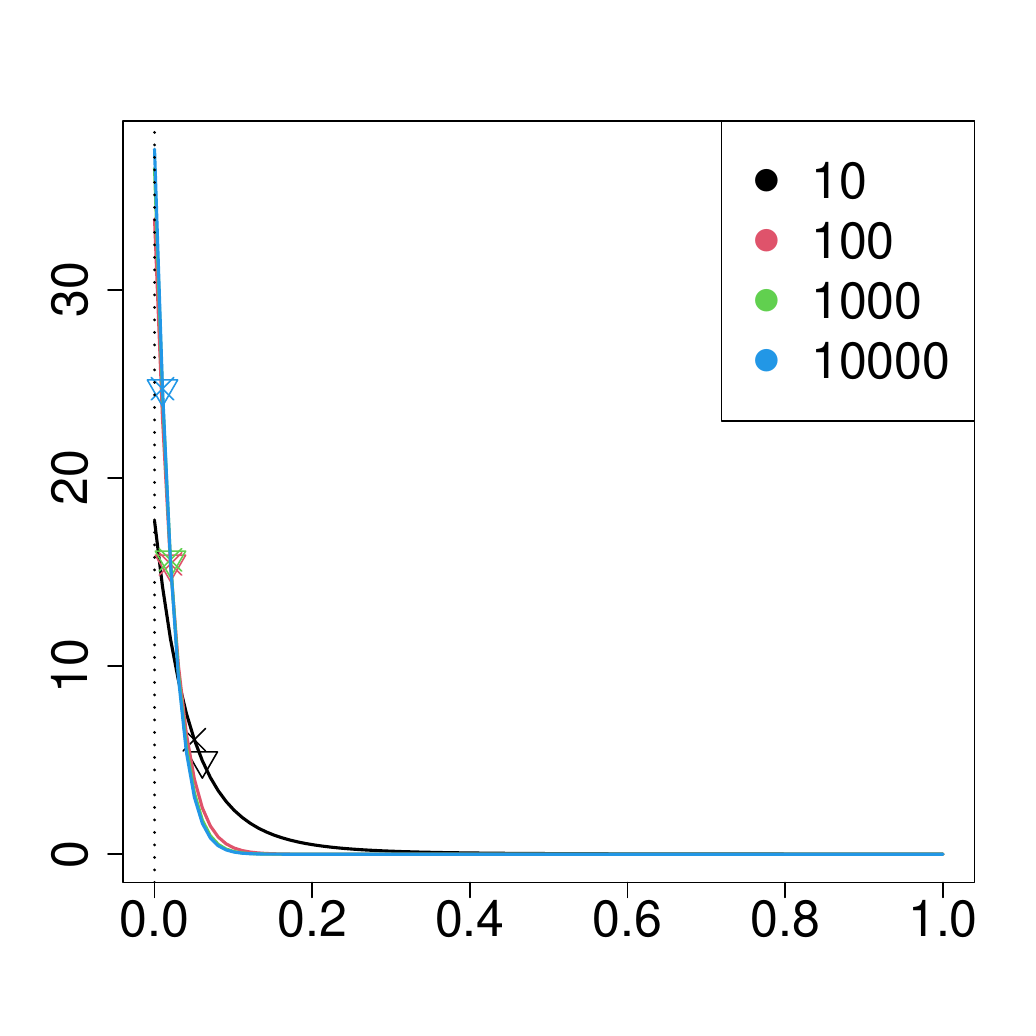}  
    \put(40,0){\tiny $\gamma$ values} \put(-5,25){\rotatebox{90}{\tiny $\rho_{(1,J)}(s|\y,\x)$}}
    \end{overpic}
    \begin{overpic}[ width=.32\linewidth,height=0.17\textheight ]{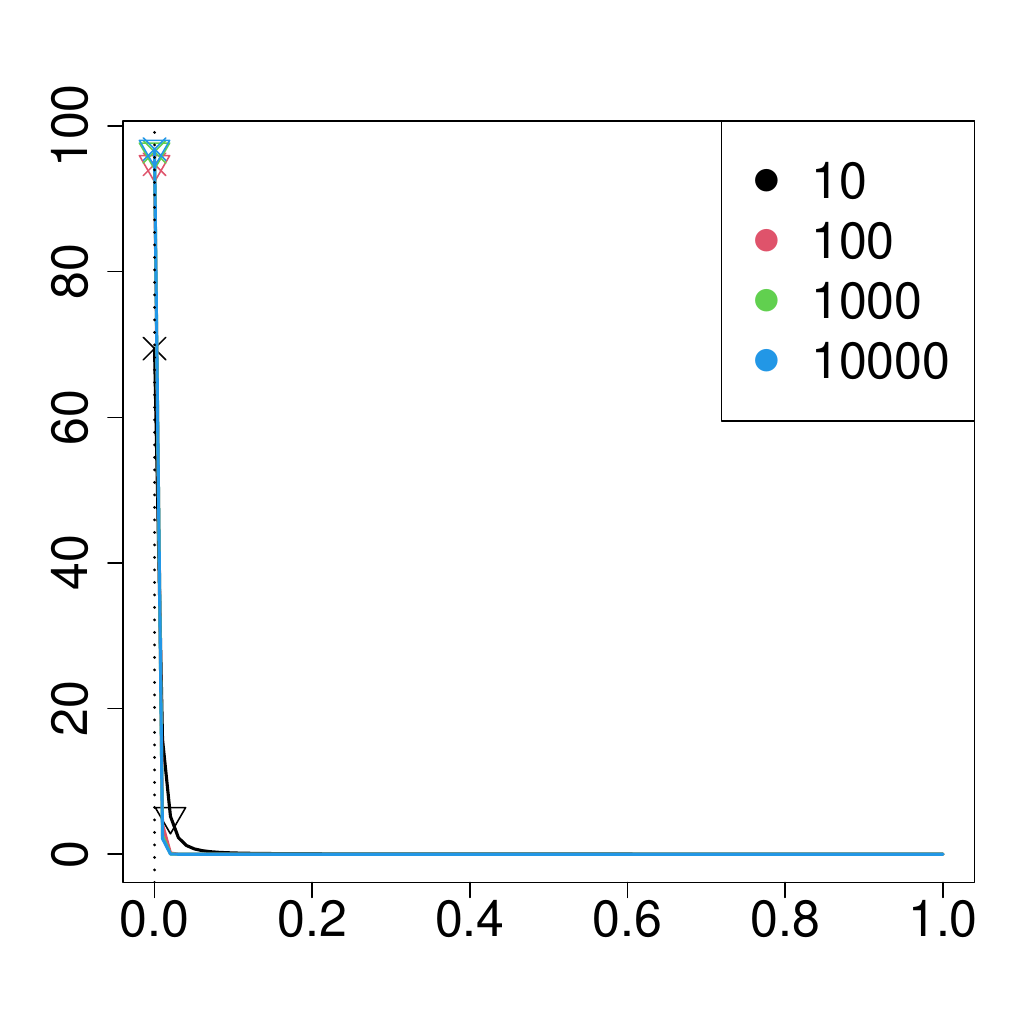}  
     \put(40,0){\tiny $\eta$ values}  \put(-5,25){\rotatebox{90}{\tiny $\rho_{(1,J)}(\eta|\y,\x)$}}
    \end{overpic}
    \begin{overpic}[width=.32\linewidth,height=0.17\textheight ]{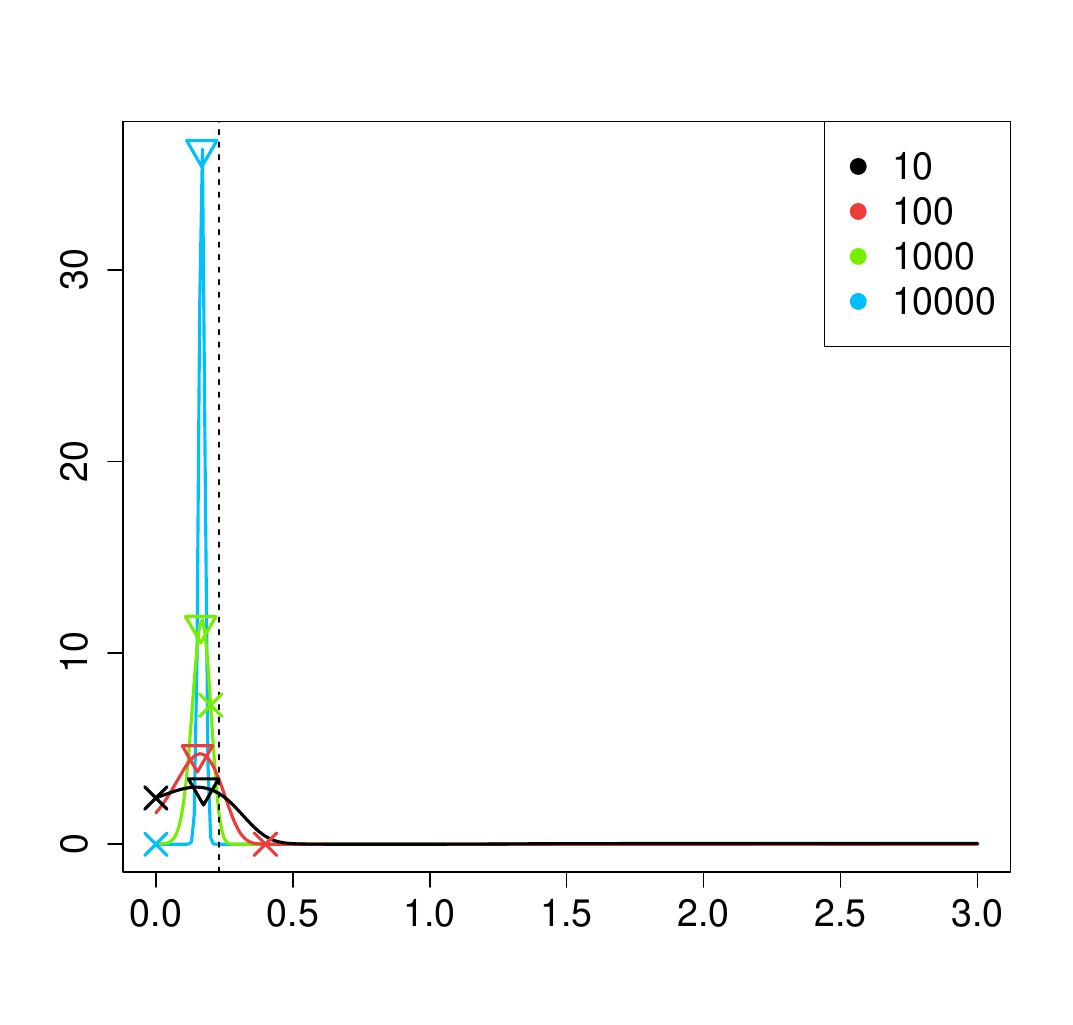}  
     \put(40,0){\tiny $b$ values}  \put(-5,25){\rotatebox{90}{\tiny $\rho_{(1,J)}(b|\y,\x)$}}
    \end{overpic}
    \begin{overpic}[width=.32\linewidth,height=0.17\textheight ]{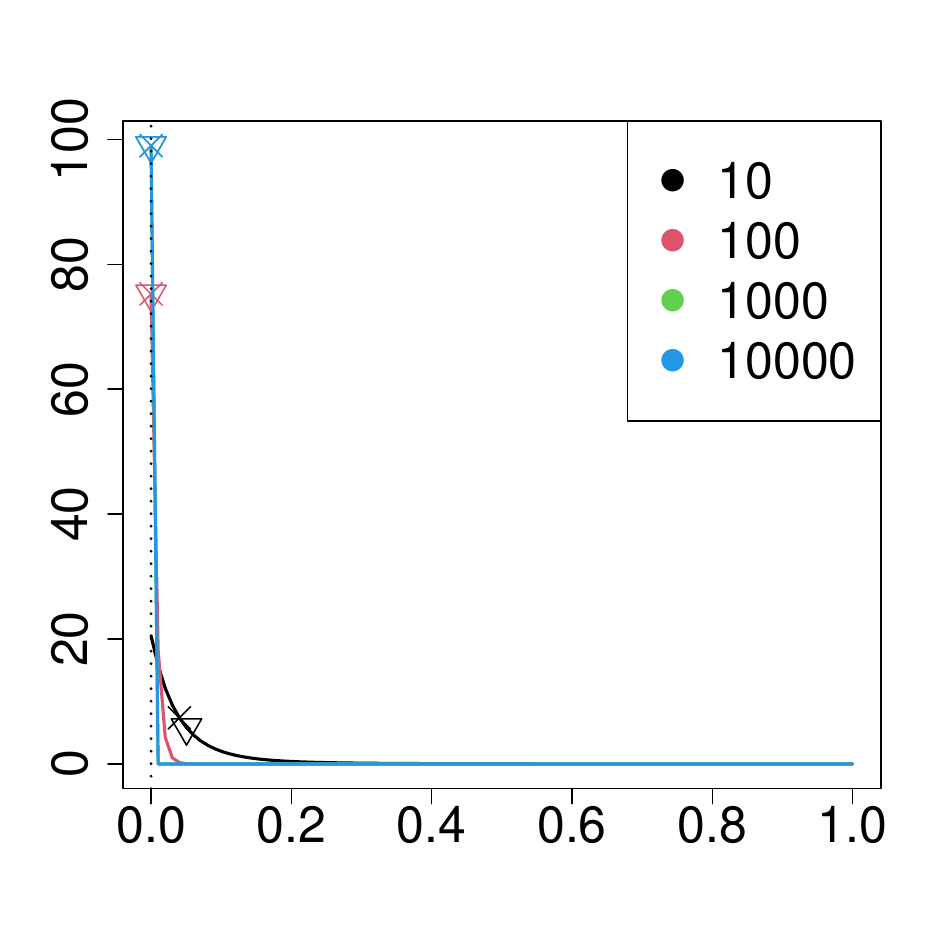}  
    \put(40,0){\tiny $\gamma$ values}  \put(-5,25){\rotatebox{90}{\tiny $\rho_{(J,1)}(\gamma|\y,\x)$}}
    \end{overpic}
    \begin{overpic}[width=.32\linewidth,height=0.17\textheight ]{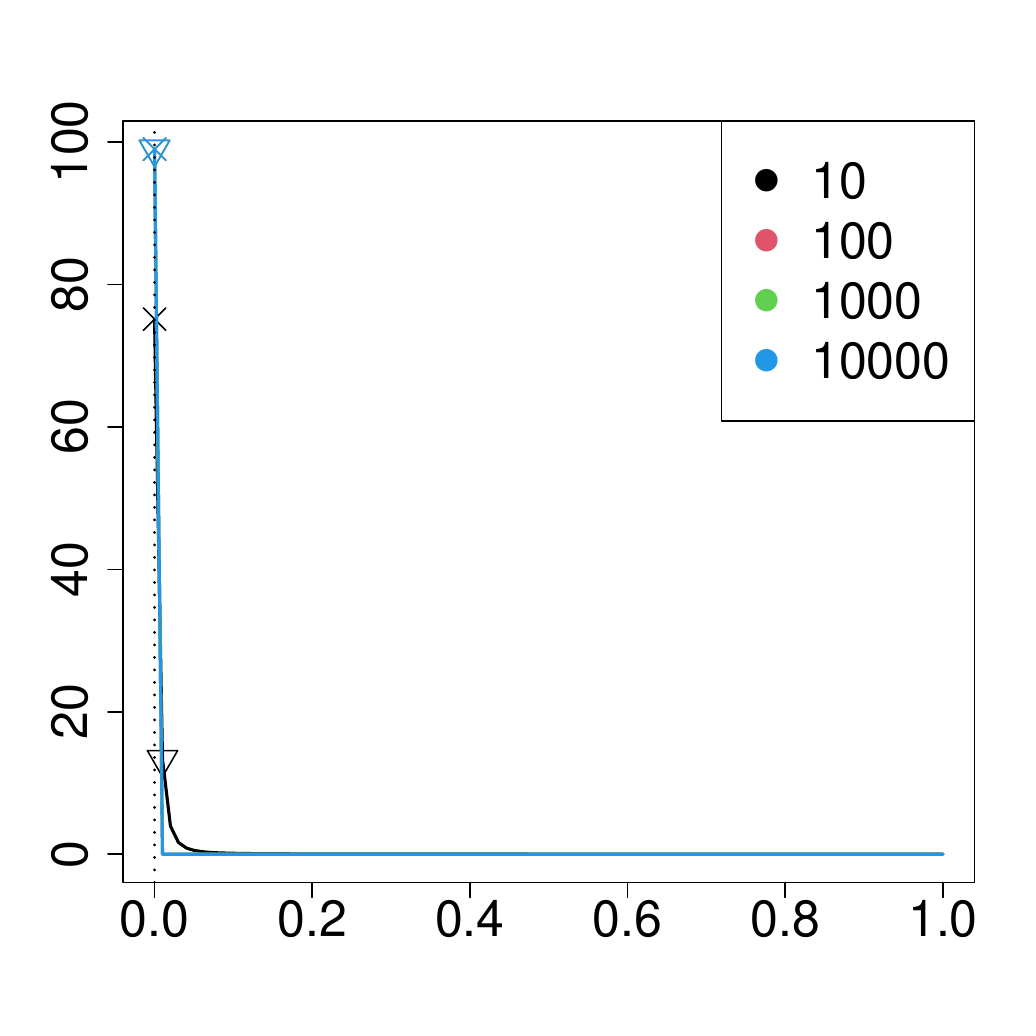}  
     \put(40,0){\tiny $\eta$ values} \put(-5,25){\rotatebox{90}{\tiny $\rho_{(J,1)}(\eta|\y,\x)$}}
    \end{overpic}
    \begin{overpic}[width=.32\linewidth,height=0.17\textheight ]{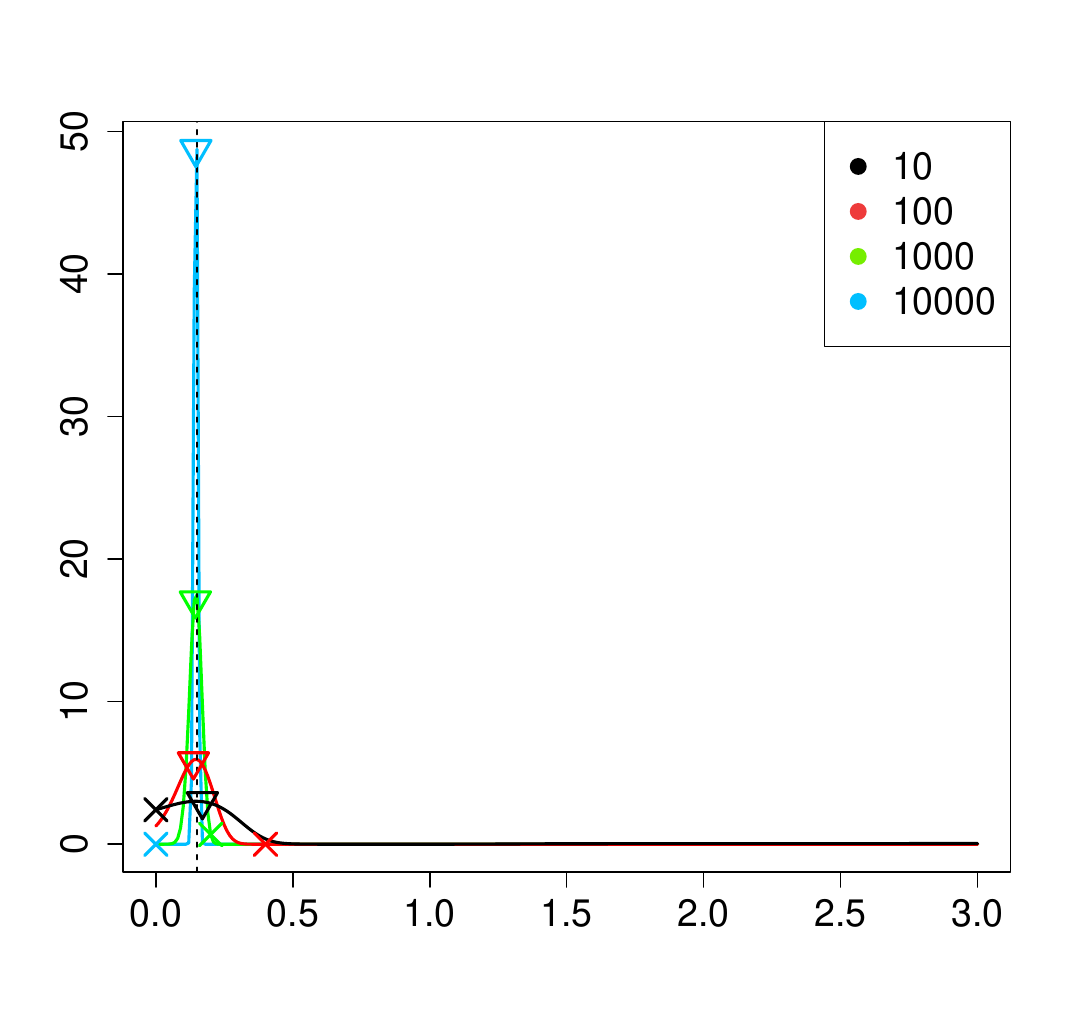}  
    \put(40,0){\tiny $b$ values} \put(-5,25){\rotatebox{90}{\tiny $\rho_{(J,1)}(b|\y,\x)$}}
    \end{overpic}
     \caption{Posterior density for $\eta$, $\gamma$ and $b$ with various sizes of calibration data as described in Section \ref{sec:normal2} under $\lambda^*=0.08$; colour code represents $J$. Top/bottom rows are pooled/product calibration. The vertical dotted lines, triangle marks and cross marks indicate the optimal value ($\widetilde{s}$), posterior mean ($\widehat{s}$) and harmonic mean ($\widehat{s}_{KL}$), respectively.}
    \label{fig:normal_loss_cut}
\end{figure}

\begin{figure}  
    \centering
    \begin{overpic}[ width=.32\linewidth,height=0.17\textheight ]{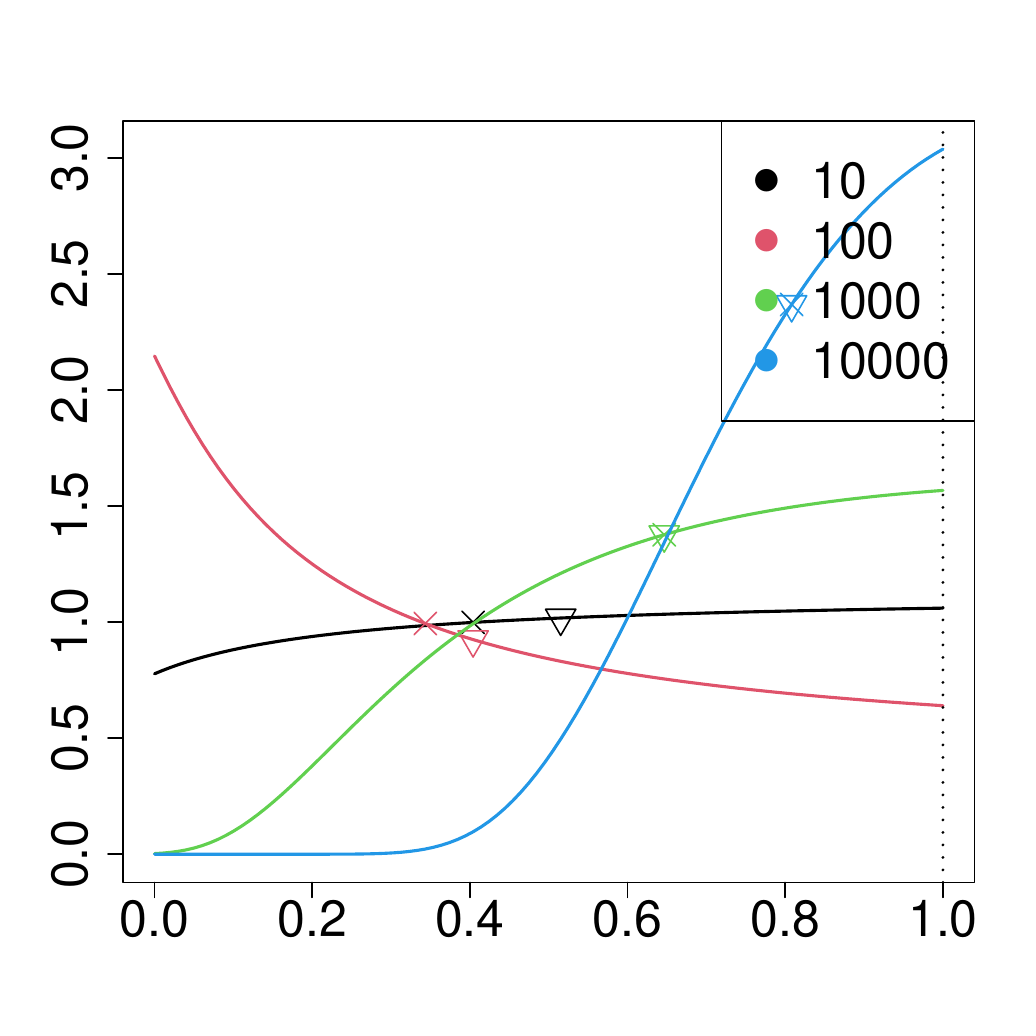}  
    \put(40,0){\tiny $\gamma$ values} \put(-5,25){\rotatebox{90}{\tiny $\rho_{(J,1)}(\gamma|\y,\x)$}}
    \end{overpic}
    \begin{overpic}[ width=.32\linewidth,height=0.17\textheight ]{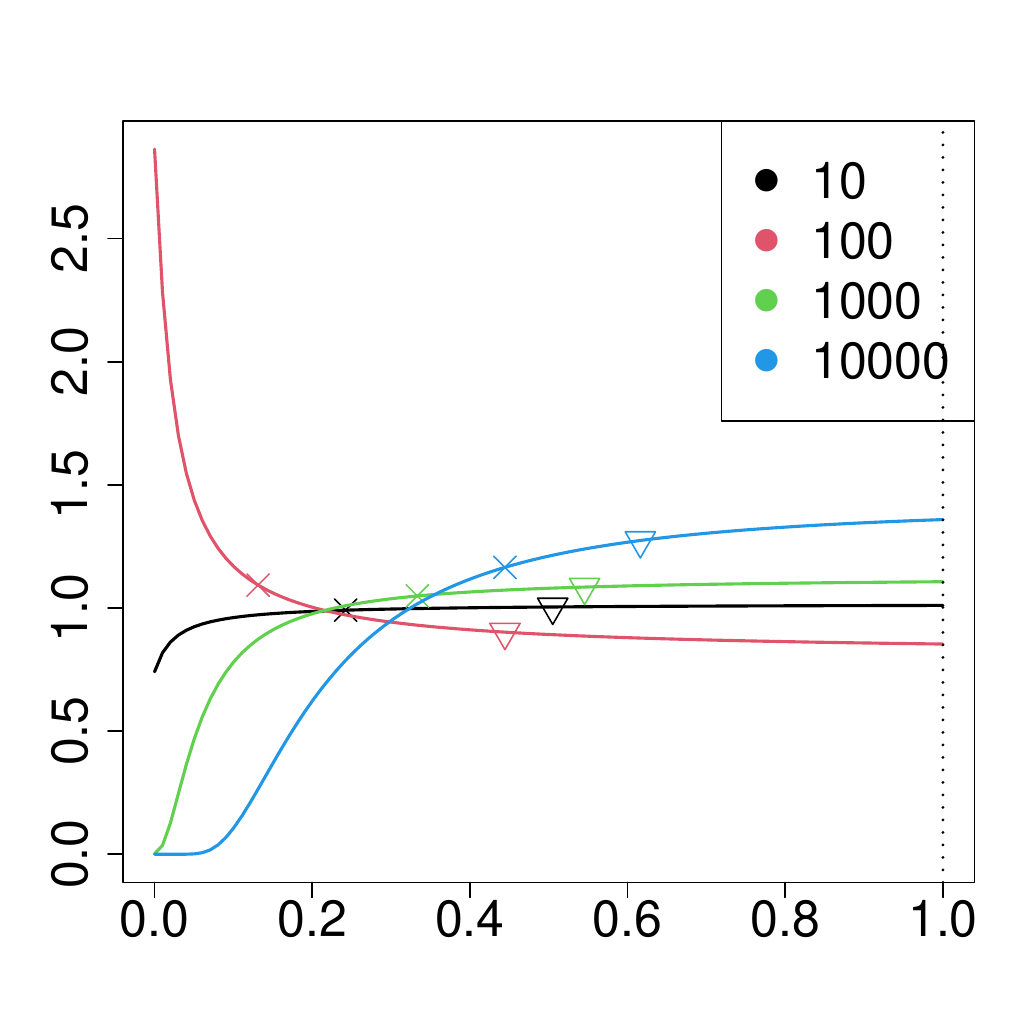}  
     \put(40,0){\tiny $\eta$ values}  \put(-5,25){\rotatebox{90}{\tiny $\rho_{(J,1)}(\eta|\y,\x)$}}
    \end{overpic}
    \begin{overpic}[width=.32\linewidth,height=0.17\textheight ]{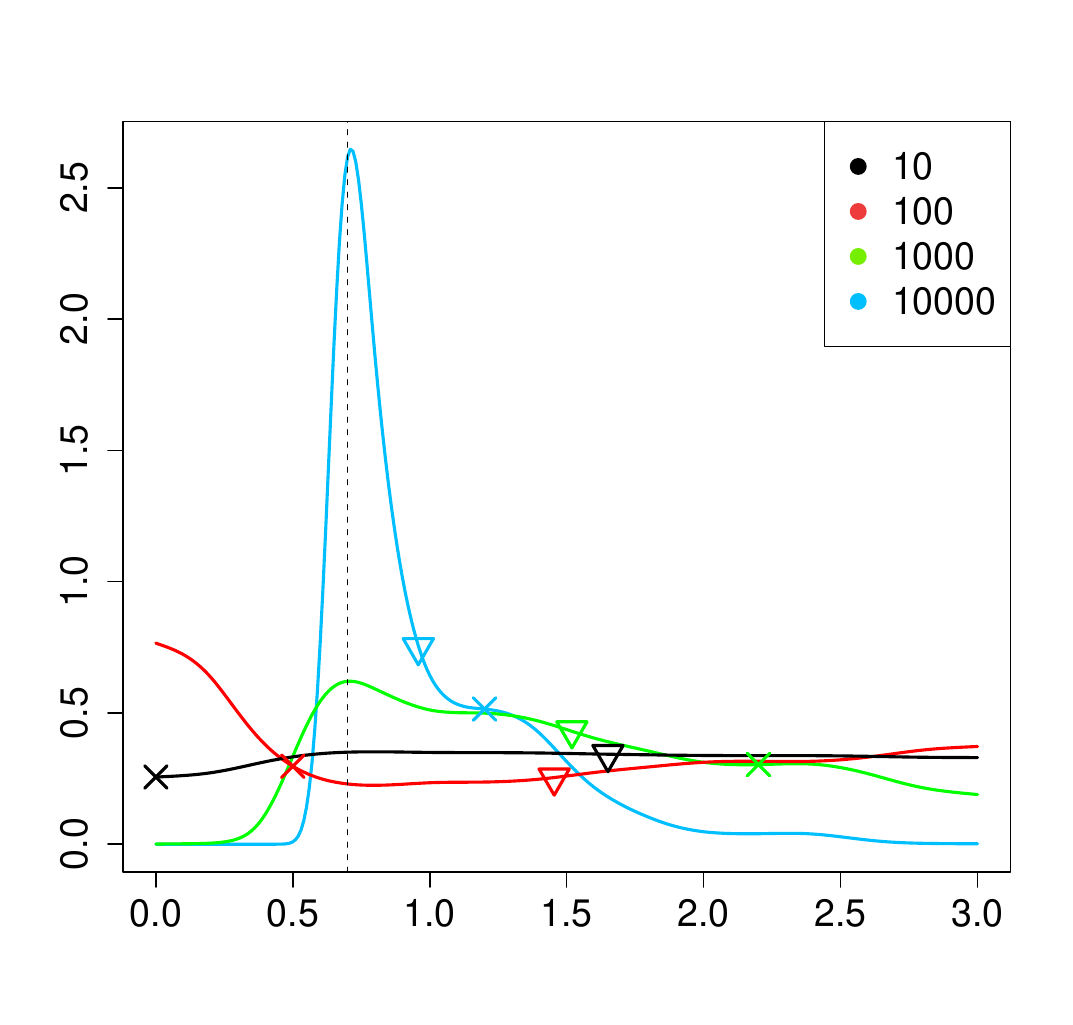}  
     \put(40,0){\tiny $b$ values}  \put(-5,25){\rotatebox{90}{\tiny $\rho_{(J,1)}(b|\y,\x)$}}
    \end{overpic}
    \caption{Product-posterior density for $\eta$, $\gamma$ and $b$ with various sizes of calibration data as described in Section \ref{sec:normal2} under $\lambda^*=1$; colour code represents $J$. The vertical dotted lines, triangle marks and cross marks indicate the optimal value ($\widetilde{s}$), posterior mean ($\widehat{s}$) and harmonic mean ($\widehat{s}_{KL}$), respectively.}
    \label{fig:normal_loss_overparam}
\end{figure}

The similar behavior of the $\eta$ and $\gamma$ posteriors in Figure~\ref{fig:normal_eta_rep} and Figure~\ref{fig:normal_loss_gamma_rep} are expected from Proposition~\ref{prop:eta-gamma-first-order} but the result only holds as the size of the training data $n_1$ and $n_2$ grow at fixed $\alpha=n_1/n_2$. The two right columns of Figure~\ref{fig:normal_loss_b_gam_eta} also show qualitatively similar behaviours for $\eta$ and $\gamma$ posteriors. However, the result applies to $\pi_s(\varphi,\theta|\x)$ and only secondarily to $\rho(s|\y,\x)$. 

In Figure~\ref{fig:normal_loss_b_gam_eta} we have only $n_1=30$ and $n_2=60$. Looking directly at the joint distribution of $\pi_s(\varphi,\theta|\x)$ at much larger $n_1$-values we see in Figure~\ref{fig:smipost} that the $\gamma$ and $\eta$ SMI-posteriors converge as $n_1\to\infty$ for a fixed $\alpha=0.5$, consistent with Proposition~\ref{prop:eta-gamma-first-order}. This will mean the $\rho(s|\y,\x)$ posteriors converge for both product and pooled losses as the posterior predictive distributions on which they are based converge.

\begin{figure}    
    \includegraphics[width=5cm,height=5cm]{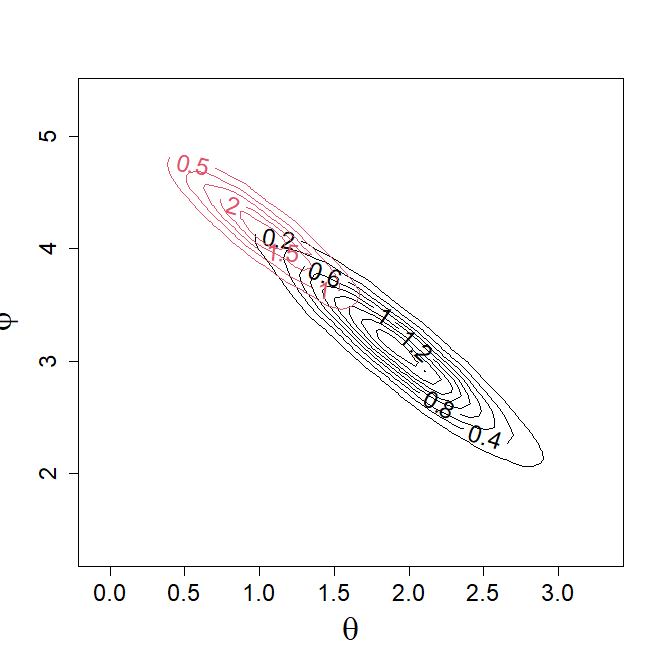}
    \includegraphics[width=5cm,height=5cm]{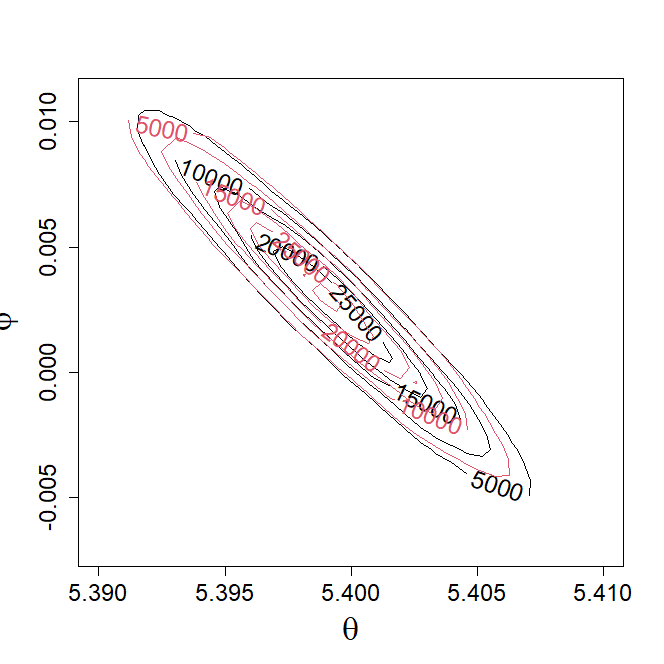}
     \caption{Joint $\gamma$-SMI posterior density (black contour) and $\eta$-SMI posterior density (red contour) comparison for $n_1=30$ (left plot) and $n_1=10^6$ (right plot) while $\alpha=0.5$ and $\gamma=\eta=0.3$ as described in Section \ref{sec:normal2}.}
    \label{fig:smipost}
\end{figure}


\subsubsection{Expected risk ratio comparison for Normal-mixture experiment}\label{sec:normal3}

When we learn the hyperparameters $s$, we split the data into training and calibration data sets. We don't in practice have the test data to which the fitted belief update will ultimately be applied. In this section we ask whether we gain anything over simply using the Bayesian or Cut-model belief update on the combined training and calibration data. We also investigate how performance depends on the split fraction into training and test data.

We split the available data $\overline{\x}=\{\overline{x}_{1,1:n_1},\overline{x}_{2,1:n_2}\}$ into a training dataset $\x$ and calibration dataset $\y$. Given $0<k<1$, the training dataset is defined as $\x=(\x_1,\x_2)$ with $\x_1=\bar x_{1,1:kn_1}$ and $\x_2=\bar x_{2,1:n_2}$. We calibrate by learning to predict $\x_1$-module data only, so the calibration dataset $\y$, is the complement of $\x_1$ within $\overline{\x}_1$. Here $\y=(y_{1,1},\dots,y_{1,J})$ with $J=(1-k)n_1$. 

We examine the posterior prediction quality and expected risk ratio using the test data for $\x_1$-module, since calibration targeted prediction of $\y_1$. Since we only have test data for the $\x_1$ module, denote by $\z=(z_{1,1},\dots,z_{1,J^{(z)}})$ so there are $J^{(z)}$ samples in the test data. The optimal $s$ value {\it for predicting the test data using all the training and calibration data} will be 
$\bar{s}_{(J^{(z)},1)} = \argmin_s l_{(J^{(z)},1)}(s;\z, \overline{\x})$ (in the product posterior) and $\bar{s}_{(1,J^{(z)})} = \argmin_s l_{(1,J^{(z)})}(s;\z, \overline{\x})$ (in the pooled case).

Let $n_2=60$, $J^{(z)}=n_1=30$ be the data set size parameters before splitting so we have $k\times 30$ training samples in $\x$ and $J=(1-k)\times 30$ calibration samples in $\y$, for a range of values of $k$ between $0$ and $1$. We generate 100 realizations of $\overline{\x}$. For each $\overline{\x}$, we estimate the expected risk ratio $R(s,s_o)$, defined in (\ref{eq:s_relative}), using 30 test datasets with misspecification parameter $\lambda^*=0.9$ (mild misspecification) and split proportions $k=(0.2,0.4,0.6,0.8)$. The reference value $s_o$ corresponding to the cut ($\gamma_o=\eta_o=b_o=0$) or conventional Bayes ($\gamma_o=\eta_o=b_o=1$) posterior is considered and compared against different estimators for $\widetilde{s}$ in different forms of SMI. 

In Figure~\ref{fig:normal_b_gamma_eta_prod_pred} we plot the distribution of $R$-values. Rows are SMI-variants $\beta$, $\gamma$ and $\eta$ and columns compare the two $s$-estimators MAP (with a bar) and posterior mean (with a hat) with Cut and Bayes. The dependence on the train/calibrate split fraction ($k$) is shown on the $x$-axis. 
A box sitting above the horizontal line at $R=1$ indicates the SMI-variant/estimator pair predicts with lower expected risk ratio than the comparator. 

In summary, both estimators, MAP and mean of $\rho(s|\y,\x)$ do about equally well. The $\eta$ and $\gamma$ SMI results (in the bottom two rows) are slightly better than $\beta$ (top row) and uniformly out-perform Bayes and Cut. 
There is little dependence on the split proportion $k$ in the bottom two rows. Interestingly, $b$-estimation for $\beta$-SMI does best when only $20\%$ of the data are used for training and $80\%$ are used for test. This is surprising because $\hat b$ and $\bar b$ (and all the other estimators) are estimated using an SMI-posterior for $\varphi$ and $\theta$ based on training data alone and we simply assume this will be the right adjustment when we come to predict test data using a posterior based on $\x$ and $\y$. Using a small fraction of data for training and a large fraction for calibration tests this assumption strongly, but it seems to hold up well.
\begin{figure}
    \centering
    \begin{overpic}[ width=.24\linewidth,height=0.13\textheight ]{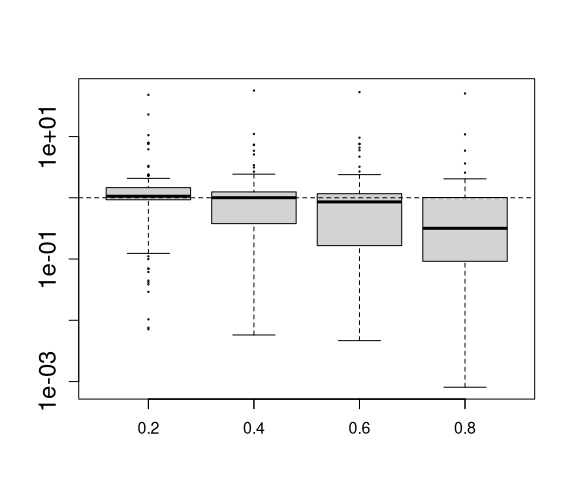}\put(30,78){\tiny $R(\widehat{b}_{(J,1)},0)$} 
    \end{overpic}   
    \begin{overpic}[ width=.24\linewidth,height=0.13\textheight ]{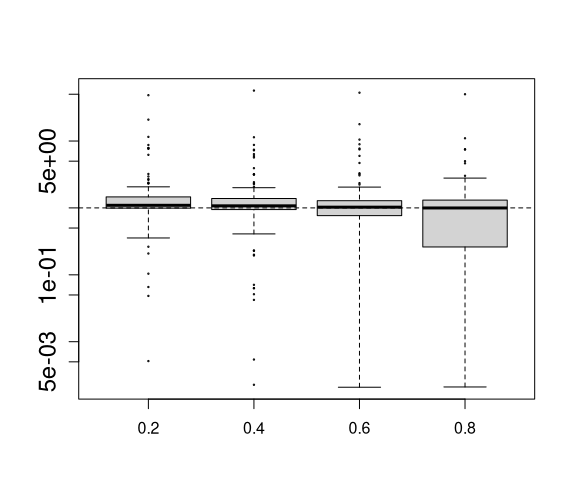}\put(30,78){\tiny $R(\overline{b}_{(J,1)},0)$}  
    \end{overpic}
    \begin{overpic}[ width=.24\linewidth,height=0.13\textheight ]{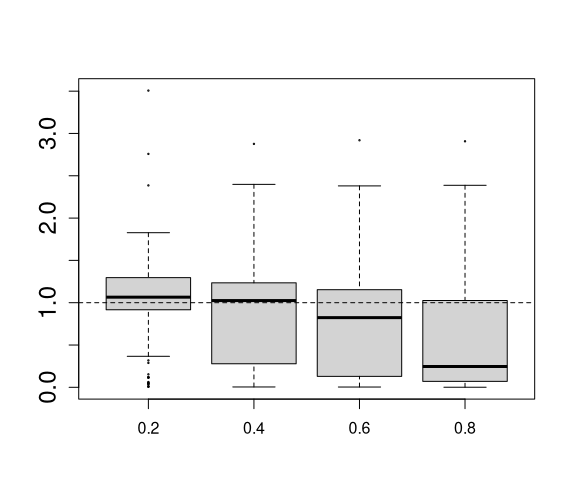}\put(30,78){\tiny $R(\widehat{b}_{(J,1)},1)$} 
    \end{overpic}
    \begin{overpic}[ width=.24\linewidth,height=0.13\textheight ]{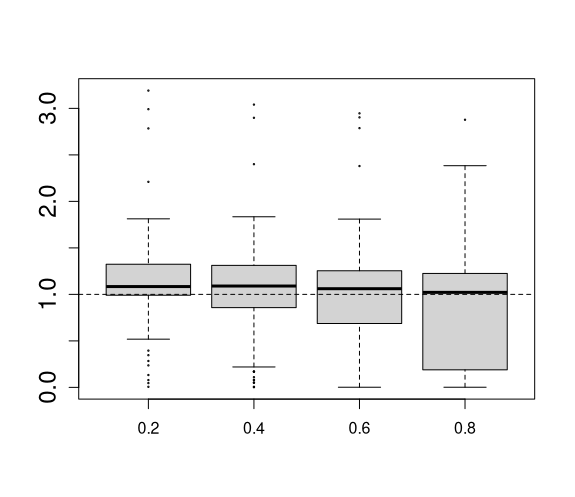}\put(30,78){\tiny $R(\overline{b}_{(J,1)},1)$}  
    \end{overpic} 

    \vspace{0.4cm}
    \begin{overpic}[ width=.23\linewidth,height=0.12\textheight ]{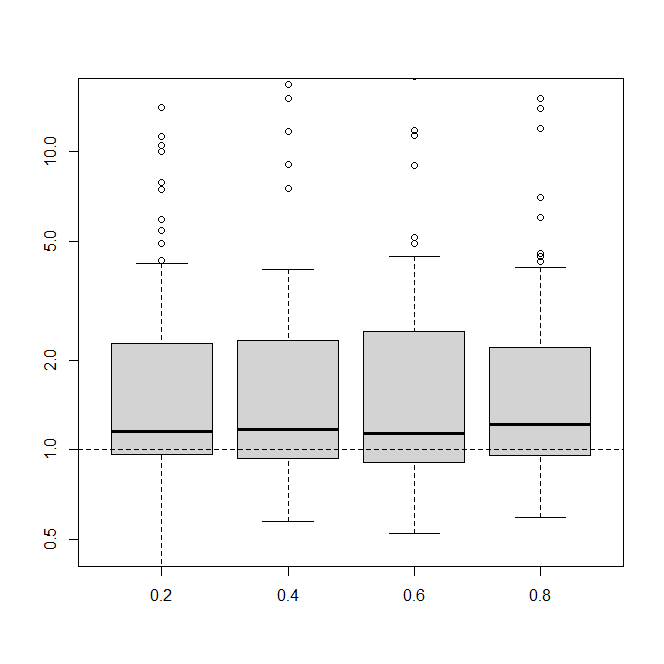}\put(30,78){\tiny $R(\widehat{\gamma}_{(J,1)},0)$}  
    \end{overpic}
    \begin{overpic}[ width=.23\linewidth,height=0.12\textheight ]{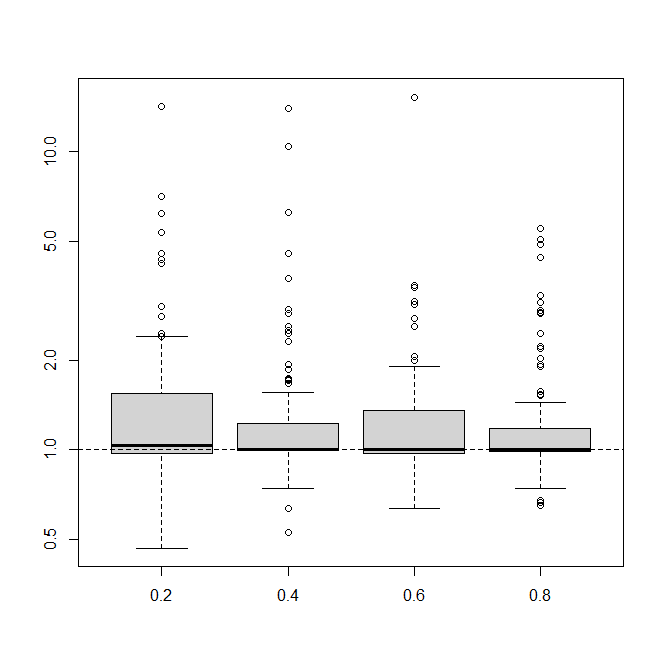}\put(30,78){\tiny $R(\overline{\gamma}_{(J,1)},0)$}  
    \end{overpic}
    \begin{overpic}[ width=.23\linewidth,height=0.12\textheight ]{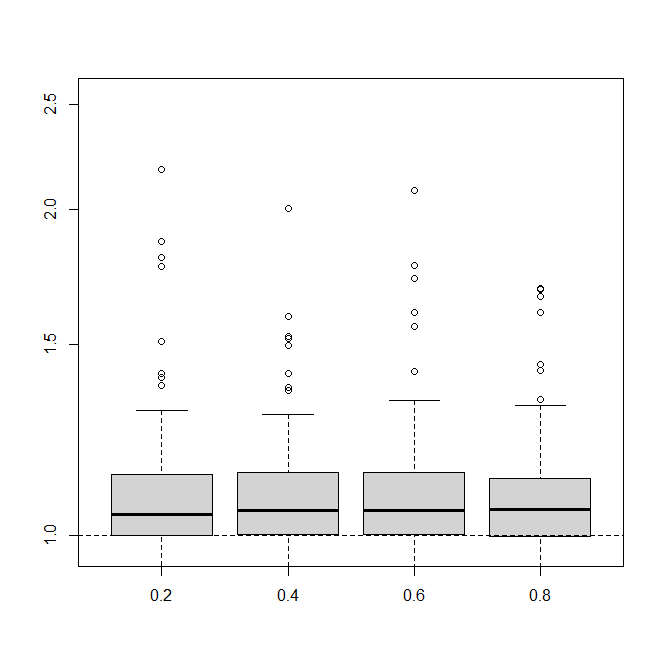}\put(30,78){\tiny $R(\widehat{\gamma}_{(J,1)},1)$}  
    \end{overpic}
    \begin{overpic}[ width=.23\linewidth,height=0.12\textheight ]{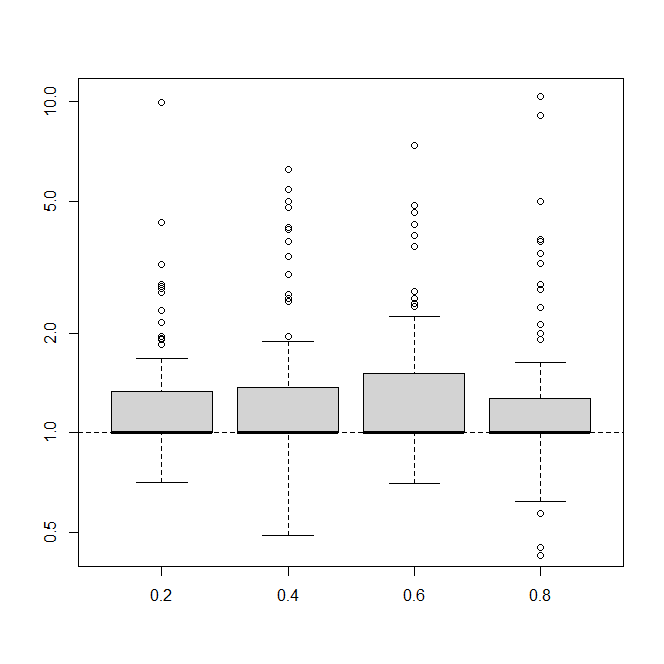}\put(30,78){\tiny $R(\overline{\gamma}_{(J,1)},1)$}  
    \end{overpic}
    \vspace{0.4cm}
    \begin{overpic}[ width=.23\linewidth,height=0.12\textheight ]{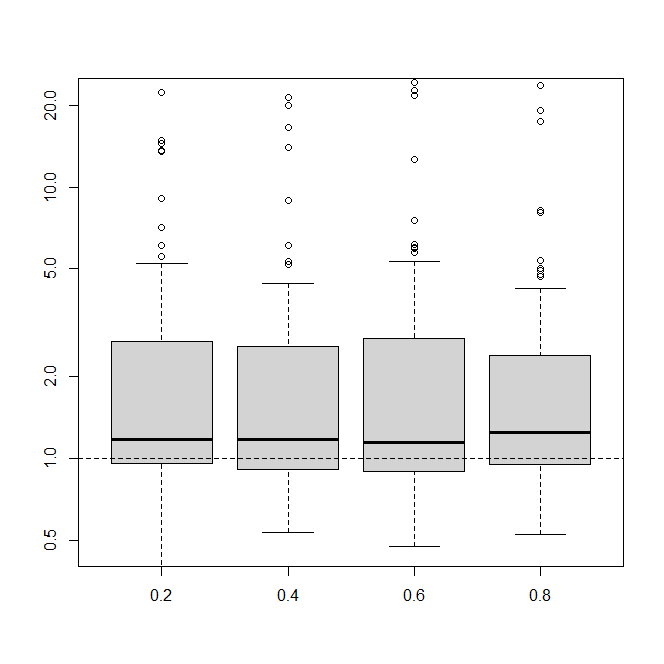}\put(30,78){\tiny $R(\widehat{\eta}_{(J,1)},0)$}  
    \end{overpic}
    \begin{overpic}[ width=.23\linewidth,height=0.12\textheight ]{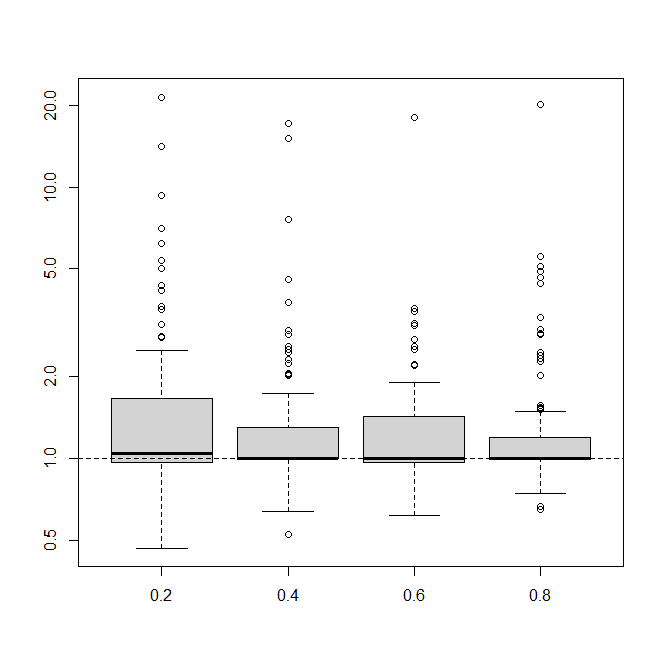}\put(30,78){\tiny $R(\overline{\eta}_{(J,1)},0)$}  
    \end{overpic}
    \begin{overpic}[ width=.23\linewidth,height=0.12\textheight ]{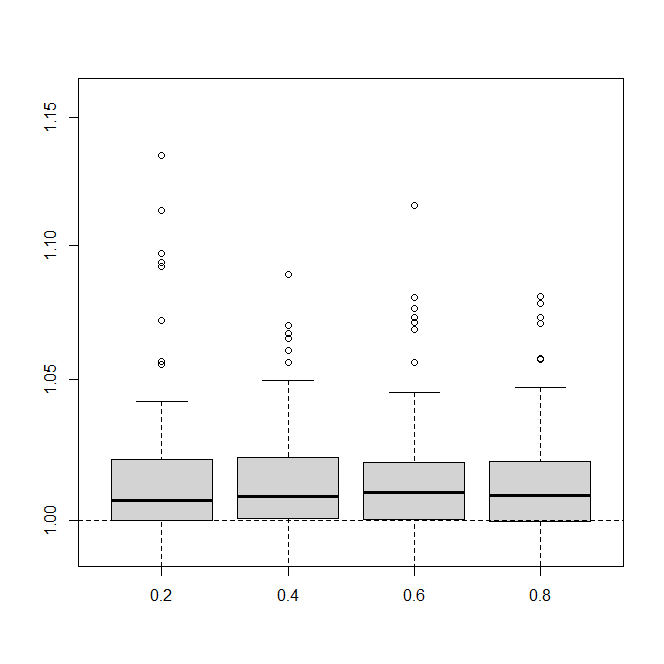}\put(30,78){\tiny $R(\widehat{\eta}_{(J,1)},1)$}  
    \end{overpic}
    \begin{overpic}[ width=.23\linewidth,height=0.12\textheight ]{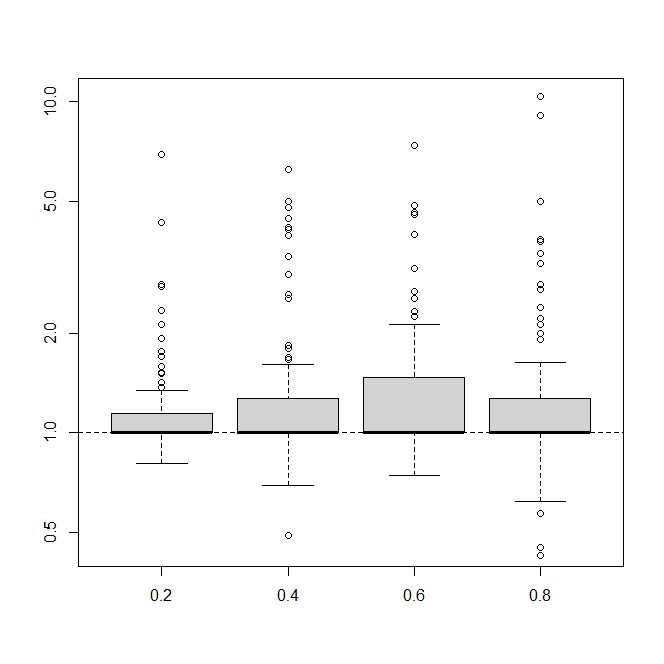}\put(30,78){\tiny $R(\overline{\eta}_{(J,1)},1)$}  
    \end{overpic}
    \caption{Columns left to right: estimates of $R(\widehat{s}_{(J,1)},0)$, $R(\bar{s}_{(J,1)},0)$, $R(\widehat{s}_{(J,1)},1)$ and $R(\bar{s}_{(J,1)},1)$ across the split rate $k$, with $J=(1-k)n_1$ calibration points and $kn_1$ training, for misspecification level $\lambda^*=0.9$ based on 100 replicates as described in Section \ref{sec:normal3}. Rows top to bottom: $R$ ratios for $b$, $\gamma$ and $\eta$ product loss posterior predictives. First two columns compare risk against Cut and last two compare Bayes. Horizontal dotted lines indicate $R=1$.}
\label{fig:normal_b_gamma_eta_prod_pred}
\end{figure}

\subsection{State space example}\label{Appendix:state-space}

Here we report three further experiments with the $\eta$ and $(\eta,\beta)$-SMI belief updates applied to the State Space model example of Section~\ref{sec:ssm}. 
In Appendix~\ref{sec: asym ssm} we investigate convergence of the pooled and product-loss posteriors for $\eta$ in $\eta$-SMI. This repeats for the State Space model the experiment in Section~\ref{sec:normal2} and Figure~\ref{fig:normal_loss_b_gam_eta} which considered the normal mixture example.
In Appendix~\ref{Appendix:risk-ratios-state-space} we estimate, for different pairs of belief updates, the expected risk ratios defined in \eqref{eq:s_relative}. This is done using test data. The experiment setup matches the setup for expected risk-ratio estimation in Figure~\ref{fig:ssm_pred} in Section~\ref{sec: post pred ssm}. In Appendix~\ref{Appendix:zero_constraint}, we demonstrate the asymptotics and estimate the risk in the case where the space of $\eta$ is $\Omega_\eta=[0,\infty)$.

\subsubsection{Asymptotic property validation for $\eta$-SMI}\label{sec: asym ssm}

Example behaviour of a posterior density for $\eta$ with increasing $J$ is shown in Figure~\ref{fig:post pred n,m new}, under the three levels of misspecification $\varphi_M^*=0.5,0.7,1$. Training data is fixed with $n=10$. We vary the size of the calibration data set, taking $J=k{n}$ with $k=0.5,...,100$. The optimal values are estimated using $J^{(z)}=1000\, n$ test data and, for this particular example, $\tilde{\eta}_{(1,\infty)}$ and $\tilde{\eta}_{(\infty,1)}$ are similar. 
We see that $\rho_{(1,J)}(\eta|\y,\x)$ converges in shape with increasing $J$ to $\pi_{\eta}(\bar{\phi}|\x)$ up to normalization as we expect for the pooled case by Theorem \ref{thm:loss_s_single}. In the bottom row, $\rho_{(J,1)}(\eta|\y,\x)$ is converging to a normal distribution as we expect from Theorem~\ref{thm:product} (except bottom right where the expected risk ratio is minimised on the boundary so the regularity conditions are not satisfied).

Figure \ref{fig:ssm_pred} compares $\eta$ estimates with high-precision estimates $\widetilde{\eta}_{\,\x}$ of $\widetilde{\eta}_{(1,\infty)}$ or $\widetilde{\eta}_{(\infty,1)}$ depending on the chosen loss. If a similarly high-precision estimate of the optimal value conditioned on $\bar{\x}=(\x,\y)$,  $\widetilde{\eta}_{\,\bar\x}=\argmin_\eta l(\eta;\z,\bar{\x})$, is similar to $\widetilde{\eta}_{\,\x}$, the expected risk of the SMI posterior using the $\eta$-estimate will likely perform well. Figure \ref{fig:diff_eta} shows that the difference of the $\eta$ estimate from the optimal value conditioned on $(\x,\y)$ is indeed small. Under misspecification, most $\eta$ estimates, except $\widehat{\eta}_{hm}$, are not significantly higher or lower than the optimal value conditioned on $(\x,\y)$. Although $\widehat{\eta}_{hm}$ tends to be smaller, the penalty in prediction is rather small and does not result poor predictive performance. 

\begin{figure}
    \centering
   \begin{overpic}[width=.32\linewidth,height=0.2\textheight]{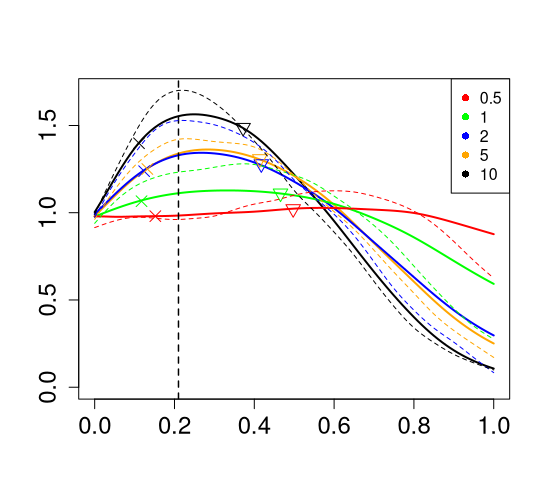}
    \put(-10,25){\rotatebox{90}{\tiny $\rho_{(1,J)}(\eta|\y,\x)$}}
    \put(35,90){\tiny $\varphi^*_M=0.5$}
    \end{overpic}
    \begin{overpic}[width=.32\linewidth,height=0.2\textheight]{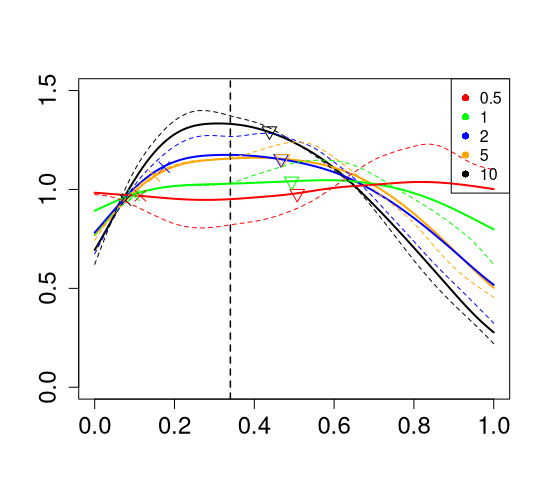}
    \put(-10,25){\rotatebox{90}{\tiny $\rho_{(1,J)}(\eta|\y,\x)$}}
    \put(35,90){\tiny $\varphi^*_M=0.7$}
    \end{overpic}
\begin{overpic}[width=.32\linewidth,height=0.2\textheight]{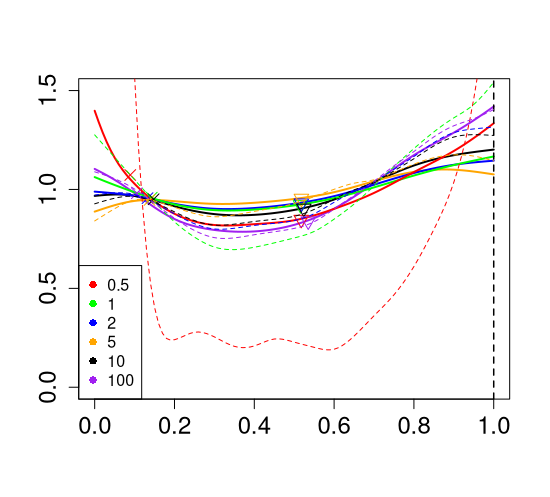}
    \put(35,90){\tiny $\varphi^*_M=1$}
    \end{overpic}
    
    \begin{overpic}[width=.32\linewidth,height=0.2\textheight]{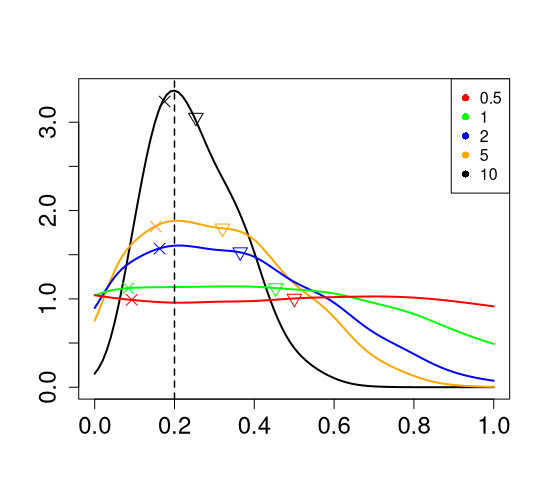}
    \put(-10,25){\rotatebox{90}{\tiny $\rho_{(J,1)}(\eta|\y,\x)$}}
    \end{overpic}
    \begin{overpic}[width=.32\linewidth,height=0.2\textheight]{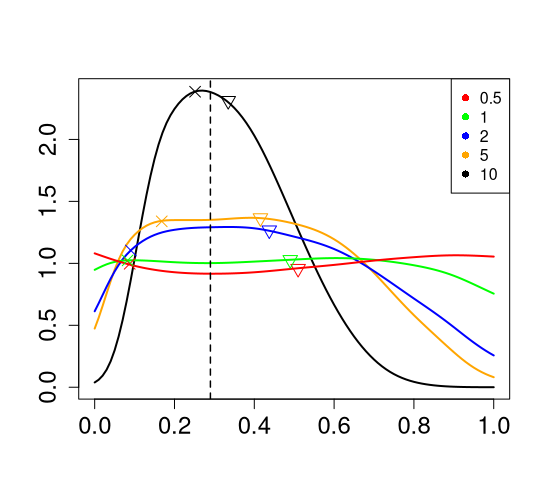}
    \put(-10,25){\rotatebox{90}{\tiny $\rho_{(J,1)}(\eta|\y,\x)$}}
    \end{overpic}
    \begin{overpic}[width=.32\linewidth,height=0.2\textheight]{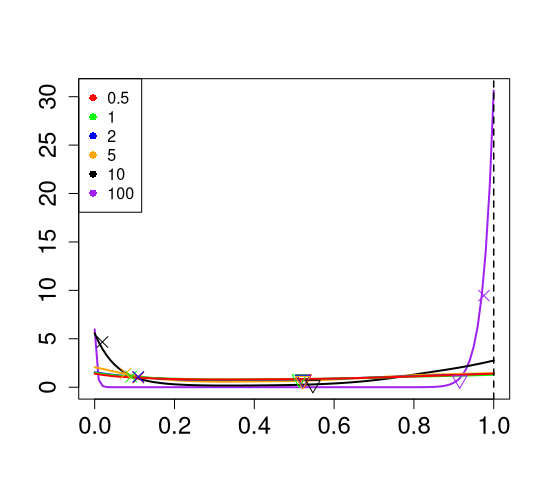}
    \end{overpic}
    \caption{Pooled (top row) and product (bottom row) posterior densities for $\eta$ for various sizes of calibration data $J=k{n}$ with $k=0.5,...,100$ (training data size is fixed at $n=10$, there is no data-splitting here); line colour indicates $J$-value. The vertical dotted lines, triangle marks and cross marks indicate $\widetilde{\eta}$ (true), $\widehat{\eta}$ (mean) and $\widehat{\eta}_{hm}$ (harmonic mean) respectively. The dotted curves in the top row are $\pi_\eta(\bar{\phi}|\x)/
    \int\pi_\eta(\bar{\phi}|\x)d\eta$ from (\ref{eq:loss_laplace}). These should coincide with the solid line of the same colour at large $k$.}
    \label{fig:post pred n,m new}
\end{figure}

\begin{figure}
    \centering
    \begin{overpic}
[width=.4\linewidth,height=0.25\textheight]{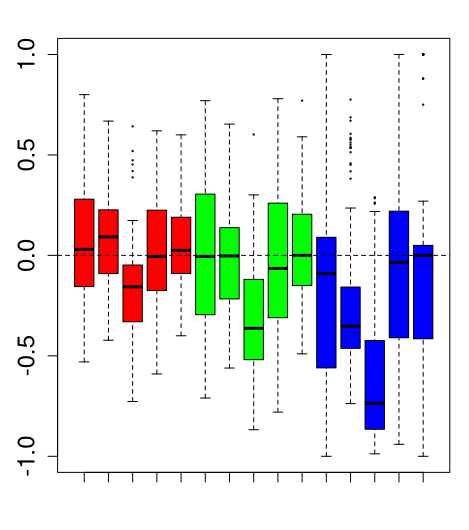}
\put(16,5){\footnotesize a}
\put(21,5){\footnotesize b}
\put(27,5){\footnotesize c}
\put(31,5){\footnotesize d}
\put(36,5){\footnotesize e}
    \end{overpic}  
    \begin{overpic}
[width=.4\linewidth,height=0.25\textheight]{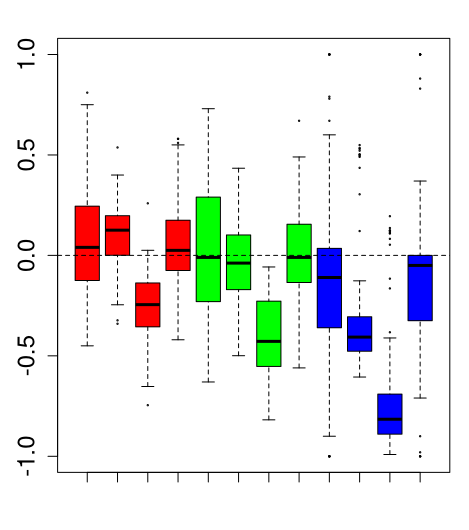}
\put(16,5){\footnotesize a}
\put(23,5){\footnotesize b}
\put(31,5){\footnotesize c}
\put(37,5){\footnotesize e}
    \end{overpic}
    \caption{Deviation of $\eta$ estimates in Figure \ref{fig:ssm_pred} from $\widetilde{\eta}_{\,\bar\x}$ for the three misspecification levels using product (left) and pooled (right) losses; $\varphi^*_M=0.5$ (red, high), $\varphi^*_M=0.7$ (green, medium), and $\varphi^*_M=1$ (blue, no misspecification). $\widetilde{\eta}_{\,\bar\x}$ is the optimal value estimated using $J^{(z)}=10^3\,n$ observation vectors $z_i=(z_{i,1},\dots,z_{i,d_x}),\ i=1\,\dots,J^{(z)}$ conditioned on $\bar\x=(\x,\y)$. For each misspecification level, four $\eta$-differences are shown: (a) $\overline{\eta}-\widetilde{\eta}_{\,\bar\x}$ (posterior mode), (b) $\widehat{\eta}-\widetilde{\eta}_{\,\bar\x}$ (mean), (c) $\widehat{\eta}_{hm}-\widetilde{\eta}_{\,\bar\x}$  (harmonic mean), (d) $\eta_{\rm \scriptscriptstyle WAIC}-\widetilde{\eta}_{\,\bar\x}$ and (e) $\widetilde{\eta}_{\,\x}-\widetilde{\eta}_{\,\bar\x}$. }   
    \label{fig:diff_eta}
\end{figure}

\subsubsection{Expected risk ratio comparisons for the State Space model}\label{Appendix:risk-ratios-state-space}

In Figure~\ref{fig:eta b ratio} we compare the posterior predictive performance of $(\eta,\beta)$-SMI to $\eta$-SMI, $\beta$-SMI, Bayes and Cut for estimation using the pooled loss. This experiment took $\varphi^*=0.7$ and $n=J=10$ blocks of calibration and training data. Each bar in Figure~\ref{fig:eta b ratio} summarises the distribution of $R$ in 100 simulated data sets. Although it considers fewer estimators and misspecification levels it is nevertheless typical of what we have seen. Prediction with $(\eta,\beta)$-SMI and any estimator is clearly preferred to Bayes (top left) and Cut (top right) and is slightly better than $\beta$-SMI (bottom right, this is $(\eta,\beta)$-SMI with $\eta=1$). It is almost identical to $\eta$-SMI (bottom left, at least for the Bayes estimator $R((\bar{\eta},\bar{b})_{(1,J)},\bar{\eta}_{(1,J)})$) and slightly better than $\beta$-SMI (bottom right), at least for the posterior mean $(\eta,b)$-estimator.

We can understand this using the single example training data set in Figure~\ref{fig:eta b joint} which has the same level of misspecification. To match Figure~\ref{fig:eta b ratio} we look at the pooled case with $J=10$ at top left in Figure~\ref{fig:eta b joint}. Although the distribution is diffuse, we see an $\eta$-value equal one is ruled out, so $\eta$ is needed for good prediction. This explains why $\beta$-SMI does less well in Figure~\ref{fig:eta b ratio}. On the other hand when we take $b\to 1$ in $(\eta,\beta)$-SMI we get $\eta$-SMI; a value of $b=1$ has high probability in the top row of Figure~\ref{fig:eta b joint}, so $b$ is not needed. This explains why $\eta$ and $(\eta,\beta)$-SMI have similar performance in Figure~\ref{fig:eta b ratio} lower left but $\beta$-SMI does slightly worse than $(\eta,\beta)$-SMI in Figure~\ref{fig:eta b ratio} lower right.

\begin{figure}
    \centering
\begin{overpic}
        [ width=.3\linewidth,height=0.2\textheight ]{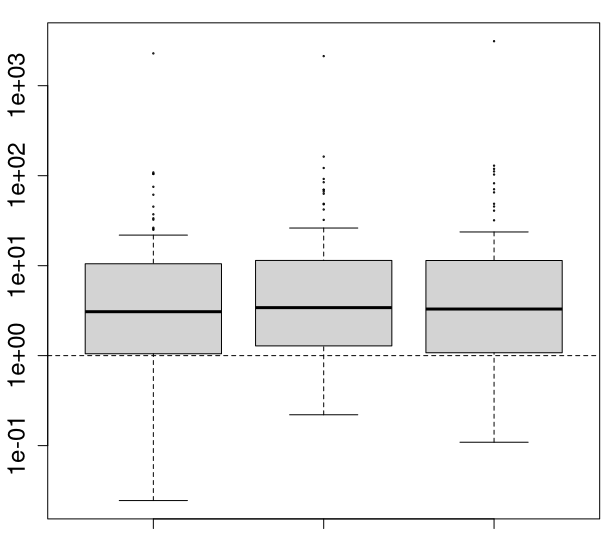}  
        \put(20,-3){\footnotesize $\widehat{s}$}
        \put(50,-3){\footnotesize $\overline{s}$}    
        \put(77,-3){\footnotesize $\widehat{s}_{hm}$}
        \put(25,96){\tiny $R((\eta,b)_{(1,J)},1)$}
\end{overpic}
\begin{overpic}
        [ width=.3\linewidth,height=0.2\textheight ]{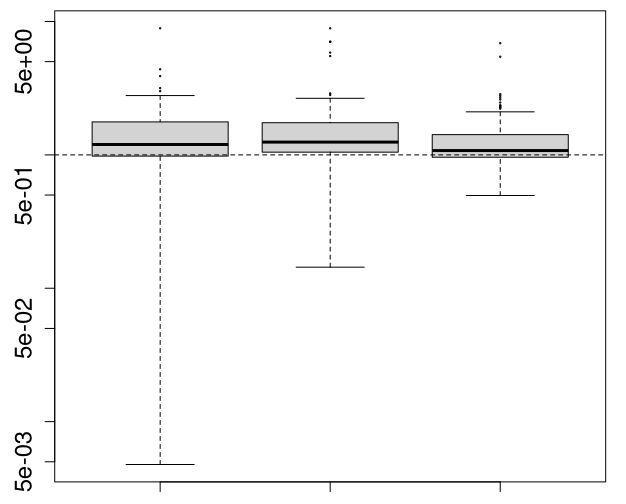}  
        \put(20,-3){\footnotesize $\widehat{s}$}
        \put(50,-3){\footnotesize $\overline{s}$}    
        \put(77,-3){\footnotesize $\widehat{s}_{hm}$}
        \put(25,98){\tiny $R((\eta,b)_{(1,J)},0)$}
\end{overpic}

\vspace{0.7cm}
\begin{overpic}
        [ width=.3\linewidth,height=0.2\textheight ]{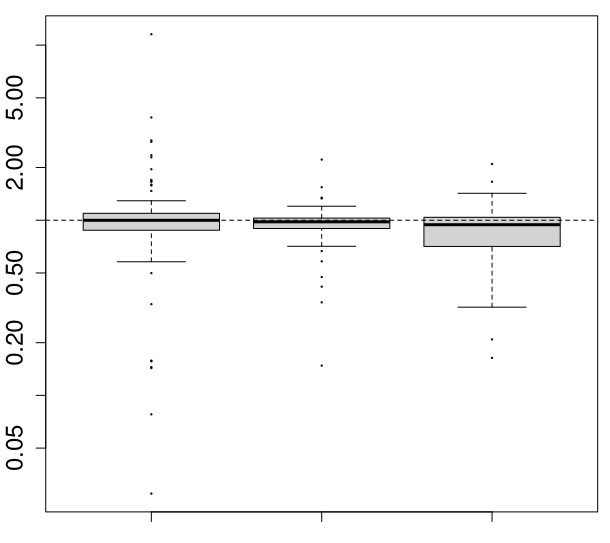} 
        \put(20,0){\footnotesize $\widehat{s}$}
        \put(50,0){\footnotesize $\overline{s}$}    
        \put(77,0){\footnotesize $\widehat{s}_{hm}$}
        \put(20,98){\tiny $R((\eta,b)_{(1,J)},\eta_{(1,J)})$}
\end{overpic}
\begin{overpic}
        [ width=.3\linewidth,height=0.2\textheight ]{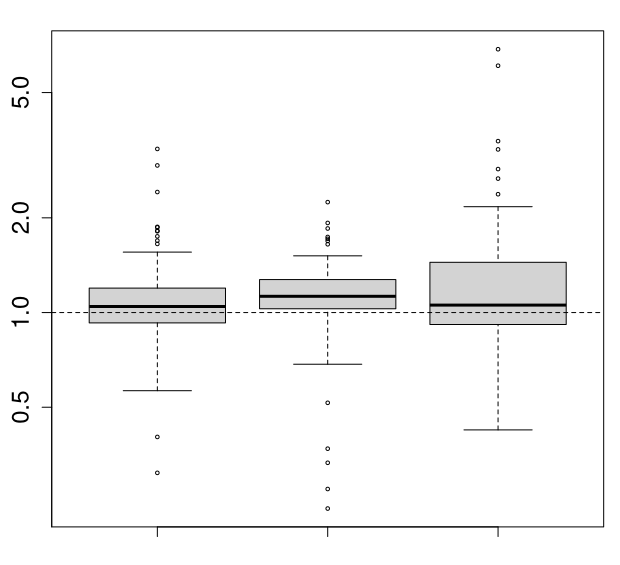}  
        \put(20,0){\footnotesize $\widehat{s}$}
        \put(50,0){\footnotesize $\overline{s}$}    
        \put(77,0){\footnotesize $\widehat{s}_{hm}$}
        \put(20,95){\tiny $R((\eta,b)_{(1,J)},b_{(1,J)})$}
\end{overpic}
        \caption{Risk for joint posterior of $(\eta,b)$ in $(\eta,\beta)$-SMI compared against other estimators (indicated on $x$-axis). In each case the same estimator-type is used for both arguments of $R$. For example, the left bar in the lower left graph is the distribution of $R((\widehat{\eta},\widehat{b})_{(1,J)},\widehat{\eta}_{(1,J)})$. The training and calibration data have $n=10$ and $J=10$ and the misspecification parameter is $\varphi^*_M=0.7$ (mild).  
        The first row shows the expected risk-ratio distributions for different $(\eta,b)$-estimators against Bayes (left) and Cut (right) posterior prediction. The second row shows comparisons with $\eta$-SMI (left) and $\beta$-SMI (right). Hyperparameter estimation is based on the pooled-data loss. Horizontal dotted lines indicate $R=1$.
         }
    \label{fig:eta b ratio}
\end{figure}

\subsubsection{Hyperparameter posteriors on $[0,\infty)$}\label{Appendix:zero_constraint}

In this section we examine the state-space example for $\eta\in [0,\infty)$ in $\eta$-SMI. We carry out inference for $\eta$ without any domain constraint (other than positivity). In Appendix~\ref{Appendix:ssm-finite-eta-star} below we prove these unconstrained estimates are always finite.

Figure \ref{fig:post_eta_infUnif} repeats the large-$J$ behaviour of the posterior as in Figure \ref{fig:post pred n,m new} for an improper prior on $[0,\infty)$. As $J$ increases, the posterior exhibits strong concentration around the optimal values; the posterior modes are preserved, while tails reduce. For the well-specified model, the optimal values are slightly greater than 1. 

The behavior in Figure \ref{fig:post_eta_infUnif} is typical, as we now illustrate using 100 replicates of $\x$ and $\y$. In Figure \ref{fig:post_eta_100rep} the posterior support typically moves towards the lower end of $0<\eta<1$ as the level of misspecification increases (left two graphs). Without any misspecification (at right), the posteriors tend to cluster around 1 though (as the training data $\x$ has only ten blocks of six observations - see Figure~\ref{fig:ssm-hmm-block-structure}) very small values are also seen. Some of these posteriors are U-shaped, like the example bottom right in Figure \ref{fig:post_eta_infUnif}, black curve. This disappears and concentrates in the usual way at large $J$ (see Figure \ref{fig:post_eta_infUnif} bottom right, purple curve).

Figures \ref{fig:eta_ratio_exp3} and \ref{fig:eta diff positive 50calib} show the expected risk ratio using the prior $\eta\sim \mbox{Exp}(1/3)$ (mean 3) and $\eta$-estimate error $\eta-\widetilde{\eta}_{\,\x}$ from the optimal value. Why Exp$(1/3)$? 
It is more important to distinguish $\eta=0.5$ from $\eta=1$ and $\eta=1$ from $\eta=3$ than to resolve larger values like $\eta=9$, as $\pi_\eta(\phi|\x)$ is only sensitive to $\eta$ at small $\eta$. The prior should be uninformative over the sensitive range of $\eta$ and penalise large values on subjective grounds - we don't see them often in the real analyses we have conducted.  Results for prediction performance ($ie$, $R$ in Figure \ref{fig:eta_ratio_exp3})  are very similar to the corresponding results computed in $[0,1]$ in Figure~\ref{fig:ssm_pred}.

\begin{figure}
    \centering
   \begin{overpic}[width=.32\linewidth]{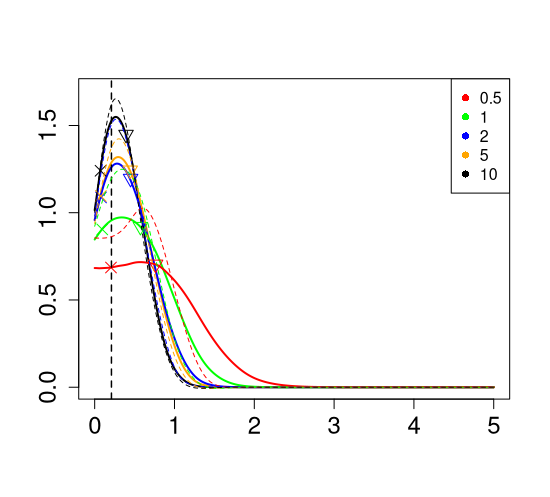}
    \put(-5,25){\rotatebox{90}{\tiny $\rho_{(1,J)}(\eta|\y,\x)$}}
    \put(35,85){\tiny $\varphi^*_M=0.5$}
    \end{overpic}
    \begin{overpic}[width=.32\linewidth]{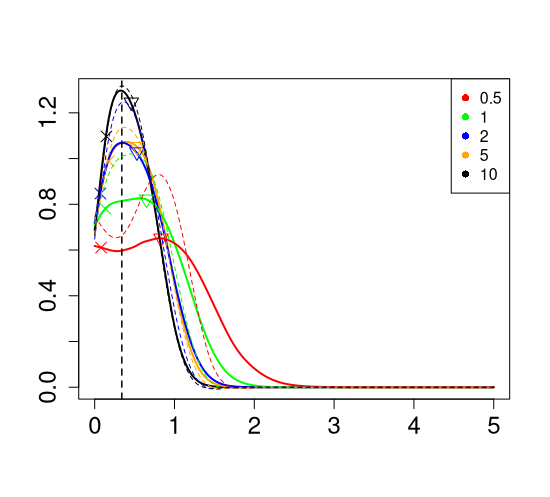}
    \put(35,85){\tiny $\varphi^*_M=0.7$}
    \end{overpic}
\begin{overpic}[width=.32\linewidth]{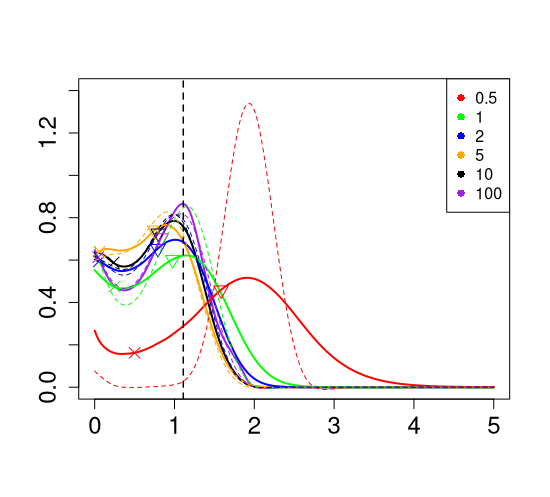}
    \put(35,85){\tiny $\varphi^*_M=1$}
    \end{overpic}
    \newline\\[-0.4in]
    
    \begin{overpic}[width=.32\linewidth]{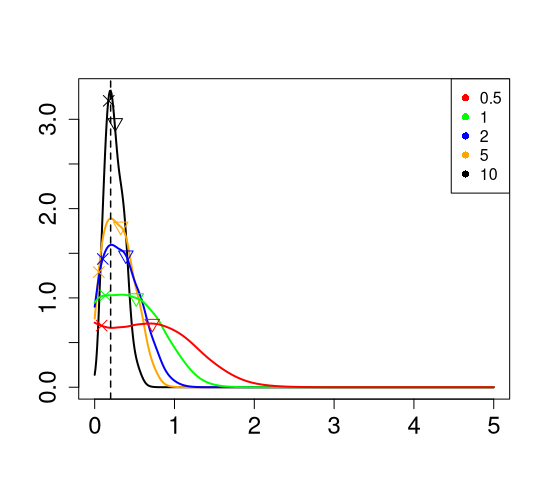}
    \put(-4,25){\rotatebox{90}{\tiny $\rho_{(J,1)}(\eta|\y,\x)$}}
    \end{overpic}
    \begin{overpic}[width=.32\linewidth]{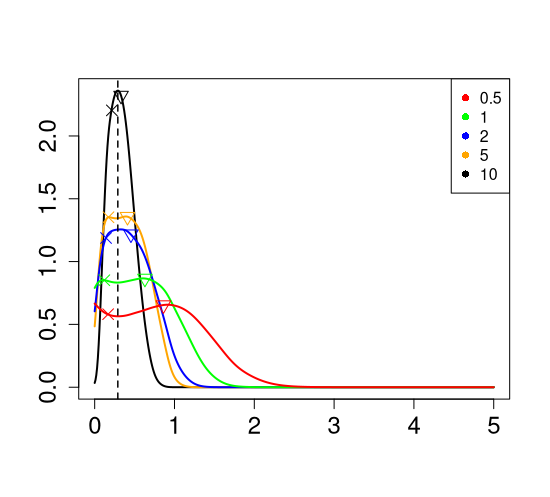}
    \end{overpic}
    \begin{overpic}[width=.32\linewidth]{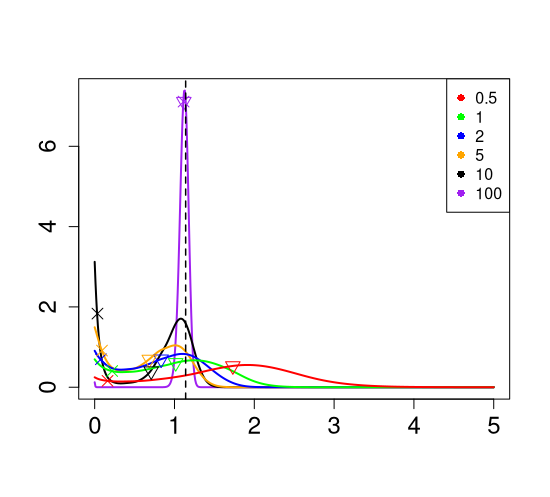}
    \end{overpic}
    \caption{Posterior densities for $\eta\ge 0$. These are the same posteriors for $\eta$ shown in Figure~\ref{fig:post pred n,m new} (the same training and calibration data $\x,\y$). The prior is the improper flat prior $\rho(\eta)\propto1$ on $[0,\infty)$.  At $\varphi_M^*=0.5$, $\widetilde{\eta}_{(1,\infty)}=0.21,\  \widetilde{\eta}_{(\infty,1)}=0.2$. At $\varphi_M^*=0.7$, $\widetilde{\eta}_{(1,\infty)}=0.34,\ \widetilde{\eta}_{(\infty,1)}=0.29$.
     At $\varphi_M^*=1$, $\widetilde{\eta}_{(1,\infty)}=1.14,\ \widetilde{\eta}_{(\infty,1)}=1.11$. These ground-truth $\widetilde{\eta}_{(\infty,1)}$-values are precise estimates based on $10,000$ blocks of calibration data.}
    \label{fig:post_eta_infUnif}
\end{figure}

\begin{figure}
    \centering
    \begin{overpic}
[width=.32\linewidth,height=0.2\textheight]{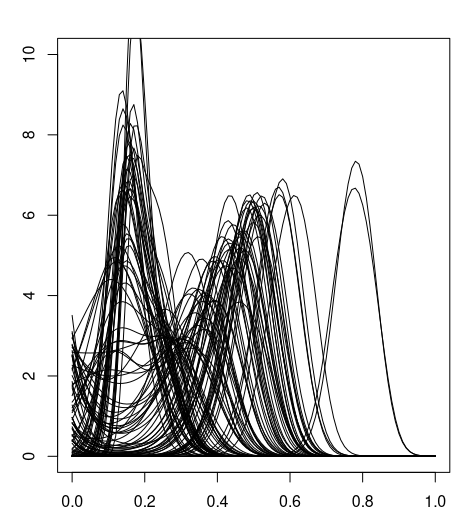}
\put(40,90){\tiny $\varphi_M^*=0.5$}
    \end{overpic}
    \begin{overpic}
[width=.32\linewidth,height=0.2\textheight]{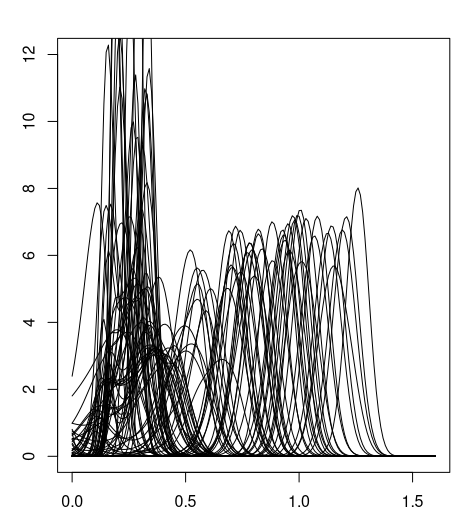}
\put(40,90){\tiny $\varphi_M^*=0.7$}
    \end{overpic}
    \begin{overpic}
[width=.32\linewidth,height=0.2\textheight]{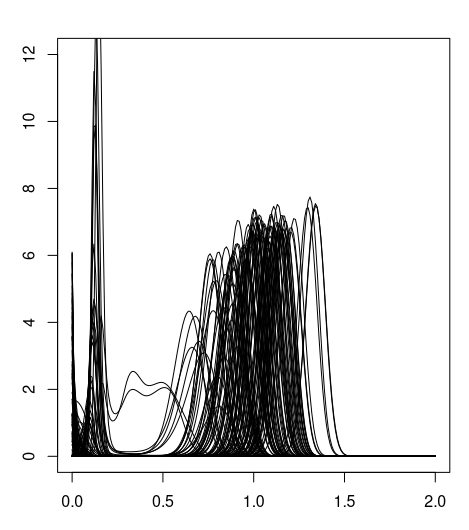}
\put(40,90){\tiny $\varphi_M^*=1$}
    \end{overpic}
    \caption{Posterior densities for $\eta$ with prior $\rho(\eta)\propto1$ on $[0,\infty)$, from 100 realizations of training blocks of $J^{(\x)}=10$ and calibration blocks of $J^{(\y)}=50$.}
    \label{fig:post_eta_100rep}
\end{figure}

\begin{figure}[t]
    \centering
\begin{overpic}[ width=.3\linewidth,height=0.22\textheight ]{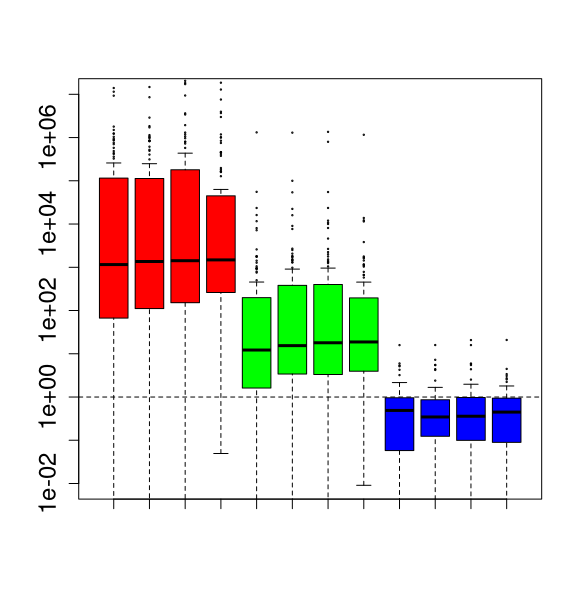} 
\put(32,92){\tiny $R(\eta_{(J,1)},1)$}
\put(16,5){\tiny a}
\put(23,5){\tiny b}
\put(28,5){\tiny c}
\put(35,5){\tiny d}
\end{overpic}
\begin{overpic}[ width=.3\linewidth,height=0.22\textheight ]{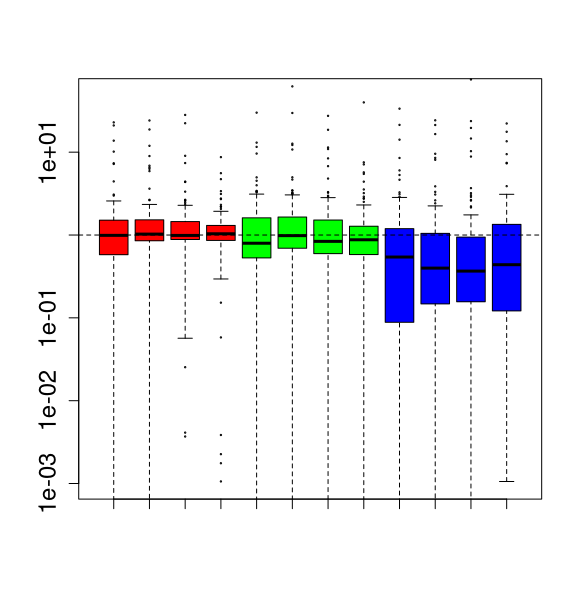} 
\put(32,92){\tiny $R(\eta_{(J,1)},0)$}

\end{overpic}
\begin{overpic}[ width=.3\linewidth,height=0.22\textheight ]{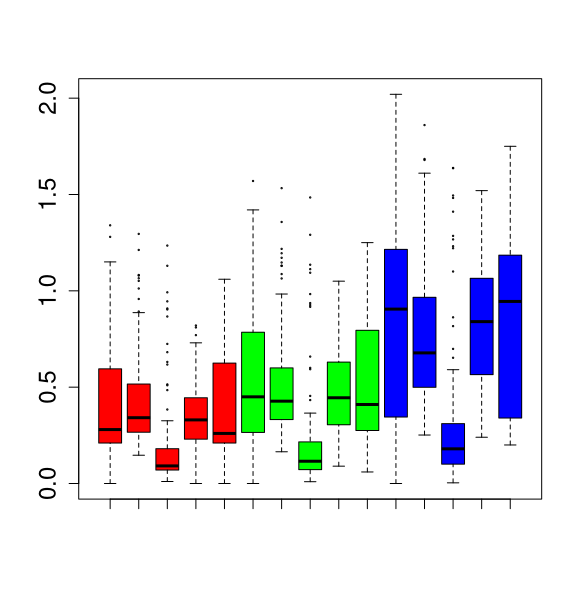} 
\put(32,92){\tiny $\eta$-estimates}
\put(16,5){\tiny a}
\put(19,5){\tiny b}
\put(25,5){\tiny c}
\put(30,5){\tiny d}
\put(35,5){\tiny e}
\end{overpic}

\vspace{0.6cm}
\begin{overpic}[ width=.3\linewidth,height=0.22\textheight ]{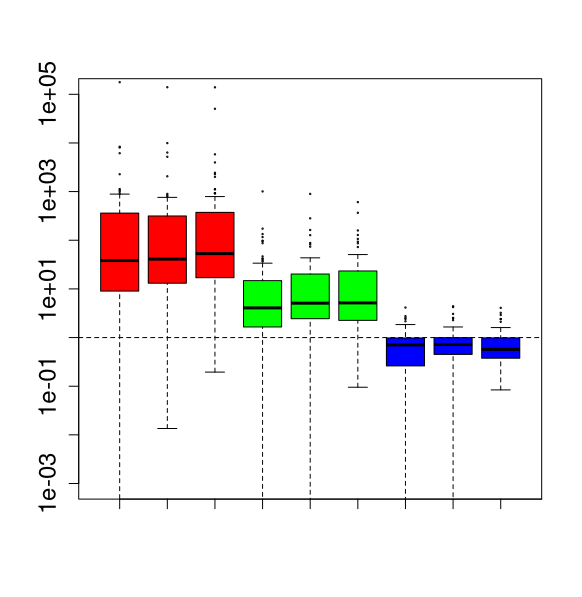} 
\put(32,92){\tiny $R(\eta_{(1,J)},1)$}
\put(18,5){\tiny a}
\put(28,5){\tiny b}
\put(35,5){\tiny c}
\end{overpic}
\begin{overpic}[ width=.3\linewidth,height=0.22\textheight ]{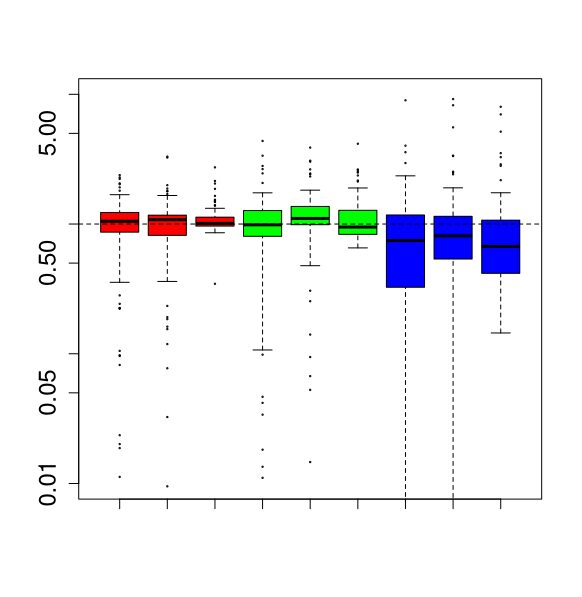} 
\put(32,92){\tiny $R(\eta_{(1,J)},0)$}
\end{overpic}
\begin{overpic}[ width=.3\linewidth,height=0.22\textheight ]{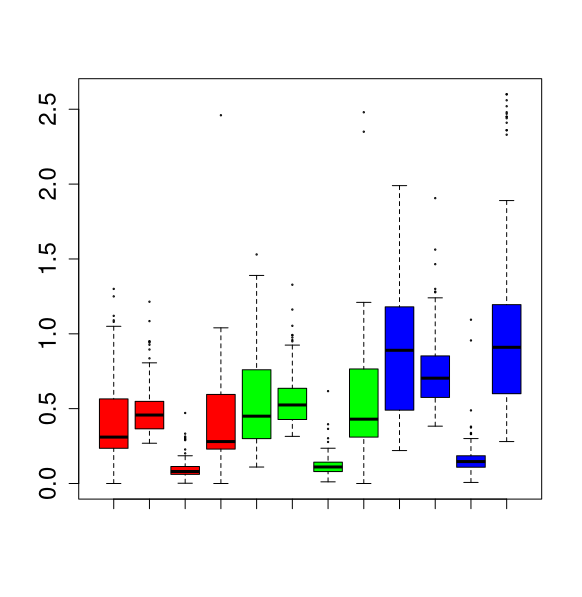}  
\put(32,92){\tiny $\eta$-estimates}
\put(17,5){\tiny a}
\put(23,5){\tiny b}
\put(28,5){\tiny c}
\put(34,5){\tiny e}
\end{overpic}
\caption{Distributions of expected risk ratio \eqref{eq:s_relative} and $\eta$ estimates using a prior $\eta\sim Exp(1/3)$ for the three misspecification levels; $\varphi^*_M=0.5$ (red, high), $\varphi^*_M=0.7$ (green, medium), and $\varphi^*_M=1$ (blue, no misspecification). Rows show comparisons against estimates of $\eta$ computed using product (top row) and pooled (bottom row) $\eta$-posteriors. For each misspecification level, four $\eta$ estimators, (a) $\overline{\eta}$ (posterior mode), (b) $\widehat{\eta}$ (mean), (c) $\widehat{\eta}_{hm}$  (harmonic mean) and (d) $\eta_{\rm \scriptscriptstyle WAIC}$ are compared against Bayes ($\eta=1$, left column) and cut ($\eta=0$, middle column) and $\widetilde{\eta}$ (fourth box (e) in each color group, estimated using $J=10^3$ test blocks, right column). Horizontal dotted lines at $R=1$.} \label{fig:eta_ratio_exp3} 
\end{figure}

\begin{figure}
    \centering
    \begin{overpic}[ width=.4\linewidth,height=0.3\textheight ]{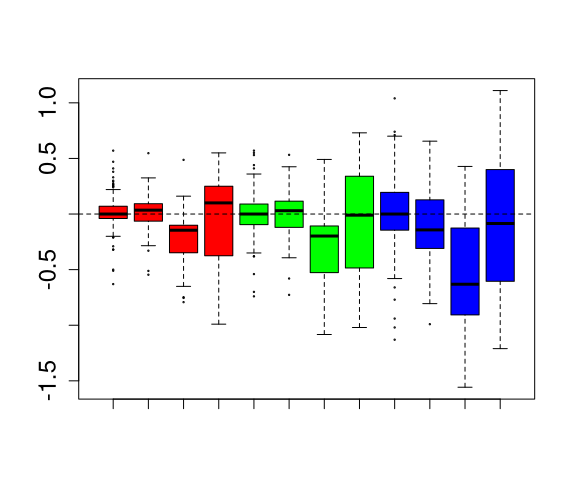} 
\put(16,5){\footnotesize a}
\put(22,5){\footnotesize b}
\put(28,5){\footnotesize c}
\put(33,5){\footnotesize d}
\end{overpic}
\begin{overpic}[ width=.4\linewidth,height=0.3\textheight ]{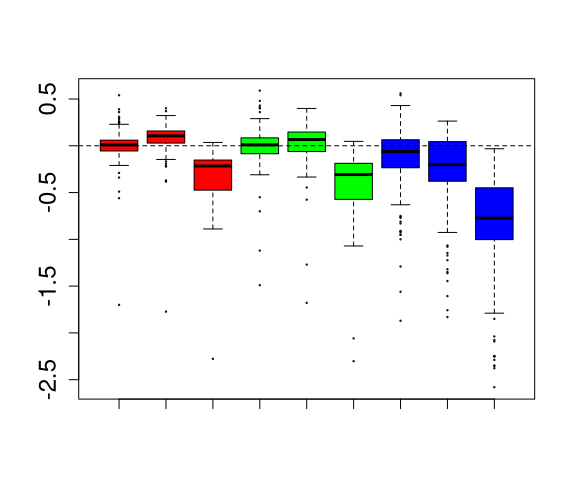} 
\put(18,5){\footnotesize a}
\put(25,5){\footnotesize b}
\put(33,5){\footnotesize c}

\end{overpic}
\caption{Deviation of $\eta$ estimates in Figure \ref{fig:eta_ratio_exp3} from $\widetilde{\eta}$ for the three misspecification levels using product (left) and pooled (right) losses; $\varphi^*_M=0.5$ (red, high), $\varphi^*_M=0.7$ (green, medium), and $\varphi^*_M=1$ (blue, no misspecification). For each misspecification level, four $\eta$ estimators, (a) $\overline{\eta}$ (posterior mode), (b) $\widehat{\eta}$ (mean) and (c) $\widehat{\eta}_{hm}$  (harmonic mean). Horizontal dotted lines at zero.}   
    \label{fig:eta diff positive 50calib}
\end{figure}

\subsubsection{Finite optimal hyperparameters}\label{Appendix:ssm-finite-eta-star}

In this subsection we show that $\widetilde{\eta}_{(1,\infty)}$ and $\widetilde{\eta}_{(\infty,1)}$ are almost surely finite in the State Space Model example of Section~\ref{sec:ssm}. This chimes with the results for $\eta$-estimation above.
Let $|A|$ and $|M|$ be the number of anchors and missing $\theta$'s. Using the reparametrization $\theta_M \mid \theta_A \sim N(B\theta_A,\Sigma_A)$ for the AR(1) prior, integrating out $\theta_M$ we obtain
\begin{align*}
\pi(\varphi^2 | \theta_A, x_A, x_M) & \propto \pi(\varphi^2)\,(\varphi^2)^{-(\eta-1)|M|/2 - |A|/2}
\left|\Sigma_A+\tfrac{\varphi^2}{\eta}I_{|M|}\right|^{-1/2} \\   
& \times \exp\!\left( -\tfrac{1}{2}(x_M - B\theta_A)^\top (\Sigma_A+\tfrac{\varphi^2}{\eta}I_{|M|})^{-1} (x_M - B\theta_A) -\tfrac{1}{2\varphi^2}\|x_A - \theta_A\|^2 \right).
\end{align*}
The prior for $\varphi^2$ is $\mbox{InvG}(a,b)$. For large $\eta$, 
$
|\Sigma_A+\tfrac{\varphi^2}{\eta}I_{|M|}|^{-1/2} 
=|\Sigma_A|^{-1/2}+O(\eta^{-1})
$
so using the Neumann expansions
\begin{align*}
(x_M - B\theta_A)^\top (\Sigma_A+\frac{\varphi^2}{\eta}I_{|M|})^{-1} (x_M - B\theta_A)  &= (x_M - B\theta_A)^\top \left(\Sigma_A^{-1}-\frac{\varphi^2}{\eta} \Sigma_A^{-2} +O\left(\eta^{-2}\right) \right)  (x_M - B\theta_A) \\ 
&=(x_M - B\theta_A)^\top \Sigma_A^{-1}(x_M - B\theta_A)+O(\eta^{-1})\,.
\end{align*}
It follows that $\pi(\varphi^2|\cdot) = InvG(a_\eta,b_A)(1+O(1/\eta))$ where $b_A=b+\frac{1}{2}\|x_A-\theta_A\|^2$ and \[a_\eta=a+\frac{|A|}{2}+\frac{(\eta-1)|M|}{2}=\frac{\eta|M|}{2}+O(1).\] 

Suppose that there are $J$ held-out anchors and $S_J=\|x'_{A}-\theta'_A\|^2$ denotes L2 norm of these held-out calibration data for the pooled loss. Integrating with $\varphi^2$, the posterior predictive for all the calibration data (pooled case) is
\begin{align}
p_\eta(x'_{A}|x_{A},\theta_A,\theta'_A,x_{M}) &= \int (2\pi\varphi^2)^{-J/2}\exp(-\frac{S_J}{2\varphi^2})\,\mbox{InvG}(\varphi^2;a_\eta,b_A) d\varphi^2 \ \times\ (1+O(\eta^{-1})) \nonumber\\
&= (2\pi)^{-J/2} \frac{b_A^{a_\eta}}{\Gamma(a_\eta)}\frac{\Gamma(a_\eta+J/2)}{(b_A+S_J/2)^{a_\eta+J/2}}\ \times\ (1+O(\eta^{-1}))\nonumber\\
&\propto \eta^{J/2} \left(1+{S_J}/{(2b_A)}\right)^{-\eta|M|/2-J/2}(1+O(\eta^{-1}))\,,
\label{eq:ssm_pooled_post_pred_big_eta}
\end{align}
where the last step uses $\frac{\Gamma(a_\eta+J/2)}{\Gamma(a_\eta)}=a_\eta^{J/2}(1+O(a_\eta^{-1}))$ and $a_\eta = \eta|M|/2+O(1)$.  
The product form of posterior predictive is 
\begin{align} \prod_{j=1}^J p_\eta(x'_{A,j}|x_{A},\theta_A,\theta'_{A,j},x_{M}) \propto \eta^{J/2} \prod^J_{j=1} \left(1+\frac{\|x'_{A,j}-\theta'_{A,j}\|^2}{2b_A} \right)^{-\eta|M|/2-1/2}(1+O(\eta^{-1})) \,. \label{eq:ssm_prod_post_pred_big_eta}\end{align}

Unless $x'_{A,j}=\theta'_{A,j}$ for every $j=1,\dots,J$, both the pooled and product posterior predictive densities go to zero in \eqref{eq:ssm_pooled_post_pred_big_eta} and \eqref{eq:ssm_prod_post_pred_big_eta} as $\eta\to\infty$. Since these predictives are positive at any finite $\eta_0$, there exists $\eta_1<\infty$ such that predictives for all $\eta>\eta_1$ is smaller than at $\eta_0$. Hence the global maximum must lie in $[0,\eta_1]$. Continuity in $\eta$ then guarantees that this maximum is in this compact interval, so the optimal $\eta$ is finite almost surely in $\x$. 
 

\subsection{EDiSC experiments}\label{Appendix:EDiSC}

\subsubsection{EDiSC model details}\label{Appendix:EDiSC-model}

The EDiSC model infers evolving word meanings over time in an unsupervised, bag-of-words framework, using context words to determine the sense of a target word. For a detailed introduction of the model and notation, we refer the reader to \citet[Section 2 and 4]{zafar2024embedded}
 and \citet[Section 4]{zafar2024exploring}. Here, we summarize the key components for the current paper. The training data $\x=(x_1,\dots,x_n)$ for a target word consists of $n$ snippets, each a set of about $2L$ words $x_i,\ i=1,\dots,n$ from a vocabulary $\V=\cup_i x_i$ from just before and after the target word, with $V=|\V|$ co-occurring words in all (the number of words $n_i\le 2L$ in a snippet $x_i$ is usually less than $2L$ as stopwords and very low frequency words are filtered out) and spans $T$ discrete contiguous time
periods and $G$ text genres (such as news, magazine, fiction, non-fiction). Each snippet $x_i$, $i = 1, \ldots, n$, belongs to time period $\tau_i\in \{1,\dots,T\}$ and genre $\gamma_i\in \{1,\dots,G\}$ and the full data are $\x=(x_1,\dots,x_n)$. In the generative model, for each $x_i$, we first sample its sense $\c_i$, and then sample context words given the sense. Sense $\c_i$ is sampled from a multinomial sense-prevalence distribution $\widetilde{\phi}^{\gamma_i, \tau_i}$ over $K$ senses (indexed by genre and time) so that
\[
p(\c_i = k|\widetilde{\phi}^{\gamma_i, \tau_i}) = \widetilde{\phi}_k^{\gamma_i, \tau_i}
\enspace \text{for all senses } k \in \{1, \ldots, K\}.
\]
Given the sense $\c_i$, at each position/word $j=1,\dots,n_i$ in $x_i$, context words $x_{i,j}$ are sampled independently from a multinomial distribution $\widetilde{\psi}^{\c_i, \tau_i}$ over the $V$-sized lemmatized vocabulary (indexed by sense and time) so that
\[
p(x_{i,j} = v | \c_i, \widetilde{\psi}^{\c_i, \tau_i}) = \widetilde{\psi}^{\c_i, \tau_i}_{v}
\enspace \text{for all words } v \in \{1, \ldots, V\}.
\]
The inferential task is to learn $\widetilde{\phi}$ and $\widetilde{\psi}$ given $\x$ and correctly group the snippets into clusters by meaning.

We use the notation $(\c_i, \gamma_i, \tau_i)$ for snippet-specific sense-genre-time triples, and $(k, g, t)$ for their generic equivalents.
Let $\widetilde{\phi}^{g,t} = \mathrm{softmax}(\phi^{g,t})$ and $\widetilde{\psi}^{k,t} = \mathrm{softmax}(\psi^{k,t})$.  The real-valued arrays $\phi$ and $\psi$ evolve over time $t=1,\dots,T$ so their priors (given in \citet[Section 4]{zafar2024embedded}) model that stochastic process. The likelihood
\begin{align}\label{eq:EDiSC-gb-posterior}
p(\x | \phi, \psi) &= \prod_{i=1}^{n} \sum_{k=1}^{K} 
p(\c_i = k| \phi)\, p(x_i |\c_i = k, \psi)\\ 
&= \prod_{i=1}^{n} \sum_{k=1}^{K} 
\widetilde{\phi}^{\gamma_i, \tau_i}_{k}
\prod_{w \in x_i} \widetilde{\psi}^{k, \tau_i}_{w}\nonumber
\end{align}
can be exactly marginalized over the unknown latent sense-label parameters $\c$ as the number $K$ of possible values is typically small.
Snippets $\x$ are observed data, whereas sense assignments $\c = (\c_1, \ldots, \c_n)$ are missing data (in fact these are known for our data, as the data have been hand-labeled, but the point of the analysis is to group the snippets without knowledge of true sense labels). The generalized Bayesian posterior is
\[\pi_\eta(\phi, \psi |\x) \propto \pi(\phi, \psi) p(\x | \phi, \psi)^\eta
\]
by raising the likelihood to a power of $\eta$.


In Table~\ref{tab:EDiSC-experiment-settings} we give details of the data sets and model settings used in the experiments reported in Section~\ref{sec:edisc}. {Inference is based on 1,400 samples obtained by taking every 10th draw from four chains of length 4,000, using STAN \citep{Carpenter2018}, after discarding the first 500 iterations of each chain as burn-in. We checked ESS values were reliably large.}

\begin{table}[t]
\caption{Data and Train/Calibrate splits for experiments reported in Section~\ref{sec:edisc}. This table is reproduced with minor changes from Table~2 in \cite{zafar2024exploring}. In column 2, $N=n+J+J_z$ is the total data set size before splitting. In column 8, $n=|\x|$, $J=|\y|$ and $J_z=|\z|$ give the numbers of snippets in the training, calibration and test data respectively.\\ 
} \label{tab:EDiSC-experiment-settings}
\centering
\begin{adjustbox}{max width=\textwidth}
\begin{tabular}{lrrcccclrl}
\toprule
        &Snippets & Vocab & Length & True senses & Model senses & Genres & {Train, Cal, Test} & \multicolumn{2}{l}{Time periods} \\
Target word & \multicolumn{1}{c}{($N$)} & \multicolumn{1}{c}{($V$)} & ($L$) & ($K^*$) & ($K$) & ($G$) & $(n, J, J_z)$ & ($T$) & detail \\ \\[-1.25em]
bank split 1 &   704\phantom{N} &   736 & 14 & 2 & 2 & 1 & 470, 117, 117 & 10 & 1810-2010 \\
bank split 2 &   708\phantom{N} &   717 & 14 & 2 & 2 & 1 & 472, 118, 118 & 10 & 1810-2010 \\
bank split 3 &   703\phantom{N} &   728 & 14 & 2 & 2 & 1 & 469, 117, 117 & 10 & 1810-2010 \\
bank split 4 &   704\phantom{N} &   742 & 14 & 2 & 2 & 1 & 470, 117, 117 & 10 & 1810-2010 \\
bank split 5 &   706\phantom{N} &   735 & 14 & 2 & 2 & 1 & 471, 118, 117 & 10 & 1810-2010 \\
chair   &   745\phantom{N} & 3,180 & 20 & 2 & 2 & 4 & 497, 124, 124 & 10 & 1820-2020 \\
apple   & 1,154\phantom{N} & 3,737 & 20 & 2 & 2 & 4 & 770, 192, 192 & 6 & 1960-2020\\
gay   &     650\phantom{N} & 3,071 & 20 & 2 & 4 & 1 & 434, 108, 108 & 5 & 1920-2020\\
mouse &     584\phantom{N} & 2,439 & 20 & 2 & 3 & 1 & 390, 97, 97 & 4 & 1940-2020\\
bug   &     522\phantom{N} & 2,475 & 20 & 4 & 4 & 1 & 348, 87, 87 & 8 & 1980-2020\\
\bottomrule
\end{tabular}
\end{adjustbox}
\end{table}

\subsubsection{EDiSC - further results}
\label{Appendix:EDiSC-pooled-loss}

\begin{figure}[ht]
    \centering
    \hbox{
    \begin{overpic}
    [width=.45\linewidth,height=0.23\textheight]{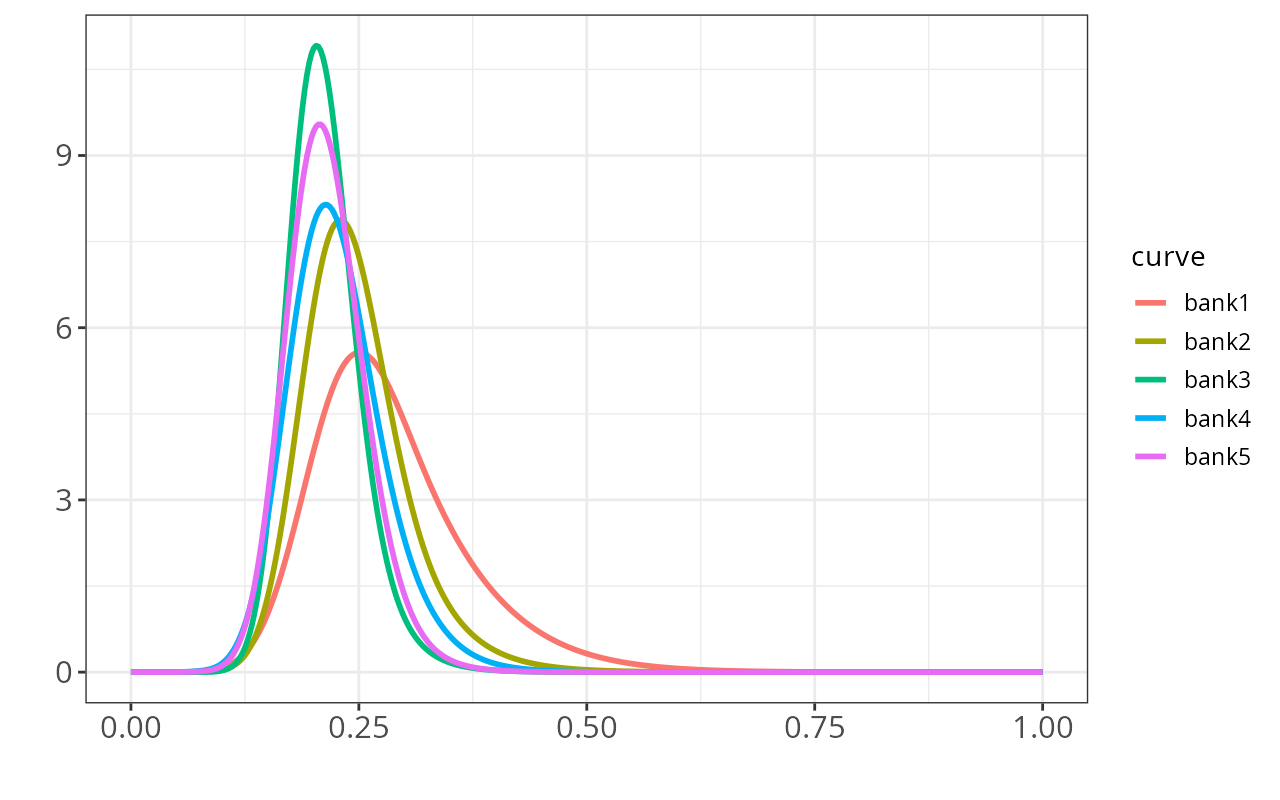}
    \put(-5,25){\rotatebox{90}{\small $\rho_{(J,1)}(\eta|\y, \x)$}}
        \put(50,2){\small $\eta$}
    \end{overpic}
    \hspace{0.2cm}
    \begin{overpic}
    [width=.45\linewidth,height=0.23\textheight]{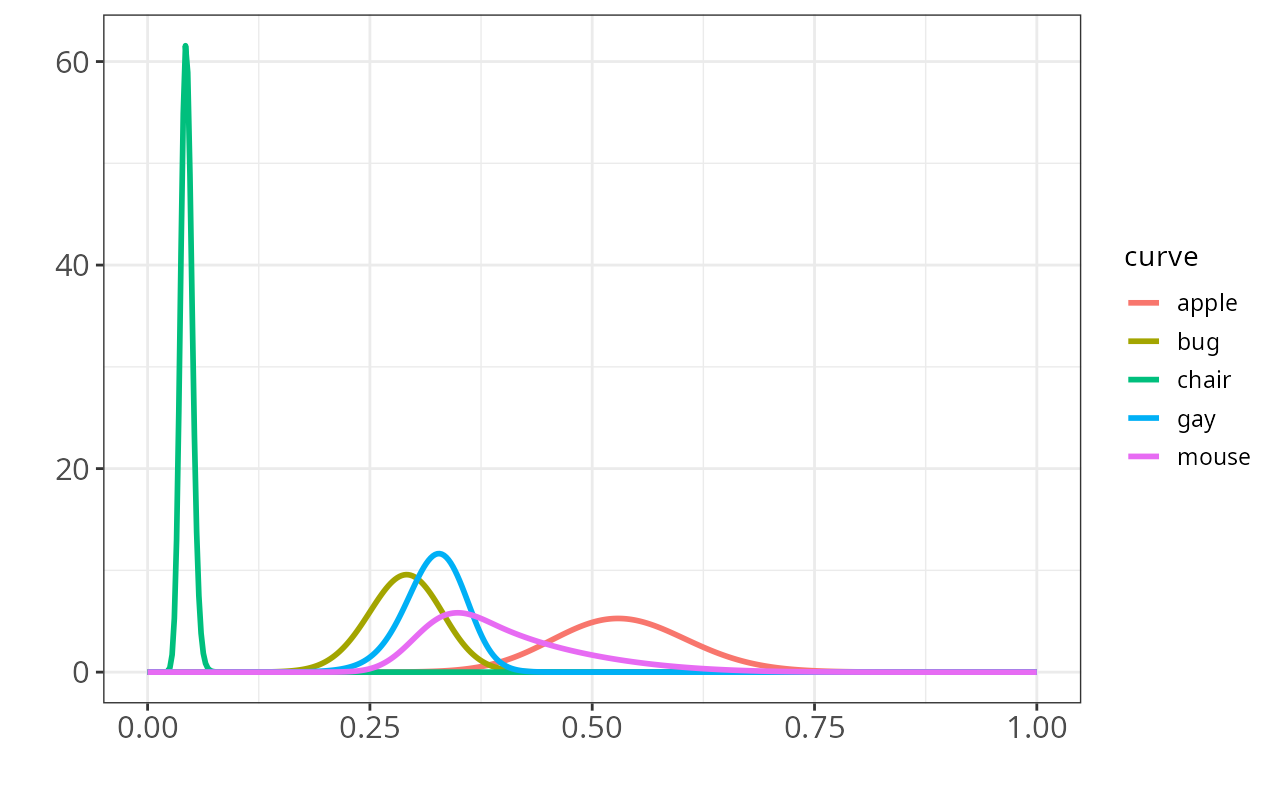}
    \put(-5,25){\rotatebox{90}{\small $\rho_{(J,1)}(\eta|\y, \x)$}}
        \put(50,2){\small $\eta$}
    \end{overpic}
    }
    \caption{Product-posterior densities for $\eta$ for 5 subsets of data all with target word ``bank'' (left) and 5 other target words (right). The posteriors in Figure~\ref{fig:bank split} are estimated using one third of the data as calibration data; the posteriors above use one sixth and are generally less normal-looking and more dispersed. All other details are the same.
    }
    \label{fig:bank split sixth}
\end{figure}

In this section, we report findings for risk ratios in the EdiSC analysis of Section~\ref{sec:edisc}. We first compare predictive performance on held-out test data, and then examine the sensitivity of inference for $\eta$ and of the downstream Brier score to the training/calibration split.

The experiment settings are given in Table~\ref{tab:EDiSC-experiment-settings}, and in particular the split from the full data $\bar \x$ into train/calibrate/test blocks $\x,\y,\z$ uses a 4:1:1 split in contrast to the 2:1 split into train/calibrate blocks $\x$ and $\y$ used to form the $\eta$-posteriors in Figure~\ref{fig:bank split}. There we didn't need test data as we compared on Brier scores. However, here we need to hold out test data in order to compute risk ratios. The $\eta$-posteriors from the 4:1:1 split are shown in Figure~\ref{fig:bank split sixth}. As expected these have somewhat greater variance and are a little skewed, as we are further from the asymptotic regime ($J$ halves in going from Figure~\ref{fig:bank split} to Figure~\ref{fig:bank split sixth}).
We evaluate $R(\eta_1,\eta_2)=\frac{p_{\eta_1}(\z|\x,\y)}{p_{\eta_2}(\z|\x,\y)}$ (Equation \ref{eq:s_relative}), taking as our baseline $\eta_1=\hat\eta$ (the posterior mean for each density in Figure~\ref{fig:bank split sixth}) and compare against Bayes (with $\eta_2=1$) and against the $\eta$-values $\eta_2=\eta^\dagger$ reported in \cite{zafar2024exploring}. The $R$-values given in Table~\ref{tab:edisc-R-bank} are logs, $\tilde{R}:=\log R$. Gains over Bayes are substantial and our estimator gives improved prediction of help out data over $\eta^\dagger$ in $8/10$ experiments.
\begin{table}[htbp]
    \centering
    \begin{tabular}{c|c|c}
    Target word     &  $ \tilde{R}_{(\hat{\eta},1)}$ & $\tilde{R}_{(\hat{\eta},\eta^\dagger)}$\\
    \hline
    bank 1 &  10.6  &  1.0  \\
    bank 2 &  18.2  & 3.9 \\
    bank 3 &  26.0  & 6.0  \\
    bank 4 &  19.7  & 2.5   \\
    bank 5 &  27.2  & 6.7   \\
    \end{tabular}\qquad \qquad
    \begin{tabular}{c|c|c}
    Target word & $\tilde{R}_{(\hat{\eta},1)}$     & $\tilde{R}_{(\hat{\eta},\eta^\dagger)}$ \\
    \hline
    chair  & 160.3    & 111.7 \\
    apple &  20.0  &  25.9 \\
    gay  & 65.0  &  19.1 \\
    mouse  & 12.1  &  0.58 \\
    bug  &  53.5  &  14.8 \\
    \end{tabular}
    \caption{%
    Log Posterior predictive ratios on test data in the 4:1:1 split. Positive values favor our method. Here $\hat{\eta}$ is the posterior mean learning rate in the product-loss posterior with a uniform prior and $\eta^\dagger$ is the value estimated in \cite{zafar2024exploring}.}
    \label{tab:edisc-R-bank}
\end{table}


We next examine sensitivity to the proportion of data allocated to training and calibration. For this experiment, we use the Bank 1 data set and vary the training proportion $n/(n+J)$, with the remaining observations used for calibration. For each split we recompute the product posterior $\rho_{(J,1)}(\eta| \y;\x)$. The left panel of Figure~\ref{fig:bank1_split_results} shows that the posterior location increases and posterior uncertainty generally decreases as more calibration data are made available. This reflects the trade-off in data splitting between learning information for $\varphi$ and $\eta$. 

To assess the effect of this sensitivity on the downstream clustering performance, let $BS_\eta$ denote the Brier score in~(\ref{eq:edisc-brier-score}), evaluated for the clustering obtained under learning rate $\eta$. Since $\eta$ is itself uncertain under the product posterior, a posterior-averaged Brier score is
\[ \E\!\left[BS_\eta \mid \y,\x\right]=\int BS_\eta \, \rho_{(J,1)}(\eta| \y;\x)\,d\eta .\]
We estimate this quantity by Monte Carlo, $\widehat{BS}$.

The right panel of Figure~\ref{fig:bank1_split_results} compares three summaries across the different training/calibration splits: the Brier score obtained by plugging in the posterior mean $\hat{\eta}$, the score obtained by plugging in the loss minimizer $\bar{\eta}$, and the posterior-averaged Brier score $\widehat{BS}$. The Brier score based on $\hat{\eta}$ is very close to the posterior-averaged score over the range of splits considered. In contrast, the score based on the posterior mode $\bar{\eta}$ is more sensitive to the split. Thus, although the posterior for $\eta$ itself changes as the amount of calibration data is varied, downstream predictive performance is comparatively stable when either the posterior mean or posterior averaging is used.

\begin{figure} 
\centering
\includegraphics[width=0.48\linewidth]{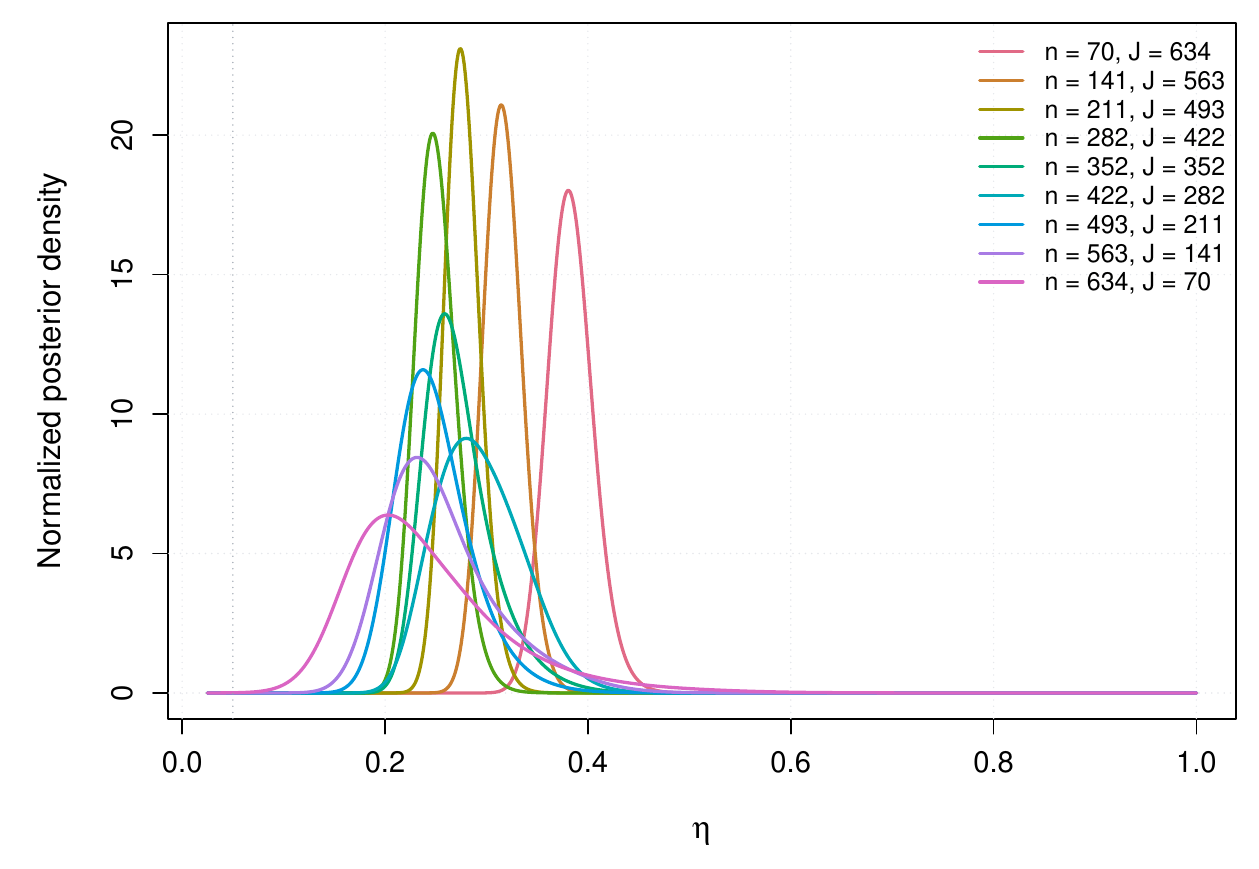}
\includegraphics[width=0.48\linewidth]{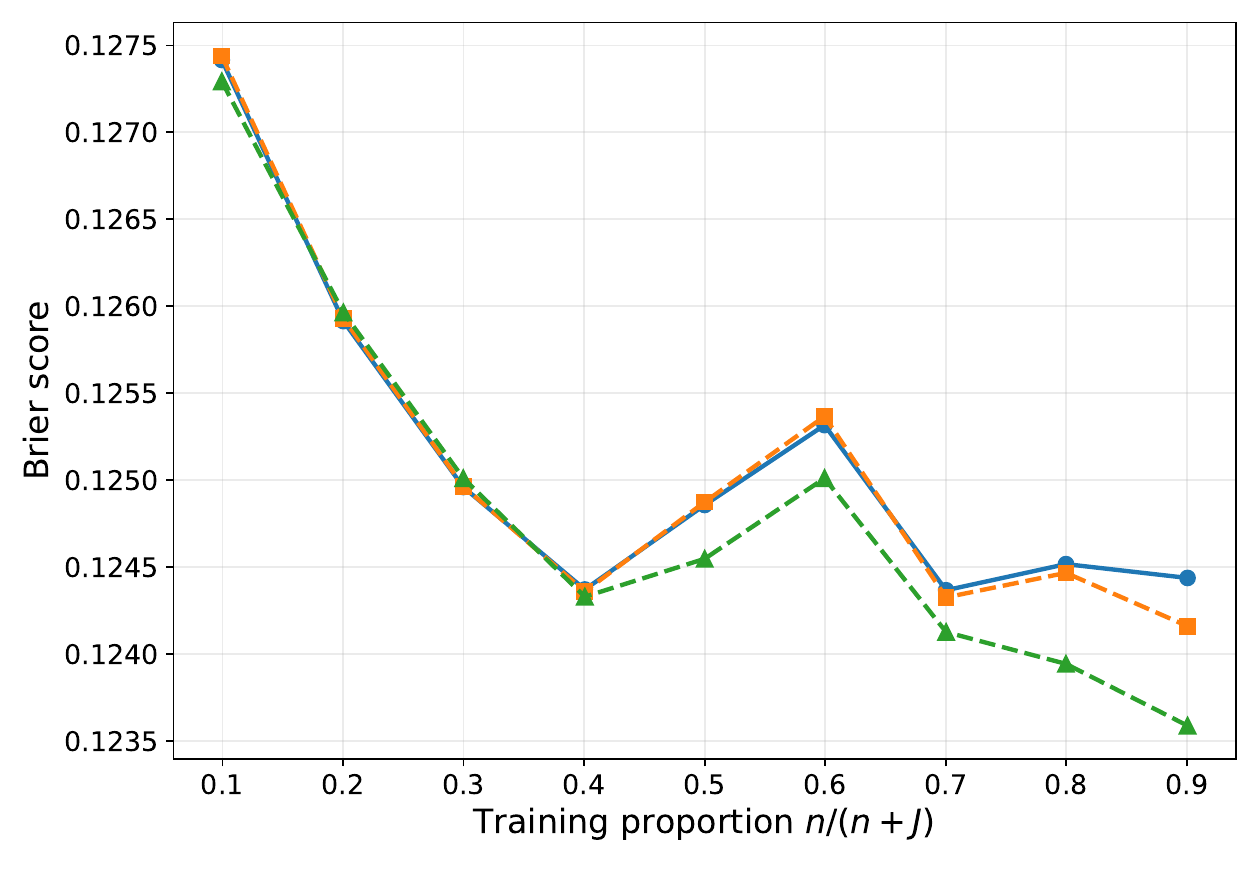}
\caption{Sensitivity to the training/calibration split for the Bank 1 dataset. \textit{Left:} $\rho_{(J,1)}(\eta| \y;\x)$ for different training proportions $n/(n+J)$. \textit{Right:} Brier scores obtained by plugging in the posterior mean $\hat{\eta}$ (orange line), plugging in the loss minimizer $\bar{\eta}$ (green line), and averaging the Brier score over the posterior for $\eta$ (blue line).}
    \label{fig:bank1_split_results}
\end{figure}

\subsubsection{EDiSC - clustering results at different $\eta$}\label{Appendix:edisc-clustering-bug}

To examine the model output, we look at the context words with the highest probabilities of co-occuring with the target word under each model sense. We show the top 10 context words of the word \textbf{bug} under model fits for $\eta=1$ (Bayes), $\eta=0.4$ (optimal) and $\eta=0.3$ (our estimate) and $\eta=0.1$ (too small) in Table~\ref{tab:context_words}. \textbf{bug} takes four meanings, i.e. (A) insect, (B) microorganism, (C) software glitch or (D) wire-tapping device. In the output for $\eta=0.4$ the target-word senses can be identified: clusters 1-4 are respectively A, B, C and D. For $\eta=0.1$, clusters 1 and 4 are both A, 3 is C, 2 is meaningless and senses B and D are lost.  Our method chooses $\eta=0.3$ which is a little small, but still gives fair sense-separation of top-words.

\begin{table}[!t]
\tiny
\noindent
\begin{tabularx}{\textwidth}{c X X X X X X X X X X}
\toprule
\multicolumn{11}{l}{\textbf{Sense \hspace{1em} Top 10 context words $\eta = 1$}} \\
\midrule
1 & insect & spray & bug & find & mosquito & eat & little & say & make & beetle \\
2 & p & cause & make & say & get & virus & bacterium & people & also & one\\
3 & fix & system & computer & software & new & use & update & company & say & device \\
4 & p & say & new & computer & year & get & company & find & make & also  \\
\midrule
\multicolumn{11}{l}{\textbf{Sense \hspace{1em} Top 10 context words $\eta=0.4$ (brier-score optimal)}} \\
\midrule
1 & insect & bug & spray &   mosquito & find & beetle & say & like & make & little  \\
2 & p & cause & say & make & bacterium & virus & get & people & one & find \\ 
3 & p & computer & new &     say & year & company & software & make & get &    system   \\
4 & say & new & p & federal & security & system & agent & get & computer & phone \\
\midrule
\multicolumn{11}{l}{\textbf{Sense \hspace{1em} Top 10 context words $\eta =\hat{\eta}= 0.3$ (our estimate)}} \\
\midrule
1 & insect & bug & spray & say & mosquito & find & like & p & beetle & get \\
2 & p & cause & virus & say & make & bacterium & get & one & year & find\\
3 & computer & p & new & say & company & year & software & system & make & get\\
4 & say & p & new & get & make & find & year & use & computer & know  \\
\midrule
\multicolumn{11}{l}{\textbf{Sense \hspace{1em} Top 10 context words $\eta = 0.1$}} \\
\midrule
1 & insect & bug & p & mosquito & say & like & spray & find & one & make \\
2 & say & p & get & make & one & new & know & like & use & find  \\
3 & p & computer & new & say & year & make & software & use & get & company  \\
4 & say & p & bug & make & get & like & insect & one & new & find \\
\bottomrule
\end{tabularx}
\caption{Top 10 context words across senses and $\eta$ values for \textbf{bug}. EDiSC runs unsupervised classification, so Brier-scores are not available in the real applications as there are no true meaning labels. We wish to estimate a Brier-score-optimal-$\eta$ without access to the Brier scores. \cite{zafar2024exploring} hand-labeled the snippets with their true senses. These data are not used in the inference but can be used to check it.}
\label{tab:context_words}
\end{table}


\subsection{Epidemiological data}\label{sec:hpv}

The model has two modules: in each population $i=1,\dots,n$, $n=13$, a Poisson response for the number of cancer cases $\x_{2,i}$ in $T_i$ women-years of followup, and a Binomial model for the number $\x_{1,i}$ of women infected with HPV in a sample of size $N_i$ from the $i$’th population; $\x_{2,i} \sim Poisson(T_i \exp(\beta_1+\beta_2\phi_i))$ and $\x_{1,i} \sim Binomial(N_i,\phi_i)$. The $\x_2$-module is known to be a suspect, and the $\x_1$-module is well-specified. The shared parameter is $\varphi=(\phi_1,\dots,\phi_n)$ and the $\x_2$-module specific parameter (see Figure~\ref{fig:Multimodular_model}) is $\theta=(\beta_1,\beta_2)$ following the two modular setup in Section 2. A uniform prior is assigned for each $\phi$, $\phi_i\sim U[0,1]$, $i=1,\dots,n$ and $\beta_1,\beta_2 \sim N(0,10^3)$.

To fit our purpose, the data from a Binomial model $(\x_{1,i},N_i)$ is represented by Bernoulli trials $\{\x_{1,i,j}\}^{N_i}_{j=1}$; $\x_{1,i,j} \sim \mbox{Bernoulli}(\phi_i)$ for $j=1,\dots, N_i$. This is to expand the data and divide it into training, calibration, and test datasets. The 1st, 3rd, 4th, 7th and 9th populations are selected because their empirical probabilities are not too close to 0, and there are enough samples per population. For each of those populations, $50\%$ of $\{\x_{1,i,j}\}^{N_i}_{j=1}$, $i\in\{1,3,4,7,9\}$ is used as the training data, $25\%$ as the calibration data $\y$, and $25\%$ as the test data $\z$. The $\x_{1,i}$ data for the other populations ($i\in \{2,5,6,8,10,11,12,13\}$) and all data from the Poisson model ($\x_2$) are used as the training data. 

For $\gamma$-SMI, estimation of the marginal loss is time-consuming (so it is helpful that it is well approximated by $\eta$-SMI!) so we restrict comparison to calibration with pooled losses where this only has to be done once, rather than $J$ times. We randomly partitioned data of the selected populations, estimated a pooled-loss posterior mean denoted by $\widehat{\eta}_{(1,J)}$ and compared the $\eta$-SMI posterior predictive for the test data. This was repeated 100 times, and the result is summarized in Figure \ref{fig:hpv} (we replicate using random splits of a single real data set, whereas in the normal and State Space Model comparisons we simulate new data for each replicate). The posterior mean estimate ($\widehat{\eta}_{(1,J)}$) tends to be smaller than 0.5 on average, and the relative risk shows that the cut posterior results in better prediction on average, consistent with \citet{carmona20}. 

For the $\gamma$-SMI posteriors, the marginalized loss $\ell_\gamma(\varphi;\x)$ is estimated using the AGHQ method with $k=5$ quadrature points.  The $\eta$-SMI posterior was simulated using the R-package, \texttt{rstan} \citep{Carpenter2018} and $2\times 10^4$ posterior samples (after a burn-in of $10^4$) are used to estimate $\ell_{\eta}(\varphi,\theta';\x_1,N,\x_2,T)$ and $R({\widehat{\eta}_{(1,J)},\eta_{o}})$ in Figure \ref{fig:hpv}.
The $\gamma$-SMI posterior was simulated using the R-package, \texttt{adaptMCMC} \citep{Scheidegger2024}, implementing the adaptive Metropolis algorithm \citep{Vihola2012}. The marginalized loss was estimated using the adaptive Gaussian Hermite quadrature method with 5 quadrature points using the R-package \texttt{aghq} \citep{Stringer2021} and $2\times 10^5$ posterior samples thinned by 10 are used to estimate $\ell_{\gamma}(\varphi;\x_1,N,\x_2,T)$ in Figure \ref{fig:hpv}.

The right plot in Figure \ref{fig:hpv} compares the logged negative loss functions for $\eta$ and $\gamma$ from one simulation and they are very similar. Consequently they perform similarly in $\eta$ and $\gamma$ estimation and predicting the test data; $\widehat{\eta}_{(1,J)}=0.366$, $\widehat{\gamma}_{(1,J)}=0.364$, $R(\widehat{\gamma}_{(1,J)},0)=0.178$, $R(\widehat{\eta}_{(1,J)},0)=0.195$, $R(\widehat{\gamma}_{(1,J)},1)=0.982$ and $R(\widehat{\eta}_{(1,J)},1)=1.078$. If we use the MLE we get $\bar{\eta}_{(1,J)}=\bar{\gamma}_{(1,J)}=0$, $R(\bar{\gamma}_{(1,J)},0)=R(\bar{\eta}_{(1,J)},0)=1$, $R(\widehat{\gamma}_{(1,J)},1)=R(\widehat{\eta}_{(1,J)},1)=0.001$ so this is a case where the MLE/MAP estimate does better than the posterior mean. This is because the minimum loss value is moderately strongly favored by loss and sits on the boundary at $\eta/\gamma=0$.

\begin{figure}[ht]   
\centering
    \begin{overpic}[width=0.4\linewidth,height=0.25\textheight]{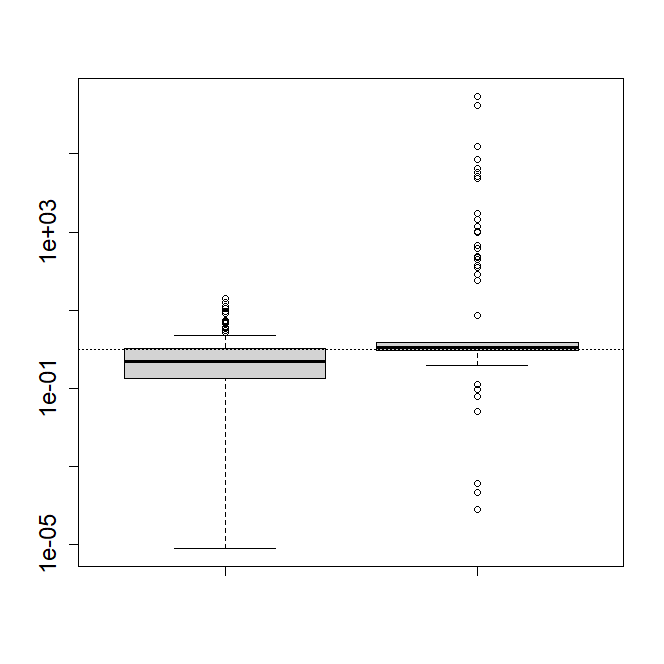}
    \put(20,5){\tiny $R(\widehat{\eta}_{(1,J)},0)$}
    \put(60,5){\tiny $R(\widehat{\eta}_{(1,J)},1)$}
    \end{overpic}
    \begin{overpic}[width=0.25\linewidth,height=0.25\textheight]{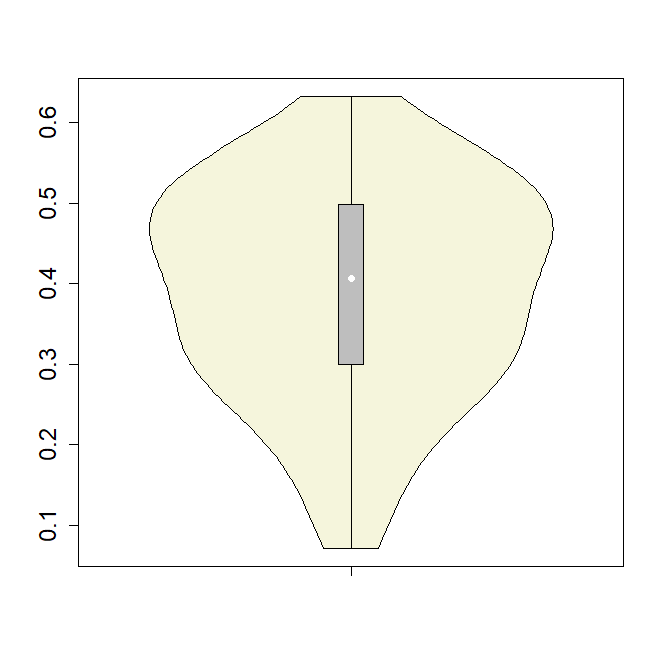}
    \put(32,5){\tiny $\widehat{\eta}_{(1,J)}$}
    \end{overpic}
    \hspace*{0.3cm}
    \begin{overpic}[width=0.25\linewidth,height=0.25\textheight]{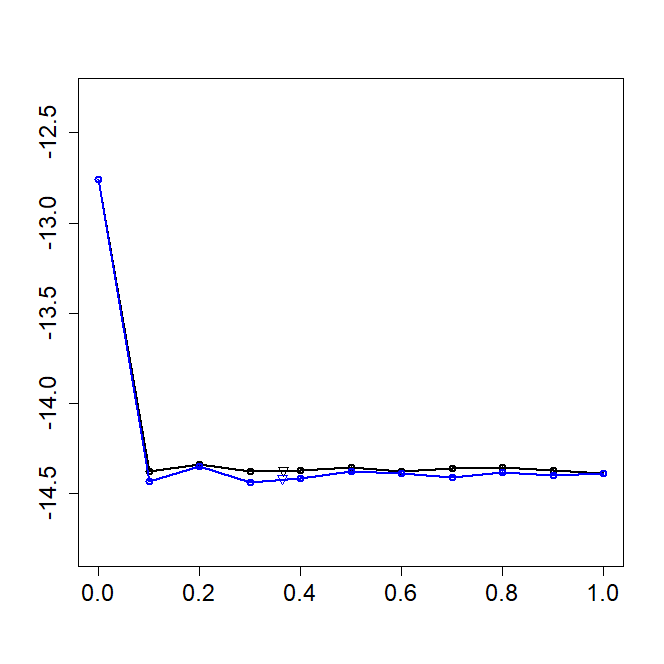}  
    \put(25,5){\tiny $\gamma$,$\eta$ values}
    \put(-5,30){\rotatebox{90}{\tiny $-l_{(1,J)}(s;\y,\x)$}}
    \end{overpic}
    \caption{Summary of relative posterior predictive (left plot) and the posterior mean estimates $\widehat{\eta}_{(1,J)}$ (middle plot) from 100 replicates. The plot on the right compares $-l_{(1,J)}(\eta;\y,\x)$ (black line) and $-l_{(1,J)}(\gamma;\y,\x)$ (blue line). The horizontal dotted line in the left plot indicates $R=1$.}
    \label{fig:hpv}
\end{figure}


\section{Remarks on the training/calibrate split}\label{Appendix:normal4-TC-split}

Our inference for the hyperparameters $s$ uses a split into training data $\x$ and calibration data $\y$. For fixed $s$, $\x$ determines the GBI posterior $\pi_s(\phi|\x)$, while $\y$ is used to update $s$ through $\rho(s\mid \y;\x)$. After estimating $s$, we recombine $\x$ and $\y$ for the final inference. The split therefore involves a trade-off; too little training data weakens $\pi_s(\phi|\x)$, while too little calibration data provides limited information about $s$. 

The effect of the split is illustrated first by the normal mixture
experiment in Appendix~\ref{sec:normal3}. Figure~\ref{fig:normal_b_gamma_eta_prod_pred} varies the training proportion $0<k<1$, so that $n=k n_1$ observations are used for training and
$J=(1-k)n_1$ for calibration. Predictive performance is fairly insensitive
to the split over most of the range of $k$, but deteriorates when the
calibration sample becomes very small. In this experiment, $n_1=30$, so
$k=0.8$ leaves only $J=6$ calibration observations, and the resulting
hyperparameter posterior is correspondingly weakly informed. By contrast,
at $k=0.2$ there are $J=24$ calibration observations, which are sufficient
to distinguish between values of $s$, even though the training sample is small. Figure~\ref{fig:normal_b_gamma_eta_prod_pred} shows that the resulting estimates can nevertheless
generalize well when $\x$ and $\y$ are subsequently recombined.

The EDiSC sensitivity experiment in Appendix~\ref{Appendix:EDiSC-pooled-loss} gives a complementary picture in a substantially larger real-data problem. There, we vary the
training proportion $k=n/(n+J)$ for the Bank 1 dataset ($n+J=704$ from Table~\ref{tab:EDiSC-experiment-settings}). The left panel
of Figure~\ref{fig:bank1_split_results} shows that changing the split changes both the location and the concentration of $\rho_{(J,1)}(\eta\mid y;x)$. Changing the split simultaneously changes the posterior predictive
$p_\eta(\,\cdot|\x)$ being calibrated and the amount of calibration
information available for learning $\eta$. The split affects
the inferred value of $\eta$ itself, and then the resulting downstream
inference (the right panel of Figure~\ref{fig:bank1_split_results}). The resulting Brier scores are relatively insensitive to the split over most of the range. However, when the training sample is very small, limited information about $(\phi,\varphi)$ can outweigh the benefit of having more calibration information about $\eta$, leading to weaker predictive performance.

The limiting cases also clarify why some splitting is essential. If we take $k=1$ then we simply get the prior for $s$ and no belief update. At the other extreme, using the same observations both to construct $\pi_s(\phi|\x)$ and to calibrate $s$ would no longer provide a meaningful inferential goal in the belief update for $s$. We were critical of this in the introduction. A fuller investigation of how to choose the split is left for future work.


\section{Assumptions 1-5 in a normal model}\label{app:gbi-normal-example}

In this section we give a very simple example where we can check the 
assumptions in Section~\ref{sec:bayes_s}. 
We take a normal model with true data-distributions
\[
x_i,\ y_j \stackrel{\mathrm{iid}}{\sim} N(\mu^*,1),\ i=1,\dots,n,\quad j=1,\dots,J\,.
\]
The mean and variance are unknown and assigned a conjugate prior
\[
\pi(\theta,\sigma^2)\propto 1\times \operatorname{IG}(\sigma^2;a,b),
\qquad a>0,\quad b>0.
\]
The fitted observation model is
\[
x_i\mid \theta,\sigma^2 \sim N(\theta,\sigma^2),
\]
so $p(\y|\theta,\sigma^2)$ is $C^\infty$ for $(\theta,\sigma^2)\in \R\times (0,\infty)$ in Assumption~\ref{A:laplace}(i). Take $r=n\eta$ and let
\[
\bar x = \frac1n\sum_{i=1}^n x_i,
\quad
s_x = \frac1n\sum_{i=1}^n (x_i-\bar x)^2,
\quad
\bar y = \frac1J\sum_{j=1}^J y_j,
\quad
s_y = \frac1J\sum_{j=1}^J (y_j-\bar y)^2, \quad \Delta:=\bar x-\mu^*.
\]
By conjugacy, $\pi_r(\theta,\sigma^2|\x)\! =\!
N\!\left(\theta;\bar x,{\sigma^2}/{r}\right)\,
\operatorname{IG}(\sigma^2;\alpha_r,B_r)$ so
\begin{align}
\pi_r(\theta,\sigma^2|\x) =
\sqrt{\frac{r}{2\pi}}\,
\frac{B_r^{\alpha_r}}{\Gamma(\alpha_r)}
(\sigma^2)^{-\alpha_r-3/2}
\exp\left\{
-\frac{B_r+r(\theta-\bar x)^2/2}{\sigma^2}
\right\}.
\label{eq:post-param-appendix-F}
\end{align}
with
\[
\alpha_r:=a+\frac{r-1}{2},
\qquad
B_r:=b+\frac r2 s_x.
\]
We need $\alpha_r>0$ so $r>r_0=\max\{0,1-2a\}$ and the space for the learning rate 
is $\Omega_\eta=\left(\frac{r_0}{n},\infty\right)$, so if $a\ge 1/2$ then $\Omega_\eta=(0,\infty)$.

The pooled posterior predictive, $p_r(\y|\x)=\int p(\y|\theta,\sigma^2)\pi_r(\theta,\sigma^2|\x)d\theta d\sigma^2$ is
\begin{equation*}
p_r(\y|\x)=\frac{\Gamma\{(\nu_r+J)/2\}}
 {\Gamma(\nu_r/2)(\nu_r\pi)^{J/2}|\Sigma_r|^{1/2}}
 {\left\{1+\frac{(\y-\bar x\,\one_J)^\top\Sigma_r^{-1}(\y-\bar x\,\one_J)}{\nu_r}\right\}.}^{\hspace*{-0.1in}-(\nu_r+J)/2}\hspace*{-0.2in}
\end{equation*}
This is a $J$-variate Student-$t$ density $t_J(\y;\bar x\,\one_J,\Sigma_r,\nu_r)$ with
$\nu_r=2a+r-1$ degrees of freedom, location $\bar x\,\one_J$, and scale matrix
\[
\Sigma_r
=
\frac{B_r}{\alpha_r}
\left(
I_J+\frac1r\one_J\one_J^\top
\right).
\]
For the product posterior predictive we set $J=1$ and take a product so
\begin{align}
\widetilde p_r(\y|\x)
=
\prod_{j=1}^J t_1(y_j;\bar x, \lambda_r,\nu_r)
\label{eq:post-pred-prod-app-F}
\end{align}
with $\lambda_r^2 = (1+1/r)B_r/\alpha_r$ a scalar and $\nu_r$ unchanged.

\subsection{Pooled case}\label{app:gbi-normal-example-pooled}

The normal location-scale family is a full, regular, minimal exponential
family after the smooth natural reparameterisation
$\phi=({\theta}/{\sigma^2},-{1}/{2\sigma^2})$,
with sufficient statistic
$T(y)=(y,y^2)^\top$. The true distribution $\pstar=N(\mu^\star,1)$ has a strictly positive
Lebesgue density and finite moments of $T(Y)$. Moreover, the
normal--inverse-gamma posterior in \eqref{eq:post-param-appendix-F} is $C^\infty$ in
$(\theta,\sigma^2)$ and jointly continuous in
$(r,\theta,\sigma^2)$ on every compact subset of
$
(r_0,\infty)\times\mathbb R\times(0,\infty)$.
Its normalising constant is finite and positive for $r>r_0$.
Consequently, Proposition~\ref{prop:ass1-reg-exp-fam} verifies
Assumption~\ref{A:laplace}(i - iv). Assumption~\ref{A:laplace}(v) holds by Proposition~\ref{prop:ef-uniform-laplace}. The mean at $\phi(\tau)=(m/w,-1/2w)$ is
$\tau(\phi)=\mathbb{E}\{T(Y)\}=(m,m^2+w)^\top$ so
the contrast function in condition (A) at $\phi$ is
\begin{align}
I_x(\theta,\sigma^2;m,w)\coloneqq F_\mu(\phi)-F_\mu\{\phi(\mu)\}
=
\frac{1}{2}
\left\{
\log\left(\frac{\sigma^2}{w}\right)
+
\frac{w}{\sigma^2}
-1
+
\frac{(\theta-m)^2}{\sigma^2}
\right\}.
\label{eq:contrast-normal-case-app-F}
\end{align}
On a compact neighborhood $\tau\in U_\tau$ of
$(\mu^\star,(\mu^\star)^2+1)^\top$ within $\operatorname{Int}(\mathcal M)$, $m$ is bounded and $w$ is
bounded above and away from zero. Bounding $I_x(\theta,\sigma^2;m,w)\le c$ over $\theta,v$ at fixed $m,w$, as in condition (A), forces
$v/w$ to be bounded above and away from zero, and then forces
$\theta-m$ to be bounded. Hence its level sets are uniformly contained
in a compact subset of the $(\theta,\sigma^2)$-space $\mathbb{R}\times(0,\infty)$. That set maps to a compact subset of $\Omega_\phi$, so condition~(A) holds. Condition (B) is also satisfied as $\nabla_\phi\pi_r(\phi|\mathbf{x})$ in \eqref{eq:post-param-appendix-F} is jointly continuous on $K_s\times N_\phi$ and requiring $N_\phi$ to be convex is no restriction here.

It remains to verify
Assumption~\ref{A:log_min_pooled}. Since \(n\) is fixed and \(r=n\eta\), we write
$\bar r=n\bar\eta$,
 $\widetilde r_{(1,\infty)}
=n\widetilde{\eta}_{(1,\infty)}$,
and use a compact set \(K_r\) for \(r\) with \(K_s=K_r/n\).

We begin with Assumption~\ref{A:log_min_pooled}(i): Let $g(r)\coloneqq \log(\pi_r(\mu^\star,1|\x))$; we need a unique $\widetilde r_{(1,\infty)}
 =\arg\max_{r>r_0} g(r)$, satisfying
 \begin{equation}\label{eq:true-r-tilde-interior-f-appendix}
 r_0<\widetilde r_{(1,\infty)}<\infty,
 \qquad\text{almost surely in }x
\end{equation}
to exist. At \((\theta,\sigma^2)=(\mu^\star,1)\), directly from \eqref{eq:post-param-appendix-F},
\[
g(r)=
\frac{1}{2}\log r
-\frac{1}{2}\log(2\pi)
-\frac{r\Delta^2}{2}
+\alpha_r\log B_r
-\log\Gamma(\alpha_r)
-B_r.
\]
Differentiation gives
\begin{align*}
 g'(r)
 &=\frac1{2r}-\frac{\Delta^2}{2}
   +\frac12\log B_r+\frac{\alpha_r s_x}{2B_r}
   -\frac12\psi(\alpha_r)-\frac{s_x}{2},
 \\
 g''(r)
 &=-\frac1{2r^2}+\frac{s_x}{2B_r}
   -\frac{\alpha_r s_x^2}{4B_r^2}-\frac14\psi_1(\alpha_r),
\end{align*}
where \(\psi_1\) is the trigamma function. Since
\(\psi_1(\alpha)>1/\alpha\) for \(\alpha>0\),
\begin{align*}
 g''(r)
 &< -\frac1{2r^2}
 +\frac{s_x}{2B_r}-\frac{\alpha_r s_x^2}{4B_r^2}
 -\frac1{4\alpha_r}\\
 &=-\frac1{2r^2}
   -\frac{(B_r-\alpha_r s_x)^2}{4\alpha_r B_r^2}<0.
\end{align*}
Thus \(g\) is strictly concave on \((r_0,\infty)\).

At the lower tail, \(g(r)\to-\infty\): recall 
$\alpha_r=a+(r-1)/2$ and 
$r_0=\max\{0,1-2a\}$;
if \(a<1/2\), then \(r_0=1-2a>0\) and
\(-\log\Gamma(\alpha_r)\to-\infty\); if \(a\ge 1/2\) then $r_0=0$ so \(\tfrac12\log r\to-\infty\) (with \(-\log\Gamma(r/2)\to-\infty\) at $a=1/2$). Recall $B_r=b+r s_x/2$. At the upper boundary,
\[
 \lim_{r\to\infty}g'(r)
 =\frac12\{\log s_x+1-s_x-\Delta^2\}<0
\]
unless \((\Delta,s_x)=(0,1)\), using \(\log s_x\leq s_x-1\). For \(n\geq2\),
we have \(s_x>0\) almost surely, and the exceptional equality event has
probability zero. Hence \(g(r)\to-\infty\) as \(r\to\infty\).

Concavity and the boundary limits mean there is a unique
interior maximiser of $g$, equivalently the pooled-risk minimizer. This establishes \eqref{eq:true-r-tilde-interior-f-appendix} and Assumption~\ref{A:log_min_pooled}(i).

It remains to verify Assumption~\ref{A:log_min_pooled}(ii). Work on the probability-one
event
\[
s_x>0,
\qquad
\widehat\phi_x:=(\bar x,s_x)\neq\widetilde\phi:=(\mu^\star,1).
\]
Choose a compact neighbourhood
$
N_\phi\subset\mathbb{R}\times(0,\infty)
$
of $\widetilde\phi$ which does not contain $\widehat\phi_x$. It follows from \eqref{eq:post-param-appendix-F} that
\begin{align}
\sup_{\phi\in N_\phi}\pi_r(\phi|\mathbf{x})
\longrightarrow0
\qquad\text{as }r\downarrow r_0
\quad\text{or}\quad
r\longrightarrow\infty.
\label{eq:normal-posterior-r-tails}
\end{align}
At the lower boundary this follows from the factor $\sqrt r$ (when $a\ge 1/2$ so $r_0=0$) or
$1/\Gamma(\alpha_r)\to0$ (when $a\le 1/2$ so $\alpha_r\to 0$ at $r\to r_0$). At the upper boundary,
\[
\log\pi_r(\theta,v|\mathbf{x})
=
-rI_x(\theta,v;\bar x,s_x)+O(\log r)
\]
uniformly on $N_\phi$, with $I_x$ the contrast in \eqref{eq:contrast-normal-case-app-F}.
Since $\log(t)+1/t-1\ge 0$ for $t>0$, the contrast is a positive function of $\theta,v$, vanishing only at $(\theta,v)=(\bar x,s_x)$, or in other words at $\widehat\phi_x\notin N_\phi$, 
which proves the upper-boundary part of
\eqref{eq:normal-posterior-r-tails}.

Define
\[
Q_J(r)
:=
{p_r(\mathbf{y}|\mathbf{x})}/{C_J(\mathbf{y})},
\qquad
q_\star
:=
\pi_{\widetilde r_{(1,\infty)}}
(\widetilde\phi|\mathbf{x})>0,\qquad L_J(\phi):=p(\mathbf{y}|\phi).
\]
We now set $K_r=[r_-,r_+]$, choosing $r_\pm$ so that both $r_0<r_-<\widetilde r_{(1,\infty)}<r_+<\infty$ and 
\begin{align*}
c
:=
\sup_{r\notin K_r}
\sup_{\phi\in N_\phi}
\pi_r(\phi|\mathbf{x})
<
q_\star-\epsilon,\qquad\mbox{for some } 0<\epsilon<q_\star.
\end{align*}
This can always be done by pushing $r_\pm$ into the tails and using \eqref{eq:normal-posterior-r-tails}. For $r\notin K_r$,
\[
\begin{split}
Q_J(r)
&=
\frac{1}{C_J(\mathbf{y})}
\int_{N_\phi}
L_J(\phi)\pi_r(\phi|\mathbf{x})\,d\phi+
\frac{1}{C_J(\mathbf{y})}
\int_{N_\phi^c}
L_J(\phi)\pi_r(\phi|\mathbf{x})\,d\phi
\\
&\leq
c\,
\frac{1}{C_J(\mathbf{y})}
\int_{N_\phi}L_J(\phi)\,d\phi
+
\frac{\sup_{\phi\notin N_\phi}L_J(\phi)}
{C_J(\mathbf{y})}.
\end{split}
\]
Equations \ref{eq:ef-tail-laplace} and \ref{eq:ef-unit-laplace} therefore give
\begin{align}
\limsup_{J\to\infty}
\sup_{r\notin K_r}Q_J(r)
\leq c
<
q_\star-\epsilon
\qquad\text{almost surely}.
\label{eq:normal-r-tail-bound}
\end{align}
On the other hand, Assumption~\ref{A:laplace}(v), continuity and
$\bar\phi_J\to\widetilde\phi$ give
\begin{align}
Q_J\{\widetilde r_{(1,\infty)}\}
\longrightarrow
\pi_{\widetilde r_{(1,\infty)}}
(\widetilde\phi|\mathbf{x})
=
q_\star
\qquad\text{almost surely}.
\label{eq:normal-r-target}
\end{align}
It follows from \eqref{eq:normal-r-tail-bound} and
\eqref{eq:normal-r-target} that
$
Q_J\{\widetilde r_{(1,\infty)}\}
>
\sup_{r\notin K_r}Q_J(r),
$, almost surely for all sufficiently
large $J$,
so $Q_J$ attains its
maximum somewhere in the compact set $K_r$. Thus there exists
\[
\bar r_{(1,J)}
\in
\argmax_{r\in K_r}Q_J(r).
\]
Since we have now established that, almost surely for all sufficiently large $J$,
\[
\max_{r\in K_r}Q_J(r)
\geq
Q_J\{\widetilde r_{(1,\infty)}\}
>
\sup_{r\notin K_r}Q_J(r),
\]
every global maximiser of $Q_J$ eventually lands in $K_r$.
Because $C_J(\mathbf{y})$ is positive and does not depend on $r$, the
same statements hold for $p_r(\mathbf{y}|\mathbf{x})$. 
We have established that $\bar r_{(1,J)}$ exists and 
\[
\Pr\{\bar r_{(1,J)}\in K_r\}\longrightarrow1.
\]
Since $\bar r_{(1,J)}$ is an interior point of the $r$-domain and the
loss is differentiable,
\[
\nabla_r l_{(1,J)}
\{\bar r_{(1,J)};\mathbf{y},\mathbf{x}\}
=
0
\]
which completes Assumption~\ref{A:log_min_pooled}(ii).

\subsection{Product case}
\label{app:gbi-normal-example-product}

We now consider Assumptions~\ref{A:product-regular} and
\ref{A:product-separation}. Picking up from the product posterior predictive in \eqref{eq:post-pred-prod-app-F},
with \(r=n\eta\) and $\alpha_r,B_r,r_0,\nu_r$ and $\lambda_r$ defined there, the single-observation posterior predictive, $p_r(y|\x)=t_1(y;\bar x,\lambda_r,\nu_r)$ in \eqref{eq:post-pred-prod-app-F} is
\begin{align}
p_r(y|\x)
=
\frac{\Gamma\{(\nu_r+1)/2\}}
{\Gamma(\nu_r/2)(\pi\nu_r\lambda_r^2)^{1/2}}
\left\{
1+\frac{(y-\bar x)^2}{\nu_r\lambda_r^2}
\right\}^{-(\nu_r+1)/2}.
\label{eq:product-one-step-predictive}
\end{align}
Define the one-observation loss 
$
\ell_r(y|\x):=-\log p_r(y|\x)
$, the risk
\begin{align}
\mathcal R(r|\x)
:=
\mathbb E_{Y\sim N(\mu^\star,1)}
\{\ell_r(Y|\x)\}
\label{eq:product-risk-r}
\end{align}
and the average loss
$
L_J(r|\x)
:=
J^{-1}\sum_{j=1}^J\ell_r(y_j|\x)
$.
In the notation of Section~\ref{asymptotic_regular}, $L_J(r|\x)=J^{-1}l_{(J,1)}(r/n;\y,\x)$ and $\mathcal R(r|\x)=\widetilde l(r/n,\x)$ with $\bar r_{(J,1)}=\argmin L_J(r|\x)$ and $\widetilde r_{(\infty,1)}=\argmin R(r|\x)$.

\paragraph{Local regularity.} For every fixed \(y\), the map
$r\longmapsto\ell_r(y|\x)$
is analytic on \((r_0,\infty)\). 
Moreover, on every compact interval
\(M=[r_-,r_+]\) with
\(r_0<r_-<r_+<\infty\), differentiation of
\eqref{eq:product-one-step-predictive} gives, for
\(k=0,1,2,3\), 
\[
\sup_{r\in M_{\epsilon,R}}
\left|
\frac{\partial^k}{\partial r^k}\ell_r(y|\x)
\right|
\leq
C_{M}
\{1+(y-\bar x)^2\},
\]
where \(C_{M}<\infty\) doesn't depend on $J$ or \(y\). The RHS is
integrable under \(Y\sim N(\mu^\star,1)\).
Integrability for $k=0$ gives a risk and for $k=1$ the mean-value theorem gives an integrable Lipschitz
bound over $M$. Taken with the ordinary strong law, this gives the uniform strong law of
large numbers,
\begin{align}
\sup_{r\in M}
|L_J(r|\x)-\mathcal R(r|\x)|
\longrightarrow0
\qquad\text{almost surely}.
\label{eq:loss-to-risk-strong}
\end{align}
Since this holds for all $\epsilon>0$ and $r_0+\epsilon<\mathcal R<\infty$ it (more than) verifies Assumption~\ref{A:product-regular}(i).

The same bounds and strong law give, for every fixed compact
interval \(M\),
\[
\sup_J\sup_{r\in M}
|L_J'''(r)|<\infty
\qquad\text{almost surely}.
\]
This covers any bounded open interval \(E_r\) whose closure is
contained in \((r_0,\infty)\) so it verifies Assumption~\ref{A:product-regular}(iii) on \(E_r\). What remains is the existence of an interior risk minimiser and separation from the boundaries.

\paragraph{The two boundaries.}

The Student-\(t\) density in
\eqref{eq:product-one-step-predictive} is maximised at \(y=\bar x\).
Its maximum satisfies
$
p_r(\bar x|\x)\rightarrow0$ as $r\downarrow r_0
$
by the same reasoning as for \eqref{eq:normal-posterior-r-tails}. Then from \eqref{eq:product-risk-r} with
$\ell_r(y|\x)=-\log p_r(y|\x)$ we have
$
\mathcal R(r|\x)\rightarrow\infty
$
as
$r\downarrow r_0$ so the lower boundary cannot minimise the risk or
the loss.

As $r\to\infty$, the number of degrees of freedom $\nu_r=2a+r-1$ in the $t$-distribution in \eqref{eq:product-one-step-predictive} goes to infinity and it converges to a normal. Let 
$
D(y):={(y-\bar x)^2}/{s_x}
$.
A large-\(r\) asymptotic expansion, obtained from Taylor expansions in \(1/r\) and the Stirling expansion for the gamma-function ratio gives
\begin{align}
\log p_r(y|\x)
=
\log N(y;\bar x,s_x)
+
\frac{a_1(D(y))}{r}
+
\mathcal H_r(y),
\label{eq:product-one-step-expansion}
\end{align}
where
\begin{align}
a_1(D)
=
\frac14
\left[
D^2+
\left(2-4a+\frac{4b}{s_x}\right)D
+4a-5-\frac{4b}{s_x}
\right].
\label{eq:product-a1}
\end{align}
and
\begin{align}
|\mathcal H_r(y)|
\leq
\frac{C}{r^2}\{1+D(y)^3\}.
\label{eq:product-expansion-remainder}
\end{align}
Let
$
L_{J,\infty}
:=
-J^{-1}\sum_{j=1}^J
\log N(y_j;\bar x,s_x)
$
and
$
\bar A_{J,1}
=
J^{-1}\sum_{j=1}^J a_1(D_j)$
with 
$D_j=D(y_j)$. 
Equations~\eqref{eq:product-one-step-expansion} - %
\eqref{eq:product-expansion-remainder} give
\begin{align}
L_J(r|\x)
=
L_{J,\infty}
-
\frac{\bar A_{J,1}}{r}
+
O(r^{-2}),
\label{eq:product-empirical-upper-expansion}
\end{align}
where the remainder has a coefficient which is almost surely finite.
By the strong law,
\begin{align}
L_{J,\infty}
\rightarrow
\mathcal R_\infty
:=
\mathbb E_{Y\sim N(\mu^\star,1)}
\{-\log N(Y;\bar x,s_x)\}
\label{eq:loss-r-inf-risk-strong}
\end{align}
and
$
\bar A_{J,1}
\rightarrow
\bar A_{\infty,1}
:=
\mathbb E\bigl[a_1\{D(Y)\}\bigr]
$
almost surely. Recalling that
$
\Delta=\bar x-\mu^\star,
$
we have
\[
\mathbb E\{D(Y)\}
=
\frac{1+\Delta^2}{s_x},
\qquad
\mathbb E\{D(Y)^2\}
=
\frac{3+6\Delta^2+\Delta^4}{s_x^2},
\]
and hence
\begin{align}
\bar A_{\infty,1}
&=
\frac{3+6\Delta^2+\Delta^4}{4 s_x^2}
+
\left(
\frac12-a+\frac{b}{s_x}
\right)
\frac{1+\Delta^2}{s_x}
+
a-\frac54-\frac{b}{s_x}.
\label{eq:product-A-limit}
\end{align}
Taking expectations in
\eqref{eq:product-one-step-expansion} gives
\begin{align}
\mathcal R(r|\x)
=
\mathcal R_\infty
-
\frac{\bar A_{\infty,1}}{r}
+
O(r^{-2})
\qquad
\text{as }r\to\infty.
\label{eq:product-risk-upper-expansion}
\end{align}
To determine the eventual direction of the risk, we differentiate the
exact predictive density \emph{before} taking its large-$r$ expansion. 
Since
$\ell_r(y\mid\x)=-\log p_r(y\mid\x)$, it follows by straightforward but lengthy calculation that
\[
\frac{\partial}{\partial r}\ell_r(y\mid\x)
=
\frac{a_1\{D(y)\}}{r^2}
+
O\left(\frac{1+D(y)^3}{r^3}\right).
\]
Under $Y\sim N(\mu^\star,1)$,
$
\mathbb E\{1+D(Y)^3\}<\infty
$,
so we can differentiate under the expectation in \eqref{eq:product-risk-r}, then expand and take the expectation to obtain
\begin{align}
\mathcal{R}'(r\mid\x)
=\frac{\bar A_{\infty,1}}{r^2}
+O(r^{-3}).
\label{eq:H-product-risk-derivative-expansion}
\end{align}

If $\bar A_{\infty,1}>0$, the risk is eventually increasing in the upper tail at large $r$. Since
$\mathcal R(r|\x)\to\infty$ as $r\downarrow r_0$ in the lower tail, the risk must have at least
one finite interior minimiser. 
Since $\bar A_{J,1}\to \bar A_{\infty,1}$, the same holds for the loss at all sufficiently large $J$. We conclude that if $\bar A_{\infty,1}>0$ then an interior minimum occurs. If it is unique and the Hessian condition holds then Assumption~\ref{A:product-regular}(ii) and the compact
uniform convergence and boundary estimates above gives
Assumption~\ref{A:product-separation}. 
Uniqueness and non-degeneracy do not follow from the tail expansion. 
Numerical experiments plotting $\mathcal R(r|\x)$ in \eqref{eq:product-risk-r} at different
$a,b,s_x$ and $\Delta$ values show that multiple \emph{local} minima do occur, but to get non-unique \emph{global} minima we need special tuning of the data-dependent $s_x$ and $\Delta$. 
If $\bar A_{\infty,1}<0$ then the risk is eventually decreasing at large $r$, so the upper boundary is at least a local minimum. If this is a global minimum then Assumptions~\ref{A:product-regular}(ii) and \ref{A:product-separation} do not hold. 
We show below that the non-regular asymptotics of Theorem~\ref{thm:post.boundary} then apply.

At large but finite $n$ we will have $\bar x\simeq\mu^*$ and $s_x\simeq 1$ so $\Delta\simeq 0$ and $\bar A_{\infty,1}\simeq 0$. This suggests the probability $\widetilde r_{(\infty,1)}=\infty$ is around a half (under resampling of $\x$) at large $n$ and this is borne out in simulation. At large $J$ the signs of $\bar A_{\infty,1}$ and $\bar A_{J,1}$ are the same so $\bar r_{(J,1)}$ has the same behavior at large $J$. This doesn't depend on the prior hyperparameters $a,b$.
At lower $n$ and $J$ values the choice of $a$ and $b$ do make a difference. In Table~\ref{tab:prob-eta-finite-prod} we parameterise using $\mu_{\sigma^2}=E(\sigma^2)=b/(a-1)$ and $v_{\sigma^2}=\mbox{var}(\sigma^2)=b^2/(a-1)^2(a-2)$ are the prior mean and variance of $\sigma^2$. 

We tried prior variance $v_{\sigma^2}=0.1$ and mean $\mu_{\sigma^2}$ below, at and above the true value to investigate the effect of misspecification and $(\mu_{\sigma^2},v_{\sigma^2})=(1,1)$ and $(2,4)$ for something like typical subjective values. What is happening here is that if $\bar x\simeq \mu^*$ and $s_x\simeq 1$ then the simple MLE $\bar\theta=\bar x$ is very close to the true value of $\theta$ (ie, $\mu^*$) so $\y$ is well-predicted by shrinking $\theta$ onto $\bar x$, and this is achieved by taking $\eta=\infty$.

\begin{table}[H]
    \centering
    \begin{tabular}{l|rrrrr}
        \toprule
        &\multicolumn{4}{c}{$(\mu_{\sigma^2},v_{\sigma^2})$}\\[-0.1in]
        \\
        & (0.5, 0.1) & (1, 0.1) & (4, 0.1) & (1, 1) & (2, 4)\\
        \midrule
        $\Pr(\bar \eta_{(J,1)}<\infty)$ & 0.62 & 0.72 & 0.62 & 0.60 & 0.59 \\[0.05in]
        $\Pr(\widetilde \eta_{(\infty,1)}<\infty)$ & 0.73 & 0.94 & 0.71 & 0.77 & 0.71\\
        \bottomrule
    \end{tabular}
    \caption{Estimated probabilities for optimal $\eta$ to be interior at $n=10$ and $J=10$ at different values of the prior mean and variance for $\sigma^2$. We illustrate with small values of $n$ and $J$ as the probabilities are one half at large $n$ and $\Pr(\bar \eta_{(J,1)}<\infty)\simeq \Pr(\widetilde \eta_{(\infty,1)}<\infty)$ at large $J$. The true value of $\sigma^2=1$ and $\mu^*=0$. The entries have standard errors $\le 0.002$ and are based on $N=10^5$ replicate $\x,\y$-datasets.}
    \label{tab:prob-eta-finite-prod}
\end{table}

\paragraph{Non-regular asymptotics.} Suppose $\bar A_{\infty,1}<0$ and that $r=\infty$
is the unique global minimiser of $\mathcal{R}(r|\x)$ so the first part of Assumption~\ref{A:s_boundary}(i) holds. Non-regular asymptotics only matter for a boundary minimiser so the real assumption is uniqueness. 
We now verify the other conditions.
Let
$
u=1/(r-r_0)$, $L_J^u(u):=L_J(r_0+u^{-1}|\x)$ and $\mathcal{R}_u(u|\x):=R(r_0+u^{-1}|\x)$,
extended continuously to $u=0$. Defining $\ell_r(y\mid\x):=-\log p_r(y\mid\x)$ in
\eqref{eq:product-one-step-predictive}, 
differentiating  and applying the same expansions gives
\begin{align}
{\partial_u}
\ell_{r(u)}(y|\x)
=
-a_1\{D(y)\}+uB(u,y),
\qquad
\sup_{0\leq u\leq u_0}|B(u,y)|
\leq C\{1+D(y)^3\},
\label{eq:normal-IG-loss-deriv-expand}
\end{align}
for some $u_0>0$ with $a_1(D)$ given in \eqref{eq:product-a1}. Equations~\eqref{eq:product-empirical-upper-expansion} and
\eqref{eq:product-risk-upper-expansion} give
\[
L_J^u(u)
=
L_{J,\infty}-u\bar A_{J,1}+O(u^2)
\qquad
\mbox{and}
\qquad
\mathcal{R}_u(u|\x)
=
R_\infty-u\bar A_{\infty,1}+O(u^2).
\]
Now $\partial_u\mathcal{R}_u(u|\x)=-(r-r_0)^2\mathcal{R}'(r|\x)$ with $\mathcal{R}'(r|\x)$ given in \eqref{eq:H-product-risk-derivative-expansion}. It follows that 
\[
\lim_{u\to 0} \partial_u\mathcal{R}_u(u|\x)=\lim_{r\to\infty}-(r-r_0)^2\mathcal{R}'(r|\x)=c_\infty,
\]
with $c_\infty=-\bar A_{\infty,1}$ by \eqref{eq:H-product-risk-derivative-expansion}. Now $c_\infty$ is finite and positive, as $\bar A_{\infty,1}<0$ in \eqref{eq:product-A-limit} is needed for $\mathcal{R}(r|\x)$ to have even a local minimum at $r=\infty$. This verifies $\partial_u\mathcal{R}_u(0_+|\x)>0$ and the second part of Assumption~\ref{A:s_boundary}(i). 

For Assumption~\ref{A:s_boundary}(ii), recall
$
\bar A_{J,1}
=
J^{-1}\sum_{j=1}^J a_1(D_j)
$ and use \eqref{eq:normal-IG-loss-deriv-expand} to get
\[
\sup_{0\leq u\leq\delta_J}
\left|
\partial_uL_J^u(u\mid\x)-c_\infty
\right|
\leq
\left|
\bar A_{J,1}-\bar A_{\infty,1}
\right|
+
C\delta_J
\frac1J\sum_{j=1}^J\{1+D(y_j)^3\}\xrightarrow{a.s.} 0,
\]
as the normal distribution has finite moments and
$\bar A_{J,1}\to\bar A_{\infty,1}$ almost surely. 

If the prior probability distribution has $\alpha$-polynomial tails,
in the sense that 
$
\rho_r(r)
=
D r^{-\alpha-1}\{1+o(1)\}
$ as $r\to\infty
$
with $\alpha>0$, the posterior is proper and
\[
\rho_u(u)
=
u^{-2}\rho_r(r_0+u^{-1})
=
D u^{\alpha-1}\{1+o(1)\},
\qquad u\downarrow0
\]
satisfying Assumption~\ref{A:s_boundary}(iii) and verifying Assumption~\ref{A:s_boundary} given the assumed uniqueness.

We now verify the tail condition at the start of Theorem~\ref{thm:post.boundary} in which $J\to\infty$ and $\delta_J\to 0$ with $J\deltaJ\to\infty$. Split the $\Delta_{\rm tail}(\deltaJ)$ integral at $u_0$ in \eqref{eq:normal-IG-loss-deriv-expand},
\[
\Delta_{\rm tail}(\deltaJ)=J^\alpha
\left[\int_{\deltaJ}^{u_0} + \int_{u_0}^\infty\right]
\exp\left[
 -J\{L_J^u(u\mid\x)-L_J^u(0\mid\x)\}
\right]\rho_u(u)\,du,
\]
and consider the lower tail integral.
The derivative expansion in \eqref{eq:normal-IG-loss-deriv-expand} and
the strong law allow \(u_0>0\) to be chosen sufficiently small so that,
with probability tending to one,
\[
\partial_uL_J^u(u\mid\x)\geq{c_\infty}/{2},
\qquad 0\leq u\leq u_0.
\]
It follows by integration that
$
L_J^u(u\mid\x)-L_J^u(0\mid\x)
\geq{c_\infty u}/{2}$ for all 
$0\leq u\leq u_0$ so
\[
J^\alpha
\int_{\delta_J}^{u_0}
\exp\left[
-J\{L_J^u(u)-L_J^u(0)\}
\right]\rho_u(u)\,du
\leq
C\int_{J\delta_J}^{\infty}
e^{-c_\infty t/2}t^{\alpha-1}\,dt \rightarrow 0,
\]
using the change of variables $t=Ju$ and taking $J\delta_J\to\infty$. 

For the upper tail integral in \(\Delta_{\rm tail}(\deltaJ)\), we lower-bound $\inf_{u\ge u_0}\{L_J^u(u)-L_J^u(0)\}$. We are
assuming \(u=0\) is the unique global risk minimiser and we saw
above that $\mathcal R(r\mid\x)\to\infty$ as $r\downarrow r_0$, so $\mathcal R_u(u\mid\x)\to\infty$ as $u\to\infty$.
Continuity and uniqueness give
\[\textstyle
d_{u_0}
:=
\inf_{u\geq u_0}
\left\{
\mathcal R_u(u\mid\x)-\mathcal R_u(0\mid\x)
\right\}>0.
\]
Now fix
\(U>u_0\). The interval \(u\in[u_0,U]\) corresponds to a compact
\(r\)-interval.
Consequently, \eqref{eq:loss-to-risk-strong}, together with
\(L_J^u(0)=L_{J,\infty}\to R_\infty=\mathcal R_u(0\mid\x)\) from \eqref{eq:loss-r-inf-risk-strong}, gives
\begin{align*}
\sup_{u_0\leq u\leq U}
\bigg|
\{L_J^u(u\mid\x)-L_J^u(0\mid\x)\}
-
\{\mathcal R_u(u\mid\x)-\mathcal R_u(0\mid\x)\}
\bigg|
\xrightarrow{a.s.} 0,
\end{align*}
so with probability tending to one,
$
\inf_{u_0\leq u\leq U}
\left\{
L_J^u(u\mid\x)-L_J^u(0\mid\x)
\right\}
\geq {d_{u_0}}/{2}
$.

It remains to control \(u\geq U\), equivalently
\(r_0<r\leq r_0+U^{-1}\). The Student-\(t\) density
\(p_r(y\mid\x)\) is maximised at \(y=\bar x\), so
$L_J(r\mid\x)
\geq -\log p_r(\bar x\mid\x)$ and we showed above that
$
-\log p_r(\bar x\mid\x)\to\infty$ as $r\downarrow r_0
$.
We may therefore choose \(U\) sufficiently large that
\[
\inf_{u\geq U}
\{-\log p_{r(u)}(\bar x\mid\x)\}
>
\mathcal R_u(0\mid\x)+d_{u_0}.
\]
By \eqref{eq:loss-r-inf-risk-strong}, \(\Pr\{L_J^u(0)<\mathcal R_u(0\mid\x)+d_{u_0}/2\}\to 1\) so, with probability tending to one,
\begin{align*}
\inf_{u\geq U}
\left\{
L_J^u(u\mid\x)-L_J^u(0\mid\x)
\right\}
&\geq \inf_{u\geq U}
\left\{-\log p_{r(u)}(\bar x\mid\x)
-L_J^u(0\mid\x)
\right\} \\
&>  \inf_{u\geq U}
\left\{\mathcal R_u(0\mid\x)+d_{u_0}
-\mathcal R_u(0\mid\x)-d_{u_0}/2
\right\}={d_{u_0}}/{2}
\end{align*}
Combining the two regions $[u_0,M]$ and $[M,\infty)$ gives
\[
\begin{split}
J^\alpha
\int_{u_0}^{\infty}
\exp\left[
-J\{L_J^u(u\mid\x)-L_J^u(0\mid\x)\}
\right]\rho_u(u)\,du
&\leq
J^\alpha e^{-Jd_{u_0}/2}
\rho_u([u_0,\infty))
\\
&\longrightarrow0.
\end{split}
\]
This completes the tail verification and shows that all conditions of Theorem~\ref{thm:post.boundary} hold given the one remaining unverified assumption that $r=\infty$ is the unique risk minimiser. If that holds then $r/J\xrightarrow{d} \operatorname{IG}\!\left(
\alpha,-\bar A_{\infty,1}\right)$ for $r\sim\rho_{(J,1)}(\cdot|\y,\x)$. 

\end{document}